\documentclass[a4paper,12pt]{report}

\RequirePackage[T2A]{fontenc}
\RequirePackage[utf8]{inputenc}
\RequirePackage[english,russian]{babel}
\RequirePackage{euscript}
\RequirePackage{microtype}

\RequirePackage{amssymb,amsmath,amsthm}
\RequirePackage{textcomp}

\renewcommand{\baselinestretch}{1.3}    
\PassOptionsToPackage{lowtilde}{url}
\RequirePackage[%
    colorlinks=true, %
    pdfstartview=FitH,%
    linkcolor=blue,citecolor=blue,filecolor=blue,urlcolor=blue,linktoc=page]{hyperref}
\hypersetup{%
    pdftitle={Диссертационная работа на соискание учёной степени кандидата физико-математических наук по специальности
        01.01.09 <<Дискретная математика и математическая кибернетика>>},
    pdfauthor={Закаблуков Дмитрий Владимирович},
    pdfkeywords={обратимые схемы, схемная сложность, вычисления с ограниченной памятью},
    pdfsubject={Методы синтеза обратимых схем из функциональных элементов NOT, CNOT и 2-CNOT}}

\RequirePackage{bookmark}
\bookmarksetup{
  open,
  numbered,
  addtohook={
    \ifnum\bookmarkget{level}=1 %
      \bookmarksetup{color=blue}%
    \fi
  },
}

\newlength{\tocleftmargin}
\RequirePackage{geometry}
\RequirePackage{fancyhdr}
\fancypagestyle{plain}{%
    \fancyhf{} 
    \fancyhead[C]{\thepage} 
}
\newcommand{\forceindent}{{\noindent \hbox to \parindent{}}}

\renewcommand{\thesection}{\arabic{section}}

\RequirePackage{xifthen}

\newcommand{\sectionenumerated}[2][]
{
    \refstepcounter{section}
    \addcontentsline{toc}{section}{\thesection\hspace{\tocleftmargin}#2}
    \section*{\thesection\ifthenelse{\isempty{#1}}{}{\label{#1}}\quad#2}
}

\newcounter{sectionnonum}
\newcommand{\sectionnoenumeration}[1]
{
    \refstepcounter{sectionnonum}
    \section*{#1}
    \addcontentsline{toc}{section}{#1}
}

\newcounter{subsectionnonum}
\newcommand{\subsectionnoenumeration}[1]
{
    \refstepcounter{subsectionnonum}
    \subsection*{#1}
    \addcontentsline{toc}{subsection}{#1}
}

\newcounter{myparagraph}[subsection]
\renewcommand{\themyparagraph}{\thesubsection.\arabic{myparagraph}}

\newcommand{\myparagraph}[2][]
{
    \refstepcounter{myparagraph}

    \subsubsection*{%
        \texorpdfstring%
            {\themyparagraph\ifthenelse{\isempty{#1}}{}{\label{#1}}\quad#2}%
            {\themyparagraph\hspace{12pt}#2}%
    }

    \addcontentsline{toc}{subsubsection}{\themyparagraph\hspace{\tocleftmargin}#2}
}

\newtheorem{theorem}{Теорема}[section]
\newtheorem{lemma}[theorem]{Лемма}
\newtheorem{predicate}[theorem]{Утверждение}
\newtheorem{define}[theorem]{Определение}
\newtheorem{corollary}[theorem]{Следствие}
\newtheorem{hypothesis}[theorem]{Гипотеза}
\newtheorem{replacement}{Замена}[section]

\newcommand{\compose}{\mathop{*}\limits}
\newcommand{\mcirc}{\mathop{\circ}\limits}
\newcommand{\smallwedge}{\mathop{\wedge}\limits}
\newcommand{\NN}{\mathbb N}
\newcommand{\ZZ}{\mathbb Z}
\newcommand{\FF}{\mathbb F}
\newcommand{\vv}{\mathbf}
\newcommand{\balpha}{\boldsymbol\alpha}
\newcommand{\frS}{\mathfrak S}
\newcommand{\fp}{\FF_2[x] \mathop / f(x)}
\newcommand{\fpm}{\FF_2^*[x] \mathop / f(x)}
\newcommand{\LC}{L^{(C)}}
\newcommand{\LT}{L^{(T)}}
\newcommand{\WC}{W^{(C)}}
\newcommand{\WT}{W^{(T)}}

\newcommand{\gate}{Ф.\,Э.}

\RequirePackage{longtable}
\RequirePackage[dvips]{graphicx,xcolor}
\RequirePackage{calc}
\RequirePackage{caption}
\DeclareCaptionLabelFormat{figure_format}{\textbf{\small Рис. #2}}
\DeclareCaptionLabelFormat{table_format}{\textbf{\small Таблица #2}}
\DeclareCaptionLabelSeparator{null_sep}{ }
\RequirePackage{array}
\RequirePackage{multirow}

\RequirePackage[square,numbers,sort&compress]{natbib}

\makeatletter
    \@addtoreset{subsection}{section}
    \@addtoreset{myparagraph}{subsection}
    \@addtoreset{figure}{section}
    \@addtoreset{table}{section}
    \@addtoreset{equation}{section}
\makeatother

\newcounter{synthalgchapter}
\newcounter{synthalg}[synthalgchapter]

\renewcommand{\thesynthalg}{\arabic{synthalgchapter}.\arabic{synthalg}}
\newcommand{\nextalg}[1]{\refstepcounter{synthalg}\textbf{A\thesynthalg\label{#1}}}
\newcommand{\algref}[1]{A\ref{#1}}

\newcounter{optcounter}
\newcommand{\nextopt}{\refstepcounter{optcounter}\textit{O\arabic{optcounter}}}

\RequirePackage{totcount}
\newcounter{fakecounter} 
\newtotcounter{tablecount}
\newtotcounter{figurecount}
\newtotcounter{citecount}

\newcommand{\Table}{\refstepcounter{tablecount}\begin{table}}
\newcommand{\Longtable}{\refstepcounter{tablecount}\begin{longtable}}
\newcommand{\Figure}{\refstepcounter{figurecount}\begin{figure}}
\newcommand{\Bibitem}[1]{\refstepcounter{citecount}\bibitem{#1}}

\begin{document}

\clearpage
\begin{titlepage}
    \begin{center}
        \renewcommand{\baselinestretch}{1.2}
        \selectfont
        
        Федеральное государственное бюджетное образовательное учреждение\\
        высшего профессионального образования\\
        <<Московский государственный технический университет имени Н.\,Э. Баумана>>\\
        (МГТУ им. Н.\,Э. Баумана)

        \bigskip
        Факультет: <<Информатика и системы управления>>\\
        Кафедра: <<Информационная безопасность>>
        
        \medskip
        \bigskip
        \begin{flushright}
            \large
            На правах рукописи\\
            \vspace{2cm}
            
        \end{flushright}
 
        \vspace{0.1cm}
        {
            \Large
            Закаблуков Дмитрий Владимирович
            
        }
       
        \bigskip
        {
            \LARGE \textbf{
                Методы синтеза обратимых схем\\
                из функциональных элементов\\
                NOT, CNOT и 2-CNOT
            }
            
        }
        
        \bigskip
        \bigskip
        {
            \large
            Специальность 01.01.09 --- дискретная математика\\
            и математическая кибернетика
            
        }
 
        \bigskip
        \bigskip
        \bigskip
        {
            \Large
            ДИССЕРТАЦИЯ\\
            на соискание учёной степени\\
            кандидата физико-математических наук
            
        }

        \bigskip
        \bigskip
        {
            \begin{flushright}
                \large
                Научный руководитель:\\
                кандидат физико-математических наук,\\
                доцент А.\,Е. Жуков
                
            \end{flushright}
        }
        
        \vspace*{\fill}

        {
            \large
            Москва --- 2018
            
        }

    \end{center}
\end{titlepage}
\clearpage
\setcounter{page}{2}
\section*{Обозначения}
\begin{longtable}{lcl}
    \gate{} & --- & функциональный элемент\\
    NOT     & --- & инвертор\\
    $N_t^n$ & --- & инвертор c $n$ входами, $t$~--- контролируемый выход\\
    $N_t$ & --- & инвертор, $t$~--- контролируемый выход\\
    CNOT    & --- & управляемый инвертор\\
    $C_{i;t}^n$ & --- & управляемый инвертор с $n$ входами, $t$~--- контролируемый выход,\\
            && $i$~--- контролирующий вход\\
    $C_{i;t}$ & --- & управляемый инвертор, $t$~--- контролируемый выход,\\
            && $i$~--- контролирующий вход\\
    2-CNOT  & --- & элемент Тоффоли\\
    $C_{i_1,i_2;t}^n$ & --- & элемент Тоффоли с $n$ входами, $t$~--- контролируемый выход,\\
            && $i_1$, $i_2$~--- контролирующие входы\\
    $C_{i_1,i_2;t}$ & --- & элемент Тоффоли, $t$~--- контролируемый выход,\\
            && $i_1$, $i_2$~--- контролирующие входы\\
    $k$-CNOT & --- & обобщённый элемент Тоффоли с $k$ контролирующими входами\\
    $C_{I;t}^n$ & --- & обобщённый элемент Тоффоли с $n$ входами, $t$~--- контролируемый\\
            && выход, $I$~--- множество контролирующих входов\\
    $C_{I;t}$ & --- & обобщённый элемент Тоффоли, $t$~--- контролируемый выход,\\
            && $I$~--- множество контролирующих входов\\
    $E(t)$ & --- & инвертор, $t$~--- контролируемый выход\\
    $E(t,I)$ & --- & обобщённый элемент Тоффоли, $t$~--- контролируемый выход,\\
               && $I$~--- множество прямых контролирующих входов\\
    $E(t,I,J)$ & --- & обобщённый элемент Тоффоли, $t$~--- контролируемый выход,\\
               && $I$~--- множество прямых контролирующих входов,\\
               && $J$~--- множество инвертированных контролирующих входов\\
    $\Omega_n^2$ & --- & множество всех элементов NOT, CNOT и 2-CNOT с $n$ входами\\
    $\Omega_n$ & --- & множество всех элементов NOT и $k$-CNOT с $n$ входами\\
    $\Omega_*^2$ & --- & множество всех элементов NOT, CNOT и 2-CNOT\\
    $\ZZ_2^n$ & --- & множество двоичных векторов длины $n$\\
    $S(M)$ & --- & симметрическая группа подстановок на множестве $M$\\
    $A(M)$ & --- & знакопеременная группа подстановок на множестве $M$\\
    $\frS$ & --- & обратимая схема\\
    $L(\frS)$ & --- & сложность обратимой схемы $\frS$\\
    $\LC(\frS)$ & --- & количество элементов NOT и CNOT в обратимой схеме $\frS$\\
    $\LT(\frS)$ & --- & количество элементов 2-CNOT в обратимой схеме $\frS$\\
    $D(\frS)$ & --- & глубина обратимой схемы $\frS$\\
    $W(E)$ & --- & квантовый вес функционального элемента $E$\\
    $\WC$ & --- & квантовый вес элемента NOT/CNOT\\
    $\WT$ & --- & квантовый вес элемента 2-CNOT\\
    $W(\frS)$ & --- & квантовый вес обратимой схемы $\frS$\\
    $Q(\frS)$ & --- & количество дополнительных входов схемы $\frS$\\
    $P_2(n,n)$ & --- & множество всех булевых отображений $\ZZ_2^n \to \ZZ_2^n$\\
    $L(f,q)$ & --- & минимальная сложность обратимой схемы, реализующей булево\\
        && отображение $f$ с $q$ дополнительными входами\\
    $L(n,q)$ & --- & функция сложности по Шеннону обратимой схемы\\
        && с $q$ дополнительными входами\\
    $D(f,q)$ & --- & минимальная глубина обратимой схемы, реализующей булево\\
        && отображение $f$ с $q$ дополнительными входами\\
    $D(n,q)$ & --- & функция глубины по Шеннону обратимой схемы\\
        && с $q$ дополнительными входами\\
    $W(f,q)$ & --- & минимальный квантовый вес обратимой схемы, реализующей\\
        && булево отображение $f$ с $q$ дополнительными входами\\
    $W(n,q)$ & --- & функция квантового веса по Шеннону обратимой схемы\\
        && с $q$ дополнительными входами
\end{longtable}
\clearpage

\tableofcontents

\sectionnoenumeration{Введение}

\subsectionnoenumeration{История вопроса}

\forceindent
При проектировании цифровых микросхем всё чаще во главу угла ставится требование компактности микросхемы
и её низкого энергопотребления. Если тепловые потери во время вычислительного процесса будут относительно высокими,
производители не смогут выпускать новые масштабируемые решения для рынка, имеющие достаточно низкую рабочую температуру,
а пользователи будут вынуждены искать подходящий источник питания для вычислительного устройства.
При разработке мобильных систем ограниченное энергопотребление становится уже критически важным.

С одной стороны, проблема тепловых потерь во время вычислительного процесса связана с несовершенством современных
технологий производства цифровых микросхем и используемых для этого материалов.
Однако с течением времени все эти недостатки постепенно устраняются. Возникает вопрос: если полностью устранить все
технологические недостатки, то можно ли добиться нулевого уровня тепловых потерь во время вычислительного процесса?
Ответ на этот вопрос даёт фундаментальный физический принцип, предсказанный Дж. фон Нейманом в 1949 году%
~\cite{von_neyman} и сформулированный Р. Ландауэром в 1961 году~\cite{landauer}:
в любой вычислительной системе, независимо от её физической реализации, при потере 1 бита информации
выделяется минимум $kT \ln 2$ Дж тепла, где $k$~--- постоянная Больцмана, $T$~--- абсолютная температура, при которой
происходят вычисления.

Продолжительное время данный принцип оставался всего лишь чистой теорией, не подкреплённой
результатами экспериментов. Это было связано с трудностями измерения малых объёмов выделяемой энергии.
Однако в 2012 году учёным удалось провести эксперимент с коллоидной частицей, впервые подтвердивший принцип Р.~Ландауэра%
~\cite{berut_landauer_verification}. В 2014 был проведён ещё один эксперимент, показавший, что при уменьшении
возможных макроскопических состояний системы в 2 раза выделяется минимум $kT \ln 2$ Дж тепла~\cite{jun_landauer_verification}.
Это также подтверждает принцип Р.~Ландауэра, поскольку потерю 1 бита информации можно рассматривать
как уменьшение возможных состояний системы в 2 раза.
В других теоретических работах также было показано~\cite{vacarro_info_erasure}, что стирание 1 бита информации
невозможно без увеличения общей энтропии системы.

С математической точки зрения, уменьшение состояний системы в 2 раза можно рассматривать как результат работы
сюръективной функции: если для двух различных наборов входных значений $X_1 \ne X_2$ значение функции совпадает,
$f(X_1) = f(X_2) = Y$, то можно считать, что произошла потеря 1 бита информации о входном значении функции $f$.
Если к значению $Y$ добавить один бит и считать, что он равен 0,
если на вход функции $f$ был подан набор $X_1$, и равен 1, если на вход функции $f$ был подан набор $X_2$,
то потери информации не происходит.
Таким образом, вычисление значения булевой функции от $n$ переменных, равной константе, приводит к потере $n$ бит информации.
Мы тем самым пришли к тепловым потерям, как результату \textit{необратимости} вычислений: только биективная функция является
обратимой на всём множестве входных значений.

Однако насколько величина в $kT \ln 2$ Дж является существенной? Несложно посчитать, что при комнатной температуре она
будет равна $2,8 \cdot 10^{-21}$ Дж, что само по себе крайне мало. Тем не менее, в современных вычислительных устройствах
одновременно работают миллионы транзисторов. Если предположить, что на каждом такте работы каждый из этих транзисторов,
реализуя необратимое вычисление, теряет 1 бит информации и выделяет указанное количество энергии,
то общая величина тепловых потерь уже не будет столь незначительной.
К примеру, процессор Intel Core i7 содержит более 700 млн. транзисторов и работает на частоте выше 2 ГГц.
При пиковой нагрузке температура этого процессора может доходить до 80\textcelsius. Если все вычисления в этом процессоре
будут необратимы, мы получим примерно $4,7$ мВт выделяемой энергии. На фоне общего энергопотребления данная величина
продолжает выглядеть незначительной.

Современные технологии развиваются очень быстро. Эмпирический закон Мура~\cite{moore_law}, сформулированный им ещё в 1965 году,
гласит, что примерно каждые 2 года количество транзисторов на единицу площади удваивается. В последнее время скорость роста
плотности транзисторов несколько снизилась, но всё ещё остаётся экспоненциальной~\cite{evans_moore_law_validity}.
Таким образом, уже примерно к 2030 году размеры транзисторов при соблюдении закона Мура достигнут атомарного уровня.
Некоторые разрабатываемые технологии теоретически могут позволить достичь плотности размещения логических устройств в $10^{17}$
на кубический сантиметр~\cite{merkle_helical_logic}. Согласно принципу Р.~Ландауэра,
если все производимые вычисления будут необратимы, такое количество вычислительных устройств
при комнатной температуре во время работы на частоте в 10 ГГц должно выделять более $3 \cdot 10^6$ Вт.
В то же время, компьютер, содержащий в 1000 раз больше логических устройств
с такой плотностью, должен будет выделять $3 \cdot 10^9$ Вт из-за необратимости вычислений, имея при этом всё ещё
допустимые физические размеры (10 см\textsuperscript{3}). Отвод такого количества тепловой энергии представляет собой
неразрешимую технологическую проблему.

Если потеря информации во время вычислительного процесса приводит к тепловым потерям, то логично предположить,
что полностью обратимый процесс без потери информации должен снижать общий уровень выделяемой энергии.
Ч.~Беннет показал~\cite{bennet_reversibility}, что нулевой уровень тепловых потерь возможен только тогда,
когда все логические устройства схемы являются обратимыми, другими словами, когда они реализуют биективное отображение.
Отметим, что обратимость вычислений является \textit{необходимым}, но не \textit{достаточным} условием нулевого уровня
выделяемой энергии во время вычислительного процесса.
Краткая история обратимых вычислений может быть найдена в работе~\cite{bennet_history}.

Однако обратимость важна не только для снижения энергопотребления вычислительных устройств.
В некоторых случаях вычисления должны быть обратимыми в силу происходящих физических процессов.
Примером могут служить квантовые вычисления~\cite{nielsen_quantum,preskill_lecture},
представляющие особый интерес в связи с тем, что с их помощью некоторые экспоненциально сложные проблемы
могут быть решены за полиномиальное время~\cite{nielsen_quantum,halgren,shor,shmidt,boneh,tame,gavinsky,kedlaya}.
К примеру, для дискретного логарифмирования и факторизации чисел известен полиномиальный квантовый алгоритм Шора~\cite{shor}.
В 2013 году появился первый коммерческий 512-кубитный квантовый компьютер~\cite{dwave_computer}.
Такие успехи современных технологий позволяют надеяться, что квантовые алгоритмы найдут широкое применение в ближайшем будущем.
Схемы же из обратимых функциональных элементов являются строгой математической моделью физических процессов,
происходящих во время квантовых вычислений.

Существуют и другие способы создания обратимых логических устройств и элементов цифровых схем, помимо квантовых технологий.
Среди них можно выделить КМОП технологии~\cite{cmos_reversible} (в частности, адиабатическая~\cite{cmos_adiabatic}
и термодинамическая~\cite{thermodinamic_logic} обратимая логика), оптические технологии~\cite{optical_reversible_logic},
нанотехнологии~\cite{nanotechnology_reversible_logic} и технологии с использованием молекул ДНК~\cite{dna_reversible_logic}.
Таким образом, обратимые схемы являются не только математической абстракцией, но и реальными вычислительными устройствами.

Большинство функциональных элементов, рассматриваемых в отечественной литературе по синтезу управляющих систем,
не являются обратимыми. К примеру, конъюнктор и дизъюнктор реализуют необратимые отображения. Если быть совсем точным,
ни один функциональный элемент, реализующий булеву функцию более чем одной переменной, не является обратимым.
Это следует из того, что обратимые функциональные элементы должны реализовывать биективное булево отображение.
Единственным обратимым классическим функциональным элементом является инвертор (в англоязычной литературе обозначаемый как NOT).
За последние десятелетия было предложено несколько новых обратимых функциональных элементов, среди которых:
элемент Фейнмана~\cite{feynman}, именуемый также контролируемой инверсией (Controlled NOT, CNOT);
элемент Тоффоли~\cite{toffoli}, именуемый также контролируемой контролируемой инверсией (Controlled Controlled NOT, CCNOT
или 2-CNOT); элемент Фредкина~\cite{fredkin} и ряд других.

Интерес к функциональным элементам CNOT и 2-CNOT был обусловлен развитием теории квантовых вычислений.
В работе~\cite{barenco_elementary_gates} было показано, что каждому элементу NOT соответствует однокубитный квантовый вентиль,
элементу CNOT~--- двухкубитный квантовый вентиль, а элемент Тоффоли может быть реализован в виде композиции
5 двухкубитных квантовых вентилей.
В дальнейшем такие квантовые реализации стали определять \textit{квантовый вес} обратимых функциональных элементов:
количество одно- и двухкубитных квантовых вентилей, необходимых для их реализации. Появилась потребность в эффективных
алгоритмах синтеза обратимых схем с минимальным квантовым весом. В той же работе~\cite{barenco_elementary_gates}
было доказано, что элемент 2-CNOT является универсальным в том плане, что с его помощью можно реализовать любую булеву
функцию от $n$ переменных. Однако полученная схема в некоторых случаях будет содержать больше, чем $n$ входов.
Дальнейшие исследования показали~\cite{shende_synthesis,my_lemma_prove}, что с помощью элементов NOT, CNOT и 2-CNOT
можно реализовать любую чётную подстановку на множестве $\ZZ_2^n$ в обратимой схеме ровно с $n$ входами,
а если при этом использовать один дополнительный вход, то можно реализовать любую подстановку на множестве $\ZZ_2^n$.

Таким образом, как и в случае схемной сложности булевой функции, в качестве меры сложности подстановки на множестве $\ZZ_2^n$
можно рассматривать сложность реализующей её обратимой схемы. Задача же синтеза обратимой схемы может свестись к поиску
минимального представления элемента (подстановки) в системе образующих (множество подстановок, задаваемых обратимыми
функциональными элементами) соответствующей группы подстановок.
В работах~\cite{shende_synthesis,iwama_transform_rules,iterative_compositions,fast_synthesis_exact_minimal,
miller_spectral,miller_spectral_two_place,miller_transform_based,saeedi_novel,group_based,maslov_rm_synthesis} были предложены
различные алгоритмы синтеза обратимых схем, состоящих из функциональных элементов NOT, CNOT и 2-CNOT.
Более подробно эти алгоритмы будут рассмотрены во второй главе. Часть из них является переборными алгоритмами,
другие используют для синтеза либо теорию групп подстановок, либо изменение таблицы истинности для
входного булева преобразования. Однако стоит отметить, что для случая, когда заданная чётная подстановка на множестве $\ZZ_2^n$
имеет малое количество подвижных точек, не было предложено эффективных методов синтеза реализующей её обратимой схемы.
В этом случае существующие алгоритмы либо требуют значительного времени для синтеза, либо сложность полученной обратимой
схемы является слишком высокой по сравнению со схемами, синтезированными другими алгоритмами.

Открытым вопросом на текущий момент также является поиск эффективных алгоритмов снижения сложности обратимой схемы.
Практически во всех существующих алгоритмах синтеза описывается этап снижения сложности синтезируемой схемы%
~\cite{shende_synthesis,iwama_transform_rules,miller_transform_based,saeedi_novel}.
Во всех рассмотренных автором работах для этой цели используются заранее построенные таблицы эквивалентных замен композиций
функциональных элементов. Такие таблицы строятся либо по некоторому набору правил~\cite{iwama_transform_rules,
miller_transform_based,saeedi_novel},
либо поиском минимальных схем короткой длины полным перебором~\cite{shende_synthesis}.
Преимущество такого подхода заключается в экономии времени при снижении сложности обратимой схемы. Однако,
по мнению автора, один недостаток данного подхода не позволяет эффективно его использовать на практике:
таблицы эквивалентных замен обычно строятся для фиксированного значения числа входов обратимой схемы и при росте этого значения
начинают требовать значительного объёма памяти для своего хранения.

Ещё одним открытым вопросом является изучение зависимости сложности синтезируемой обратимой схемы от количества используемых
дополнительных входов в общем случае. Исторически сложилось, что почти все существующие работы по синтезу обратимых схем
ставят перед собой цель получить обратимую схему без дополнительных входов, а в остальных работах упор делается на
снижение количества дополнительных входов~\cite{maslov_thesis}. Это связано с тем, что в квантовых вычислениях,
в отличие от классических, технологически сложно и дорого добавлять дополнительные входы в схему.
С другой стороны, эффект снижения сложности и глубины обратимых схем за счёт использования дополнительных входов
уже известен~\cite{miller_reducing_complexity,reducing_depth}:
к примеру, обобщённый элемент Тоффоли~\cite{toffoli} с $k$ контролирующими входами
может быть представлен в виде композиции либо $8(k-3)$ элементов 2-CNOT без использования дополнительных входов,
либо $(k-1)$ элементов 2-CNOT, но с использованием $(k-2)$ дополнительных входов~\cite{barenco_elementary_gates}.
Таким образом, для произвольной чётной подстановки на множестве $\ZZ_2^n$ связь <<сложность-память>>
для реализующей её обратимой схемы до сих пор установлена не была.

Теория схемной сложности берёт своё начало с работы К.~Шеннона~\cite{shannon}, в которой он предложил
в качестве меры сложности булевой функции рассматривать сложность минимальной контактной схемы,
реализующей эту функцию. Им же было показано, что почти все булевы функции от $n$ переменных реализуются со
сложностью порядка $2^n \mathop / n$ в базисе функциональных элементов, соответствующих всем двуместным булевым функциям.
Асимптотически оптимальный метод синтеза схем в этом базисе функциональных элементов был разработан О.\,Б. Лупановым%
~\cite{lupanov_one_method,yablonsky}. Им также была установлена асимптотика функции Шеннона для сложности реализации
булевых функций во всех основных классах схем: класс формул~\cite{lupanov_formuls}, класс контактных~\cite{lupanov_contact}
и релейно-контактных~\cite{lupanov_rele} схем. Во всех этих работах не рассматривался вопрос связи <<сложность-память>>
для синтезируемых схем.

В классических необратимых схемах не запрещено ветвление входов и выходов функциональных элементов.
Однако в случае обратимых схем такие ветвления запрещены. В работах~\cite{lupanov_one_class,korshunov,shiganov_neighbor}
был рассмотрен вопрос синтеза схем с ограничением на количество соединений для одного функционального элемента.
Однако, как и в предыдущем случае, полученные оценки не учитывали количество использованных дополнительных входов схемы
(дополнительной памяти).

Вопрос о вычислениях с ограниченной памятью (ограниченным числом <<регистров>>/ячеек памяти)
рассматривался Н.\,А. Карповой в работе~\cite{karpova}.
Ею было доказано, что в базисе классических функциональных элементов, реализующих все \mbox{$p$-местные} булевы функции,
асимптотическая оценка функции Шеннона сложности схемы с тремя и более регистрами памяти зависит от значения $p$,
но не изменяется при увеличении количества используемых регистров памяти. Также было показано, что существует булева функция,
которая не может быть реализована в маломестных базисах с использованием менее, чем двух регистров памяти.
В случае же базиса из всех элементов NOT, CNOT и 2-CNOT уже конъюнкция трёх переменных не может быть реализована
с использованием менее, чем 5-ти регистров памяти: 3 регистра для хранения значения входных переменных, 1 регистр для
промежуточного и 1 для итогового результата.
В работе~\cite{konovodov_limited_width} было показано, что для некоторого класса булевых функций ослабление ограничения
на ширину схемы с 2 до 3 позволяет снизить асимптотическую сложность схемы.
В работе~\cite{nikitin_conveer} было показано, что для некоторых классов конвейерных схем увеличение
количества допустимой памяти приводит к асимптотическому снижению сложности схемы.

О.\,Б. Лупановым также были рассмотрены схемы из функциональных элементов с задержками~\cite{lupanov_delay}. Было доказано, 
что в регулярном базисе функциональных элементов любая булева функция может быть реализована схемой,
имеющей задержку $T(n) \sim \tau n$, где $\tau$~--- минимум приведённых задержек всех элементов базиса, при сохранении
асимптотически наилучшей сложности. Однако не рассматривался вопрос зависимости $T(n)$
от количества используемых регистров памяти.
В работах~\cite{hrapchenko_depth_and_delay,hrapchenko_difference} было доказано, что если различать задержку и глубину схемы,
то даже в минимальной схеме задержка может быть почти в 2 раза меньше глубины (при одинаковых единицах измерения).
В работе~\cite{hrapchenko_new} было доказано, что в некоторых случаях эти величины могут различаться на порядок.
Вопрос асимптотической глубины в различных управляющих системах был рассмотрен в работах%
~\cite{lozhkin_formula_depth,lozhkin_danilov}.
Тем не менее, связь <<глубина-память>> или <<задержка-память>> в данных работах также не рассматривалась.
В работах~\cite{lupanov_delay,lozhkin_formula_depth} было показано, что в некоторых классах управляющих систем
удаётся построить схему, сложность и задержка/глубина которой не превосходят асимптотически наилучшие оценки.
В случае обратимых схем такого результата добиться, по-видимому, не удастся, т.\,к. в обратимых схемах запрещено
ветвление входов и выходов функциональных элементов. Но никаких конкретных доказательств данного утверждения
до настоящего времени получено не было.

В настоящее время одним из основных направлений научных работ в отечественной литературе, связанных со сложностью
управляющих систем из некоторых классов, является получение асимптотических оценок функции Шеннона высокой степени точности%
~\cite{lozhkin_disser,lozhkin_high_precise,shupletsov_predicate,shiganov_neighbor}.
В работе~\cite{maslov_thesis} были получены асимптотические верхние и нижние оценки сложности обратимых схем, состоящих из
обобщённых элементов Тоффоли с прямыми и инвертированными контролирующими входами, а также был предложен
асимптотически оптимальный метод синтеза обратимых схем из элементов mEXOR (терминология автора работы~\cite{maslov_thesis}).
Тем не менее, в данной работе не рассматривался вопрос асимптотической глубины,
а также не было выявлено связей <<сложность-память>> и <<глубина-память>> для обратимых схем, в том числе, состоящих
из функциональных элементов NOT, CNOT и 2-CNOT.
В работе~\cite{vinokurov} была доказана нижняя оценка сложности обратимых схем, состоящих из обобщённых элеметов Тоффоли
и не имеющей дополнительных входов.
В работе~\cite{shende_synthesis} была доказана нижняя асимптотическая оценка сложности обратимой схемы,
состоящей из функциональных элементов NOT, CNOT и 2-CNOT и не имеющей дополнительных входов.
В работе~\cite{maslov_rm_synthesis} была доказана наилучшая известная на сегодняшний день верхняя асимптотическая оценка
сложности обратимой схемы, состоящей из функциональных элементов NOT, CNOT и 2-CNOT и не имеющей дополнительных входов.

Ещё одним многообещающим направлением исследований является изучение однонаправленности (one-wayness) преобразований через
построение реализующих их схем~\cite{zhukov_asym,interlando}. В обратимой схеме при использовании дополнительных входов
возможно появление так называемого <<вычислительного мусора>> на выходах: ненулевых значений, не являющихся частью результата.
В работе~\cite{kitaev_vyaliy} было показано, что любое биективное отображение можно реализовать обратимой схемой
без порождения вычислительного мусора. Впоследствии был предложен подход по изучению асимметричных преобразований через
построение реализующих их обратимых схем~\cite{zhukov_model} и было сделано предположение,
что сложность прямого и обратного преобразований определяется сложностью подсхем по уборке
вычислительного мусора для этих преобразований~\cite{zhukov_asym}.
Если построить подсхему по уборке вычислительного мусора (обнулить значения на соответствующих выходах),
можно получить обратимую схему, в которой проявляется структура как прямого, так и обратного преобразования:
если зеркально отобразить данную схему, получится схема, реализующая обратное преобразование.

На сегодняшний день были получены обратимые схемы для следующих асимметричных преобразований:
\begin{enumerate}
    \item
        Линейные преобразования и нелинейные подстановки из работ~\cite{boppana_one_way,hiltgen}:
        обратимые схемы, реализующие эти преобразования без порождения вычислительного мусора на выходах,
        описаны в работах~\cite{zhukov_asym,zakablukov_zasorina_chikin}.
        Асимптотическая сложность этих схем согласуется с полученными теоретическими данными о сложности соответствующих
        преобразований.
        
    \item
        Двоичное сложение и вычитание. Обратимая схема, реализующая двоичный сумматор без порождения вычислительного мусора
        на выходах, описана в работе~\cite{zakablukov_zasorina_chikin}.
        Асимптотическая сложность этой схемы согласуется с полученными теоретическими данными о сложности соответствующих
        преобразований, полученных в работе~\cite{redkin_summator}.
        
    \item
        Умножение и деление в кольце многочленов и конечном поле характеристики 2.
        Обратимые схемы, реализующие эти преобразования, описаны в работе~\cite{zakablukov_zasorina_chikin}.
        В данном случае асимптотическая сложность этих схем оказалась
        выше, чем полученные теоретические данные о сложности этих преобразований~\cite{sergeev_log_depth}.
\end{enumerate}
Таким образом, остаётся открытым вопрос об эффективной реализации асимметричных преобразований обратимыми схемами
без порождения вычислительного мусора. Следующим таким асимметричным преобразованием, представляющим наибольший интерес,
по мнению автора, является алгоритм дискретного логарифмирования (неквантовый), для которого не удалось найти какие-либо
опубликованные результаты по его реализации в обратимых схемах.
Однако стоит отметить, что на сегодняшний день известны квантовые алгоритмы дискретного логарифмирования
(к примеру, алгоритм Шора) с полиномиальной временн\'{о}й сложностью~\cite{shor}.
Данные алгоритмы могут быть естественным образом реализованы обратимой схемой.

\pagebreak
\subsectionnoenumeration{Цели и задачи работы}

\forceindent
Целью работы является изучение обратимых схем из функциональных элементов NOT, CNOT и 2-CNOT, разработка
новых методов синтеза таких схем и изучение зависимости их сложности и глубины от количества используемых
дополнительных входов схемы.
В работе используются методы теории синтеза управляющих систем, методы теории групп подстановок,
мощностные методы установления нижних оценок.

Работа носит не только теоретический, но и практический характер. Предложенные методы синтеза и способы снижения
сложности обратимых схем были реализованы в программном обеспечении~\cite{my_program}
по синтезу обратимых схем без дополнительных входов.
Данное программное обеспечение, по мнению автора, может быть применено в будущем при решении задач синтеза квантовых схем
малой сложности. С другой стороны, разработанные методы снижения сложности обратимых схем позволяют
изучать структуру подстановок на множестве двоичных векторов при помощи изучения структуры реализующих их обратимых схем.

Все полученные в диссертации результаты являются новыми. В настоящей работе впервые систематически изучается вопрос
синтеза схем из обратимых функциональных элементов при различном количестве используемых в схеме дополнительных входов
(дополнительной памяти). Разработан новый быстрый алгоритм синтеза обратимой схемы,
реализующей заданную чётную подстановку с малым числом подвижных точек.
Предложены и систематизированы различные способы снижения сложности обратимых схем,
состоящих из обобщённых элементов Тоффоли. Получены асимптотические оценки сложности, глубины и квантового веса обратимых схем
и показано, что данные оценки существенно зависят от количества используемых дополнительных входов схемы.
Разработан асимптотически оптимальный метод синтеза обратимых схем без дополнительных входов.
Предложены различные способы синтеза обратимых схем, реализующих алгоритм дискретного логарифмирования
в конечном поле характеристики 2.

Основные результаты диссертации опубликованы автором в работах~\cite{zakablukov_zasorina_chikin,my_thesis_dm9,
my_fast_group_based_algorithm,my_complexity_bounds,my_lemma_prove,my_thesis_circuit_synthesis,
my_thesis_bit_bmstu,my_complexity_reduction,my_program,my_jcss_complexity_no_memory,my_asymp_depth,my_discret_matem_complexity,
my_lncs_application_of_group_theory},
из которых статьи~\cite{my_fast_group_based_algorithm,my_complexity_reduction,my_complexity_bounds,
my_asymp_depth,my_discret_matem_complexity} ---
в рецензируемых научных изданиях из перечня ВАК.
Результаты диссертации докладывались и обсуждались на спецсеминаре кафедры математической кибернетики факультета ВМК МГУ,
на семинаре отдела <<Интеллектуальных систем>> ФИЦ ИУ РАН и на следующих конференциях:
\begin{enumerate}
    \item
        ХX Всероссийская научно-практическая конференция
        <<Проблемы информационной безопасности в системе высшей школы>> (Москва, МИФИ, февраль 2013).
        
    \item
        V Международная конференция <<Безопасные информационные технологии - 2014>> (Москва, МГТУ им. Баумана, ноябрь 2014).
        
    \item
        9-я Международная конференция <<Дискретные модели в теории управляющих систем>> (Москва и Подмосковье, МГУ, май 2015).

    \item
        8\textsuperscript{th} Conference on Reversible Computation (RC 2016) (Италия, Болонья, июль 2016).
\end{enumerate}

\subsectionnoenumeration{Краткое содержание работы}

\forceindent
Диссертация состоит из введения, 5 глав, заключения и списка литературы.
Текст диссертации изложен на~\pageref{page_last} странице, содержит~\total{figurecount} иллюстрации и~\total{tablecount} таблиц.
Список литературы включает~\total{citecount} наименований.

В первой главе даются базовые определения обратимых функциональных элементов NOT, CNOT и 2-CNOT,
обобщённого элемента $k$-CNOT, а также обратимых схем, состоящих из этих элементов. Вводится множество
$\Omega_n^2$, состоящее из всех элементов NOT, CNOT и 2-CNOT с $n$ входами. Доказывается, что, во-первых,
каждый \gate{} из множества $\Omega_n^2$ задаёт некоторую подстановку на множестве $\ZZ_2^n$, а во-вторых,
что множество подстановок, задаваемых всеми \gate{} из множества $\Omega_n^2$, при $n < 4$ генерирует симметрическую
группу подстановок $S(\ZZ_2^n)$, а при $n \geqslant 4$~--- знакопеременную группу подстановок $A(\ZZ_2^n)$.
Также в первой главе вводится понятие обратимой схемы, реализующей заданное булево отображение с использованием и
без использования дополнительной памяти (дополнительных входов). Доказывается, что для реализации нечётной подстановки
из $S(\ZZ_2^n)$ при $n \geqslant 4$ в обратимой схеме, состоящей из \gate{} множества $\Omega_n^2$, требуется
как минимум один дополнительный вход.
Вводится понятие значимых входов и выходов обратимой схемы, а также понятие \textit{вычислительного мусора} на незначимых выходах.
Показывается связь между схемной сложностью реализации прямого и обратного отображения через сложность обратимой схемы,
реализующей прямое отображение.

Во второй главе рассматриваются различные существующие алгоритмы синтеза обратимых схем:
переборные алгоритмы, непереборные быстрые алгоритмы и алгоритмы снижения сложности обратимой схемы.
Приводится сравнение данных алгоритмов по основным параметрам: время синтеза, количество требуемой для синтеза памяти,
сложность синтезированной обратимой схемы.
На основании данного сравнения делается вывод, что до текущего момента не было разработано быстрых и эффективных
алгоритмов синтеза обратимой схемы, реализующей подстановку с малым числом подвижных точек. Существующие же алгоритмы
синтеза в этом случае либо требуют значительного времени для своей работы, либо синтезируемая схема имеет избыточную сложность.
Даётся описание двух новых быстрых алгоритмов синтеза, использующих теорию групп подстановок. Принцип работы этих алгоритмов
основывается на доказательстве леммы из первой главы о том, что множество подстановок, задаваемых всеми \gate{}
из множества $\Omega_n^2$, при $n \geqslant 4$ генерирует знакопеременную группу подстановок $A(\ZZ_2^n)$.
Доказывается, что наилучший из разработанных алгоритмов синтеза позволяет получить обратимую схему со сложностью,
асимптотически не превышающей $7n2^m$ для любой чётной подстановки из $A(\ZZ_2^n)$, у которой не более $2^m$ подвижных точек.
В конце второй главы сравниваются существующие и предлагаемые быстрые алгоритмы синтеза, использующие
теорию групп подстановок. На основании этого сравнения делается вывод об эффективности новых алгоритмов
с точки зрения их быстродействия и сложности синтезируемых обратимых схем.

В третьей главе рассматриваются различные способы снижения сложности обратимых схем. Доказывается необходимое и
достаточное условие коммутируемости двух обратимых элементов $E(t_1, I_1, J_1)$ и $E(t_2, I_2, J_2)$:
либо $t_1 \notin I_2 \cup J_2$ и $t_2 \notin I_1 \cup J_1$, либо $(I_1 \cap J_2) \cup (I_2 \cap J_1) \ne \varnothing$.
Предлагаются также различные эквивалентные замены композиций обратимых \gate{}
с доказательством корректности таких замен при помощи операций на множествах. Часть замен позволяет снизить сложность
обратимой схемы, состоящей из элементов $E(t, I, J)$, оставшиеся же замены позволяют получить композицию новых \gate{}
Описывается алгоритм снижения сложности обратимых схем, состоящих из элементов $E(t, I, J)$, использующий описанные
эквивалентные замены композиций \gate{} Даётся оценка снизу временн\'{о}й сложности данного алгоритма.
Далее рассматриваются различные способы снижения сложности обратимой схемы на этапе её синтеза. Первый способ:
поиск грани булева куба $\mathbb B^n$, такой что для найденной грани размерности $k$ в синтезируемой схеме
можно заменить композицию порядка $2^{n-k-2}$ подряд идущих \gate{} на композицию не более $n$ элементов $E(t, I, J)$.
Второй способ: эффективное разбиение циклов в представлении исходной подстановки в виде произведения независимых циклов.
Описывается алгоритм быстрого поиска такого разбиения. Последний способ: рассмотрение произведения справа и слева
в представлении исходной подстановки в виде произведения транспозиций. Показывается эффективность предложенных
способов снижения сложности обратимой схемы на практике при помощи разработанного программного обеспечения.

В четвёртой главе рассматривается вопрос асимптотической сложности и глубины обратимых схем, состоящих из \gate{}
множества $\Omega_n^2$ и реализующих некоторое отображение $\ZZ_2^n \to \ZZ_2^n$.
Вводится множество $F(n,q)$ всех отображений $\ZZ_2^n \to \ZZ_2^n$, которые могут быть реализованы такими
обратимыми схемами с $(n+q)$ входами. Рассматриваются обратимые схемы, реализующие отображение $f \in F(n,q)$ 
с использованием $q$ дополнительных входов (дополнительной памяти). Вводятся функции Шеннона сложности $L(n,q)$,
глубины $D(n,q)$ и квантового веса $W(n,q)$ обратимой схемы как функции от $n$ и количества
дополнительных входов схемы $q$.
При помощи мощностного метода Риордана--Шеннона доказываются следующие нижние оценки:
\begin{align*}
    L(n,q) &\geqslant \frac{2^n(n-2)}{3\log_2(n+q)} - \frac{n}{3} \; , \\
    D(n,q) &\geqslant \frac{2^n(n-2)}{3(n+q)\log_2(n+q)} - \frac{n}{3(n+q)} \; , \\
    W(n,q) &\geqslant \min \left( \WC, \WT \right) \cdot \left( \frac{2^n(n-2)}{3\log_2(n+q)} - \frac{n}{3} \right) \; .
\end{align*}
Предлагается обобщение алгоритма синтеза обратимых схем, описанного во второй главе:
исходная подстановка из $A(\ZZ_2^n)$ представляется в виде
произведения не пар независимых транспозиций, а групп по $K$ независимых транспозиций в каждой группе.
Доказывается, что любая такая группа может быть задана композицией одного обобщённого элемента Тоффоли с большим количеством
контролирующих входов и множества элементов CNOT и 2-CNOT. Оценивается сложность и глубина обратимой схемы,
синтезируемой данным алгоритмом, откуда получаются следующие верхние оценки для обратимых схем без дополнительной памяти:
\begin{gather*}
    L(n,0) \leqslant \frac{3n2^{n+4}}{\psi(n)}
            \left( 1 + \varepsilon_L(n) \right) \; , \\
    D(n,0) \leqslant \frac{n2^{n+5}}{\psi(n)}
            \left( 1 + \varepsilon_D(n) \right) \; , \\
    W(n,0) \leqslant \frac{n2^{n+4} \left( \WC(1 + \varepsilon_C(n)) + 2\WT(1 + \varepsilon_T(n)) \right)}{\psi(n)} \; ,
\end{gather*}
где $\psi(n) = \log_2 n - \log_2 \log_2 n - \log_2 \phi(n)$,
$\phi(n)$~--- любая сколь угодно медленно растущая функция такая,
что $\phi(n) < n \mathop / \log_2 n$,
\begin{align*}
    \varepsilon_L(n) &= \frac{1}{6\phi(n)} +\left(\frac{8}{3} - o(1)\right)
            \frac{\log_2 n \cdot \log_2 \log_2 n}{n} \; , \\
    \varepsilon_D(n) &= \frac{1}{4\phi(n)} +(4 + o(1))\frac{\log_2 n \cdot \log_2 \log_2 n}{n} \; , \\
    \varepsilon_C(n) &= \frac{1}{2\phi(n)} - \left( \frac{1}{2} - o(1) \right) \cdot \frac{\log_2 \log_2 n }{n} \; , \\
    \varepsilon_T(n) &= (4 - o(1))\frac{\log_2 n \cdot \log_2 \log_2 n}{n} \; .
\end{align*}

Далее описывается асимптотически оптимальный метод синтеза обратимых схем с дополнительной памятью, аналогичный методу
Лупанова для классических схем. При описании данного метода подсчитывается количество используемых дополнительных
входов обратимой схемы (дополнительной памяти) при достижении минимальной сложности и минимальной глубины схемы.
Для данного метода удалось получить верхнюю оценку сложности обратимых схем с дополнительной памятью
$$
    L(n, q_0) \lesssim 2^n \text{ \,при\, } q_0 \sim n 2^{n-\lceil n \mathop / \phi(n)\rceil} \; ,
$$
где $\phi(n)$ и $\psi(n)$~--- любые сколь угодно медленно растущие функции такие,
что $\phi(n) \leqslant n \mathop / (\log_2 n + \log_2 \psi(n))$.

Получены верхние оценки глубины обратимых схем с дополнительной памятью.
$$
    D(n,q_1) \lesssim 3n \text{ \,при\, } q_1 \sim 2^n \; ,
$$
при этом показано, что обратимая схема $\frS$, реализующая отображение $f \in F(n, q_1)$ с глубиной $D(\frS) \sim 3n$,
имеет сложность $L(\frS) \sim 2^{n+1}$
и квантовый вес $W(\frS) \sim \WC \cdot 2^{n+1} + \WT \cdot n2^{n - \lceil n \mathop / \phi(n)\rceil}$,
где $\phi(n)$ и $\psi(n)$~--- любые сколь угодно медленно растущие функции такие,
что $\phi(n) \leqslant n \mathop / (\log_2 n + \log_2 \psi(n))$.
$$
    D(n,q_2) \lesssim 2n \text{ \,при\, } q_2 \sim \phi(n)2^n  \; ,
$$
где $\phi(n)$~--- любая сколь угодно медленно растущая функция такая, что $\phi(n) = o(n)$;
при этом показано, что обратимая схема $\frS$, реализующая отображение $f \in F(n, q_2)$ с глубиной $D(\frS) \sim 2n$,
имеет сложность $L(\frS) \sim \phi(n)2^{n+1}$
и квантовый вес $W(\frS) \sim \WC \cdot \phi(n)2^{n+1} + \WT \cdot 2^{n - \lceil n \mathop / \phi(n)\rceil}$.

Получена верхняя оценка квантового веса обратимых схем с дополнительной памятью
$$
    W(n,q_0) \lesssim \WC \cdot 2^n + \WT \cdot n2^{n - \lceil n \mathop / \phi(n)\rceil},
        \text{ \,при\, } q_0 \sim n 2^{n-\lceil n \mathop / \phi(n)\rceil}  \;  ,
$$
где $\phi(n)$ и $\psi(n)$~--- любые сколь угодно медленно растущие функции такие,
что $\phi(n) \leqslant n \mathop / (\log_2 n + \log_2 \psi(n))$.

При сопоставлении верхних и нижних оценок для сложности обратимых схем получаются следующие соотношения:
\begin{gather*}
    L(n,0) \asymp n2^n \mathop / \log_2 n \; , \\
    L(n,q_0) \asymp 2^n \text{ \,при\, } q_0 \sim n 2^{n-\lceil n \mathop / \phi(n)\rceil} \; ,
\end{gather*}
где $\phi(n)$ и $\psi(n)$~--- любые сколь угодно медленно растущие функции такие,
что $\phi(n) \leqslant n \mathop / (\log_2 n + \log_2 \psi(n))$.

Получены общие верхние оценки сложности, глубины и квантового веса обратимых схем.
Для любого значения $q$ такого, что $8n < q \lesssim n 2^{n-\lceil n \mathop / \phi(n)\rceil}$,
где $\phi(n)$ и $\psi(n)$~--- любые сколь угодно медленно растущие функции такие,
что $\phi(n) \leqslant n \mathop / (\log_2 n + \log_2 \psi(n))$, верны соотношения
\begin{gather*}
    L(n,q) \lesssim 2^n + \frac{8n2^n}{\log_2 (q-4n) - \log_2 n - 2} \; , \\
    D(n,q) \lesssim 2^{n+1}(2,5 + \log_2 n - \log_2 (\log_2 (q - 4n) - \log_2 n - 2))  \; , \\
    W(n,q) \lesssim \WT \cdot \left(2^n + \frac{8 \cdot 2^n}{\log_2 (q - 4n) - \log_2 n - 2} \right)
        + \frac{32 \WC n2^n}{\log_2 (q - 4n) - \log_2 n - 2} \; , \\
    W(n,q) \lesssim \WT \cdot \left(2^n + \frac{8n2^n}{\log_2 (q-4n) - \log_2 n - 2}\right)
        + \frac{32 \WC 2^n}{\log_2 (q-4n) - \log_2 n - 2} \; .
\end{gather*}

Получена оценка порядка роста функции $L(n,q)$.
Для любого значения $q$ такого, что $n^2 \lesssim q \lesssim 2^{n-\lceil n \mathop / \phi(n)\rceil + 1}$,
где $\phi(n)$ и $\psi(n)$~--- любые сколь угодно медленно растущие функции такие, что
$\phi(n) \leqslant n \mathop / (\log_2 n + \log_2 \psi(n))$, верно соотношение
$$
    L(n,q) \asymp \frac{n2^n}{\log_2 q} \; .
$$

Также получена оценка порядка роста функции Шеннона $L_A(n,q)$ сложности обратимых схем,
реализующих отображения $\ZZ_2^n \to \ZZ_2^n$ только из знакопеременной группы $A(\ZZ_2^n)$.
Для любого значения $q$ такого, что $0 \leqslant q \lesssim 2^{n-\lceil n \mathop / \phi(n)\rceil + 1}$,
где $\phi(n)$ и $\psi(n)$~--- любые сколь угодно медленно растущие функции такие, что
$\phi(n) \leqslant n \mathop / (\log_2 n + \log_2 \psi(n))$,
верно соотношение
$$
    L_A(n,q) \asymp \frac{n2^n}{\log_2 (n+q)} \; .
$$

На основании полученных асимптотических оценок делается вывод, что использование дополнительной памяти в обратимых схемах,
состоящих из \gate{} множества $\Omega_n^2$, почти всегда позволяет существенно снизить сложность, глубину
и квантовый вес таких схем, в отличие от классических схем, состоящих из необратимых \gate{}

В пятой главе показывается применение обратимых схем для решения задачи схемной реализации некоторых
вычислительно асимметричных преобразований. Подробно рассматривается алгоритм дискретного логарифмирования
по основанию примитивного элемента в конечном поле характеристики 2 на примере фактор-кольца $\fp$,
где $f(x)$~--- неприводимый многочлен степени $n$, и его реализация обратимой схемой.
При помощи разработанного программного обеспечения строятся обратимые схемы
без дополнительной памяти и с дополнительной памятью, реализующие алгоритм дискретного логарифмирования.
Показывается, что уже при использовании $n$ дополнительных входов сложность схемы существенно снижается.
Доказывается верхняя асимптотическая оценка сложности обратимой схемы $\frS$, реализующей алгоритм дискретного логарифмирования,
$$
    L(\frS_\mathrm{log}) \lesssim \frac{2^{n+1} \cdot \log_2 n}{n}
$$
при использовании $q \sim 2^{n-\lceil n \mathop / \phi(n)\rceil + 2} \cdot \log_2 n$ дополнительных входов,
где $\phi(n)$ и $\psi(n)$~--- любые сколь угодно медленно растущие функции такие,
что $\phi(n) \leqslant n \mathop / (\log_2 n + \log_2 \psi(n))$.
Данная оценка асимптотически ниже, чем для произвольного булева преобразования, и достигается при асимптотически
меньшем количестве дополнительных входов.

В конце пятой главы рассматривается вопрос схемной сложности реализации алгоритма, обратного к заданному,
и делается попытка объяснить разницу в схемной сложности для прямого и обратного алгоритмов через
необратимость и потерю части информации во время работы прямого алгоритма.
В подтверждение гипотезы приводятся примеры обратимых схем, реализующие такие вычислительно асимметричные преобразования,
как сложение в кольце многочленов, умножение и возведение в степень в конечном поле характеристики 2.

В заключении рассматриваются открытые вопросы и предлагаются различные направления дальнейших исследований:
\begin{itemize}
    \item
        улучшение нижней оценки для глубины обратимой схемы $D(n,q)$ при помощи более точного подсчёта количества схем
        заданной глубины;
    \item
        улучшение константы $3 \cdot 2^4$ в верхней оценке $L(n,0)$;
    \item
        разработка нового алгоритма синтеза обратимых схем без дополнительной памяти с лучшей верхней асимптотической оценкой
        глубины схемы;
    \item
        разработка нового алгоритма синтеза обратимых схем с дополнительной памятью с лучшей верхней асимптотической оценкой
        глубины схемы;
    \item
        изучение вопроса асимметричности преобразований через построение реализующих их обратимых схем.
\end{itemize}

Автор хотел бы выразить благодарность научному руководителю А.\,Е. Жукову за постановку задач и всестороннюю помощь,
а также родным за помощь и поддержку.

\sectionenumerated{Основы теории обратимой логики}

\forceindent В данной главе будут даны основные определения из теории обратимой логики. Будет показана связь подстановок с
двоичной обратимой логикой. Будет приведено доказательство того, что любая чётная подстановка из $A(\ZZ_2^n)$ при $n \geqslant 4$
и любая подстановка из $S(\ZZ_2^n)$ при $n < 4$, может быть реализована с помощью обратимых функциональных
элементов NOT, CNOT и 2-CNOT.

\subsection{Определение обратимого функционального элемента и схемы}

\forceindent
Здесь и далее будем рассматривать только схемы, состоящие из функциональных элементов.
\begin{define}\label{def_logic_gate}
    Функциональный элемент (\gate) $n \times m$~--- идеальная модель вычислительного устройства с $n$ входами и $m$ выходами,
    задающего на выходах результат некоторого булевого отображения $\ZZ_2^n \to \ZZ_2^m$ над входами.
\end{define}
\noindent Примером \gate{} может служить инвертор NOT,
задающий булево преобразование $f\colon \ZZ_2^1 \to \ZZ_2^1$ вида $f(\langle x \rangle) = \langle x \oplus 1 \rangle$.

Базовое определение обратимых \gate{} было введено в работах~\cite{toffoli,feynman}.
Будем использовать следующее формальное определение:
\begin{define}\label{def_reverse_logic_gate}
    Обратимый \gate{} $n \times n$ (далее просто обратимый \gate)~--- Ф.\,Э. $n \times n$, для которого
    задаваемое им булево отображение биективно.
\end{define}

Очевидно, что инвертор NOT является обратимым \gate{} Другим примером обратимого \gate{} является
обобщённый элемент Тоффоли с $k$ контролирующими входами.
В литературе для этого элемента принято стандартное обозначение $k$-CNOT~\cite{shende_synthesis}.
Определение обратимых \gate{} NOT и~$k$-CNOT, а также обратимых схем, состоящих из этих элементов,
было дано в нескольких работах, к примеру, в работах~\cite{maslov_thesis,shende_synthesis}.
Будем использовать следующее формальное определение элемента $k$-CNOT:
\begin{define}\label{def_k_not}
    Обобщённый элемент Тоффоли с $k$ контролирующими входами~--- обратимый \gate{} $(k+1) \times (k+1)$,
    задающий булево преобразование
    $f\colon \ZZ_2^{k+1} \to \ZZ_2^{k+1}$ следующего вида:
    $$
        f(\langle x_1, \ldots, x_k, x_{k+1} \rangle) =
        \langle x_1, \ldots, x_k, x_{k+1} \oplus x_1 \land \ldots \land x_k \rangle \; .
    $$
\end{define}

Частными случаями элемента $k$-CNOT являются управляемый инвертор CNOT (\mbox{1-CNOT}),
предложенный Р. Фейнманом~\cite{feynman}, и элемент Тоффоли (2-CNOT)~\cite{toffoli}.
Графическое обозначение этих элементов дано на рис.~\ref{pic_gates}.
Графическое обозначение элемента $k$-CNOT схоже с обозначением элемента Тоффоли:
контролирующие входы обозначаются символом {\large $\bullet$},
контролируемый выход~--- символом $\boldsymbol{\oplus}$.

\medskip
\Figure[ht]
    \centering
    \includegraphics[scale=1.2]{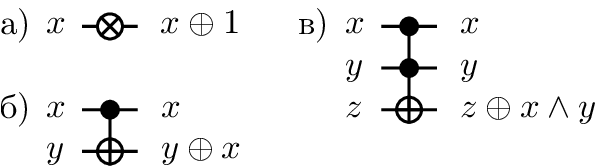}
    \caption
    {
        \small Графическое обозначение обратимых \gate: а) NOT (инвертор);\\
        б) 1-СNOT (управляемый инвертор); в) 2-CNOT (элемент Тоффоли).
    }\label{pic_gates}
\end{figure}

Можно расширить определение элементов NOT и $k$-CNOT таким образом, чтобы все \gate{} имели ровно $n$ входов и выходов.
Для описания таких \gate{} примем обозначение, предложенное в работе~\cite{group_based}:

\begin{define}\label{def_not_n}
    $N_j^n$~--- инвертор NOT с $n$ входами, инвертирующий свой $j$-й вход:
    $$
        N_j^n(\langle x_1, \ldots, x_j, \ldots, x_n \rangle) = \langle x_1, \ldots, x_j \oplus 1, \ldots x_n \rangle \; .
    $$
\end{define}

\begin{define}\label{def_k_not_n}
    $C_{i_1, i_2, \ldots, i_k;j}^n$~--- обобщённый элемент Тоффоли $k$-CNOT с $n$ входами,
    инвертирующий свой $j$-й вход тогда и только тогда, когда
    значение на всех контролирующих входах $i_1, \ldots, i_k$ равно 1, $j \ne i_1, \ldots, i_k$:
    $$
        C_{i_1, i_2, \ldots, i_k;j}^n(\langle x_1, \ldots, x_j, \ldots, x_n \rangle) =
        \langle x_1, \ldots, x_j \oplus x_{i_1} \wedge \ldots \wedge x_{i_k}, \ldots, x_n \rangle \; .
    $$
\end{define}
\noindent В случае, если значение $n$ понятно из контекста, верхний индекс в обозначении этих элементов будем опускать:
$N_j$ и $C_{i_1, i_2, \ldots, i_k;j}$.

В работе~\cite{saeedi_rule_based} было предложено обобщить представление элемента $k$-CNOT для случая нулевого значения
на некоторых контролирующих входах. Будем обозначать такой \gate{} через $C_{I;J;t}$.

\begin{define}\label{define_common_k_cnot} $C_{I;J;t}$~--- \gate, задающий булево преобразование
    $\ZZ_2^n \to \ZZ_2^n$ вида
    $$
        C_{I;J;t} \left( \langle x_1, \ldots, x_t, \ldots, x_n \rangle \right) =
            \left \langle x_1, \ldots,
            x_t \oplus \left( \smallwedge_{i \in I}{x_i} \right) \wedge \left( \smallwedge_{j \in J}{\bar x_j} \right), \ldots,
            x_n \right \rangle \; ,
    $$
    где $I$~--- множество прямых контролирующих входов, $J$~--- множество инвертированных контролирующих входов,
    $t \notin I \cup J$; $I \cap J = \varnothing$.
\end{define}
\noindent
Элементы $C_{I;J;t}$ будут подробнее рассмотрены в главе~\ref{section_complexity_minimization}.
    
Для удобства примем следующие обозначения: $E(t)$~--- элемент $N_t$; $E(t,I)$~--- элемент $C_{I;t}$;
$E(t,I,J)$~--- элемент $C_{I;J;t}$. Тогда можно считать, что $E(t,I) = E(t,I,\varnothing)$,
$E(t) = E(t, \varnothing, \varnothing)$.

Схема из \gate{} классически определяется как ориентированный граф без циклов с помеченными рёбрами и вершинами.
В обратимых схемах, состоящих из элементов NOT и $k$-CNOT, запрещено ветвление и произвольное подключение
входов и выходов \gate{}
Мы будем использовать следующее формальное определение обратимой схемы.
\begin{define}\label{def_reversible_circuit}
    Правильно сформированная обратимая схема $\frS$~--- ациклическая комбинационная логическая схема,
    в которой все \gate{} обратимы и соединены друг с другом последовательно без ветвлений.
\end{define}
\noindent
Мы будем рассматривать только те обратимые схемы, в которых все \gate{} имеют одинаковое количество входов и выходов $n$.
В ориентированном графе, описывающем такую обратимую схему,
все вершины, соответствующие \gate, имеют ровно $n$ занумерованных входов и выходов.
Эти вершины нумеруются от 1 до $l$, при этом $i$-й выход $m$-й вершины, $m < l$,
соединяется только с $i$-м входом $(m+1)$-й вершины. Входами обратимой схемы являются входы первой вершины,
а выходами~--- выходы $l$-й вершины.
Соединение \gate{} друг с другом будем также называть композицией элементов и обозначать $*$.
К примеру, запись $N_4 * C_{1;2} * C_{1,2,4;3} * C_{3;4}$ при $n=5$ соответствует схеме,
показанной на рис.~\ref{pic_scheme_example}.

\Figure[ht]
    \centering
    \includegraphics[scale=1.2]{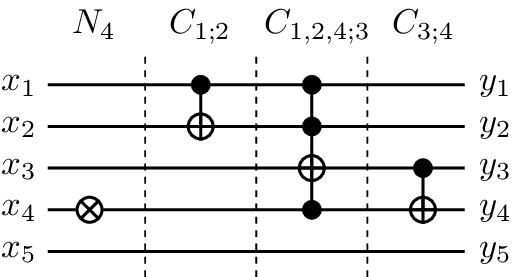}
    \caption
    {
        \small Cхема $\frS = N_4 * C_{1;2} * C_{1,2,4;3} * C_{3;4}$ с $n = 5$ входами.
    }\label{pic_scheme_example}
\end{figure}
\noindent
Стоит отметить, что любую обратимую схему, состоящую из двух и более \gate, можно рассматривать как композицию её подсхем.
К примеру, схему $\frS = C_{1;2} * C_{1,2,4;3} * C_{3;4}$ можно рассматривать как композицию схем
$\frS_1 = C_{1;2}$ и $\frS_2 = C_{1,2,4;3} * C_{3;4}$.

Н.\,А. Карповой в работе~\cite{karpova} были рассмотрены схемы с ограниченной памятью, состоящие из классических \gate{}
Каждому входу и выходу вершины графа, описывающего такую схему, приписывается некоторый символ
из множества $R = \{\,r_1, \ldots, r_n\,\}$. Символ $r_i$ можно интерпретировать как имя регистра памяти
(номер ячейки памяти), значение из которого поступает на вход \gate{} или в который записывается значение с выхода \gate{}
Данная модель функциональных схем с ограниченной памятью ($\mathbf {M_1}$) и модель обратимых схем ($\mathbf {M_2}$),
описанная нами выше, весьма похожи. Регистрам памяти модели $\mathbf {M_1}$ соответствуют
линии обратимой схемы модели $\mathbf {M_2}$, поскольку и те, и другие хранят результат вычислений на каждом шаге работы схемы.
Однако между данными моделями есть существенное отличие.

В модели $\mathbf {M_1}$ символ $r_i$, приписанный выходу некоторого \gate, может совпадать с символом,
приписанным одному из входов этого же \gate{} Другими словами, модель $\mathbf {M_1}$ позволяет
перезаписывать значения в регистрах памяти.

В модели $\mathbf {M_2}$ контролируемому выходу \gate{} приписывается номер линии, значение на которой будет инвертировано,
если значение булевой функции от значений на контролирующих входах элемента будет равно 1,
причём номер такой линии не может совпадать ни с одним из номеров линий, приписанных контролирующим входам элемента.
Другими словами, модель $\mathbf {M_2}$ не позволяет перезаписывать значения на линиях обратимой схемы,
а лишь инвертировать их в некоторых случаях.

Среди всех характеристик обратимой схемы для нас представляют интерес сложность, глубина и квантовый вес схемы.
Пусть обратимая схема $\frS$ с $n$ входами представляет собой композицию $l$ элементов
$\frS = \compose_{k=1}^l {E_k(t_k, I_k, J_k)}$.
\begin{define}\label{def_scheme_complexity}
    Сложность $L(\frS)$ обратимой схемы $\frS$~--- количество \gate{} в схеме.
\end{define}

Классически глубина схемы из \gate{} определяется как длина максимального пути на графе,
описывающем данную схему, между какими-либо входными и выходными вершинами. В рассматриваемой модели обратимой схемы граф,
описывающий такую схему, представляет собой просто одну цепочку последовательно соединённых вершин. Поэтому, если использовать
классическое определение глубины схемы, получится, что в нашем случае глубина обратимой схемы равна её сложности.

Для того чтобы не менять модель обратимой схемы, введём следующее определение глубины обратимой схемы.
Будем считать, что обратимая схема $\frS$ имеет глубину $D(\frS) = 1$,
если для любых двух её элементов $E_1(t_1,I_1,J_1)$ и $E_2(t_2,I_2,J_2)$ выполняется равенство
$$
    \left(\{\,t_1\,\} \cup I_1 \cup J_1 \right) \cap \left(\{\,t_2\,\} \cup I_2 \cup J_2 \right) = \varnothing \; .
$$
Также будем считать, что обратимая схема $\frS$ имеет глубину $D(\frS) \leqslant d$, если её можно разбить на
$d$ непересекающихся подсхем, каждая из которых имеет глубину 1:
\begin{equation}\label{formula_decomposition_of_scheme_for_depth}
    \frS = \bigsqcup_{i=1}^d{{\frS'}_i}, \text{ }D({\frS'}_i) = 1 \; .
\end{equation}
Тогда можно ввести следующее определение глубины обратимой схемы.
\begin{define}\label{def_scheme_depth_on_one_input}
    Глубина $D(\frS)$ обратимой схемы $\frS$~--- минимально возможное количество $d$
    непересекающихся подсхем глубины 1 в разбиении схемы $\frS$ по формуле~\eqref{formula_decomposition_of_scheme_for_depth}.
\end{define}
Используя это определение, можно вывести простое соотношение, связывающее сложность и глубину обратимой схемы $\frS$,
имеющей $n$ входов:
\begin{equation}\label{formula_complexity_and_depth_connection}
    \frac{L(\frS)}{n} \leqslant D(\frS) \leqslant L(\frS) \; . 
\end{equation}

На рис.~\ref{pic_scheme_example_depth} показан пример обратимой схемы
$\frS = C_{1;2} * C_{3;1} * N_2 * N_4 * C_{1,4;2} * N_3$, имеющей сложность $L(\frS) = 6$ и глубину $D(\frS) = 3$,
поскольку мы можем представить данную схему композицией трёх непересекающихся подсхем глубины 1:
$\frS = (C_{1;2}) * (C_{3;1} * N_2 * N_4) * (C_{1,4;2} * N_3)$, и не можем представить в виде композиции
двух непересекающихся подсхем глубины 1.
\Figure[ht]
    \centering
    \includegraphics[scale=1.2]{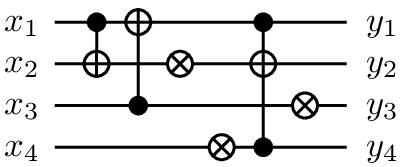}
    \caption
    {
        \small Обратимая схема $\frS = C_{1;2} * C_{3;1} * N_2 * N_4 * C_{1,4;2} * N_3$ сложности $L(\frS) = 6$
        и глубины $D(\frS) = 3$.
    }\label{pic_scheme_example_depth}
\end{figure}

В рассматриваемой модели обратимой схемы все \gate{} соединяются последовательно, поэтому максимальное время,
через которое на выходах схемы появится результат, определяется только её сложностью.
Однако если схему представлять в виде композиции подсхем, то время работы схемы может быть меньше
за счёт параллельной работы некоторых подсхем. В примере выше (см. рис.~\ref{pic_scheme_example_depth})
такими подсхемами могут быть $\frS_1 = C_{3;1}$, $\frS_2 = N_2$ и $\frS_3 = N_4$.
Именно поэтому глубина схемы рассматривается в настоящей работе наравне с её сложностью.

Если рассматривать обратимые \gate{} безотносительно к технологии, при помощи которых они могут быть получены
в реальной жизни, то можно считать, что все они имеют вес 1. Однако как показывает практика, в квантовых
технологиях, к примеру, реализовать элемент 2-CNOT намного сложнее, чем CNOT~\cite{barenco_elementary_gates}.
Поэтому будем считать, что обратимый элемент $E$ имеет вес $W(E)$, значение которого зависит от технологии
производства обратимых \gate{}

Наибольший интерес, по мнению автора, представляют квантовые вычисления, поэтому будем называть величину $W(E)$ для
обратимого элемента $E$ \textit{квантовым весом} этого элемента, равным количеству квантовых вентилей,
необходимых для его реализаци. Тогда можно ввести понятие квантового веса обратимой схемы.
\begin{define}
    Квантовый вес $W(\frS)$ обратимой схемы $\frS$~--- сумма квантовых весов всех её \gate{}
\end{define}
\noindent
Согласно работе~\cite{barenco_elementary_gates}, элементы NOT и CNOT имеют квантовый вес 1, а элемент 2-CNOT~--- квантовый вес 5.
Следовательно, чем меньше в обратимой схеме элементов $k$-CNOT с $k > 1$, тем проще её реализовать при помощи
квантовых технологий.\label{page_quantum_weight}

В настоящей работе не рассматривается вопрос реализации обратимой схемы при помощи различных квантовых вентилей
и возможного снижения количества этих вентилей в схеме: здесь и далее будем считать,
что квантовый вес $W(E)$ элемента $E$ фиксирован, задан заранее и не зависит от расположения данного элемента в схеме.
В этом случае, задавая различные значения $W(E)$, по значению величины $W(\frS)$ можно будет делать выводы
о количестве различных обратимых элементов в схеме $\frS$.

Используя определение квантового веса обратимой схемы, можно вывести простое соотношение,
связывающее сложность и квантовый вес обратимой схемы $\frS$:
\begin{equation}\label{formula_quantum_weight_connection_with_complexity_common_case}
    W(\frS) \geqslant \left( \min_{E \in \frS} {W(E)} \right) \cdot L(\frS) \; .
\end{equation}

\subsection{Связь обратимых схем с подстановками}

\forceindent
Несложно показать, что обратимые схемы непосредственно связаны с подстановками из симметрической группы $S(\ZZ_2^n)$.
\begin{define}\label{def_cycle}
    Пусть $S(A)$ является симметрической группой над множеством $A = \lbrace a_1, a_2, \ldots, a_m \rbrace$.
    Тогда $(a_{i_1}, a_{i_2}, \ldots, a_{i_k})$ называется циклом длины $k$, где $k \leqslant m$, $a_{i_j} \in A$,
    $a_{i_j} \ne a_{i_l}$.
\end{define}
\noindent
Цикл длины 2 также называется \textit{транспозицией}.

Любую подстановку можно представить в виде произведения независимых циклов. Каждый цикл в свою очередь можно представить
в виде произведения транспозиций.
\begin{define}\label{def_even_permutation}
    Подстановка является чётной, если она представима в виде произведения чётного числа транспозиций,
    и нечётной в противном случае.
\end{define}
\noindent
Операцию умножения (композиции) подстановок будем обозначать символом $\circ$.

Как было сказано выше, любой \gate{} NOT, CNOT и 2-CNOT задаёт биективное булево отображение
$f\colon \ZZ_2^n \to \ZZ_2^n$.
Рассмотрев все возможные вектора входных значений \gate{} и соответствующие им вектора выходных значений,
можно говорить о задаваемой этим элементом подстановке $h$ из симметрической группы $S(\ZZ_2^n)$.

Композиция \gate{} также задаёт подстановку, равную произведению подстановок, задаваемых этими элементами.
Рассмотрим схему $\frS_f = \compose_{i=1}^{m}{E_{f_i}}$, где элемент $E_{f_i}$ задаёт булево преобразование $f_i$
и соответствующую ему подстановку $h_i \in S(\ZZ_2^n)$.
Обозначим через $f$ задаваемое схемой $\frS_f$ булево отображение,
через $h$~--- задаваемую ей (если существует) подстановку.
Пусть $f_i\colon \ZZ_2^n \to \ZZ_2^n$, $1 \leqslant i \leqslant m$, тогда
\begin{gather}
    f(\vv x) = f_m(f_{m-1}( \ldots f_2(f_1(\vv x)) \ldots )), \vv x \in \ZZ_2^n \label{formula_scheme_f} \; , \\
    h = h_1 \circ h_2 \circ \ldots \circ h_{m-1} \circ h_m \label{formula_scheme_permutation} \; .
\end{gather}
(Здесь и далее используется соглашение, согласно которому произведение подстановок действует слева направо,
т.\,е. $(h \circ g)(x) = g(h(x))$.)

Следовательно, любая обратимая схема с $n$ входами задаёт подстановку из $S(\ZZ_2^n)$, и наоборот,
любая подстановка из $S(\ZZ_2^n)$ задаёт семейство обратимых схем с $n$ входами.
Обратимую схему с $n$ входами в других работах также называют $n$-битовым обратимым
логическим вентилем~\cite{group_based}. Всего существует $(2^n)!$ таких схем.

\begin{predicate}\label{predicate_reverse_scheme}
    Пусть обратимая схема $\frS_f = \compose_{i=1}^{m}{E_{f_i}}$ задаёт булево преобразование $f$,
    тогда композиция элементов $\compose_{i=m}^{1}{E_{f_i^{-1}}}$ задаёт булево преобразование $f^{-1}$.
\end{predicate}
\begin{proof}
Согласно формуле~\eqref{formula_scheme_f}
$
    f(\vv x) = f_m(f_{m-1}( \ldots (f_1(\vv x) \ldots )) = \vv y \; .
$
\\
$
    g(\vv y) = f_1^{-1}(f_2^{-1}( \ldots f_m^{-1}(\vv y) \ldots )) =
    f_1^{-1}(f_2^{-1}( \ldots f_m^{-1}(f_m(f_{m-1}( \ldots (f_1(\vv x) \ldots ))) \ldots )) = \vv x \; .
$
\\
$g(f(\vv x)) = \vv x$ для всех $\vv x \in \ZZ_2^n \Rightarrow g(\vv x) = f^{-1}(\vv x)$.
\end{proof}

Из определений~\ref{def_not_n} и~\ref{def_k_not_n} видно, что для элементов NOT и $k$-CNOT
задаваемые ими булевы преобразования обратны сами к себе: $f_i= f_i^{-1}$.
\begin{predicate}\label{predicate_reverse_scheme_particular}
    Если схема $\frS_f$ состоит из элементов NOT и $k$-CNOT, то схему $\frS_{f^{-1}}$ можно получить,
    соединив все элементы схемы $\frS_f$ в обратном порядке (зеркально отобразить).
\end{predicate}
\noindent
К примеру, схема, реализующая обратное преобразование к преобразованию, задаваемому схемой,
показанной на рис.~\ref{pic_scheme_example}, задаётся композицией элементов $C_{3;4} * C_{1,2,4;3} * C_{1;2} * N_4$.

\subsection{%
    \texorpdfstring{Чётность подстановок для NOT и $k$-CNOT}%
       {Чётность подстановок для NOT и k-CNOT}}

\forceindent
Возникает вопрос, какие подстановки задаются элементами $N_j^n$ и $C_{i_1, i_2, \ldots, i_k;j}^n$,
а также какие подстановки из $S(\ZZ_2^n)$ могут быть реализованы при помощи схем, состоящих из этих \gate ?

\begin{predicate}\label{predicate_parity_of_not}
    Элемент $N_j^n$ задаёт нечётную подстановку при $n = 1$ и чётную при $n > 1$.
\end{predicate}
\begin{proof}
    Элемент $N_j^n$ задаёт подстановку $h_{\text{NOT}}$ вида
    \begin{gather*}
        h_{\text{NOT}} = (\vv x^{(1)}, \vv y^{(1)}) \circ (\vv x^{(2)}, \vv y^{(2)}) \circ \ldots
        \circ (\vv x^{(m)}, \vv y^{(m)}),\text{\;\;}\vv x^{(i)}, \vv y^{(i)} \in \ZZ_2^n \; , \\
        \vv x^{(i)} = \langle x^{(i)}_1, \ldots, x^{(i)}_j, \ldots, x^{(i)}_n\rangle,\text{\;\;}
        \vv y^{(i)} = \langle x^{(i)}_1, \ldots, x^{(i)}_j \oplus 1, \ldots, x^{(i)}_n\rangle,\text{\;\;}x^{(i)}_k \in \ZZ_2 \; .
    \end{gather*}
    \noindent
    Таким образом, количество различных транспозиций $(\vv x^{(i)}, \vv y^{(i)})$, $\vv x^{(i)} \neq \vv y^{(i)}$,
    равно количеству различных векторов
    $\langle x^{(i)}_1, \ldots, x^{(i)}_{j-1}, x^{(i)}_{j+1}, \ldots, x^{(i)}_n\rangle$, т.\,е. $m = 2^{n-1}$,
    откуда следует нечётность подстановки $h_{\text{NOT}}$ при $n = 1$ и чётность при $n > 1$.
\end{proof}

\begin{predicate}\label{predicate_parity_of_$k$-CNOT}
    Элемент $C_{i_1, i_2, \ldots, i_k;j}^n$ задаёт нечётную подстановку при $k = n - 1$ и чётную при $k < n - 1$.
\end{predicate}
\begin{proof}
    Элемент $C_{i_1, i_2, \ldots, i_k;j}^{k+1}$ задаёт подстановку $h_{\text{$k$-CNOT}}$ вида
    \begin{gather*}
        h_{\text{$k$-CNOT}} = (\vv x^{(1)}, \vv y^{(1)}) \circ (\vv x^{(2)}, \vv y^{(2)}) \circ \ldots
        \circ (\vv x^{(m)}, \vv y^{(m)}),\text{\;\;}\vv x^{(i)}, \vv y^{(i)} \in \ZZ_2^n \; , \\
        \vv x^{(i)} = \langle x^{(i)}_1, \ldots, x^{(i)}_j, \ldots, x^{(i)}_n\rangle,\text{\;\;}
        \vv y^{(i)} = \langle x^{(i)}_1, \ldots, x^{(i)}_j \oplus x^{(i)}_{i_1} \land \ldots\land x^{(i)}_{i_k},
        \ldots, x^{(i)}_n\rangle,\text{\;\;}x^{(i)}_k \in \ZZ_2 \; .
    \end{gather*}
    \noindent
    Таким образом, количество различных транспозиций $(\vv x^{(i)}, \vv y^{(i)})$, $\vv x^{(i)} \neq \vv y^{(i)}$,
    равно количеству различных векторов
    $\langle x^{(i)}_1, \ldots, x^{(i)}_{j-1}, x^{(i)}_{j+1}, \ldots, x^{(i)}_n\rangle$,
    для которых $x^{(i)}_t = 1$ при $t \in \{ i_1, \ldots, i_k\}$,
    т.\,е. $m = 2^{n-(k+1)}$, откуда следует нечётность подстановки $h_{\text{$k$-CNOT}}$ при $k=n-1$ и чётность при $k<n-1$.
\end{proof}

\begin{corollary}\label{corollary_parity_of_scheme}
    Обратимая схема с $n$ входами, состоящая из элементов NOT, CNOT и 2-CNOT, при $n < 4$ задаёт некоторую подстановку
    из симметрической группы $S(\ZZ_2^n)$, а при $n \geqslant 4$~--- некоторую чётную подстановку
    из знакопеременной группы $A(\ZZ_2^n)$.
\end{corollary}
\begin{proof}
    Следует из утверждений~\ref{predicate_parity_of_not}, \ref{predicate_parity_of_$k$-CNOT}
    и формулы~\eqref{formula_scheme_permutation}.
\end{proof}
\noindent
Другими словами, при помощи обратимой схемы с $n \geqslant 4$ входами, состоящей из элементов NOT, CNOT и 2-CNOT,
нельзя реализовать нечётную подстановку.

\begin{corollary}\label{corollary_parity_of_xor_matrix}
    Обратимое линейное преобразование над координатами векторов линейного векторного пространства $\ZZ_2^n$
    задаёт чётную подстановку при $n > 2$.
\end{corollary}
\begin{proof}
    Линейное преобразование над координатами векторов линейного векторного пространства $\ZZ_2^n$
    можно задать при помощи матрицы $A_{n \times n} = (a_{ij})$, $a_{ij} \in \ZZ_2$, $0 \leqslant i, j \leqslant n-1$.

    По условию линейное преобразование обратимо, следовательно, элементарными преобразованиями столбцов (сложение и перестановка)
    матрицу $A_{n \times n}$ можно привести к единичной матрице $E_{n \times n}$.

    Схема с $n$ входами и выходами, не содержащая ни одного элемента, задаёт линейное преобразование с матрицей $E_{n \times n}$.
    Прибавление $i$-го столбца матрицы к $j$-му столбцу матрицы задаётся элементом $C_{i;j}^n$.
    Перестановка $i$-го и $j$-го столбцов матрицы задаётся композицией элементов $C_{i;j}^n * C_{j;i}^n * C_{i;j}^n$.
    Таким образом, ставя в соответствие каждому из элементарных преобразований столбцов,
    которые приводят матрицу $A_{n \times n}$ к матрице $E_{n \times n}$, композицию элементов CNOT,
    можно построить схему, реализующую заданное обратимое линейное преобразование.

    Построенная таким способом обратимая схема будет содержать только элементы CNOT.
    Из утверждения~\ref{predicate_parity_of_$k$-CNOT} и формулы~\eqref{formula_scheme_permutation} следует,
    что такая схема задаёт чётную подстановку при $n > 2$.
\end{proof}

\subsection{%
    \texorpdfstring{Множество обратимых \gate, порождающее $S(\ZZ_2^n)$}%
        {Множество обратимых Ф.Э., порождающее симметрическую группу}%
}

\forceindent
Обозначим через $\Omega_n^2$ множество всех элементов $N_j^n$, $C_{i;j}^n$ и
$C_{i_1, i_2;j}^n$ при фиксированном значении $n$, а через $S_{\Omega_n^2}$~--- множество подстановок,
задаваемых этими \gate{} Мощности этих множеств равны следующей величине:
$$
    |\Omega_n^2| = |S_{\Omega_n^2}| = n + {\binom{n}{1}} \cdot (n-1) + {\binom{n}{2}} \cdot (n-2) = O(n^3) \; .
$$
В работе~\cite{shende_synthesis} было доказано, что множество $S_{\Omega_n^2}$ порождает симметрическую группу
$S(\ZZ_2^n)$ при $n < 4$ и знакопеременную группу $A(\ZZ_2^n)$ при $n \geqslant 4$.
Здесь будет приведено другое доказательство, на основе которого разработан новый быстрый алгоритм синтеза обратимых схем
с использованием теории групп подстановок (см. раздел~\ref{subsection_my_synthesis_algorithms}).
Частично это доказательство было опубликовано в работе~\cite{my_lemma_prove}.

\begin{lemma}\label{basis_on_s_n}
    Множество подстановок $S_{\Omega_n^2}$ при $n < 4$ порождает симметрическую группу $S(\ZZ_2^n)$.
\end{lemma}

Прежде чем перейти непосредственно к доказательству Леммы~\ref{basis_on_s_n}, отметим, что
для множества $A = \{\, a_1, a_2, \ldots, a_m \,\}$
множество подстановок $\{\,(a_1, a_2), (a_1, a_2, \ldots, a_m)\,\}$ порождает симметрическую группу $S(A)$.

Подстановку $(a_1, a_2, \ldots, a_m)$ можно представить в виде следующего произведения транспозиций:
$$
    (a_1, a_2, \ldots, a_m)=(a_1, a_2) \circ (a_1, a_3) \circ \ldots \circ (a_1, a_{m-1}) \circ (a_1, a_m) \; .
$$
В доказательстве ниже будут построены все транспозиции вида
$(\vv 0, \vv x)$, $\vv x \in \ZZ_2^n \setminus \{\, \vv 0 \,\}$, $\vv 0 = \langle 0, \ldots, 0 \rangle$,
при помощи подстановок из множества $S_{\Omega_n^2}$.

Действие сопряжением подстановкой $g$ на подстановку $h$ обозначим через $h^g$:
$$
    h^g = g^{-1} \circ h \circ g \; .
$$
Известно, что действие сопряжением сохраняет цикловую структуру подстановки.

\begin{proof}[Доказательство Леммы~\ref{basis_on_s_n}]
    Для всех значений $n < 4$ получим базовую транспозицию $h=(\vv 0, \vv x_1)$,
    где $\vv x_1 = \langle 0, \ldots, 0, 1 \rangle$.

    При $n=1$ искомая транспозиция $h=(\langle 0 \rangle, \langle 1 \rangle)$ задаётся элементом $N_1$.

    При $n=2$ элемент $C_{2;1}$ задаёт подстановку $h_1 = (\langle 1, 0 \rangle, \langle 1, 1 \rangle)$.
    Элемент $N_2$ задаёт подстановку
    $h_2 = (\langle 0, 0 \rangle, \langle 1, 0 \rangle) \circ (\langle 0, 1 \rangle, \langle 1, 1 \rangle)$.
    Искомая транспозиция $h = (\langle 0, 0 \rangle, \langle 0, 1 \rangle) = h_1^{h_2}$ задаётся композицией элементов
    $N_2 * C_{2;1}*N_2$.

    При $n=3$ элемент $C_{2,3;1}$ задаёт подстановку $g_1 = (\langle 1, 1, 0 \rangle, \langle 1, 1, 1 \rangle)$.
    Элемент $N_2$ задаёт подстановку $g_2$ вида
    $$
        g_2 = (\langle 0, 0, 0 \rangle, \langle 0, 1, 0 \rangle) \circ
        (\langle 0, 0, 1 \rangle, \langle 0, 1, 1 \rangle) \circ (\langle 1, 0, 0 \rangle, \langle 1, 1, 0 \rangle) \circ
        (\langle 1, 0, 1 \rangle, \langle 1, 1, 1 \rangle) \; .
    $$
    Элемент $N_3$ задаёт подстановку $g_3$ вида
    $$
        g_3 = (\langle 0, 0, 0 \rangle, \langle 1, 0, 0 \rangle) \circ
        (\langle 0, 0, 1 \rangle, \langle 1, 0, 1 \rangle) \circ (\langle 0, 1, 0 \rangle, \langle 1, 1, 0 \rangle) \circ
        (\langle 0, 1, 1 \rangle, \langle 1, 1, 1 \rangle) \; .
    $$
    Искомая транспозиция $h = (\langle 0, 0, 0 \rangle, \langle 0, 0, 1 \rangle) = g_1^{g_2\circ g_3}$ задаётся композицией элементов
    $N_3 * N_2 * C_{2,3;1} * N_2 * N_3$.

    Для того, чтобы получить произвольную транспозицию $g = (\vv 0, \vv x)$,
    $\vv x \in \ZZ_2^n \setminus \{\, \vv 0, \vv x_1 \,\}$, будем действовать сопряжением на базовую транспозицию
    $h=(\vv 0, \vv x_1)$.

    Для всех $x_i = 1$ вектора $\vv x = \langle x_1, \ldots, x_n \rangle$, $i \ne 1$, действуем сопряжением на $h$ подстановкой,
    задаваемой элементом $C_{1;i}$, получим $h'$. Если $x_1 = 1$, то $h' = g$.
    В противном случае существует $x_j = 1$, $j \ne 1$;
    действуем сопряжением на $h'$ подстановкой, задаваемой элементом $C_{j;1}$, получим $h'' = g$.

    Таким образом, для всех значений $n < 4$ были получены все транспозиции вида
    $(\vv 0, \vv x)$, $\vv x \in \ZZ_2^n \setminus \{\, \vv 0 \,\}$,
    при помощи подстановок из множества $S_{\Omega_n^2}$.
    Следовательно, $S_{\Omega_n^2}$ порождает симметрическую группу $S(\ZZ_2^n)$ при $n < 4$.
\end{proof}

К примеру, транспозиция $(\langle 0, 0, 0 \rangle, \langle 1, 1, 0 \rangle)$ при $n = 3$
задаётся следующей композицией \gate:
$$
    C_{2;1} * (C_{1;3} * C_{1;2}) * (N_3 * N_2 * C_{2,3;1} * N_2 * N_3) * (C_{1;2} * C_{1;3}) * C_{2;1} \; .
$$

\subsection{%
    \texorpdfstring{Множество обратимых \gate, порождающее $A(\ZZ_2^n)$}%
        {Множество обратимых Ф.Э., порождающее знакопеременную группу}%
}

\forceindent Прежде чем перейти к случаю $n \geqslant 4$, докажем ещё несколько утверждений.

\begin{predicate}\label{predicate_recursive_on_independent_product}
    Подстановка $h \in S(\ZZ_2^n)$ вида
    $$
        h = (\langle 0,0,1,1,\ldots,1 \rangle, \langle 1,0,1,1,\ldots,1 \rangle) \circ
        (\langle 0,1,1,1,\ldots,1 \rangle, \langle 1,1,1,1,\ldots,1 \rangle)
    $$
    при $n \geqslant 4$ задаётся композицией элементов $C_{2,3;1} * C_{n, \ldots,4;2} * C_{2,3;1} * C_{n, \ldots,4;2}$.
\end{predicate}
\begin{proof}
    $ $

    Элемент $C_{2,3;1}$ задаёт подстановку $h_1 = \mcirc_{\vv x, \vv y \in \ZZ_2^n}{(\vv x, \vv y)}$, где
    $\vv x = \langle 0, 1, 1, x_4, \ldots, x_n \rangle$, $\vv y = \langle 1, 1, 1, x_4, \ldots, x_n \rangle$.

    Элемент $C_{n, \ldots,4;2}$ задаёт подстановку $g = \mcirc_{\vv x, \vv y \in \ZZ_2^n}{(\vv x, \vv y)}$,
    где $\vv x = \langle x_1, 0, x_3, 1, \ldots, 1 \rangle$, $\vv y = \langle x_1, 1, x_3, 1, \ldots, 1 \rangle$.

    Подстановка $h = h_1 \circ h_1^g$, откуда следует истинность доказываемого утверждения.
\end{proof}

\begin{corollary}\label{corollary_recursive_on_independent_product_swap}
    Подстановка $h \in S(\ZZ_2^n)$ вида
    $$
        h = (\langle 0,0,1,1,\ldots,1 \rangle, \langle 1,0,1,1,\ldots,1 \rangle) \circ
        (\langle 0,1,1,1,\ldots,1 \rangle, \langle 1,1,1,1,\ldots,1 \rangle)
    $$
    при $n \geqslant 4$ также задаётся композицией элементов $C_{n, \ldots,4;2} * C_{2,3;1} * C_{n, \ldots,4;2} * C_{2,3;1}$.
\end{corollary}
\begin{proof}
    Подстановка $h$ представляет собой произведение независимых транспозиций, следовательно $h = h^{-1}$.
    Используя обозначения из доказательства утверждения~\ref{predicate_recursive_on_independent_product}, можно записать, что
    $h = h_1 \circ h_1^g \Rightarrow h^{-1} = (h_1 \circ g^{-1} \circ h_1 \circ g)^{-1}$.

    $h = h^{-1} = g^{-1} \circ h_1^{-1} \circ g \circ h_1^{-1} = g^{-1} \circ h_1 \circ g \circ h_1 = h_1^g \circ h_1$,
    откуда следует, что подстановка $h$ задаётся композицией элементов
    $C_{n, \ldots,4;2} * C_{2,3;1} * C_{n, \ldots,4;2} * C_{2,3;1}$.
\end{proof}

\begin{corollary}\label{corollary_recursive_on_independent_product_common}
    В общем случае при $n \geqslant 4$ элемент $C_{i_1,\ldots,i_k;j}^n$,
    $k \leqslant n - 2$, можно заменить без изменения результирующего булевого преобразования на композицию элементов
    \begin{equation}
        C_{i_k,l;j} * C_{i_1,\ldots,i_{k-1};l} * C_{i_k,l;j} * C_{i_1,\ldots,i_{k-1};l} \; ,
        \label{formula_k_CNOT_recursive}
    \end{equation}
    или на композицию элементов
    \begin{equation}
        C_{i_1,\ldots,i_{k-1};l} * C_{i_k,l;j} * C_{i_1,\ldots,i_{k-1};l} * C_{i_k,l;j} \; ,
        \label{formula_k_CNOT_recursive_mirrored}
    \end{equation}
    где $l \ne i_1, \ldots, i_k, j$.
\end{corollary}
\begin{proof}
    Построив биективное отображение
    $$
        \langle i_1,i_2,\ldots,i_k,l,j \rangle \text{} \to \text{} \langle n,n-1,\ldots,3,2,1 \rangle \; ,
    $$
    можно перейти к случаю $k=n-2$, доказанному в утверждении~\ref{predicate_recursive_on_independent_product}
    и следствии~\ref{corollary_recursive_on_independent_product_swap}. При $k < n-2$ суть доказательства остаётся той же.
\end{proof}

Эквивалентные замены~\eqref{formula_k_CNOT_recursive} и~\eqref{formula_k_CNOT_recursive_mirrored}
элемента $C_{i_1,\ldots,i_k;j}^n$ также были доказаны в работе~\cite{group_based} через явное вычисление
значений на выходах этого элемента.

\begin{predicate}[{\cite[Следствие~7.4]{barenco_elementary_gates}}]\label{fast_realization_of_k_cnot}
    Элемент $C_{i_1,\ldots,i_k;j}^n$ при $n \geqslant 4$ и $k \leqslant n - 2$ можно заменить
    без изменения результирующего булевого преобразования на композицию не более чем $8(k-3)$ элементов CNOT и 2-CNOT.
\end{predicate}
Отметим, что данная замена возможна за счёт одного из входов элемента, не являющегося ни контролирующим, ни контролируемым
(условие $k \leqslant n - 2$).

Рассмотрим теперь произведение двух зависимых транспозиций.

\begin{predicate}\label{predicate_recursive_on_dependent_product}
    Подстановка $h \in S(\ZZ_2^n)$ вида
    $$
        h = (\langle 0,1,1,1,\ldots,1 \rangle, \langle 1,0,1,1,\ldots,1 \rangle) \circ
        (\langle 0,1,1,1,\ldots,1 \rangle, \langle 1,1,1,1,\ldots,1 \rangle)
    $$
    при $n \geqslant 4$ задаётся композицией элементов $C_{n, \ldots,5,4,1;2} * C_{2,3;1} * C_{n, \ldots,5,4,1;2} * C_{2,3;1}$
    и композицией $C_{1,3;2} * C_{n, \ldots,5,4,2;1} * C_{1,3;2} * C_{n, \ldots,5,4,2;1}$.
\end{predicate}
\begin{proof}
    $ $

    Элемент $C_{2,3;1}$ задаёт подстановку $h_1 = \mcirc_{\vv x, \vv y \in \ZZ_2^n}{(\vv x, \vv y)}$, где
    $\vv x = \langle 0, 1, 1, x_4, \ldots, x_n \rangle$, $\vv y = \langle 1, 1, 1, x_4, \ldots, x_n \rangle$.

    Элемент $C_{n, \ldots,5,4,1;2}$ задаёт подстановку
    $g_1 = \mcirc_{\vv x, \vv y \in \ZZ_2^n}{(\vv x, \vv y)}$,
    где $\vv x = \langle 1, 0, x_3, 1, \ldots, 1 \rangle$, $\vv y = \langle 1, 1, x_3, 1, \ldots, 1 \rangle$.

    Подстановка $h = h_1^{g_1} \circ h_1$, откуда следует истинность первой части доказываемого утверждения.

    С другой стороны, элемент $C_{1,3;2}$ задаёт подстановку
    $h_2 = \mcirc_{\vv x, \vv y \in \ZZ_2^n}{(\vv x, \vv y)}$, где
    $\vv x = \langle 1, 0, 1, x_4, \ldots, x_n \rangle$, $\vv y = \langle 1, 1, 1, x_4, \ldots, x_n \rangle$.

    Элемент $C_{n, \ldots,5,4,2;1}$ задаёт подстановку
    $g_2 = \mcirc_{\vv x, \vv y \in \ZZ_2^n}{(\vv x, \vv y)}$,
    где $\vv x = \langle 0, 1, x_3, 1, \ldots, 1 \rangle$, $\vv y = \langle 1, 1, x_3, 1, \ldots, 1 \rangle$.

    Подстановка $h = g_2^{h_2} \circ g_2$, откуда следует истинность второй части доказываемого утверждения.
\end{proof}

Теперь можно перейти ко второй лемме данной главы.
\begin{lemma}\label{basis_on_a_n}
    Множество подстановок $S_{\Omega_n^2}$ при $n \geqslant 4$ порождает знакопеременную группу $A(\ZZ_2^n)$.
\end{lemma}
\begin{proof}
    Рассмотрим произвольную чётную подстановку $h \in A(\ZZ_2^n)$. По определению подстановка $h$ представляется в виде
    произведения чётного числа транспозиций, следовательно, эти транспозиции можно разбить на пары:
    \begin{equation}
        h = \mcirc_{\vv x_i, \vv y_i \in \ZZ_2^n}
        {\left( (\vv x_1, \vv y_1) \circ (\vv x_2, \vv y_2) \right )} \; .
        \label{formula_even_perm_as_product_of_pairs}
    \end{equation}
    Рассмотрим по отдельности каждую такую пару транспозиций
    $p = (\vv x_1, \vv y_1) \circ (\vv x_2, \vv y_2)$.

    \bigskip
    \bigskip
    Пусть $p$ является произведением двух независимых транспозиций: $p = (\vv x, \vv y) \circ (\vv z, \vv w)$.
    Для простоты доказательства приведём $p$ действием сопряжения к виду
    $$
        g = (\langle 0,0,1,1,\ldots,1 \rangle, \langle 1,0,1,1,\ldots,1 \rangle) \circ
        (\langle 0,1,1,1,\ldots,1 \rangle, \langle 1,1,1,1,\ldots,1 \rangle) \; .
    $$

    Сперва применим к $p$ действие сопряжением подстановкой $h_1$, задаваемой композицией элементов
    $\frS_1 = \compose_{x_j = 1}{N_j}$:
    $$
        p^{(1)} = p^{h_1} = (\langle 0, \ldots, 0 \rangle, \vv y^{(1)}) \circ (\vv z^{(1)}, \vv w^{(1)}) \; .
    $$

    $\vv y^{(1)} \ne \langle 0, \ldots, 0 \rangle \Rightarrow$ существует $y^{(1)}_i = 1$.
    Применяем к $p^{(1)}$ действие сопряжением подстановкой $h_2$, задаваемой композицией элементов
    $\frS_2 = \compose_{y^{(1)}_j = 1, j \ne 1}{C_{1;j}}$, если $y^{(1)}_1 = 1$, или композицией элементов
    $$
        \frS_2 = C_{i;1} *  \left( \compose_{y^{(1)}_j = 1, j \ne 1}{C_{1;j}} \right ) * C_{i;1} \; ,
    $$
    если $y^{(1)}_1 = 0$. Получим новую подстановку
    \begin{gather*}
        p^{(2)} = \left( p^{(1)} \right)^{h_2} =
        (\vv x^{(2)}, \vv y^{(2)}) \circ (\vv z^{(2)}, \vv w^{(2)}) \; , \\
        \vv x^{(2)} = \langle 0, \ldots, 0 \rangle, \vv y^{(2)} = \langle 1, 0, \ldots, 0 \rangle \; .
    \end{gather*}

    $\vv z^{(2)} \ne \vv x^{(2)}, \vv y^{(2)} \Rightarrow$ существует $z^{(2)}_i = 1$, $i \ne 1$.
    Применяем к $p^{(2)}$ действие сопряжением подстановкой $h_3$, задаваемой композицией элементов
    $\frS_3 = \compose_{z^{(2)}_j = 1, j \ne 2}{C_{2;j}}$, если $z^{(2)}_2 = 1$, или композицией элементов
    $$
        \frS_3 = C_{i;2} *  \left( \compose_{z^{(2)}_j = 1, j \ne 2}{C_{2;j}} \right ) * C_{i;2} \; ,
    $$
    если $z^{(2)}_2 = 0$. Получим новую подстановку
    \begin{gather*}
        p^{(3)} = \left( p^{(2)} \right)^{h_3} =
        (\vv x^{(3)}, \vv y^{(3)}) \circ (\vv z^{(3)}, \vv w^{(3)}) \; , \\
        \vv x^{(3)} = \langle 0, \ldots, 0 \rangle, \vv y^{(3)} = \langle 1, 0, \ldots, 0 \rangle,
        \vv z^{(3)} = \langle 0, 1, 0,\ldots, 0 \rangle \; .
    \end{gather*}

    Т.\,к. $\vv w^{(3)} \ne \vv x^{(3)}, \vv y^{(3)}, \vv z^{(3)}$, то возможны 2 варианта:

    \begin{enumerate}
        \item
            $w^{(3)}_1 = w^{(3)}_2 = 1 \Rightarrow$ применяем к $p^{(3)}$ действие сопряжением подстановкой $h_4$,
            задаваемой композицией элементов $\frS_4 = \compose_{w^{(3)}_j = 1, j \ne 1,2}{C_{1,2;j}}$.

        \item
            Либо $w^{(3)}_1 = 0$, либо $w^{(3)}_2 = 0$, но тогда существует $w^{(3)}_i = 1$, $i \ne 1,2$.

            Применяем к $p^{(3)}$ действие сопряжением подстановкой $h_4$,
            задаваемой композицией элементов
            $$
                \frS_4 = C_{i;1} *  \left( \compose_{w^{(3)}_j = 1, j \ne 1,2}{C_{1,2;j}} \right ) * C_{i;1} \; ,
            $$
            если $w^{(3)}_1 = 0$, $w^{(3)}_2 = 1$; композицией элементов
            $$
                \frS_4 = C_{i;2} *  \left( \compose_{w^{(3)}_j = 1, j \ne 1,2}{C_{1,2;j}} \right ) * C_{i;2} \; ,
            $$
            если $w^{(3)}_1 = 1$, $w^{(3)}_2 = 0$; и композицией элементов
            $$
                \frS_4 = C_{i;1} *  C_{i;2} *  \left( \compose_{w^{(3)}_j = 1, j \ne 1,2}{C_{1,2;j}} \right )
                    * C_{i;2} *  C_{i;1} \; ,
            $$
            если $w^{(3)}_1 = 0$, $w^{(3)}_2 = 0$.
    \end{enumerate}

    В итоге получим новую подстановку
    \begin{align*}
        &p^{(4)} = \left( p^{(3)} \right)^{h_4} =
        (\vv x^{(4)}, \vv y^{(4)}) \circ (\vv z^{(4)}, \vv w^{(4)})  \; , \\
        &\vv x^{(4)} = \langle 0, 0, 0 \ldots, 0 \rangle, \vv y^{(4)} = \langle 1, 0, 0, \ldots, 0 \rangle \; , \\
        &\vv z^{(4)} = \langle 0, 1, 0,\ldots, 0 \rangle, \vv w^{(4)} = \langle 1, 1, 0, \ldots, 0 \rangle \; .
    \end{align*}

    На последнем шаге применяем к $p^{(4)}$ действие сопряжением подстановкой $h_5$, задаваемой композицией элементов
    $\frS_5 = \compose_{3 \leqslant j \leqslant n}{N_j}$:
    $$
        g = \left( p^{(4)} \right)^{h_5} = (\langle 0,0,1,\ldots,1 \rangle, \langle 1,0,1,\ldots,1 \rangle) \circ
        (\langle 0,1,1,\ldots,1 \rangle, \langle 1,1,1,\ldots,1 \rangle) \; .
    $$

    Таким образом, верны следующие равенства:
    \begin{align*}
        & g = p^{h_1 \circ h_2 \circ h_3 \circ h_4 \circ h_5} \; , \\
        & p = g^{h_5^{-1} \circ h_4^{-1} \circ h_3^{-1} \circ h_2^{-1} \circ h_1^{-1}} \; , \\
        & p = g^{h_5 \circ h_4 \circ h_3 \circ h_2 \circ h_1} \; .
    \end{align*}

    Согласно утверждению~\ref{predicate_recursive_on_independent_product}, подстановка $g$
    задаётся композицией элементов $C_{2,3;1} * C_{n, \ldots,4;2} * C_{2,3;1} * C_{n, \ldots,4;2}$.
    Следовательно, пара независимых транспозиций $p = (\vv x, \vv y) \circ (\vv z, \vv w)$
    задаётся следующей композицией элементов:
    $$
        \left( \compose_{i = 1}^5 \frS_i \right ) * C_{2,3;1} * C_{n, \ldots,4;2} * C_{2,3;1} * C_{n, \ldots,4;2} *
        \left( \compose_{i = 5}^1 \frS_i \right ) \; ,
    $$
    где композиция элементов $\left( \compose_{i = 1}^5 \frS_i \right )$ задаёт подстановку
    $h_1 \circ h_2 \circ h_3 \circ h_4 \circ h_5$,
    а композиция элементов $\left( \compose_{i = 5}^1 \frS_i \right )$~--- подстановку
    $h_5 \circ h_4 \circ h_3 \circ h_2 \circ h_1$.

    Согласно утверждению~\ref{fast_realization_of_k_cnot}, элемент $C_{n, \ldots,4;2}$ может быть выражен
    через композицию элементов CNOT и 2-CNOT.

    В случае, когда $p$ является произведением зависимых транспозиций, мы можем выразить $p$ через
    произведение двух пар независимых транспозиций:
    $$
        p = (\vv x, \vv y) \circ (\vv x, \vv z) = ((\vv x, \vv y) \circ (\vv a, \vv b)) \circ
            ((\vv a, \vv b) \circ (\vv x, \vv z)) \; .
    $$
    Каждую из пар независимых транспозиций можно выразить через композицию элементов множества $\Omega_n^2$
    описанным выше способом.

    Таким образом, любая пара транспозиций $p = (\vv x_1, \vv y_1) \circ (\vv x_2, \vv y_2)$
    (из представления~\eqref{formula_even_perm_as_product_of_pairs} произвольной чётной подстановки $h \in A(\ZZ_2^n)$
    в виде произведения пар транспозиций) задаётся композицией элементов из множества $\Omega_n^2$.
    Следовательно, множество подстановок $S_{\Omega_n^2}$ порождает знакопеременную группу $A(\ZZ_2^n)$ при $n \geqslant 4$.
\end{proof}

\subsection{Схемы с дополнительной памятью}

\forceindent 
Как было показано в предыдущих разделах, обратимая схема с $n \geqslant 4$ входами, состоящая из элементов NOT и $k$-CNOT,
$ k < n - 1$, всегда реализует чётную подстановку на множестве $\ZZ_2^n$. Возникает вопрос: возможно ли реализовать
при помощи такой схемы произвольное булево отображение $f\colon \ZZ_2^n \to \ZZ_2^m$,
в общем случае не биективное?

Введём следующие отображения:
\begin{enumerate}

    \item
        \textit{Расширяющее} отображение $\phi_{n,n+k}\colon \ZZ_2^n \to \ZZ_2^{n+k}$ вида
        $$
            \phi_{n,n+k}( \langle x_1, \ldots, x_n \rangle ) = \langle x_1, \ldots, x_n, 0, \ldots, 0 \rangle \; .
        $$
    \item
        \textit{Редуцирующее} отображение $\psi_{n+k,n}^\pi\colon \ZZ_2^{n+k} \to \ZZ_2^n$ вида
        $$
            \psi_{n+k,n}^\pi( \langle x_1, \ldots, x_{n+k} \rangle ) =
            \langle x_{\pi(1)}, \ldots, x_{\pi(n)} \rangle \; ,
        $$
        где $\pi$~--- некоторая подстановка на множестве $\ZZ_{n+k}$.

\end{enumerate}

Рассмотрим произвольное булево отображение $f\colon \ZZ_2^n \to \ZZ_2^m$.
\begin{define}\label{def_scheme_with_additional_memory}
    Обратимая схема $\frS_g$ с $(n+q) \geqslant m$ входами, задающая булево преобразование $g\colon \ZZ_2^{n+q} \to \ZZ_2^{n+q}$,
    реализует отображение $f$ c использованием $q \geqslant 0$ дополнительных входов (дополнительной памяти),
    если существует такая подстановка $\pi \in S(\ZZ_{n+q})$, что
    $$
        \psi_{n+q,m}^\pi(g( \phi_{n,n+q}(\vv x))) = f(\vv x) \; ,
    $$
    где $\vv x \in \ZZ_2^n$, $f(\vv x) \in \ZZ_2^m$ (см. рис.~\ref{pic_scheme_realization}).
\end{define}
\noindent
Если в определении~\ref{def_scheme_with_additional_memory} количество дополнительных входов $q=0$,
то будем говорить, что схема реализует отображение $f$ \textit{без дополнительной памяти}.

\medskip
\Figure[ht]
    \centering
    \includegraphics[scale=1.2]{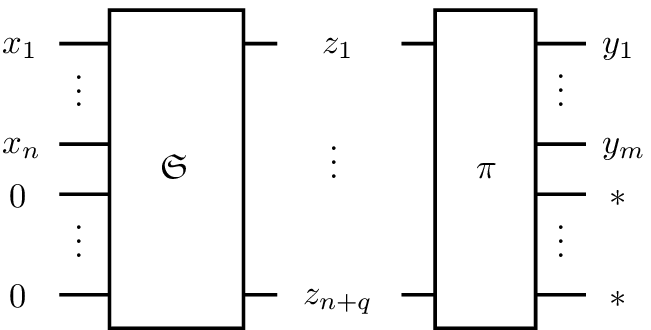}
    \caption
    {
        \small Обратимая схема $\frS$, реализующая булево отображение $f\colon \ZZ_2^n \to \ZZ_2^m$ с $q$
        дополнительными входами. Для всех $\vv x \in \ZZ_2^n$ верно равенство $f(\vv x) = \vv y$, где $\vv y \in \ZZ_2^m$.
    }\label{pic_scheme_realization}
\end{figure}

Отметим, что в данной терминологии выражения <<реализует отображение>> и <<задаёт отображение>> имеют разные значения:
если обратимая схема $\frS_g$ задаёт отображение $f$, то $g(\vv x) = f(\vv x)$ для всех входных значений $\vv x$.
Также из определения выше следует, что при $m > n$ не существует обратимой схемы без дополнительной памяти,
реализующей отображение $f\colon \ZZ_2^n \to \ZZ_2^m$.

Будем называть подстановку $\pi$ из определения~\ref{def_scheme_with_additional_memory} \textit{перестановкой выходов}
обратимой схемы.
\begin{define}
    Обратимая схема $\frS$ строго реализует заданное отображение, если её перестановка выходов является тождественной,
    и нестрого в противном случае.
\end{define}
\noindent
На рис.~\ref{pic_swap} показаны обратимые схемы, реализующие нестрого~(а) и строго~(б) преобразование
$f(\vv x) = \vv y$, где $\vv x, \vv y \in \ZZ_2^2$, $y_1 = x_2$, $y_2 = x_1$.

\medskip
\Figure[ht]
    \centering
    \includegraphics[scale=1.2]{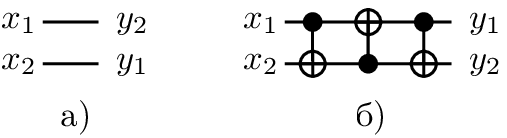}
    \caption
    {
        \small Обратимые схемы, реализующие нестрого~(а) и строго~(б) преобразование
        $f(\langle x_1, x_2 \rangle) = \langle x_2, x_1 \rangle$.
    }\label{pic_swap}
\end{figure}

Из рис.~\ref{pic_swap} следует, что из обратимой схемы $\frS^*$, нестрого реализующей заданное отображение $f$,
можно получить обратимую схему $\frS$, строго реализующую это же отображение, присоединив к выходам схемы $\frS^*$
композицию из элементов CNOT в количестве не более $3m$ штук. Следовательно, верно следующее соотношение:
$$
    L(\frS) \leqslant L(\frS^*) + 3m \; .
$$

Далее до конца этой главы все определения и утверждения будут формулироваться только для схем, строго реализующих заданное
отображение. Аналогичные определения и утверждения для схем, нестрого реализующих заданное отображение, можно получить,
перенумеровав выходы схемы в соответствии с перестановкой выходов. Подстановку $\pi$ в верхнем индексе редуцирующего
отображения $\psi_{n,m}^\pi$ будем опускать, т.\,к. для обратимой схемы, строго реализующей заданное отображение,
эта подстановка является тождественной.

\begin{define}
    Значимые входы схемы~--- входы, не являющиеся дополнительными.
    Значимые выходы схемы~--- выходы, значения на которых нужны для дальнейших вычислений.
\end{define}
\noindent
К примеру, на рис.~\ref{pic_scheme_realization} незначимые входы помечены символом 0, а незначимые выходы~--- символом $*$.

Будем считать, что обратимая схема из определения~\ref{def_scheme_with_additional_memory}
порождает \textit{вычислительный мусор} на незначимых выходах,
если для некоторого $\vv x \in \ZZ_2^n$ и $\vv y = g( \phi_{n,n+q}(\vv x))$,
$\vv y \in \ZZ_2^{n+q}$, выполняется неравенство
\begin{equation}
    \phi_{n,n+q}(\psi_{n+q,n}(\vv y)) \ne \vv y \label{formula_garbage_output_ne} \; .
\end{equation}
\noindent
Другими словами, вычислительный мусор~--- ненулевое значение на каком-либо незначимом выходе обратимой схемы,
когда значения на всех дополнительных входах этой схемы равны 0.

Если отображение $f$ инъективно, то существует такое булево отображение $f'\colon \ZZ_2^m \to \ZZ_2^n$,
что $f'(f(\vv x)) = \vv x$, $\vv x \in \ZZ_2^n$.
Пусть обратимая схема $\frS_g = \compose_{i = 1}^s {E_{f_i}}$ реализует отображение $f$ без дополнительной памяти
(с дополнительной памятью без порождения вычислительного мусора).
Тогда из определения обратимой схемы и определения~\ref{def_scheme_with_additional_memory} следует,
что схема $\frS'_{g'} = \compose_{i = s}^1 {E_{f_i}}$ реализует отображение $f'$ без дополнительной памяти
(с дополнительной памятью без порождения вычислительного мусора).
Однако если схема $\frS_g$ реализует отображение $f$ с дополнительной памятью и порождает вычислительный мусор,
то в общем случае неверно утверждение, что схема $\frS'_{g'}$ реализует отображение $f'$.
Поясним сказанное при помощи введённых отображений $\phi$ и $\psi$.
Пусть для некоторого $\vv x \in \ZZ_2^n$ и $\vv y = g( \phi_{n,n+q}(\vv x))$
выполняется неравенство~\eqref{formula_garbage_output_ne}. В этом случае верно равенство
$\psi_{n+q,n}(g'(\vv y)) = \vv x$.
Однако в общем случае
$$
    \psi_{n+q,n}(g'(\phi_{n,n+q}(\psi_{n+q,n}(\vv y)))) \ne \vv x \; ,
$$
а это равносильно утверждению, что схема $\frS'_{g'}$ в общем случае не реализует отображение $f'$.
Таким образом, условие отсутствия вычислительного мусора на незначимых выходах схемы \mbox{$\frS_g = \compose_{i = 1}^s {E_{f_i}}$}
является \textit{достаточным} (но не необходимым), чтобы композиция элементов $\compose_{i = s}^1 {E_{f_i}}$
задавала отображение $f'$. Из этого утверждения можно сделать два важных вывода.

\begin{predicate}
    Для заданного биективного булевого отображения $f\colon \ZZ_2^n \to \ZZ_2^n$ минимальная сложность
    обратимой схемы, реализующей преобразование $f^{-1}$ и состоящей из элементов множества $\Omega_n^2$,
    равна сложности обратимой схемы $\frS$, реализующей преобразование $f$ и состоящей из элементов множества $\Omega_n^2$,
    если (достаточное условие):
    \begin{enumerate}
        \item
            Схема $\frS$ реализует преобразование $f$ без дополнительной памяти
            или с дополнительной памятью без порождения вычислительного мусора.
        
        \item
            Схема $\frS$ имеет минимальную сложность среди всех обратимых схем,
            состоящих из элементов множества $\Omega_n^2$ и реализующих преобразование $f$.        
    \end{enumerate}
\end{predicate}
\noindent
Другими словами, построение обратимой схемы, реализующую прямое преобразование без дополнительной памяти
или с дополнительной памятью без порождения вычислительного мусора,
позволяет дать оценку сложности обратимой схемы, реализующей обратное преобразование.

\begin{predicate}
    Для заданного биективного булевого преобразования $f\colon \ZZ_2^n \to \ZZ_2^n$ сложность
    обратимой схемы, реализующей преобразование $f^{-1}$ и состоящей из элементов множества $\Omega_n^2$,
    может отличаться от сложности обратимой схемы $\frS$, реализующей преобразование $f$ и состоящей из элементов
    множества $\Omega_n^2$, только тогда, когда (необходимое условие):
    \begin{enumerate}
        \item
            Схема $\frS$ реализует преобразование $f$ с порождением вычислительного мусора.
        
        \item
            Схема $\frS$ имеет минимальную сложность среди всех обратимых схем,
            состоящих из элементов множества $\Omega_n^2$ и реализующих преобразование $f$.        
    \end{enumerate}
\end{predicate}

\medskip
Используя понятие дополнительной памяти, можно переформулировать основные леммы этой главы:
\begin{enumerate}
    \item
        Для любой подстановки $h \in S(\ZZ_2^n)$ при $n < 4$ существует задающая её обратимая схема
        без дополнительной памяти, состоящая из элементов множества $\Omega_n^2$ (Лемма~\ref{basis_on_s_n}).
    
    \item
        Для любой чётной подстановки $h \in A(\ZZ_2^n)$ при $n \geqslant 4$ существует задающая её обратимая схема
        без дополнительной памяти, состоящая из элементов множества $\Omega_n^2$ (Лемма~\ref{basis_on_a_n}).
    
    \item
        Для любой нечётной подстановки $h \in S(\ZZ_2^n)$ при $n \geqslant 4$ не существует реализующей её
        обратимой схемы без дополнительной памяти, состоящей из элементов множества $\Omega_n^2$
        (следствие~\ref{corollary_parity_of_scheme}).
\end{enumerate}

\begin{predicate}\label{predicate_scheme_with_additional_memory}
    Для любой нечётной подстановки $h \in S(\ZZ_2^n)$, $n \geqslant 4$, существует реализующая её обратимая схема
    с одним дополнительным входом, состоящая из элементов множества $\Omega_n^2$.
\end{predicate}
\begin{proof}
    Подстановке $h$ соответствует некоторое преобразование $f_h\colon \ZZ_2^n \to \ZZ_2^n$.
    Введём два множества: $\ZZ^{(0)} = \{\,\vv x \in \ZZ_2^{n+1}\mid x_{n+1} = 0\,\}$ и
    $\ZZ^{(1)} = \{\,\vv x \in \ZZ_2^{n+1}\mid x_{n+1} = 1\,\}$.
    Очевидно, что $\ZZ_2^{n+1} = \ZZ^{(0)} \cup \ZZ^{(1)}$,
    $\ZZ^{(0)} \cap \ZZ^{(1)} = \varnothing$.

    Зададим отображение $\phi\colon \ZZ_2^n \times \ZZ_2 \to \ZZ_2^{n+1}$ следующим образом:
    $$
        \phi(\langle x_1, \ldots, x_n \rangle, y) = \langle x_1, \ldots, x_n, y \rangle \; ,
    $$
    где $x_i, y \in \ZZ_2$.

    Зададим преобразование $g_i\colon \ZZ^{(i)} \to \ZZ^{(i)}$ следующим образом:
    $$
        g_i(\phi(\vv x, i)) = \phi(f_h(\vv x), i) \; ,
    $$
    где $\vv x \in \ZZ_2^n$, $i = 0$ или $1$.
    Подстановка $h_{g_i}$, задаваемая преобразованием $g_i$, принадлежит симметрической группе $S(\ZZ^{(i)})$.
    Очевидно, что при таком построении верно равенство
    $$
        h_{g_i}(\phi(\vv x, i)) = \phi(h(\vv x), i) \; ,
    $$
    где $\vv x \in \ZZ_2^n$. Следовательно, подстановки $h_{g_0}$, $h_{g_1}$ и $h$ имеют одинаковую цикловую структуру,
    а значит, и одинаковую чётность.

    Зададим преобразование $g\colon \ZZ_2^{n+1} \to \ZZ_2^{n+1}$ следующим образом:
    $$
        g(\langle x_1, \ldots, x_{n+1} \rangle) = \bar x_{n+1} \wedge g_0(\langle x_1, \ldots, x_n, 0 \rangle) \oplus
        x_{n+1} \wedge g_1(\langle x_1, \ldots, x_n, 1 \rangle) \; .
    $$
    Очевидно, что подстановка $h_g$, задаваемая преобразованием $g$,
    равна $h_g = h_{g_0} \circ h_{g_1} \Rightarrow$ она является чётной.

    Согласно Лемме~\ref{basis_on_a_n}, для подстановки $h_g \in A(\ZZ_2^{n+1})$ существует задающая её обратимая схема
    с $(n+1)$ входами, состоящая из элементов множества $\Omega_n^2$. А по определению~\ref{def_scheme_with_additional_memory}
    эта же схема будет реализовывать преобразование $f_h$, соответствующее нечётной подстановке $h$.
    Вход схемы с номером $(n+1)$ и будет тем самым одним дополнительным входом из условия.
\end{proof}

В этом доказательстве показан лишь один из способов построить обратимую схему $\frS_g$,
реализующую преобразование, соответствующее нечётной подстановке $h \in S(\ZZ_2^n)$, с использованием одного дополнительного
входа. Сложность схемы $\frS_g$ будет зависеть в том числе и от того, как задано преобразование $g$.
В последующих главах будет изучаться вопрос зависимости сложности и глубины обратимой схемы от вида преобразования $g$
и от количества используемых дополнительных входов.

\bigskip
В заключение данной главы рассмотрим вопрос реализации сюръективного отображения $f\colon \ZZ_2^n \to \ZZ_2^m$, $m < n$.
Введём множество $A_{\vv y}=\{\,\vv x \in \ZZ_2^n\mid f(\vv x) = \vv y \in \ZZ_2^m\,\}$.
Обозначим $d = \max\limits_{\vv y}{|A_{\vv y}|}$.

\begin{predicate}\label{predicate_max_additional_outputs}
    Не существует обратимой схемы, состоящей из элементов множества $\Omega_n^2$, реализующей сюръективное отображение $f$
    с $q < \lceil \log_2 d \rceil$ дополнительными входами.
\end{predicate}
\begin{proof}
    Докажем от противного.
    Пусть существует обратимая схема $\frS_g$, состоящая из элементов множества $\Omega_n^2$
    и реализующая сюръективное отображение $f$ с $q < \lceil \log_2 d \rceil$ дополнительными входами.

    Существует множество $A_{\vv y}$, мощность которого равна $d$: $|A_{\vv y}| = d$, $\vv y \in \ZZ_2^m$.
    Определим множество $A \subseteq \ZZ_2^{m+q}$ следующим образом:
    $$
        A = \{\,\vv x = \langle x_1, \ldots, x_n, 0, \ldots, 0 \rangle \in \ZZ_2^{m+q} \mid
            \langle x_1, \ldots, x_n \rangle \in A_{\vv y}\,\} \; .
    $$

    Рассмотрим булево преобразование $g\colon \ZZ_2^{m+q} \to \ZZ_2^{m+q}$,
    задаваемое схемой $\frS_g$.
    Для всех $\vv x \in A$ верно следующее равенство:
    \begin{align*}
        &g(\vv x) = \langle y_1, \ldots, y_m, z_1, \ldots, z_q \rangle \; , \\
        &\vv y = \langle y_1, \ldots, y_m \rangle, z_i \in \ZZ_2 \; .
    \end{align*}
    Отсюда следует, что мощность множества значений преобразования $g$ на множестве $A$
    $$
        |g[A]| \leqslant 2^q \; .
    $$
    При этом $|A| = |A_{\vv y}| = d > 2^q \Rightarrow$ на множестве $A$ преобразование $g$ сюръективно
    $\Rightarrow g$ не биективно, а значит схема $\frS_g$ не существует.

    Пришли к противоречию, следовательно, доказываемое утверждение верно.
\end{proof}

\sectionenumerated{Cинтез обратимых схем}

\forceindent В данной главе будут рассмотрены основные переборные и непереборные алгоритмы синтеза обратимых схем,
состоящих из элементов NOT и $k$-CNOT; по возможности будет приведено их краткое описание и основные характеристики.
Далее будет приведена сравнительная таблица алгоритмов синтеза по таким основным характеристикам, как временн\'{а}я сложность
алгоритма, требуемое для работы алгоритма количество памяти, количество дополнительных входов синтезированной схемы
и её сложность.
В конце главы будет представлен новый, основанный на использовании теории групп подстановок алгоритм синтеза обратимой схемы,
состоящей из элементов NOT, CNOT и 2-CNOT, реализующей заданную чётную подстановку $h \in A(\ZZ_2^n)$ при $n > 3$
со сложностью $O(n 2^m)$, где $m = \lceil \log_2|M| \rceil$,
$M = \{\,\vv x\mid h(\vv x) \ne \vv x\,\}$~--- множество подвижных точек преобразования $h$.

В настоящей работе не будут рассматриваться алгоритмы синтеза обратимых схем, реализующих частично заданные булевы функции
(например, алгоритм, описанный в работе~\cite{wave_synthesis}), и алгоритмы синтеза, использующие отличные от NOT и $k$-CNOT
обратимые элементы (к примеру, \gate{} \textit{Kerntopf} в работе~\cite{regularity_symmetry}).
Последнее ограничение связано с тем, что задаваемая любым обратимым \gate{} подстановка на множестве $\ZZ_2^n$ может быть
представлена в виде произведения транспозиций, задаваемых элементами NOT и $k$-CNOT
(см. доказательство Лемм~\ref{basis_on_s_n} и~\ref{basis_on_a_n}).

Сперва рассмотрим алгоритмы, главной задачей которых является снижение сложности обратимых схем.
Несмотря на то, что алгоритмы такого типа не предназначены для синтеза схем, они нередко являются составной частью более
сложных алгоритмов синтеза и применяются на последнем этапе для снижения сложности синтезированной схемы.

\subsection{Алгоритмы снижения сложности схем}
\label{subsection_optimization_algorithms}
\refstepcounter{synthalgchapter}

\forceindent В работе~\cite{iwama_transform_rules} было описано множество простых, но нетривиальных правил замен
композиций обратимых \gate{} определённого вида на эквивалентные им композиции \gate{} в обратимых схемах,
состоящих из элементов NOT и $k$-CNOT.
Также авторами было показано, что приведённое множество замен является полным, т.\,е. для любых двух эквивалентных схем
$\frS_1$ и $\frS_2$, задающих одно и то же булево преобразование, существует последовательность из предложенных замен,
которая приводит $\frS_1$ к $\frS_2$. Там же был предложен алгоритм~\nextalg{alg_iwama} снижения сложности
обратимых схем с применением данного множества замен.

Алгоритм~\algref{alg_iwama} работает только с обратимыми схемами $\frS_f$,
задающими преобразование $f\colon \ZZ_2^n \to \ZZ_2^n$ следующего вида:
\begin{equation}
    f(\langle x_1, \ldots, x_n\rangle) =  \langle x_1, \ldots, x_{n-1},
        x_n \oplus \phi(\langle x_1, \ldots, x_{n-1}\rangle)  \rangle,
    \text{ где }\phi\colon \ZZ_2^{n-1} \to \ZZ_2 \; .
    \label{formula_transf_for_bool_func_with_one_output}
\end{equation}

Из формулы~\eqref{formula_transf_for_bool_func_with_one_output} следует, что алгоритм~\algref{alg_iwama} не применим
к обратимым схемам, задающим преобразование произвольного вида.
Тем не менее, при разработке данного алгоритма авторами был рассмотрен вопрос эквивалентной замены композиций \gate{}
В большинстве случаев такая замена требуется, когда необходимо поменять два обратимых \gate{} местами.
Все эти замены или часть из них в том или ином виде используются во многих алгоритмах синтеза и оптимизации схем из
логических элементов NOT и $k$-CNOT, рассматриваемых далее.

Приведём здесь предложенные авторами правила эквивалентной замены композиций обратимых \gate{}
Композиция элементов $E(t_1,I_1) *E(t_2,I_2)$ может быть:
\begin{itemize}
    \item[\nextopt\label{opt_duplicates})]
        исключена из схемы без изменения результирующего преобразования,
        если $I_1 = I_2$, $t_1 = t_2$ (случай двух одинаковых \gate);

    \item[\nextopt\label{opt_full_independent})]
        заменена на $E(t_2,I_2)*E(t_1,I_1)$,
        если $t_1 \notin I_2$, $t_2 \notin I_1$;

    \item[\nextopt\label{opt_left_target_in_rigth_controls})]
        заменена на
        $E(t_2,I_2)*E(t_1,I_1)*E(t_1,I_1 \cup I_2 \setminus \{\,t_2\,\})$, если $t_1 \notin I_2$, $t_2 \in I_1$;

    \item[\nextopt\label{opt_right_target_in_left_controls})]
        заменена на
        $E(t_2,I_1 \cup I_2 \setminus \{\,t_1\,\})*E(t_2,I_2)*E(t_1,I_1)$, если $t_1 \in I_2$, $t_2 \notin I_1$;

    \item[\nextopt\label{opt_add_controls})]
        заменена на $E(t_1,I_1)*E(t_2,\{\,t_1\,\} \cup \{\,I_2 \setminus I_1\,\})$,
        если $t_1 \notin I_2$, $t_2 \notin I_1$, $I_1 \subseteq I_2$ и $x_{t_1} \equiv 0$
        (значение на $t_1$-м входе элемента $E(t_1,I_1)$) для всех возможных значений на входах схемы;
\end{itemize}

Ещё одно правило эквивалентной замены \nextopt\label{opt_null_control}: элемент $E(j,I)$ можно исключить из схемы
без изменения результирующего преобразования, если существует такое значение $i \in I$, что $x_i \equiv 0$
(значение на $i$-м входе данного элемента) для всех входных значений схемы.

В главе~\ref{section_complexity_minimization} эти правила эквивалентных замен будут расширены для случая
инвертированных контролирующих входов (для элементов $E(t, I, J)$).

\subsection{Алгоритмы полного перебора}
\refstepcounter{synthalgchapter}

\forceindent
В работе~\cite{iterative_compositions} был представлен переборный алгоритм~\nextalg{alg_khlopotine} синтеза обратимой схемы,
дающей на одном из выходов результат заданной булевой функции над входами. Булево преобразование, задаваемое такой схемой,
с точностью до перенумерации входов описано в разделе алгоритма~\algref{alg_iwama} на
с.~\pageref{formula_transf_for_bool_func_with_one_output}. 
Основной задачей данного алгоритма является уменьшение количества дополнительных входов синтезированной схемы.

Алгоритм~\algref{alg_khlopotine} основан на <<весовых функциях>> из теории информации (терминология авторов)
и предсказании наилучшего решения на один шаг вперёд.
Рабочей единицей алгоритма является каскад из $N$ обратимых \gate{} ($N$~--- параметр алгоритма,
задаётся вручную). Авторами вводится понятие \textit{наилучшего элемента}~--- \gate,
дающего максимальное количество совпадений
минимальных конъюктивных форм (\textit{минтерм}) полученной булевой фунцкции и заданной функции.
Вкратце принцип работы алгоритма можно описать следующим образом:
\begin{enumerate}
    \item \label{alg_khlopotine_first_step}
        Получить каскад из $N$ наилучших элементов (\textit{текущий} каскад).

    \item
        Для каждого из \gate{} текущего каскада ищутся $N$ наилучших элементов (\textit{следующий} каскад).
        Среди всех \gate{} из следующих каскадов ищется наилучший.

    \item \label{alg_khlopotine_last_step}
        \gate{} из текущего каскада, соответствующий найденному на предыдущем шаге наилучшему элементу из следующих каскадов,
        включается в схему; следующий каскад включённого \gate{} становится текущим.

    \item
        Шаги \ref{alg_khlopotine_first_step}--\ref{alg_khlopotine_last_step} повторяются,
        пока синтезированная обратимая схема не реализует заданную функцию.
\end{enumerate}

Авторами алгоритма не оговаривается, как выбрать параметр $N$, как зависит время синтеза от величины $N$
и всегда ли алгоритм способен синтезировать схему для любой заданной функции.
Также остаётся неясным, какова сложность синтезированной схемы. Предположительно,
алгоритму~\algref{alg_khlopotine} для синтеза схемы требуется объём памяти, намного превосходящий $N^2$
(хранение всех возможных пар каскадов), а временн\'{а}я сложность намного превосходит $(N^2)^l$, где $l$~--- сложность
синтезированной схемы.

\bigskip
\bigskip
В работе~\cite{shende_synthesis} рассматривается вопрос синтеза обратимых схем, состоящих из элементов NOT и $k$-CNOT
и задающих какую-либо чётную подстановку. Авторами данной работы конструктивно доказывается, что любая чётная подстановка
может быть задана обратимой схемой без дополнительной памяти, состоящей из элементов NOT, CNOT и 2-CNOT,
а любая нечётная~--- обратимой схемой с одним дополнительным входом, состоящей из элементов NOT, CNOT и 2-CNOT.

Также авторами данной работы предложен алгоритм~\nextalg{alg_shende} синтеза обратимых схем с \textit{минимальной} 
сложностью, состоящих из элементов NOT, CNOT и 2-CNOT.
Данный алгоритм основан на \textit{поиске в глубину с итеративным углублением}. Для его работы необходимо построить библиотеку
минимальных обратимых схем сложности $L \leqslant m$ и задаваемых ими преобразований.
Синтез схемы, задающей требуемое булево преобразование, осуществляется следующим образом:
\begin{itemize}
    \item
        проверяется наличие заданного преобразования в библиотеке минимальных схем; в случае его присутствия,
        из библиотеки извлекается соответствующая ему минимальная схема;
    
    \item
        в случае отсутствия схемы в библиотеке, начинается поиск в глубину с шагом $k$ по количеству \gate{} в схеме,
        $1 \leqslant k \leqslant m$:
        \begin{itemize}
            \item
                из библиотеки извлекается очередная минимальная схема сложности $k$;
            \item
                к выходам этой схемы подключаются всеми возможными способами элементы NOT, CNOT и 2-CNOT;
            \item
                сверяется требуемое преобразование с преобразованием, задаваемым полученной схемой;
                в случае совпадения этих преобразований все подсхемы полученной схемы заменяются на минимальные из библиотеки
                (если существуют).
        \end{itemize}
\end{itemize}
\noindent
Так же алгоритмом предусмотрена принудительная остановка поиска в глубину,
если сложность синтезированной схемы начинает превышать заранее заданную величину. 

Таким образом, алгоритм~\algref{alg_shende}, благодаря библиотеке минимальных схем, является более быстрым по сравнению
с алгоритмом полного перебора. Однако данный алгоритм не применим на практике для синтеза сложных схем, т.\,к. с ростом
количества входов схемы и сложности синтезируемой схемы экспоненциально растёт размер библиотеки минимальных схем
и время синтеза. Вторым существенным ограничением данного алгоритма является то, что для некоторых заданных преобразований
он не сможет синтезировать схему из-за принудительной остановки поиска в глубину, описанного выше.

Временн\'{а}я сложность алгоритма~\algref{alg_shende} намного превышает
величину $N^l$, где $l$~--- сложность синтезированной схемы, $N$~--- количество схем в библиотеке минимальных схем;
требуемый для синтеза объём памяти составляет порядка $O(mN)$ (хранение библиотеки минимальных схем), где $m$~--- максимальная
сложность обратимой схемы из библиотеки.

\bigskip
\bigskip
В работе~\cite{fast_synthesis_exact_minimal} был представлен быстрый алгоритм~\nextalg{alg_fast_exact} синтеза обратимых схем
с минимальной сложностью, состоящих из \gate{} заранее сформированной библиотеки элементов.
Данный алгоритм использует математическое программное обеспечение \textit{GAP} (\textit{Groups, Algorithms, Programming}),
представляющее собой пакет программ для вычислительной дискретной алгебры.
Авторами показывается, что задачу синтеза обратимой схемы можно свести к некоторой задаче из теории групп подстановок,
которую и решает \textit{GAP}.

Алгоритм~\algref{alg_fast_exact}, также как и алгоритм~\algref{alg_shende}, является переборным, однако применение
теории групп подстановок и специализированного ПО позволяет решать задачу синтеза обратимых схем,
имеющих минимальную сложность, за приемлемое время при количестве входов схемы $n=3$.
К сожалению, при б\'{о}льшем значении $n$ данный алгоритм не применим из-за чрезмерно продолжительного времени синтеза.
Также не даётся никаких оценок на требуемый алгоритмом объём памяти для синтеза схемы.

\subsection{Непереборные алгоритмы}
\refstepcounter{synthalgchapter}

\forceindent
В работах~\cite{miller_spectral} и~\cite{miller_spectral_two_place} был представлен непереборный
алгоритм~\nextalg{alg_miller_spectral} синтеза обратимых схем с близкой
к минимальной сложностью, использующий спектральный метод Радемахера-Уолша.

Для заданной функции $f(x_1, \ldots, x_n)$ спектр Радемахера-Уолша $\vv R = \vv T^n \bar f$,
где $\bar f$~--- столбец значений функции $f$, $\vv T^n$~--- матрица Адамара:
\begin{align*}
    \vv T^0 &=
    \left [ \begin{matrix}
        1
    \end{matrix} \right ] \; , \\
    \vv T^n &=
    \left [ \begin{matrix}
        \vv T^{n - 1} & \vv T^{n - 1} \\
        \vv T^{n - 1} & -\vv T^{n - 1}
    \end{matrix} \right ] \; .
\end{align*}

В качестве базовой меры сложности $C(f)$ функции $f$ авторами используется количество смежных нулей и смежных единиц
на карте Карно. Как было показано в работе~\cite{hurst}, значение $C(f)$ равно
$$
    C(f) = \frac{1}{2}\left( n2^n - \frac{1}{2^{n-2}}\sum_{v = 0}^{2^n - 1} w(\vv v) \vv r_v^2 \right) \; ,
$$
где $w(\vv v)$~--- вес двоичного вектора $\vv v$, $\vv r_v$~--- $v$-я координата спектра $\vv R$.

В качестве основной меры сложности функции $f$ авторы используют величину $D(f) = n2^n NZ(\vv R) + C(f)$,
где $NZ(\vv R)$~--- количество нулевых коэффициентов в спектре $\vv R$.

Сам алгоритм синтеза~\algref{alg_miller_spectral} работает следующим образом. Для заданного преобразования
$f\colon \ZZ_2^n \to \ZZ_2^n$ строится система выходных функций $f_i(x_1, \ldots, x_n)$ и их спектров
$\vv R_i$, $ 1 \leqslant i \leqslant n$. Для каждого спектра $\vv R_i$ запускается алгоритм поиска следующего \gate:
\begin{enumerate}
    \item
        Для каждого из ${\binom{n}{1}}(n-1)$ всех возможных элементов $C_{i_1;j}$ рассмотреть изменение значения $D(f_i)$
        и выбрать тот \gate, для которого это изменение будет максимальным положительным.

    \item
        Если на предыдущем шаге не был выбран \gate, для каждого из ${\binom{n}{2}}(n-2)$ всех возможных элементов
        $C_{i_1,i_2;j}$ рассмотреть изменение значения $D(f_i)$ и выбрать тот \gate, для которого это изменение будет
        максимальным положительным.

    \item
        Если на предыдущем шаге не был выбран \gate, для каждого из ${\binom{n}{3}}(n-3)$ всех возможных элементов
        $C_{i_1,i_2,i_3;j}$ рассмотреть изменение значения $D(f_i)$ и выбрать тот \gate, для которого это изменение будет
        максимальным положительным.

    \item
        Если ни один \gate{} не выбран на предыдущих шагах, алгоритм заканчивается с ошибкой.
        Иначе $\vv R'_j = \vv R_j$, $j \ne i$, где $x_i$~--- переменная, изменяемая выбранным \gate{}
        Новое значение $\vv R'_i$ вычисляется на основе выбранного \gate{}
\end{enumerate}

Из данного описания не ясно, почему алгоритм может завершиться с ошибкой на последнем шаге.
Если ошибки можно избежать, рассматривая элементы $k$-CNOT с б\'{о}льшим количеством контролирующих входов, то время синтеза схемы
будет расти экспоненциально и в случае рассмотрения всех возможных элементов $k$-CNOT составит $O(2^{nl})$,
где $l$~--- сложность синтезированной схемы.
Для обратимых схем с $l \sim n$ время синтеза уже будет порядка $O \left( 2^{n^2} \right )$.
Требуемый для синтеза объём памяти составляет порядка $O(n2^n)$ (хранение спектров $n$ функций).

\bigskip
\bigskip
В работе~\cite{maslov_rm_synthesis} был описан похожий алгоритм синтеза~\nextalg{alg_maslov_rm},
использующий спектр Рида-Маллера (вектор коэффициентов многочлена Жигалкина).

Для заданной функции $f(x_1, \ldots, x_n)$ спектр Рида-Маллера $\vv {Rm}(f) = \vv M^n \bar f$,
где $\bar f$~--- столбец значений функции $f$, $\vv M^n$~--- матрица следующего вида:
\begin{align*}
    \vv M^0 &=
    \left [ \begin{matrix}
        1
    \end{matrix} \right ] \; , \\
    \vv M^n &=
    \left [ \begin{matrix}
        \vv M^{n - 1} & 0 \\
        \vv M^{n - 1} & \vv M^{n - 1}
    \end{matrix} \right ] \; .
\end{align*}

В качестве базовой меры сложности $C(f)$ функции $f$ авторами используется количество различных коэффициентов в спектре
$\vv {Rm}(f)$ и в спектре Рида-Маллера тождественной функции. Сам алгоритм синтеза можно описать следующим образом:
имея таблицу истинности для заданного булева преобразования $F\colon \ZZ_2^n \to \ZZ_2^n$,
модифицировать каждую строку этой таблицы, начиная с первой, таким образом, чтобы она соответствовала строке в таблице
истинности тождественного булева преобразования. При этом каждой такой модификации строк ставится в соответствие
некоторая композиция обратимых \gate{} из множества $\Omega_n^2$.

Авторами доказывается, что алгоритм~\algref{alg_maslov_rm} всегда синтезирует обратимую схему для любого
заданного булева преобразования. Временн\'{а}я сложность алгоритма равна $O(n2^n)$ (просмотр всей таблицы истинности),
сложность синтезированной схемы $\frS$ равна $L(\frS) \lesssim 5n2^n$,
требуемый для синтеза объём памяти составляет порядка $O(n2^n)$ (хранение таблицы истинности).

\bigskip
\bigskip
В работе~\cite{miller_transform_based} был представлен непереборный алгоритм~\nextalg{alg_transform_based} синтеза обратимых
схем, состоящих из элементов NOT и $k$-CNOT. В качестве входа алгоритм принимает таблицу истинности для заданного обратимого
преобразования $f\colon \ZZ_2^n \to \ZZ_2^n$. Его работу можно описать следующим образом
($\vv 0 = \langle 0, \ldots, 0 \rangle$~--- нулевой вектор):
\begin{enumerate}
    \item
        Если $f(\vv 0) = \vv y \ne \vv 0$, то для каждого $y_j = 1$ добавить в схему элемент
        $N_j$ и инвертировать в таблице истинности $j$-й выходной столбец.

    \item
        Для каждого ненулевого вектора $\vv x \ne \vv 0$ такого,
        что $f(\vv x) = \vv y$, $\vv y \ne \vv x$:
        \begin{itemize}
            \item
                построить множество индексов $I_x = \{\,i\mid x_i = 1\,\}$;
            \item
                построить множество индексов $I_y = \{\,i\mid y_i = 1\,\}$;
            \item
                построить множество индексов $J_p = \{\,i\mid x_i = 1, y_i = 0\,\}$;
            \item
                построить множество индексов $J_q = \{\,i\mid x_i = 0, y_i = 1\,\}$;
            \item
                для каждого индекса $p \in J_p$ добавить в схему элемент $C_{I_y;p}$ и инвертировать соответствующие значения
                в таблице истинности;
            \item
                для каждого индекса $q \in J_q$ добавить в схему элемент $C_{I_x;q}$ и инвертировать соответствующие значения
                в таблице истинности;
        \end{itemize}
            
        В итоге для каждого значения вектора $\vv x$ получаем на каждом шаге функцию $f$ такую,
        что $f(\vv x') = \vv x'$ для всех векторов $\vv x'$, удовлетворяющих условию
        $$
            \sum_{k = 1}^n {x'_k 2^{k-1}} \leqslant \sum_{k = 1}^n {x_k 2^{k-1}} \; .
        $$
\end{enumerate}

В конце работы алгоритма $f$ представляет собой тождественное преобразование. Расположив \gate{} синтезированной схемы
в обратном порядке, можно получить обратимую схему, задающую преобразование $f$.

Время работы алгоритма составляет порядка $O(n2^n)$, результат синтеза гарантирован. Сложность синтезированной схемы,
по словам авторов, не превосходит $(n-1)2^n + 1$, для любого значения $n$ можно построить таблицу истинности таким образом,
что сложность синтезированной схемы будет равна $(n-1)2^n + 1$.
Требуемый для синтеза объём памяти составляет порядка $O(n2^n)$ (хранение таблицы истинности).

Для снижения сложности синтезированной схемы авторы предлагают четыре способа оптимизации алгоритма:
\begin{itemize}
    \item
        использование перестановки выходов схемы для минимизации суммы расстояний Хемминга для всех пар вход/выход
        (в этом случае, правда, синтезированная схема уже не задаёт искомое преобразование, а реализует его; см.~определение
        на с.~\pageref{def_scheme_with_additional_memory});

    \item
        уменьшение количества контролирующих входов получаемых элементов $k$-CNOT на каждом шаге алгоритма;
        критерием выбора этого количества является минимизация суммы расстояний Хемминга для всех пар вход/выход;

    \item
        проход не только от начала таблицы истинности преобразования к её концу, но и от её конца к началу на каждом шаге
        алгоритма; критерий выбора наилучшего прохода определяется также, как и в предыдущем пункте;

    \item
        использование таблицы эквивалентных замен композиций \gate{}
\end{itemize}

По заверениям авторов, алгоритм~\algref{alg_transform_based} с учётом всех способов его оптимизации
синтезирует схему со сложностью, близкой к минимальной, если количество входов схемы равно трём;
при большем количестве входов схемы алгоритм даёт очень хороший результат синтеза.
Тем не менее многие из перечисленных способов оптимизации алгоритма имеют существенный недостаток.
Так, к примеру, использование перестановки выходов (первый способ) влечёт за собой необходимость рассмотрения $n!$ подстановок,
что для очень больших значений $n$ потребует колоссальных временн\'{ы}х трудозатрат.
Во-вторых, уменьшение количества контролирующих входов элементов $k$-CNOT на каждом шаге алгоритма (второй способ)
влечёт за собой необходимость рассмотрения порядка $O(2^n)$ \gate, что в сумме на весь алгоритм может дать порядка
$O(2^{2n})$ операций.
И наконец, последний способ оптимизации требует как большое количество памяти для хранения таблицы эквивалентных замен
композиций \gate, так и большое время поиска этих композиций в синтезированной схеме.
В итоге, уже для схемы с 7-ю входами, содержащей 12 \gate, суммарное время синтеза
\label{sythesis_time_for_alg_transform_based} составляет порядка двух секунд~\cite{miller_transform_based}.

Ещё одним существенным недостатком алгоритма~\algref{alg_transform_based} является необходимость рассматривать таблицу из
$2^n$ значений, даже если эта таблица задаёт всего одну транспозицию.
Из-за этого при больших значениях $n$ данный алгоритм уже может быть не применим на практике.
В работе~\cite{alhagi_large_functions} было предложено усовершенствование данного алгоритма для больших значений $n$,
которое заключается в переупорядочивании строк в таблице. Это позволяет избавиться от хранения таблицы,
но не от её рассмотрения.

\bigskip
\bigskip
Похожий алгоритм~\nextalg{alg_novel} синтеза обратимых схем был представлен в работе~\cite{saeedi_novel}.
Как и в предыдущем случае, алгоритм принимает на вход таблицу истинности для заданного обратимого преобразования
$f\colon \ZZ_2^n \to \ZZ_2^n$. Однако синтез схемы происходит упорядочиванием в таблице истинности
не минтерм целиком (минимальных конъюктивных форм), а их разрядов. Псевдокод алгоритма следующий:
\vspace{\bigskipamount}\\
\phantom{}\quad \textbf{Вход}:
    таблица истинности для обратимого преобразования $f(\langle x_1, \ldots, x_n \rangle) = \langle y_1, \ldots, y_n \rangle$.\\
\phantom{}\quad \textbf{Выход}:
    обратимая схема из элементов $k$-CNOT, задающая преобразование $f$.\\
\phantom{}\quad \textbf{Обозначения}:
    $i$-я входная (выходная) переменная $j$-й минтермы обозначается $x_{i,mj}$ ($y_{i,mj}$);
    $i$-я минтерма $j$-й входной (выходной) переменной обозначается $m_{i,xj}$ ($m_{i,yj}$).
\\
\texttt{\small
    \phantom{}\quad $\mathtt{i = 1}$\\
    \phantom{}\quad повторить:\\
    \phantom{}\qquad пометить все $\mathtt{2^n}$ минтерм как непосещённые\\
    \phantom{}\qquad для каждой минтермы $\mathtt{m_j}$, $\mathtt{j = 1, \ldots, 2^n}$ делать:\\
    \phantom{}\quad\qquad если минтерма $\mathtt{m_j}$ не помечена, как посещённая, тогда:\\
    \phantom{}\qquad\qquad если $\mathtt{y_{i,mj} \ne x_{i,mj}}$, тогда:\\
    \phantom{}\quad\qquad\qquad начало блока\\
    \phantom{}\qquad\qquad\qquad пометить минтерму $\mathtt{m_{j,yi}}$, как посещённую\\
    \phantom{}\qquad\qquad\qquad найти минтерму $\mathtt{m_{k,yi}}$, отличающуюся от
        $\mathtt{m_{j,yi}}$ в $\mathtt{i}$-й переменной\\
    \phantom{}\qquad\qquad\qquad если $\mathtt{m_{k,yi}}$ находится ниже $\mathtt{m_{j,yi}}$, тогда:\\
    \phantom{}\quad\qquad\qquad\qquad поменять местами $\mathtt{m_{j,yi}}$ и $\mathtt{m_{k,yi}}$
        (получим $\mathtt{y_{i,mj} = x_{i,mj}}$)\\
    \phantom{}\quad\qquad\qquad\qquad пометить минтерму $\mathtt{m_{k,yi}}$, как посещённую\\
    \phantom{}\qquad\qquad\qquad иначе ($\mathtt{m_{k,yi}}$ находится выше $\mathtt{m_{j,yi}}$):\\
    \phantom{}\quad\qquad\qquad\qquad если $\mathtt{y_{p,mk} \ne y_{p,mk}}$, $\mathtt{p = 1, \ldots, n}$ хотя бы для одного
        $\mathtt{p}$, тогда:\\
    \phantom{}\qquad\qquad\qquad\qquad поменять местами $\mathtt{m_{j,yi}}$ и $\mathtt{m_{k,yi}}$
        (получим $\mathtt{y_{i,mj} = x_{i,mj}}$)\\
    \phantom{}\qquad\qquad\qquad\qquad пометить минтерму $\mathtt{m_{k,yi}}$, как посещённую\\
    \phantom{}\quad\qquad\qquad конец блока\\
    \phantom{}\qquad получить элементы $\mathtt{k}$-CNOT для произведённых перестановок минтерм\\
    \phantom{}\qquad $\mathtt{i = i+1 \text{ } mod \text{ } n}$\\
    \phantom{}\quad пока не достигнуто условие $\mathtt{y_i = x_i}$ для всех $\mathtt{i \in (1, \ldots, n)}$
}

\bigskip
Как видно из этого описания, алгоритм~\algref{alg_novel} также, как и алгоритм~\algref{alg_transform_based},
оперирует таблицей истинности, содержащей $2^n$ строк. Следовательно, при больших значениях $n$ данный алгоритм также может
быть не применим на практике.
Сложность $l$ синтезированной алгоритмом~\algref{alg_novel} схемы, по заверениям авторов, близка к минимальной;
временн\'{а}я сложность алгоритма равна $O(l 2^n)$.
Требуемый для синтеза объём памяти составляет порядка $O(n2^n)$ (хранение таблицы истинности).

\bigskip
\bigskip
В работе~\cite{group_based} авторами конструктивно доказывается, что любая чётная подстановка $h \in A(\ZZ_2^n)$ при $n > 3$
может быть задана обратимой схемой, состоящей из элементов NOT и 2-CNOT, без использования дополнительных входов.
На основе этого доказательства авторами также был разработан алгоритм~\nextalg{alg_group_based} синтеза обратимой схемы
с $n > 3$ входами, задающей какую-либо чётную подстановку $h \in A(\ZZ_2^n)$. Данный алгоритм основан на теории
групп подстановок.
Суть алгоритма следующая:
\begin{enumerate}
    \item
        Подстановка $h \in A(\ZZ_2^n)$ представляется в виде произведения циклов длины 3:
        $$
            h = \mcirc_i (u,s,t)_i \; .
        $$

    \item
        Каждый цикл $(u,s,t)_i$ представляется в виде произведения циклов длины 3 специального вида,
        называемых \textit{соседними} (терминология авторов):
        $$
            (u,s,t) = \mcirc_j (a_1,a_2,a_3)_j \; .
        $$

    \item
        Каждый \textit{соседний} цикл $(a_1,a_2,a_3)_j$ задаётся предвычисленной композицией элементов NOT (до $n-2$ штук)
        и либо элементов 2-CNOT ($3 \cdot 2^{n-4} - 2$ штук), либо одного элемента $(n-2)$-CNOT.
\end{enumerate}

Как видно из данного описания, алгоритм~\algref{alg_group_based} принимает на вход не таблицу истинности, а подстановку.
Следовательно, в случае если подстановка короткая, отпадает необходимость просматривать и/или хранить всю таблицу истинности
из $2^n$ значений. И только если в подстановке участвуют все $2^n$ двоичных векторов длины $n$, то временн\'{ы}е затраты
и затраты на память становятся такими же, как и для алгоритмов, работающих с таблицей истинности
(\algref{alg_transform_based} и~\algref{alg_novel}).
По словам авторов, временн\'{а}я сложность алгоритма~\algref{alg_group_based} в худшем случае равна $\frac{10}{3} n^2 2^n$,
верхняя оценка сложности синтезированной схемы равна
$\left (n+\lfloor \frac{n}{3} \rfloor\right )\left (3 \cdot 2^{2n-3} - 2^{n+2}\right )$
и $n\left (n+\lfloor \frac{n}{3} \rfloor\right )2^{n+2}$ для 2-CNOT и NOT соответственно;
требуемый для синтеза объём памяти зависит от вида входной подстановки, но не превышает по порядку $O(n2^n)$
(хранение всех элементов подстановки).
Отметим, что авторы данной работы рассматривали разложение элемента ($n-2$)-CNOT в композицию ($3 \cdot 2^{n-4} - 2$)
элементов 2-CNOT по формуле~\ref{formula_k_CNOT_recursive}.
Однако если воспользоваться утверждением~\ref{fast_realization_of_k_cnot}, то тогда верхняя оценка количества обратимых
элементов 2-CNOT снизится до $(n-5)\left (n+\lfloor \frac{n}{3} \rfloor\right )2^{n+4}$.

Достоинством алгоритма~\algref{alg_group_based} является существенно меньшее время синтеза схемы
по сравнению с алгоритмами~\algref{alg_khlopotine}--\algref{alg_fast_exact}
и меньшее количество требуемой для этого памяти по сравнению с алгоритмами~\algref{alg_miller_spectral}--\algref{alg_novel}.
Однако синтезированная схема имеет б\'{о}льшую сложность по сравнению со схемами, синтезированными при помощи
алгоритмов~\algref{alg_miller_spectral}--\algref{alg_novel}.

\subsection{Сравнение алгоритмов синтеза}

\forceindent Рассмотрим обратимое преобразование $f\colon\ZZ_2^n \to \ZZ_2^n$ и множество подвижных точек
этого преобразования $M = \{\,\vv x\mid f(\vv x) \ne \vv x\,\}$.
Обозначим через $m$ величину $\lceil \log_2|M| \rceil$. Очевидно, что $m \leqslant n$.

Преобразование $f$ задаёт какую-либо подстановку $h_f \in S(\ZZ_2^n)$, которую можно представить в виде произведения
транспозиций в количестве не более чем $2^m$ (см. с.~\pageref{formula_cycle_decomposition}).
Это свойство будет использоваться при сравнении основных характеристик алгоритмов синтеза.
Цель сравнения~--- показать их различие при $m = o(n)$.

Будем рассматривать только те преобразования $f$, для которых задаваемая ими подстановка $h_f$ является чётной.
В таблице~\ref{table_alg_comparison} приведено сравнение по основным характеристикам описанных в предыдущих разделах
алгоритмов синтеза обратимых схем, состоящих из элементов NOT и $k$-CNOT и задающих/реализующих преобразование $f$.
Обозначения:
$\mathrm T(A)$~--- временн\'{а}я сложность алгоритма синтеза;
$\mathrm M(A)$~--- требуемый для синтеза объём памяти;
$l=L(\frS_f)$~--- сложность синтезированной схемы;
$\Omega_n^2$~--- множество всех элементов NOT, CNOT и 2-CNOT с $n$ входами;
$\Omega_n$~--- множество всех элементов NOT и $k$-CNOT с $n$ входами ($\Omega_n^2 \subset \Omega_n$).

{
\renewcommand{\baselinestretch}{1.2}

\Table[ht]
    \small
    \centering
    \begin{tabular}{|m{0.95 cm}|m{1.75 cm}|m{1.9 cm}|m{1.9 cm}|m{2 cm}|m{6.2 cm}|}
        \hline
        \centering Алго\-ритм &
        \centering Результат &
        \centering $T(A)$ &
        \centering $M(A)$ &
        \centering $l=L(\frS_f)$ &
        \centering Примечание \tabularnewline
        \hline

        \centering \algref{alg_khlopotine} &
        \centering не гарантирован &
        \centering $\gg \left( N^2 \right )^l$ &
        \centering $N^2$ &
        \centering близкая к минимальной &
        \centering переборный, не универсальный;\\ \gate{} из $\Omega_n$;
            $N$~--- размер каскада \gate, задаётся вручную \tabularnewline
        \hline

        \centering \algref{alg_shende} &
        \centering не гарантирован &
        \centering $\gg N^l$ &
        \centering $O(kN)$ &
        \centering минималь\-ная &
        \centering переборный; \gate{} из $\Omega_n^2$ или $\Omega_n$;
            $N$~--- размер библиотеки минимальных схем со сложностью $< k$ \tabularnewline
        \hline

        \centering \algref{alg_fast_exact} &
        \centering гаранти\-рован &
        \centering неизвестно &
        \centering неизвестно &
        \centering минималь\-ная &
        \centering переборный; \gate{} из $\Omega_n$;\\ использует ПО \textit{GAP} \tabularnewline
        \hline

        \centering \algref{alg_miller_spectral} &
        \centering не гарантирован &
        \centering $O(2^{nl})$ &
        \centering $O(n2^n)$ &
        \centering близкая к минимальной &
        \centering непереборный; \gate{} из $\Omega_n$; использует спектральный метод Радемахера-Уолша \tabularnewline
        \hline

        \centering \algref{alg_maslov_rm} &
        \centering гаранти\-рован &
        \centering $O(n2^n)$ &
        \centering $O(n2^n)$ &
        \centering $O(n2^n)$ &
        \centering непереборный; \gate{} из $\Omega_n^2$; использует спектр Рида-Маллера \tabularnewline
        \hline

        \centering \algref{alg_transform_based} &
        \centering гаранти\-рован &
        \centering $O(n2^n)$ &
        \centering $O(n2^n)$ &
        \centering $O(n2^n)$ &
        \centering непереборный; \gate{} из $\Omega_n$;\\ использует таблицу истинности \tabularnewline
        \hline

        \centering \algref{alg_novel} &
        \centering гаранти\-рован &
        \centering $O(l2^n)$ &
        \centering $O(n2^n)$ &
        \centering близкая к минимальной &
        \centering непереборный; \gate{} из $\Omega_n$;\\ использует таблицу истинности \tabularnewline
        \hline

        \centering \algref{alg_group_based} &
        \centering гаранти\-рован &
        \centering $O(n^2 2^m)$ &
        \centering $O(n2^m)$ &
        \centering $O(n^2 2^m)$ &
        \centering непереборный; \gate{} из $\Omega_n^2$;\\ использует теорию групп подстановок \tabularnewline
        \hline

    \end{tabular}
    \caption{\small Сравнение алгоритмов синтеза обратимых схем.}\label{table_alg_comparison}
\end{table}

} 

Из таблицы~\ref{table_alg_comparison} видно, что только алгоритмы~\algref{alg_maslov_rm}~---~\algref{alg_group_based}
гарантированно дают результат синтеза схемы за приемлемое время. При этом лишь алгоритм~\algref{alg_novel} позволяет получить
схему с близкой к минимальной сложностью.

Тем не менее, для рассматриваемого частного случая $m = o(n)$, ни один из алгоритмов синтеза,
кроме алгоритма~\algref{alg_group_based}, не позволяет получить результат синтеза за время порядка $o(n2^n)$.
Это замечание верно и для объёма памяти, требуемого алгоритмами синтеза.

Ещё одним важным свойством алгоритма~\algref{alg_group_based} является использование в синтезированной схеме \gate{} только
из множества $\Omega_n^2$, но не из расширенного множества $\Omega_n$. В случае, когда в обратимой схеме невозможно использовать
\gate{} из $\Omega_n$, только этот алгоритм и алгоритм~\algref{alg_maslov_rm} позволят гарантированно синтезировать схему,
т.\,к. не всегда возможно привести без использования дополнительных входов схему, состоящую из \gate{} множества $\Omega_n$,
к схеме, состоящей из \gate{} множества $\Omega_n^2$.

Однако существенным недостатком алгоритма~\algref{alg_group_based} является избыточная сложность синтезированной
схемы по сравнению с другими алгоритмами, которая зависит только от $m$ и никак не зависит от вида подстановки $h$.

Далее будет представлен новый быстрый алгоритм синтеза обратимых схем, дающий результат синтеза за время порядка $o(n2^n)$
при $m = o(n)$ и позволяющий получить схему с меньшей сложностью по сравнению с алгоритмом~\algref{alg_group_based}.
Описание этого алгоритма было опубликовано автором в работе~\cite{my_fast_group_based_algorithm}.

\subsection{Новый быстрый алгоритм синтеза}\label{subsection_my_synthesis_algorithms}
\refstepcounter{synthalgchapter}

\forceindent
Рассмотрим произвольную подстановку $h \in S(\ZZ_2^n)$, $n > 3$, и множество подвижных точек этой подстановки
$M = \{\,\vv x\mid \vv x \in \ZZ_2^n, h(\vv x) \ne \vv x\,\}$. Обозначим $m = \lceil \log_2 |M| \rceil$.

Самый простой способ синтеза (алгоритм~\nextalg{alg_my_simple}) обратимой схемы, задающей подстановку $h$,
можно описать следующим образом:
\begin{enumerate}
    \item \label{simple_alg_group_theory_first_step}
        Найти представление заданной подстановки $h \in S(\ZZ_2^n)$ в виде произведения транспозиций.
    \item \label{simple_alg_group_theory_second_step}
        Каждую транспозицию $t = (\vv x, \vv y)$
        путём сопряжения привести к виду $t' = (\vv x', \vv y')$
        таким образом, чтобы нашлось такое значение $j$, что $x'_j = y'_j \oplus 1$ и $x'_i = y'_i = 1$ при $i \ne j$;
        $\vv x, \vv y, \vv x', \vv y' \in \ZZ_2^n$.
    \item \label{simple_alg_group_theory_third_step}
        Найти композиции обратимых \gate{} из множества $\Omega_n$, задающих транспозицию $t'$
        и действие сопряжением на транспозицию $t$.
\end{enumerate}

Подсчитаем максимальное количество транспозиций в представлении подстановки $h$
на шаге~\ref{simple_alg_group_theory_first_step}. 
В случае, когда подстановка $h$ представляет собой один длинный цикл, её можно представить в виде произведения не более
чем $2^m - 1$ транспозиций:
\begin{equation}
    (i_1, i_2, i_3, \ldots, i_{2^m})=(i_1,i_2) \circ (i_1, i_3) \circ \ldots \circ (i_1, i_{2^m}) \; .
    \label{formula_cycle_decomposition}
\end{equation}
В случае, когда $h$ представляет собой произведение нескольких циклов, каждый цикл можно представить в виде произведения
транспозиций по формуле~\eqref{formula_cycle_decomposition},
что при фиксированном значении $m$ даёт максимальное количество транспозиций в этом произведении не более $2^m - 1$.

Подсчитаем количество обратимых элементов, необходимых на шаге~\ref{simple_alg_group_theory_second_step}
и~\ref{simple_alg_group_theory_third_step}.

Действие сопряжением подстановкой $g$ на подстановку $h$ мы обозначили через $h^g = g^{-1} \circ h \circ g$.
На с.~\pageref{predicate_reverse_scheme_particular} было показано, что задаваемые элементами NOT и $k$-CNOT
подстановки являются обратными к самим себе.
Поэтому действие сопряжением подстановкой $g$, задаваемой таким \gate, выражается, как $h^g = g \circ h \circ g$,
что требует ровно 2 \gate{}
Действие сопряжением не меняет цикловой структуры подстановки, поэтому транспозиция $t$ в результате действия сопряжением
всегда будет оставаться одной транспозицией.

Для рассматриваемой транспозиции $t = (\vv x, \vv y)$ введём 4 множества:
$B_{00}$, $B_{01}$, $B_{10}$ и $B_{11}$, где $B_{XY} = \{\,i\mid x_i = X, y_i = Y\,\}$,
Мощности этих множеств обозначим через $b_{00}$, $b_{01}$, $b_{10}$, $b_{11}$ соответственно.
Очевидно, что $b_{00} + b_{01} + b_{10} + b_{11} = n$.

$\vv x \ne \vv y \Rightarrow$ либо $b_{01} \ne 0$, либо $b_{10} \ne 0$.
Рассмотрим 2 случая:\label{conjuntion_on_single_transp}
\begin{enumerate}
    \item
        $b_{01} \ne 0$, $b_{10} \ne 0$.
        
        В этом случае существуют индексы $j \in B_{10}$, $k \in B_{01}$.
        Для каждого $i \in B_{10}, i \ne j$ будем действовать сопряжением на $t$ подстановкой, задаваемой элементом $C_{k;i}$.
        Затем для каждого $i \in B_{01}$ будем действовать сопряжением на полученную транспозицию подстановкой,
        задаваемой элементом $C_{j;i}$.
        На последнем шаге для каждого $i \in B_{00}$ будем действовать сопряжением на полученную транспозицию подстановкой,
        задаваемой элементом $N_i$. В результате получим искомую транспозицию $t' = (\vv x', \vv y')$.

        Для сопряжения требуется $2(b_{10} - 1)$ элементов $C_{k;i}$, $2b_{01}$ элементов $C_{j;i}$ и $2b_{00}$ элементов $N_i$.
        В сумме всего требуется $2(b_{10} + b_{01} + b_{00} - 1)$ элементов NOT и CNOT.
        В худшем случае $b_{11} = 0 \Rightarrow b_{10} + b_{01} + b_{00} = n$.
        Следовательно, для получения транспозиции $t'$ при $b_{01} \ne 0$, $b_{10} \ne 0$ требуется
        не более $2(n - 1)$ элементов NOT и CNOT.

        Транспозиция $t$ задаётся следующей композицией \gate:
        $$
            \frS_t =
            \left( \compose_{\substack{i \in B_{10} \\ i \ne k}} C_{k;i} \right ) *
            \left( \compose_{i \in B_{01}} C_{j;i}                       \right ) *
            \left( \compose_{i \in B_{00}} N_i                           \right ) *
                   \frS_{t'}                                                      *
            \left( \compose_{i \in B_{00}} N_i                           \right ) *
            \left( \compose_{i \in B_{01}} C_{j;i}                       \right ) *
            \left( \compose_{\substack{i \in B_{10} \\ j \ne k}} C_{k;i} \right )  \; ,
        $$
        где схема $\frS_{t'}$ задаёт подстановку $t'$.

    \item
        $b_{01} = 0$ или $b_{10} = 0$.
        
        Без ограничения общности рассмотрим только случай $b_{01} = 0$, $b_{10} \ne 0$.
        В этом случае существует индекс $j \in B_{10}$.
        Сначала действуем сопряжением на $t$ подстановкой, задаваемой элементом $N_j$.
        Затем для каждого $i \in B_{10}, i \ne j$ будем действовать сопряжением на полученную транспозицию подстановкой,
        задаваемой элементом $C_{j;i}$.
        После этого вновь будем действовать сопряжением на полученную транспозицию подстановкой, задаваемой элементом $N_j$.
        На последнем шаге для каждого $i \in B_{00}$ будем действовать сопряжением на полученную транспозицию подстановкой,
        задаваемой элементом $N_i$.
        В результате получим искомую транспозицию $t' = (\vv x', \vv y')$.

        Для сопряжения требуется $2(b_{10} - 1)$ элементов $C_{j;i}$, $4$ элемента $N_j$ и $2b_{00}$ элементов $N_i$.
        В сумме всего требуется $2(b_{10} + b_{00} + 1)$ элементов NOT и CNOT.
        В худшем случае $b_{11} = 0 \Rightarrow b_{10} + b_{00} = n$.
        Следовательно, для получения транспозиции $t'$ при $b_{01} = 0$ или $b_{10} = 0$ требуется
        не более $2(n + 1)$ элементов NOT и CNOT.

        Транспозиция $t$ задаётся следующей композицией \gate:
        $$
            \frS_t =
            N_j *
            \left( \compose_{\substack{i \in B_{10} \\ i \ne j}} C_{j;i} \right ) *
            N_j *
            \left( \compose_{i \in B_{00}} N_i \right ) *
            \frS_{t'} *
            \left( \compose_{i \in B_{00}} N_i \right ) *
            N_j *
            \left( \compose_{\substack{i \in B_{10} \\ i \ne j}} C_{j;i} \right ) *
            N_j \; ,
        $$
        где схема $\frS_{t'}$ задаёт подстановку $t'$.
\end{enumerate}

Сама транспозиция $t'$ задаётся элементом $C_{I;j}$ с $(n-1)$ контролирующими входами:\\
$I = \{\,i\mid 1 \leqslant i \leqslant n, i \ne j\,\}$, $|I| = n - 1$, $j$~--- индекс,
для которого верно равенство $x'_j = y'_j \oplus 1$. Следовательно, для~\ref{simple_alg_group_theory_second_step}-го
и~\ref{simple_alg_group_theory_third_step}-го шага алгоритма~\algref{alg_my_simple}
требуется не более $(2n+3)$ \gate{} из множества $\Omega_n$.

Таким образом, умножая максимально возможное количество транспозиций на сложность реализации одной транспозиции, получаем,
что сложность схемы, синтезированной алгоритмом~\algref{alg_my_simple},
$L \leqslant (2^m - 1)(2(n+1)+1) \lesssim n2^{m+1}$.
Временн\'{а}я сложность алгоритма $T \lesssim n2^{m+1}$: для каждой транспозиции сначала необходимо построить множества
$B_{00}$, $B_{01}$, $B_{10}$ и $B_{11}$, а затем уже приступить к синтезу этой транспозиции.
При этом объём памяти, необходимый для синтеза обратимой схемы, равен $O(2^m)$ (хранение всех элементов подстановки).

По сравнению с алгоритмом~\algref{alg_group_based}, алгоритм~\algref{alg_my_simple} имеет на порядок меньшее время работы
($O(n2^m)$ против $O(n^2 2^m)$) и на порядок меньшую сложность синтезированной схемы
($O(n2^m)$ против $O(n^2 2^m)$), сохраняя при этом такой же объём памяти, требуемый для синтеза ($O(n2^m)$).
Тем не менее, главным недостатком алгоритма~\algref{alg_my_simple} является использование обратимых элементов $(n-1)$-CNOT,
а не только \gate{} из множества $\Omega_n^2$.
В некоторых случаях это является недопустимым, т.\,к. такой \gate{} нельзя заменить на композицию \gate{}
из множества $\Omega_n^2$ без использования дополнительных входов схемы~\cite{shende_synthesis}.
Чтобы избавиться от этого недостатка, необходимо использовать иной подход к синтезу.

\bigskip
\bigskip
Усовершенствованный итоговый алгоритм~\nextalg{alg_my_common} синтеза обратимых схем, состоящих из \gate{}
множества $\Omega_n^2$, предлагаемый в данной главе, основан на доказательстве Леммы~\ref{basis_on_a_n},
согласно которой множество подстановок, задаваемых \gate{} множества $\Omega_n^2$,
генерирует знакопеременную группу $A(\ZZ_2^n)$ при $n > 3$.

Алгоритм~\algref{alg_my_common} позволяет получить для любой чётной подстановки $h \in A(\ZZ_2^n)$ задающую её
обратимую схему, состоящую из элементов множества $\Omega_n^2$. Если задана нечётная подстановка $g \in S(\ZZ_2^n)$,
то можно найти такую чётную подстановку $h \in A(\ZZ_2^{n+1})$, для которой синтезированная алгоритмом~\algref{alg_my_common}
обратимая схема будет реализовывать подстановку $g$. Получить подстановку $h$ из $g$ можно, к примеру, способом,
описанным в утверждении~\ref{predicate_scheme_with_additional_memory}. Синтезированная схема в этом случае будет иметь
один дополнительный вход.

Произведение двух независимых циклов можно представить следующим образом:
\begin{equation}
    (i_1, i_2, \ldots, i_{k_1}) \circ (j_1, j_2, \ldots, j_{k_2}) = (i_1, i_2) \circ (j_1, j_2) \circ
    (i_1, i_3, \ldots, i_{k_1}) \circ (j_1, j_3, \ldots, j_{k_2}) \; .
    \label{formula_decompostion_of_two_cycles}
\end{equation}
Цикл длины $k \geqslant 5$ можно представить следующим образом:
\begin{equation}
    (i_1, i_2, \ldots, i_k) = (i_1, i_2) \circ (i_3, i_4) \circ (i_1, i_3, i_5, i_6, \ldots, i_k) \; .
    \label{formula_decompostion_of_k_cycle}
\end{equation}
Следовательно, имея исходное представление чётной подстановки $h$ в виде произведения независимых циклов и используя
формулы~\eqref{formula_decompostion_of_two_cycles} и~\eqref{formula_decompostion_of_k_cycle}, эту подстановку можно представить
в виде произведения пар транспозиций, из которых только одна будет парой зависимых транспозиций, остальные~--- независимых
транспозиций:
$$
    h = \mcirc_{\vv x_i, \vv y_i \in \ZZ_2^n}{((\vv x_1, \vv y_1) \circ (\vv x_2, \vv y_2))} \; .
$$
Как уже было сказано на с.~\pageref{formula_cycle_decomposition},
любую подстановку $h \in \ZZ_2^n$ можно представить в виде произведения не более чем $2^m - 1$ транспозиций
$\Rightarrow$ количество пар независимых транспозиций не превосходит $2^{m - 1}$.

Введём функцию $\phi\colon \ZZ_2^n \to \ZZ_{2^n}$ следующим образом:
$$
    \phi(\vv x) = \sum_{i = 1}^n {x_i 2^{i-1}} \; ,
$$
где $\vv x \in \ZZ_2^n$, $\phi(\vv x) \in \ZZ_{2^n}$.

Рассмотрим пару независимых транспозиций $p = (\vv x, \vv y) \circ (\vv z, \vv w)$.
Действие сопряжением не меняет цикловой структуры
подстановки, поэтому $p$ в результате действия сопряжением всегда будет оставаться парой независимых транспозиций.
Применяя такие же рассуждения, как и для транспозиции $t = (\vv x, \vv y)$ (см. с.~\pageref{conjuntion_on_single_transp})
\label{independent_pair_first_steps},
приведём пару $p$ действием сопряжения к виду $p^{(1)} = (\vv x^{(1)}, \vv y^{(1)}) \circ (\vv z^{(1)}, \vv w^{(1)})$,
где $\phi(\vv x^{(1)}) = 2^n - 1$, $\phi(\vv y^{(1)}) = 2^n - 1 - 2^{i_1 - 1}$, $i_1$~--- индекс разряда, в котором различаются
вектора $\vv x^{(1)}$ и $\vv y^{(1)}$, $y^{(1)}_{i_1} = 0$, $(\vv z^{(1)}, \vv w^{(1)})$~--- новая транспозиция,
получившаяся в результате действия сопряжением.
Для этого шага потребуется не более $2(n+1)$ \gate{} из множества $\Omega_n^2$.

Затем, применяя такой же подход для векторов $\vv x^{(1)}$ и $\vv z^{(1)}$ из пары $p^{(1)}$,
получаем в результате действия сопряжением новую пару $p^{(2)} = (\vv x^{(2)}, \vv y^{(2)}) \circ (\vv z^{(2)}, \vv w^{(2)})$,
где $\phi(\vv x^{(2)}) = 2^n - 1$, $\phi(\vv y^{(2)}) = 2^n - 1 - 2^{i_1 - 1}$, $\phi(\vv z^{(2)}) = 2^n - 1 - 2^{i_2 - 1}$,
$i_2$~--- индекс разряда, в котором различаются вектора $\vv x^{(2)}$ и $\vv z^{(2)}$, $z^{(2)}_{i_2} = 0$,
$\vv w^{(2)}$~--- новый вектор второй транспозиции в паре, получившийся в результате действия сопряжением.
Для этого шага также потребуется не более $2(n+1)$ \gate{} из множества $\Omega_n^2$.

Покажем, как можно действием сопряжения привести пару $p^{(2)}$ к виду
$q = (\vv x^{(3)}, \vv y^{(3)} ) \circ (\vv z^{(3)}, \vv w^{(3)})$, где
$\phi(\vv x^{(3)}) = 2^n - 1$, $\phi(\vv y^{(3)}) = 2^n - 1 - 2^{i_1 - 1}$, $\phi(\vv z^{(3)}) = 2^n - 1 - 2^{i_2 - 1}$,
$\phi(\vv w^{(3)}) = 2^n - 1 - 2^{i_1 - 1} - 2^{i_2 - 1}$. Рассмотрим два случая:
\begin{enumerate}
    \item\label{conjuct_step_one_in_alg_my_common}
        $w^{(2)}_{i_1} = w^{(2)}_{i_2} = 0$.

        В этом случае сначала действуем сопряжением на $p^{(2)}$ подстановками, задаваемыми элементами $N_{i_1}$ и $N_{i_2}$.
        Затем для каждого $i\colon w^{(2)}_i \ne 1, i \ne i_1, i_2$ действуем сопряжением на полученную пару транспозиций
        подстановкой, задаваемой элементом $C_{i_1,i_2;i}$.
        После этого вновь действуем сопряжением на полученную пару транспозиций подстановками,
        задаваемыми элементами $N_{i_1}$ и $N_{i_2}$.

        Для сопряжения требуется не более $2(n - 2)$ элементов $C_{i_1,i_2;i}$, $4$ элемента $N_{i_1}$ и 4 элемента $N_{i_2}$.
        Следовательно, для получения пары транспозиций $q$ при $w^{(2)}_{i_1} = w^{(2)}_{i_2} = 0$ требуется не более $2(n + 2)$
        элементов NOT и 2-CNOT.

    \item
        $w^{(2)}_{i_1} = 1$ или $w^{(2)}_{i_2} = 1$ (в том числе и одновременно).

        $\vv w^{(2)}$ не равно ни одному из остальных векторов пары транспозиций $p^{(2)}$, следовательно,
        существует такой индекс $i_3\colon w^{(2)}_{i_3} = 0$.
        Действуем сопряжением на $p^{(2)}$ подстановкой, задаваемой элементом $N_{i_3}$.
        Затем действуем сопряжением на полученную пару транспозиций подстановкой, задаваемой элементом $C_{i_3;i_1}$,
        если $w^{(2)}_{i_1} = 1$, и подстановкой, задаваемой элементом $C_{i_3;i_2}$, если $w^{(2)}_{i_2} = 1$.
        После этого вновь действуем сопряжением на полученную пару транспозиций подстановкой, задаваемой элементом $N_{i_3}$,
        и приходим к случаю~\ref{conjuct_step_one_in_alg_my_common}, описанному выше.

        Для сопряжения требуется не более четырёх элементов 2-CNOT ($C_{i_3;i_1}$ и $C_{i_3;i_2}$), $4$ элемента $N_{i_3}$ и
        не более $2(n + 2)$ элементов NOT и 2-CNOT (при переходе к случаю~\ref{conjuct_step_one_in_alg_my_common}).
        Следовательно, для получения пары транспозиций $q$ при $w^{(2)}_{i_1} = 1$ или $w^{(2)}_{i_2} = 1$
        требуется не более $2(n + 6)$ элементов NOT и 2-CNOT.
\end{enumerate}
Если действовать сопряжением на подстановку $q$ теми же подстановками, что и на $p$, но в обратном порядке,
то получится исходная подстановка $p$.

Пара независимых транспозиций $q$ задаётся обратимым элементом $C_{I;j}$,
где $I = \{\,i\mid 1 \leqslant i \leqslant n, i \ne i_1, i_2\,\}$.
Согласно утверждению~\ref{fast_realization_of_k_cnot}, элемент $k$-CNOT при $k < n - 1$ и $n > 3$
можно представить в виде композиции $8(k-3)$ элементов 2-CNOT.
Т.\,к. $|I| = n - 2$ $\Rightarrow$ элемент $C_{I;j}$ можно заменить
на композицию не более чем $8(n-5)$ элементов 2-CNOT без использования дополнительных входов схемы.

Таким образом, суммарную сложность реализации пары независимых транспозиций можно оценить как
\begin{equation}
    L_{\Omega_n^2}(p_{\mathrm{indep}}) \leqslant 4(n+1) + 2(n + 6) + 8(n-5) = 14n + O(1)
    \label{formula_independent_pair_realization_complexity} 
\end{equation}
при использовании \gate{} множества $\Omega_n^2$. Если же допускается использовать \gate{} множества $\Omega_n$,
то суммарная сложность реализации пары независимых транспозиций ограничена сверху как
$L_{\Omega_n}(p_{\mathrm{indep}}) \leqslant 6n + O(1)$.

\bigskip
\bigskip
Рассмотрим пару зависимых транспозиций $p = (\vv x, \vv y) \circ (\vv x, \vv z)$.
Такую пару можно представить в виде произведения двух пар независимых транспозиций
$$
    (\vv x, \vv y) \circ (\vv x, \vv z) = \left( (\vv x, \vv y) \circ (\vv a, \vv b) \right )
        \circ \left( (\vv a, \vv b) \circ (\vv x, \vv z) \right ) \; .
$$
Следовательно, суммарную сложность реализации пары зависимых транспозиций можно оценить как
$$
    L_{\Omega_n^2}(p_{\mathrm{dep}}) \leqslant 2 L_{\Omega_n^2}(p_{\mathrm{indep}}) \leqslant 28n + O(1)
$$
при использовании \gate{} множества $\Omega_n^2$. В случае же, если допускается использовать \gate{} множества $\Omega_n$,
то суммарная сложность реализации пары зависимых транспозиций
$L_{\Omega_n}(p_{\mathrm{dep}}) \leqslant 12n + O(1)$.
Однако пару зависимых транспозиций можно реализовывать другим способом, позволяющим получить лучшие оценки для
$L_{\Omega_n^2}(p_{\mathrm{dep}})$ и $L_{\Omega_n}(p_{\mathrm{dep}})$.

Применяя такие же рассуждения, как и для пары независимых транспозиций (см. с.~\pageref{independent_pair_first_steps}),
приведём пару $p$ действием сопряжения к виду $p^{(1)} = (\vv x^{(1)}, \vv y^{(1)}) \circ (\vv x^{(1)}, \vv z^{(1)})$,
где $\phi(\vv x^{(1)}) = 2^n - 1$, $\phi(\vv y^{(1)}) = 2^n - 1 - 2^{i_1 - 1}$, $i_1$~--- индекс разряда, в котором различаются
вектора $\vv x^{(1)}$ и $\vv y^{(1)}$, $y^{(1)}_{i_1} = 0$,
$\phi(\vv z^{(1)}) = 2^n - 1 - 2^{i_2 - 1}$, $i_2$~--- индекс разряда, в котором различаются
вектора $\vv x^{(1)}$ и $\vv z^{(1)}$, $z^{(1)}_{i_2} = 0$.
Для этого шага потребуется не более $4(n+1)$ \gate{} из множества $\Omega_n^2$.
    
Введём отображение $\psi\colon \ZZ_2^n \to \ZZ_2^2$ следующим образом:
$$
    \psi(\vv x) = \langle x_{i_1}, x_{i_2} \rangle \; .
$$
На данном этапе
$\psi(\vv x^{(1)}) = \langle 1,1 \rangle$,
$\psi(\vv y^{(1)}) = \langle 0,1 \rangle$,
$\psi(\vv z^{(1)}) = \langle 1,0 \rangle$.

Действуем сопряжением на $p^{(1)}$ подстановкой, задаваемой элементом $N_{i_1}$,
получаем пару $p^{(2)} = (\vv x^{(2)}, \vv y^{(2)}) \circ (\vv x^{(2)}, \vv z^{(2)})$,
где
$\psi(\vv x^{(2)}) = \langle 0,1 \rangle$,
$\psi(\vv y^{(2)}) = \langle 1,1 \rangle$,
$\psi(\vv z^{(2)}) = \langle 0,0 \rangle$.
Затем действуем сопряжением на $p^{(2)}$ подстановкой, задаваемой элементом $N_{i_2}$,
получаем пару $p^{(3)} = (\vv x^{(3)}, \vv y^{(3)}) \circ (\vv x^{(3)}, \vv z^{(3)})$,
где
$\psi(\vv x^{(3)}) = \langle 0,0 \rangle$,
$\psi(\vv y^{(3)}) = \langle 1,0 \rangle$,
$\psi(\vv z^{(3)}) = \langle 0,1 \rangle$.
После этого действуем сопряжением на $p^{(3)}$ подстановкой, задаваемой элементом $C_{i_2;i_1}$,
получаем пару $p^{(4)} = (\vv x^{(4)}, \vv y^{(4)}) \circ (\vv x^{(4)}, \vv z^{(4)})$,
где
$\psi(\vv x^{(4)}) = \langle 0,0 \rangle$,
$\psi(\vv y^{(4)}) = \langle 1,0 \rangle$,
$\psi(\vv z^{(4)}) = \langle 1,1 \rangle$.
В конце действуем сопряжением на $p^{(4)}$ подстановкой, задаваемой элементом $N_{i_2}$,
получаем пару $q = (\vv x^{(5)}, \vv y^{(5)}) \circ (\vv x^{(5)}, \vv z^{(5)})$,
где
$\psi(\vv x^{(5)}) = \langle 0,1 \rangle$,
$\psi(\vv y^{(5)}) = \langle 1,1 \rangle$,
$\psi(\vv z^{(5)}) = \langle 1,0 \rangle$.

Всего для получения транспозиции $q$ требуется ровно 8 \gate{} из множества $\Omega_n^2$.
Для этой пары зависимых транспозиций верны равенства:
$\phi(\vv x^{(5)}) = 2^n - 1 - 2^{i_1 - 1}$,
$\phi(\vv y^{(5)}) = 2^n - 1$,
$\phi(\vv z^{(5)}) = 2^n - 1 - 2^{i_2 - 1}$.
Если действовать сопряжением на подстановку $q$ теми же подстановками, что и на $p$, но в обратном порядке,
то получится исходная подстановка $p$.

Согласно утверждению~\ref{predicate_recursive_on_dependent_product}, подстановка $q$ задаётся следующей композицией \gate{}:
$$
    \frS_q = C_{i_2,j;i_1} * C_{I;i_2} * C_{i_2,j;i_1} * C_{I;i_2} \; ,
$$
где $j \ne i_1, i_2$, $I = \{\,i\mid 1 \leqslant i \leqslant n; i \ne j, i_2\,\}$.
Согласно утверждению~\ref{fast_realization_of_k_cnot}, элемент $k$-CNOT при $k < n - 1$ и $n > 3$
можно представить в виде композиции $8(k-3)$ элементов 2-CNOT.
Т.\,к. $|I| = n - 2$ $\Rightarrow$ элемент $C_{I;i_2}$ можно заменить
на композицию не более чем $8(n-5)$ элементов 2-CNOT без использования дополнительных входов схемы.

Отсюда следует, что суммарную сложность реализации пары зависимых транспозиций можно оценить как
$$
    L_{\Omega_n^2}(p_{\mathrm{dep}}) \leqslant 4(n+1) + 8 + 2(8(n-5) + 1) = 20n + O(1)
$$
при использовании \gate{} множества $\Omega_n^2$. Если же допускается использовать \gate{} множества $\Omega_n$,
то $L_{\Omega_n}(p_{\mathrm{dep}}) \leqslant 4n + O(1)$.

\bigskip
\bigskip
Теперь можно подсчитать суммарную сложность обратимой схемы $\frS_h$,
синтезированной при помощи алгоритма~\algref{alg_my_common},
умножая максимально возможное количество пар транспозиций одного типа (зависимых или независимых)
в представлении подстановки $h$ на сложность реализации этого типа пары транспозиций.
Сложность синтезированной схемы при использовании \gate{} множества $\Omega_n^2$
не превосходит следующей величины:
$$
    L_{\Omega_n^2}(\frS_h) \leqslant 2^{m-1}L_{\Omega_n^2}(p_{\mathrm{indep}}) + L_{\Omega_n^2}(p_{\mathrm{dep}})
        \leqslant 2^{m-1}(14n + O(1)) + 20n + O(1) \lesssim 7n 2^m \; .
$$
Если же допустимо использовать \gate{} множества $\Omega_n$, то сложность синтезированной схемы
не превосходит следующей величины:
$$
    L_{\Omega_n}(\frS_h) \leqslant 2^{m-1}L_{\Omega_n}(p_{\mathrm{indep}}) + L_{\Omega_n}(p_{\mathrm{dep}})
        \leqslant 2^{m-1}(6n + O(1)) + 4n + O(1) \lesssim 3n 2^m \; .
$$
Временн\'{а}я сложность алгоритма~\algref{alg_my_common} составляет порядка $O(n2^m)$,
а точнее $\lesssim 3n 2^m$: для каждой пары независимых транспозиций необходимо рассмотреть три транспозиции
$(\vv x, \vv y)$, $(\vv x, \vv z)$ и $(\vv x, \vv w)$ и применить к каждой из них $n$ действий для выполнения
действия сопряжением.
Требуемый для синтеза объём памяти равен $O(2^m)$ (хранение всех элементов подстановки).

По сравнению с алгоритмом~\algref{alg_group_based}, алгоритм~\algref{alg_my_common} имеет на порядок меньшее время работы
($O(n2^m)$ против $O(n^2 2^m)$) и на порядок меньшую сложность синтезированной схемы
($O(n2^m)$ против $O(n^2 2^m)$), сохраняя при этом такой же объём памяти, требуемый для синтеза ($O(n2^m)$).
Вместе с тем, алгоритм~\algref{alg_my_common} может всегда синтезировать схему из элементов множества $\Omega_n^2$
без использования дополнительных входов схемы, в отличие от алгоритма~\algref{alg_my_simple}.

\subsection{Сравнение быстрых алгоритмов синтеза, основанных на теории групп подстановок}

\forceindent Результаты сравнения алгоритмов~\algref{alg_group_based},~\algref{alg_my_simple} и~\algref{alg_my_common}
синтеза обратимой схемы $\frS$, задающей требуемую подстановку $h \in A(\ZZ_2^n)$,
приведены в таблице~\ref{table_alg_group_theory_based_comparison}. Обзначения:
$\mathrm T(A)$~--- временн\'{а}я сложность алгоритма синтеза;
$\mathrm M(A)$~--- требуемый для синтеза объём памяти;
$L_{\Omega_n}(\frS)$ и $L_{\Omega_n^2}(\frS)$~--- сложность синтезированной схемы
при использовании \gate{} из множества $\Omega_n$ и $\Omega_n^2$ соответственно;
$Q(\frS)$~--- количество дополнительных входов схемы;
$m = \lceil \log_2|M| \rceil$, $M = \{\,\vv x \in \ZZ_2^n\mid h(\vv x) \ne \vv x\,\}$~--- множество подвижных точек
подстановки $h$.

{
    \renewcommand{\baselinestretch}{1.2}

    \Table[ht]
        \small
        \centering
        \begin{tabular}{|m{3.5cm}|*{3}{c|}}
            \hline

            \multirow{2}{*}{
                \textbf{Характеристика}
            } &
            \multicolumn{3}{c|}{\textbf{Алгоритм}} \\
            \cline{2-4}

            & 
            \centering \algref{alg_group_based} &
            \centering \algref{alg_my_simple} &
            \centering \algref{alg_my_common} \tabularnewline
            \hline

            \centering $\mathrm T(A)$ &
            \centering $\lesssim \frac{10}{3}n^2 2^m$ &
            \centering $\lesssim 2n 2^m$ &
            \centering $\lesssim 3n 2^m$ \tabularnewline        
            \hline

            \centering $\mathrm M(A)$ &
            \centering $O(n2^m)$ &
            \centering $O(n2^m)$ &
            \centering $O(n2^m)$ \tabularnewline
            \hline

            \centering $L_{\Omega_n}(\frS)$ &
            \centering $\lesssim \frac{16}{3}n^2 2^m$ &
            \centering $\lesssim 2n 2^m$ &
            \centering $\lesssim 3n 2^m$ \tabularnewline        
            \hline

            \centering $L_{\Omega_n^2}(\frS)$ &
            \centering $\lesssim \frac{64}{3}n^2 2^m$ &
            \centering $\lesssim 10n 2^m$ &
            \centering $\lesssim 7n 2^m$ \tabularnewline        
            \hline

            \centering $\mathrm Q(\frS)$ &
            \centering 0 &
            \centering $\leqslant 1$ &
            \centering 0 \tabularnewline
            \hline

        \end{tabular}
        \caption{
            \small Сравнение быстрых алгоритмов синтеза обратимых схем, основанных на теории групп подстановок.
        }\label{table_alg_group_theory_based_comparison}
    \end{table}
    
} 

По результатам сравнения, приведённым в таблице~\ref{table_alg_group_theory_based_comparison}, можно сделать вывод,
что алгоритмы~\algref{alg_my_simple} и~\algref{alg_my_common} синтезируют обратимую схему со сложностью,
меньшей на порядок, за время, меньшее на порядок, по сравнению с алгоритмом~\algref{alg_group_based}.
Вместе с тем, алгоритм~\algref{alg_my_common} по сравнению с алгоритмом~\algref{alg_my_simple} позволяет получить схему
с меньшей примерно в 1.4 раза сложностью в базисе $\Omega_n^2$ за счёт увеличенного примерно в 1.5 раза
времени синтеза. При этом алгоритм~\algref{alg_my_common} гарантированно не использует дополнительные входы
в синтезированной схеме, в отличие от алгоритма~\algref{alg_my_simple}.

Однако, несмотря на все достоинства предложенных алгоритмов синтеза, у них есть один существенный недостаток:
сложность синтезированной схемы зависит только от $m$, но никак не зависит от вида конкретной подстановки $h$.
К примеру, рассмотрим преобразование следующего вида:
$$
    f\left ( \langle x_1, x_2, \ldots, x_n \rangle \right ) = \langle x_1, x_2 \oplus x_1, x_3, \ldots, x_n \rangle \; .
$$%
\phantomsection\label{formula_bad_case_for_my_synthesis_algorithm}
Соответствующую подстановку $h_f$ можно задать одним элементом $C_{1;2}$. При этом количество подвижных точек подстановки
$h_f$ будет равно половине от всех векторов множества $\ZZ_2^n$ (для которых $x_1 = 1$)
$\Rightarrow |M| = \frac{1}{2}|\ZZ_2^n| = 2^{n-1} \Rightarrow m = n - 1$. Отсюда следует, что для всех преобразований,
схожих с $f$, алгоритмы~\algref{alg_my_simple} и~\algref{alg_my_common} будут синтезировать схему
со сложностью порядка $O(n2^n)$.

В следующих главах будет рассматриваться вопрос снижения сложности синтезированных схем.

\sectionenumerated[section_complexity_minimization]{Снижение сложности обратимых схем}

\forceindent В данной главе будут рассмотрены способы снижения сложности обратимой схемы, синтезированной при помощи
алгоритма синтеза~\algref{alg_my_common}, рассмотренного в предыдущей главе.
Можно выделить два основных направления снижения сложности:
\begin{enumerate}
    \item Внесение изменений в сам алгоритм синтеза с целью получения обратимой схемы с меньшей сложностью.
    \item Снижение сложности уже синтезированной обратимой схемы.
\end{enumerate}
Первое направление зависит от конкретного алгоритма синтеза, поэтому предлагаемые для него способы снижения сложности
могут быть неприменимы к другим алгоритмам синтеза. Второе же направление может позволить получить <<\textit{универсальные}>>
способы снижения сложности обратимых схем, синтезированных произвольным алгоритмом.

Примером <<\textit{универсального}>> способа снижения сложности обратимых схем может служить применение
таблиц эквивалентных замен композиций \gate{} Такой способ применяется, к примеру, в алгоритмах синтеза обратимых
схем, описанных в работах~\cite{iwama_transform_rules, miller_transform_based}.
Правила эквивалентных замен были рассмотрены в работах~\cite{iwama_transform_rules, saeedi_rule_based} и частично были
описаны в разделе~\ref{subsection_optimization_algorithms} на с.~\pageref{opt_duplicates}.
В работе~\cite{iwama_transform_rules} эти правила рассматривались для \textit{классических} элементов $k$-CNOT:
для изменения значения на контролируемом выходе значения на всех контролирующих входах должны равняться 1.
В работе же~\cite{saeedi_rule_based} авторами рассматриваются правила эквивалентных замен для \textit{обобщённых}
элементов $k$-CNOT, у которых нет ограничений на значения на контролирующих входах. Авторами этой работы предложен
способ использования карт Карно для получения правил эквивалентных замен.

В следующем разделе будут описаны расширенные правила эквивалентных замен композиций обратимых \gate, основанные
на операциях на множествах. Также будет приведено доказательство корректности предлагаемых замен.
Результаты следующего раздела были опубликованы автором в работе~\cite{my_complexity_reduction}. В отличие от результатов,
полученных в работе~\cite{saeedi_rule_based}, в предлагаемом далее подходе получения правил эквивалентных замен
не используются карты Карно, а только описания множеств контролирующих входов обратимых элементов $k$-CNOT,
что является, на взгляд автора, более гибким подходом.

\subsection{%
    \texorpdfstring{Обобщённое представление элемента $k$-CNOT}%
       {Обобщённое представление элемента k-CNOT}}

\forceindent
Классический элемент $k$-CNOT $C_{i_1, \ldots,i_k;j}$ инвертирует значение на контролируемом $j$-м входе,
когда значение на всех контролирующих входах $i_1, \ldots, i_k$ равно 1.
В работе~\cite{saeedi_rule_based} было предложено обобщить представление элемента $k$-CNOT для случая нулевого значения
на некоторых контролирующих входах. Формальное определение~\ref{define_common_k_cnot} такого элемента было
дано на с.~\pageref{define_common_k_cnot}. Будем обозначать такой \gate{} через $C_{I;J;t}$ или $E(t,I,J)$,
где $t$~--- контролируемый выход, $I$~--- множество прямых контролирующих входов,
$J$~--- множество инвертированных контролирующих входов.

Элемент $E(t,I,J)$ инвертирует значение на контролируемом входе только тогда,
когда значение на всех прямых контролирующих входах равно 1, на всех инвертированных контролирующих входах~--- 0.
Графически прямые контролирующие входы будем обозначать символом {\large $\bullet$},
инвертированные~--- символом $\boldsymbol{\circ}$,
контролируемый вход~--- символом $\boldsymbol{\oplus}$;
элемент NOT~--- символом $\boldsymbol{\otimes}$ (рис.~\ref{pic_generalized_k_cnot}).

\Figure[ht]
    \centering
    \includegraphics[scale=1.2]{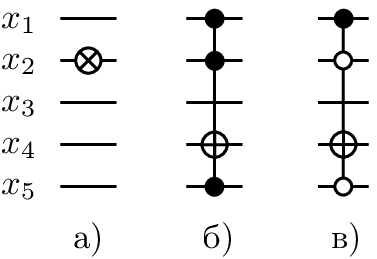}
    \caption
    {
        \small Графическое обозначение обратимых элементов ($n=5$):\\
        а) $E(2)$; б) $E(4, \{\,1, 2, 5\,\})$; в) $E(4,\{\,1\,\},\{\,2, 5\,\})$.
    }\label{pic_generalized_k_cnot}
\end{figure}

Элементы $E(t,I,J)$ могут применяться при описании алгоритмов синтеза~\cite{group_based} и~\cite{my_fast_group_based_algorithm}.
Любой элемент $E(t,I,J)$ можно представить в виде композиции классических элементов NOT и $k$-CNOT
(рис.~\ref{pic_generalized_k_cnot_representation}).

В общем случае элемент $E(t,I,J)$ можно заменить без изменения результирующего преобразования на композицию элементов
$\left ( \compose_{t \in J}{E(t)} \right ) * E\left(t, I \cup J\right) * \left( \compose_{t \in J}{E(t)} \right )$.

\Figure[ht]
    \centering
    \includegraphics[scale=1.2]{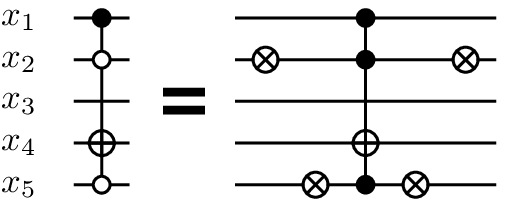}
    \caption
    {
        \small Представление элемента $E(4,\{\,1\,\},\{\,2,5\,\})$ в виде\\
        композиции элементов $E(2)$, $E(5)$ и $E(4,\{\,1,2,5\,\})$ ($n = 5$).
    }\label{pic_generalized_k_cnot_representation}
\end{figure}

\subsection{
    \texorpdfstring{Коммутирующие обратимые \gate}%
    {Коммутирующие обратимые Ф.Э.}}
\label{subsection_independend_gates}

\forceindent При решении задачи снижения сложности обратимой схемы часто необходимо выяснить,
можно ли два подряд идущих \gate{} поменять местами без изменения преобразования, реализуемого схемой.

\begin{define} Обратимые \gate{} $E_1$ и $E_2$ являются \textit{коммутирующими}, если их композицию $E_1 * E_2$ можно заменить
на композицию $E_2 * E_1$ без изменения результирующего преобразования. Иначе элементы $E_1$ и $E_2$ являются
\textit{некоммутирующими}.
\end{define}

В работе~\cite{iwama_transform_rules} были рассмотрены условия коммутируемости для классических элементов $k$-CNOT.
Рассмотрим, при каких условиях являются коммутирующими элементы $E(t,I,J)$.
Здесь и далее активно используется функция $\phi(\vv x, I, J)\colon \ZZ_2^n \to \ZZ_2$ следующего вида:
$$
    \phi(\vv x, I, J) = \left( \smallwedge_{i \in I}{x_i} \right) \wedge \left( \smallwedge_{i \in J}{\bar x_i} \right) \; .
$$
Будем считать, что $\phi(\vv x, \varnothing, \varnothing) = 1$.

\begin{lemma}\label{lemma_independend_gates}
    Элементы $E(t_1, I_1, J_1)$ и $E(t_2, I_2, J_2)$ являются коммутирующими тогда и только тогда,
    когда выполняется хотя бы одно из условий:
    \begin{enumerate}
        \item $t_1 \notin I_2 \cup J_2$ и $t_2 \notin I_1 \cup J_1$ (в частности, $t_1 = t_2$);
        \item $I_1 \cap J_2 \ne \varnothing$ или $I_2 \cap J_1 \ne \varnothing$.
    \end{enumerate}
\end{lemma}
\begin{proof}
    Докажем необходимость.

    Пусть функция $f_i(\vv x)$ задаётся элементом $E(t_i, I_i, J_i)$. Тогда $f_i(\vv x) = \vv y$,
    где $y_j=x_j$ при $j \ne t_i$, $y_{t_i} = x_{t_i} \oplus f^{(k)}_i(\vv x)$,
    $f^{(k)}_i(\vv x) = \phi(\vv x, I_i, J_i)$.

    Пусть для некоторого $\vv a$ значение $f^{(k)}_1(\vv a) = 1$, $f_1(\vv a) = \vv b$.
    По условию, элементы $E(t_1, I_1, J_1)$ и $E(t_2, I_2, J_2)$ являются коммутирующими, следовательно,
    $f_2(f_1(\vv a)) = f_1(f_2(\vv a))$.
    Отсюда следует, что либо $t_1 = t_2$, либо $t_1 \neq t_2$ и $f^{(k)}_2(\vv a) = f^{(k)}_2(\vv b)$.
    Операция $\oplus$ является коммутативной, поэтому в любом случае, вне зависимости от значений $t_1$ и $t_2$,
    верно равенство $f^{(k)}_2(\vv a) = f^{(k)}_2(\vv b)$.
    Оно может выполняться только в двух случаях:
    \begin{enumerate}
        \item
            Функция $f^{(k)}_2(\vv x)$ фиктивно зависит от переменной $x_{t_1} \Rightarrow t_1 \notin I_2 \cup J_2$.
            Мы не делали различия между двумя \gate, следовательно, аналогичное условие верно и для другого \gate:
            $t_2 \notin I_1 \cup J_1$.
        
        \item
            Функция $f^{(k)}_2(\vv x)$ существенно зависит от переменной $x_{t_1}$,
            т.\,е. $t_1 \in I_2 \cup J_2$, тогда из условия $f^{(k)}_2(\vv a) = f^{(k)}_2(\vv b)$ можно вывести равенство
            $$
                \phi(\vv a, I_2 \setminus \{\,t_1\,\}, J_2 \setminus \{\,t_1\,\}) = 0 \; .
            $$
            Отсюда следует, что для всех значений $\vv a$, для которых $f^{(k)}_1(\vv a) = 1$, верно равенство
            $f^{(k)}_2(\vv a) = 0$, таким образом, получаем следующее равенство:
            $$
                \phi(\vv a, I_2 \setminus I_1, J_2 \setminus J_1) = 0 \; .
            $$
            Выше было условлено считать, что $\phi(\vv x, \varnothing, \varnothing) = 1$, следовательно, либо
            $I_2 \setminus I_1 \neq \varnothing$, либо $J_2 \setminus J_1 \neq \varnothing$.
            Пусть $I_2 \cap J_1 = \varnothing$ и $J_2 \cap I_1 = \varnothing$.
            Рассмотрим следующее значение $\vv a'$: $a'_i = 1$ при $i \in I_1 \cup I_2$, $a'_i = 0$ при $i \in J_1 \cup J_2$.
            В этом случае $f^{(k)}_1(\vv a') = 1$ и $\phi(\vv a', I_2 \setminus I_1, J_2 \setminus J_1) = 1$~--- противоречие.
            Следовательно, либо $I_2 \cap J_1 \ne \varnothing$, либо $J_2 \cap I_1 \ne \varnothing$.
    \end{enumerate}

    \medskip
    Необходимость обоих условий доказана. Теперь докажем их достаточность.

    Для каждого из двух условий рассмотрим функции $f(\vv x) = f_2(f_1(\vv x))$,
    соответствующую композиции $E(t_1, I_1, J_1) * E(t_2, I_2, J_2)$, и $g(\vv x) = f_1(f_2(\vv x))$,
    соответствующую композиции $E(t_2, I_2, J_2) * E(t_1, I_1, J_1)$, где $\vv x$~--- входной вектор значений.

    Элементы $E(t_1, I_1, J_1)$ и $E(t_2, I_2, J_2)$ будут коммутирующими, если $f(\vv x) = g(\vv x)$.
    \begin{enumerate}
        \item
            Рассмотрим условие $t_1 \notin I_2 \cup J_2$ и $t_2 \notin I_1 \cup J_1$ (в частности, $t_1 = t_2$).
            
            $f_1(\vv x) = \vv y$; $y_i = x_i$ при $i \ne t_1$;
            $y_{t_1} = x_{t_1} \oplus \phi(\vv x, I_1, J_1)$.
            \\
            $f(\vv x) = f_2(\vv y) = \vv z$; $z_i = y_i = x_i$ при $i \ne t_1, t_2$;
            $z_{t_1} = y_{t_1} = x_{t_1} \oplus \phi(\vv x, I_1, J_1)$.
            \\
            $z_{t_2} = y_{t_2} \oplus \phi(\vv y, I_2, J_2) =
            x_{t_2} \oplus \phi(\vv x, I_2, J_2)$.

            \bigskip
            
            $f_2(\vv x) = \vv y'$; $y'_i = x_i$ при $i \ne t_2$;
            $y'_{t_2} = x_{t_2} \oplus \phi(\vv x, I_2, J_2)$.
            \\
            $f(\vv x) = f_2(\vv y') = \vv z'$; $z'_i = y'_i = x_i$ при $i \ne t_1, t_2$;
            $z'_{t_2} = y_{t_2} = x_{t_2} \oplus \phi(\vv x, I_2, J_2)$.
            \\
            $z'_{t_1} = y_{t_1} \oplus \phi(\vv y, I_1, J_1) =
            x_{t_1} \oplus \phi(\vv x, I_1, J_1)$.

            \bigskip
            \qquad $\vv z = \vv z' \Rightarrow f(\vv x) = g(\vv x)$.

        \item
            Рассмотрим условие $I_1 \cap J_2 \ne \varnothing$ или $I_2 \cap J_1 \ne \varnothing$.
            
            Пусть $I_1 \cap J_2 \ne \varnothing$, тогда $\exists k \in I_1 \cap J_2$,
            $f_1(\vv x ) = \vv x$ при $x_k = 0$, $f_2(\vv x) = \vv x$ при $x_k = 1$.
            
            \qquad Рассмотрим случай $x_k = 0$:
            \\
            $f(\vv x) = f_2(f_1(\vv x)) = f_2(\vv x)$.
            \\
            $g(\vv x) = f_1(f_2(\vv x)) = f_1(\vv y)$.
            \\
            $k \in I_1 \cap J_2 \Rightarrow y_k = x_k = 0 \Rightarrow g(\vv x) = f_1(\vv y) = \vv y =
            f_2(\vv x) \Rightarrow f(\vv x ) = g(\vv y)$.
            
            \qquad Рассмотрим случай $x_k = 1$:
            \\
            $f(\vv x) = f_2(f_1(\vv x)) = f_2(\vv y)$.
            \\
            $k \in I_1 \cap J_2 \Rightarrow y_k = x_k = 1 \Rightarrow f(\vv x) = f_2(\vv y) = \vv y = f_1(\vv x)$.
            \\
            $g(\vv x) = f_1(f_2(\vv x)) = f_1(\vv x) \Rightarrow f(\vv x ) = g(\vv y)$.
            
            Таким образом, для всех $x_k$ верно $f(\vv x ) = g(\vv y)$.
            Аналогично для случая $I_2 \cap J_1 \ne \varnothing$.
    \end{enumerate}
\end{proof}

Условие коммутируемости 2 из Леммы~\ref{lemma_independend_gates} является новым результатом,
не выводимым из условий коммутируемости для классических элементов $k$-CNOT.
Примеры коммутирующих обратимых \gate{} показаны на рис.~\ref{pic_swappable}.
\Figure[ht]
    \centering
    \includegraphics[scale=1.2]{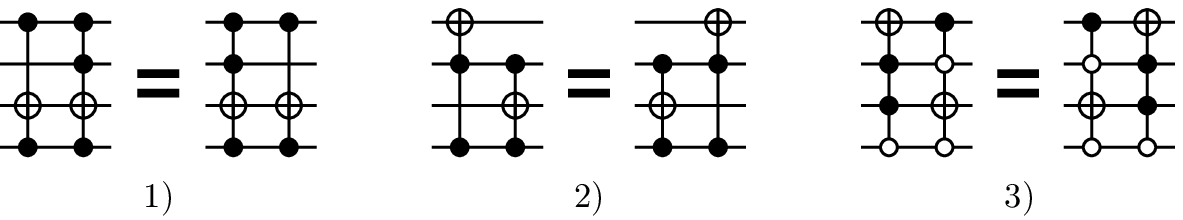}
    \caption{\small Примеры коммутирующих обратимых \gate}\label{pic_swappable}
\end{figure}

\subsection{
    \texorpdfstring{Эквивалентные замены композиций \gate}%
    {Эквивалентные замены композиций Ф.Э.}}

\forceindent В работе~\cite{iwama_transform_rules} было предложено несколько эквивалентных замен одной композиции элементов
$E(t,I)$ на другую. Эти замены были описаны в разделе~\ref{subsection_optimization_algorithms}
на с.~\pageref{opt_duplicates}.
В некоторых случаях такая замена снижает либо сложность схемы, либо количество контролирующих входов у элементов.
В работе~\cite{saeedi_rule_based} были рассмотрены некоторые эквивалентные замены композиций элементов $E(t,I,J)$.
Далее будет приведён расширенный и дополненный список таких замен.

Стоит отметить, что используемое ниже по тексту словосочетание <<композиция \gate{} может быть заменена>> означает,
что результат замены не меняет результирующего преобразования исходной композиции \gate{}


\begin{replacement}\label{replace_duplicates}
    Композиция элементов $E(t,I,J) * E(t,I,J)$ может быть исключена из схемы.
\end{replacement}
\noindent
Доказательство тривиально и следует из определения~\eqref{define_common_k_cnot}.
По сути, замена~\ref{replace_duplicates}~--- исключение дублирующих \gate{} из схемы.


\begin{replacement}\label{replace_merge} (\textit{слияние})
    Если $I_1 = I_2 \cup \{\,k\,\}$, $J_2 = J_1 \cup \{\,k\,\}$, $k \notin I_2 \cup J_1$,
    то композиция элементов $E(t,I_1,J_1) * E(t,I_2,J_2)$ может быть заменена одним элементом $E(t,I_2,J_1)$.
\end{replacement}
\begin{proof}
    Пусть функция $f_1(\vv x)$ задаётся элементом $E(t,I_1,J_1)$, а функция $f_2(\vv x)$~--- элементом $E(t,I_2,J_2)$.
    Рассмотрим функцию $f(\vv x) = f_2(f_1(\vv x))$.
    \\
    $f_1(\vv x) = \vv y$; $y_i = x_i$ при $i \ne t$;
    $y_t = x_t \oplus \phi(\vv x, I_1, J_1)$.
    \\
    $I_1 = I_2 \cup \{\,k\,\}, k \notin I_2 \Rightarrow y_t = x_t \oplus x_k \wedge \phi(\vv x, I_2, J_1)$.
    \\
    $f(\vv x) = f_2(\vv y) = \vv z$; $z_i = y_i = x_i$ при $i \ne t$;
    $z_t = y_t \oplus \phi(\vv x, I_2, J_2)$.
    \\
    $J_2 = J_1 \cup \{\,k\,\}, k \notin J_1 \Rightarrow
    z_t = y_t \oplus \bar x_k \wedge \phi(\vv x, I_2, J_1)$.
    \\
    $z_t = x_t \oplus x_k \wedge \phi(\vv x, I_2, J_1)
    \oplus \bar x_k \wedge \phi(\vv x, I_2, J_1)$.
    \\
    $z_t = x_t \oplus \left(x_k \oplus \bar x_k \right) \wedge \phi(\vv x, I_2, J_1) =
    x_t \oplus \phi(\vv x, I_2, J_1)$.

    \medskip
    Функция $f(\vv x) = \vv z$ задаётся элементом $E(t,I_2,J_1)$.
\end{proof}


\begin{replacement}\label{replace_reduce} (\textit{уменьшение количества контролирующих входов})
    Если существуют такие индексы $p$ и $q$, что $p \in I_1 \cap J_2$, $q \in J_1 \cap I_2$, 
    $I_2 = I_1 \setminus \{\,p\,\} \cup \{\,q\,\}$, $J_2 = J_1 \setminus \{\,q\,\} \cup \{\,p\,\}$,
    то композиция элементов $E(t,I_1,J_1) * E(t,I_2,J_2)$ может быть заменена композицией 
    $E(t,I_1,J_3) * E(t,I_2,J_3)$, где $J_3 = J_1 \setminus \{\,q\,\} = J_2 \setminus \{\,p\,\}$.
\end{replacement}
\begin{proof}
    Пусть функция $f_1(\vv x)$ задаётся элементом $E(t,I_1,J_1)$, а функция $f_2(\vv x)$~--- элементом $E(t,I_2,J_2)$.
    Рассмотрим функцию $f(\vv x) = f_2(f_1(\vv x))$.

    $f_1(\vv x) = \vv y$; $y_i = x_i$ при $i \ne t$;
    $y_t = x_t \oplus \phi(\vv x, I_1, J_1)$.
    \\
    $q \in J_1 \Rightarrow y_t = x_t \oplus (1 \oplus x_q) \wedge \phi(\vv x, I_1, J_1 \setminus \{\,q\,\})$.
    \\
    $q \notin I_1 \Rightarrow y_t = x_t \oplus \phi(\vv x, I_1, J_3) \oplus \phi(\vv x, I_1 \cup \{\,q\,\}, J_3)$.

    \medskip
    $f(\vv x) = f_2(\vv y) = \vv z$; $z_i = y_i = x_i$ при $i \ne t$;
    $z_t = y_t \oplus \phi(\vv x, I_2, J_2)$.
    \\
    $p \in J_2 \Rightarrow z_t = y_t \oplus (1 \oplus x_p) \wedge \phi(\vv x, I_2, J_2 \setminus \{\,p\,\})$.
    \\
    $p \notin I_2 \Rightarrow z_t = y_t \oplus \phi(\vv x, I_2, J_3) \oplus \phi(\vv x, I_2 \cup \{\,p\,\}, J_3)$.

    \medskip
    $z_t = x_t \oplus \phi(\vv x, I_1, J_3) \oplus \phi(\vv x, I_1 \cup \{\,q\,\}, J_3)
               \oplus \phi(\vv x, I_2, J_3) \oplus \phi(\vv x, I_2 \cup \{\,p\,\}, J_3)$.
    \\
    $I_2 = I_1 \setminus \{\,p\,\} \cup \{\,q\,\}$, $p \in I_1 \Rightarrow
    I_2 \cup \{\,p\,\} = I_1 \cup \{\,q\,\} \Rightarrow
    \phi(\vv x, I_1 \cup \{\,q\,\}, J_3) = \phi(\vv x, I_2 \cup \{\,p\,\}, J_3)$.
    \\
    $z_t = x_t \oplus \phi(\vv x, I_1, J_3) \oplus \phi(\vv x, I_2, J_3)$.

    \medskip
    Функция $f(\vv x) = \vv z$ задаётся композицией элементов $E(t,I_1,J_3) * E(t,I_2,J_3)$.
\end{proof}


\begin{replacement}\label{replace_dependend} (\textit{перестановка некоммутирующих \gate})
    Если $t_1 \in I_2 \cup J_2$, $t_2 \notin I_1 \cup J_1$, то композиция некоммутирующих элементов
    $E(t_1,I_1,J_1) * E(t_2,I_2,J_2)$ может быть заменена композицией
    $E(t_2, I_1 \cup I_2 \setminus \{\,t_1\,\},J_1 \cup J_2 \setminus \{\,t_1\,\}) * E(t_2, I_2, J_2) * E(t_1, I_1, J_1)$.
\end{replacement}
\begin{proof}
    $(I_1 \cup I_2) \cap (J_1 \cup J_2) = (I_1 \cap J_1) \cup (I_2 \cap J_1) \cup (I_1 \cap J_2) \cup (I_2 \cap J_2) = $
    \\
    $= (I_2 \cap J_1) \cup (I_1 \cap J_2) = \varnothing$, т.\,к. в противном случае элементы $E(t_1,I_1,J_1)$ и $E(t_2,I_2,J_2)$
    будут коммутирующими. Следовательно, элемент $E(t_2, I_1 \cup I_2 \setminus \{\,t_1\,\},J_1 \cup J_2 \setminus \{\,t_1\,\})$
    не нарушает требований, накладываемых на множества прямых и инвертированных контролирующих входов
    (см. определение~\ref{define_common_k_cnot}).

    Пусть $f_1(\vv x)$ задаётся элементом $E(t_1,I_1,J_1)$, $f_2(\vv x)$~--- элементом $E(t_2,I_2,J_2)$,
    $f_3(\vv x)$~--- элементом $E(t_2, I_1 \cup I_2 \setminus \{\,t_1\,\},J_1 \cup J_2 \setminus \{\,t_1\,\})$.
    Рассмотрим две функции: $f(\vv x) = f_2(f_1(\vv x))$ и $g(\vv x) = f_1(f_2(f_3(\vv x)))$.
    \\
    $f_1(\vv x) = \vv y$; $y_i = x_i$ при $i \ne t_1$;
    $y_{t_1} = x_{t_1} \oplus \phi(\vv x, I_1, J_1)$.
    \\
    $f(\vv x) = f_2(\vv y) = \vv z$; $z_i = y_i = x_i$ при $i \ne t_2$;
    $z_{t_2} = y_{t_2} \oplus \phi(\vv y, I_2, J_2)$.

    Пусть $t_1 \in I_2$, тогда $z_{t_2} = x_{t_2} \oplus y_{t_1} \wedge \phi(\vv x, I_2 \setminus \{\,t_1\,\}, J_2)$.
    \\
    $z_{t_2} = x_{t_2} \oplus (x_{t_1} \oplus \phi(\vv x, I_1, J_1)) \wedge \phi(\vv x, I_2 \setminus \{\,t_1\,\}, J_2)$.
    \\
    $z_{t_2} = x_{t_2} \oplus \phi(\vv x, I_2, J_2) \oplus \phi(\vv x, I_1 \cup I_2 \setminus \{\,t_1\,\}, J_1 \cup J_2)$.
    \\
    $t_1 \notin J_2 \Rightarrow z_{t_2} = x_{t_2} \oplus \phi(\vv x, I_2, J_2)
    \oplus \phi(\vv x, I_1 \cup I_2 \setminus \{\,t_1\,\}, J_1 \cup J_2 \setminus \{\,t_1\,\})$.
     
    Если же $t_1 \in J_2$, то $z_{t_2} = x_{t_2} \oplus \bar y_{t_1} \wedge \phi(\vv x, I_2, J_2 \setminus \{\,t_1\,\})$.
    \\
    $z_{t_2} = x_{t_2} \oplus ((1 \oplus x_{t_1}) \oplus \phi(\vv x, I_1, J_1)) \wedge
        \phi(\vv x, I_2, J_2 \setminus \{\,t_1\,\})$.
    \\
    $z_{t_2} = x_{t_2} \oplus \phi(\vv x, I_2, J_2) \oplus \phi(\vv x, I_1 \cup I_2, J_1 \cup J_2 \setminus \{\,t_1\,\})$.
    \\ 
    $t_1 \notin I_2 \Rightarrow z_{t_2} = x_{t_2} \oplus \phi(\vv x, I_2, J_2)
    \oplus \phi(\vv x, I_1 \cup I_2 \setminus \{\,t_1\,\}, J_1 \cup J_2 \setminus \{\,t_1\,\})$.

    Таким образом, во всех случаях:
    \\
    $z_{t_2} = x_{t_2} \oplus \phi(\vv x, I_2, J_2)
    \oplus \phi(\vv x, I_1 \cup I_2 \setminus \{\,t_1\,\}, J_1 \cup J_2 \setminus \{\,t_1\,\})$.

    \noindent $f_2(f_3(\vv x)) = \vv y'$; $y'_i = x_i$ при $i \ne t_2$;
    $y'_{t_2} = x_{t_2} \oplus \phi(\vv x, I_1 \cup I_2 \setminus \{\,t_1\,\}, J_1 \cup J_2 \setminus \{\,t_1\,\})
    \oplus \phi(\vv x, I_2, J_2)$.
    \\
    $g(\vv x) = f_1(\vv y') = \vv z'$; $z'_i = y'_i$ при $i \ne t_1$;
    $z'_{t_1} = y'_{t_1} \oplus \phi(\vv y', I_1, J_1)$.
    \\
    $t_2 \notin I_1 \cup J_1 \Rightarrow z'_{t_1} = x_{t_1} \oplus \phi(\vv x, I_1, J_1)$.

    $\vv z' = \vv z \Rightarrow f(\vv x) = g(\vv x) \Rightarrow$
    замена~\ref{replace_dependend} не меняет результирующего преобразования исходной композиции \gate{}
\end{proof}

\begin{replacement}\label{replace_dependend_by_swap} (\textit{следствие замены~\ref{replace_dependend}})
    Если в условии замены~\ref{replace_dependend} $I_1 \subseteq I_2$ и $J_1 \subseteq J_2$,
    то композиция некоммутирующих элементов $E(t_1,I_1,J_1) * E(t_2,I_2,J_2)$ может быть заменена композицией
    $E(t_2, I_2 \cup \{\,t_1\,\}, J_2 \setminus \{\,t_1\,\}) * E(t_1, I_1, J_1)$, если $t_1 \in J_2$,
    и композицией $E(t_2, I_2 \setminus \{\,t_1\,\}, J_2 \cup \{\,t_1\,\}) * E(t_1, I_1, J_1)$,
    если $t_1 \in I_2$.
\end{replacement}
\begin{proof}
    $ $

    Согласно условию замены~\ref{replace_dependend}, композиция элементов $E(t_1,I_1,J_1) * E(t_2,I_2,J_2)$
    может быть заменена композицией
    $E(t_2, I_1 \cup I_2 \setminus \{\,t_1\,\},J_1 \cup J_2 \setminus \{\,t_1\,\}) * E(t_2, I_2, J_2) * E(t_1, I_1, J_1)$.
    \\
    $I_1 \subseteq I_2$, $J_1 \subseteq J_2 \Rightarrow
    E(t_2, I_1 \cup I_2 \setminus \{\,t_1\,\},J_1 \cup J_2 \setminus \{\,t_1\,\}) =
    E(t_2, I_2 \setminus \{\,t_1\,\}, J_2 \setminus \{\,t_1\,\})$.

    Пусть $t_1 \in J_2$, тогда $E(t_2, I_2 \setminus \{\,t_1\,\}, J_2 \setminus \{\,t_1\,\}) =
        E(t_2, I_2, J_2 \setminus \{\,t_1\,\})$.
    Рассмотрим функцию $f(\vv x) = f_2(f_1(\vv x))$, задаваемой композицией
    $E(t_2, I_2, J_2 \setminus \{\,t_1\,\}) * E(t_2, I_2, J_2)$, где $f_1(\vv x)$ задаётся элементом
    $E(t_2, I_2, J_2 \setminus \{\,t_1\,\})$, а $f_2(\vv x)$~--- элементом $E(t_2, I_2, J_2)$.
    \\
    $f_1(\vv x) = \vv y$; $y_i = x_i$ при $i \ne t_2$;
    $y_{t_2} = x_{t_2} \oplus \phi(\vv x, I_2, J_2 \setminus \{\,t_1\,\})$.
    \\
    $f(\vv x) = f_2(\vv y) = \vv z$; $z_i = y_i = x_i$ при $i \ne t_2$;
    $z_{t_2} = y_{t_2} \oplus \phi(\vv x, I_2, J_2)$.
    \\
    $z_{t_2} = x_{t_2} \oplus \phi(\vv x, I_2, J_2 \setminus \{\,t_1\,\}) \oplus \phi(\vv x, I_2, J_2)
        = x_{t_2} \oplus \phi(\vv x, I_2, J_2 \setminus \{\,t_1\,\}) \oplus
            {\bar x_{t_1}} \wedge \phi(\vv x, I_2, J_2 \setminus \{\,t_1\,\})$.
    \\
    $z_{t_2} = x_{t_2} \oplus \phi(\vv x, I_2 \cup \{\,t_1\,\}, J_2 \setminus \{\,t_1\,\})$.

    Таким образом, при $t_1 \in J_2$ функция $f(\vv x)$ задаётся элементом
    $E(t_2, I_2 \cup \{\,t_1\,\}, J_2 \setminus \{\,t_1\,\})$.

    \bigskip
    Пусть $t_1 \in I_2$, тогда $E(t_2, I_2 \setminus \{\,t_1\,\}, J_2 \setminus \{\,t_1\,\}) =
        E(t_2, I_2 \setminus \{\,t_1\,\}, J_2)$. 
    Рассмотрим функцию $g(\vv x) = g_2(g_1(\vv x))$, задаваемой композицией
    $E(t_2, I_2 \setminus \{\,t_1\,\}, J_2) * E(t_2, I_2, J_2)$, где $g_1(\vv x)$ задаётся элементом
    $E(t_2, I_2 \setminus \{\,t_1\,\}, J_2)$, а $g_2(\vv x)$~--- элементом $E(t_2, I_2, J_2)$.
    \\
    $g_1(\vv x) = \vv y'$; $y'_i = x_i$ при $i \ne t_2$;
    $y'_{t_2} = x_{t_2} \oplus \phi(\vv x, I_2 \setminus \{\,t_1\,\}, J_2)$.
    \\
    $g(\vv x) = g_2(\vv y') = \vv z'$; $z'_i = y'_i = x_i$ при $i \ne t_2$;
    $z'_{t_2} = y'_{t_2} \oplus \phi(\vv x, I_2, J_2)$.
    \\
    $z'_{t_2} = x_{t_2} \oplus \phi(\vv x, I_2 \setminus \{\,t_1\,\}, J_2) \oplus \phi(\vv x, I_2, J_2) =
        x_{t_2} \oplus \phi(\vv x, I_2 \setminus \{\,t_1\,\}, J_2) \oplus
            x_{t_1} \wedge \phi(\vv x, I_2 \setminus \{\,t_1\,\}, J_2)$.
    \\
    $z_{t_2} = x_{t_2} \oplus \phi(\vv x, I_2 \setminus \{\,t_1\,\}, J_2 \cup \{\,t_1\,\})$.

    Таким образом, при $t_1 \in I_2$ функция $g(\vv x)$ задаётся элементом
    $E(t_2, I_2 \setminus \{\,t_1\,\}, J_2 \cup \{\,t_1\,\})$.
\end{proof}


\begin{replacement}\label{replace_dependend_mirrored} (\textit{зеркальное отображение замены~\ref{replace_dependend}})
    Если $t_2 \in I_1 \cup J_1$, $t_1 \notin I_2 \cup J_2$, то композиция некоммутирующих элементов
    $E(t_1, I_1, J_1) * E(t_2, I_2, J_2)$ может быть заменена композицией
    $E(t_2, I_2, J_2) * E(t_1, I_1, J_1) * E(t_1, I_1 \cup I_2 \setminus \{\,t_2\,\}, J_1 \cup J_2 \setminus \{\,t_2\,\})$.
\end{replacement}
\begin{proof}
    Доказательство аналогично доказательству для замены~\ref{replace_dependend}.
\end{proof}


\begin{replacement}\label{replace_dependend_by_swap_mirrored} (\textit{следствие замены~\ref{replace_dependend_mirrored}})
    Если в условии замены~\ref{replace_dependend_mirrored} $I_2 \subseteq I_1$ и $J_2 \subseteq J_1$,
    то композиция некоммутирующих элементов $E(t_1,I_1,J_1) * E(t_2,I_2,J_2)$ может быть заменена композицией
    $E(t_2, I_2, J_2) * E(t_1, I_1 \cup \{\,t_2\,\}, J_1 \setminus \{\,t_2\,\})$, если $t_2 \in J_1$,
    и композицией $E(t_2, I_2, J_2) * E(t_1, I_1 \setminus \{\,t_2\,\}, J_1 \cup \{\,t_2\,\})$, если $t_2 \in I_1$.
\end{replacement}
\begin{proof}
    Доказательство аналогично доказательству для замены~\ref{replace_dependend_by_swap}. 
\end{proof}


\begin{replacement}\label{replace_recursive_simple} Элемент $E(t, I, J)$ можно заменить на композицию \gate{} вида
    $$
        \left(\compose_{t \in J}{E(t)} \right) * E(t, I \cup J) * \left(\compose_{t \in J}{E(t)} \right) \; .
    $$
\end{replacement}
\begin{proof}
    См. рис.~\ref{pic_generalized_k_cnot_representation} и определение элемента $E(t, I, J)$. 
\end{proof}


\begin{replacement}\label{replace_recursive}
    Если $k \in J$, то элемент $E(t,I,J)$ можно заменить на композицию элементов
    $E(t, I \cup \{\,k\,\}, J \setminus \{\,k\,\}) * E(t, I, J \setminus \{\,k\,\})$.
\end{replacement}
\begin{proof}
    Пусть $f_1(\vv x)$ задаётся элементом $E(t, I \cup \{\,k\,\}, J \setminus \{\,k\,\})$,
    $f_2(\vv x)$~--- элементом $E(t, I, J \setminus \{\,k\,\})$.
    Рассмотрим функцию $f(\vv x) = f_2(f_1(\vv x))$.
    \\
    $f_1(\vv x) = \vv y$; $y_i = x_i$ при $i \ne t$;
    $y_t = x_t \oplus \phi(\vv x, I \cup \{\,k\,\}, J \setminus \{\,k\,\})$.
    \\
    $f(\vv x) = f_2(\vv y) = \vv z$; $z_i = y_i = x_i$ при $i \ne t$;
    $z_t = y_t \oplus \phi(\vv x, I, J \setminus \{\,k\,\})$.
    \\
    $z_t = x_t \oplus \phi(\vv x, I \cup \{\,k\,\}, J \setminus \{\,k\,\}) \oplus \phi(\vv x, I, J \setminus \{\,k\,\})$.
    \\
    $k \notin I \Rightarrow z_t = x_t \oplus x_k \wedge \phi(\vv x, I, J \setminus \{\,k\,\})
    \oplus \phi(\vv x, I, J \setminus \{\,k\,\}) = x_t \oplus \phi(\vv x, I, J)$.

    Функция $f(\vv x) = \vv z$ задаётся элементом $E(t,I,J)$.
\end{proof}

Стоит также отметить, что элемент $k$-CNOT, $k < n - 1$, может быть заменён композицией не более чем
$8(k-3)$ элементов 2-CNOT без использования дополнительных входов в схеме~\cite{barenco_elementary_gates}.
Замены \ref{replace_reduce}--\ref{replace_dependend_by_swap_mirrored} являются новым результатом, полученным автором.

Примеры эквивалентных замен показаны на рис.~\ref{pic_main_replacements}.

\Figure[ht]
    \centering
    \includegraphics[scale=1.2]{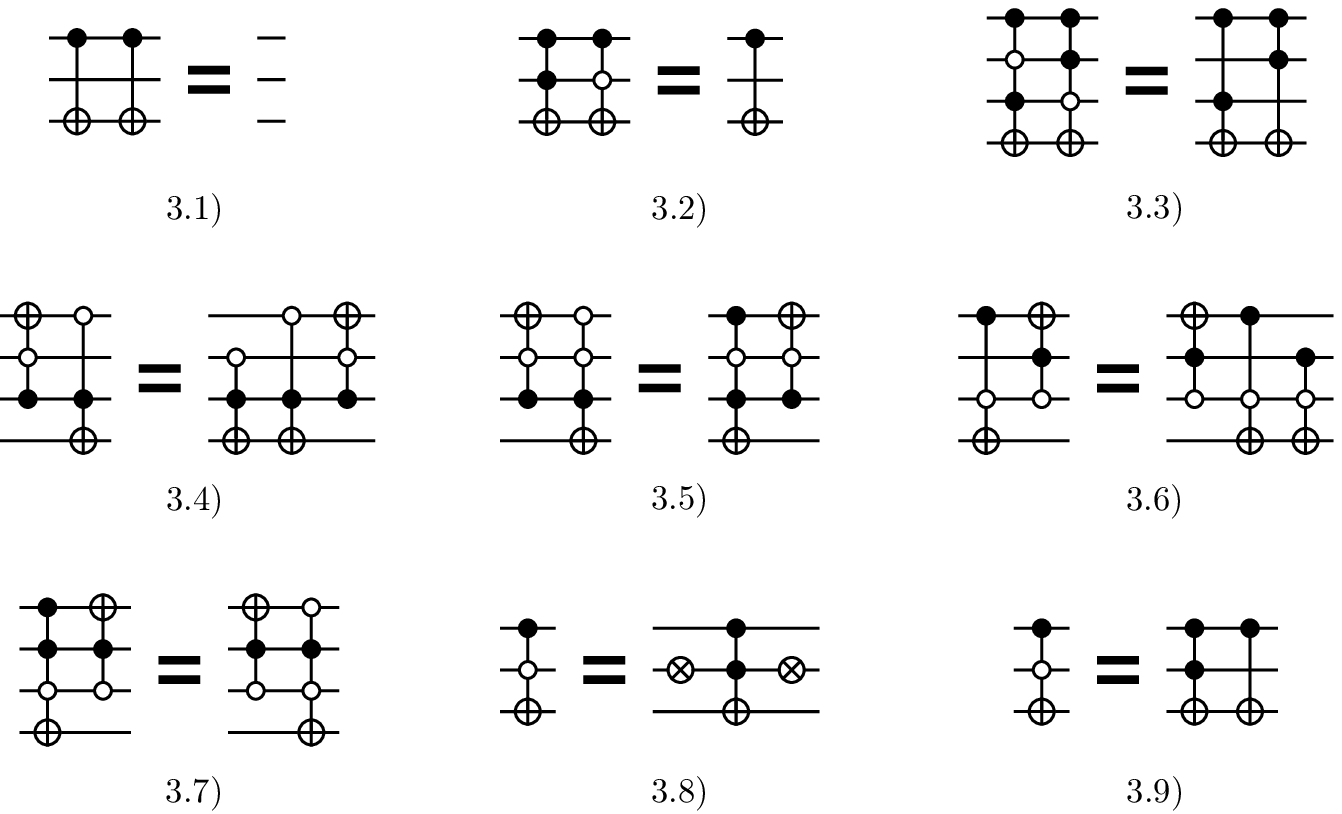}
    \caption{\small Примеры эквивалентных замен композиций \gate{}}\label{pic_main_replacements}
\end{figure}

Можно выделить ещё два частных случая замены <<слиянием>>.

\begin{replacement}\label{replace_recursive_mirrored} (\textit{обратная к замене~\ref{replace_recursive}})
    Если $I_1 = I_2 \cup \{\,k\,\}$, то композиция элементов $E(t, I_1, J) * E(t, I_2, J)$
    может быть заменена одним элементом $E(t, I_2, J \cup \{\,k\,\})$.
\end{replacement}

\begin{replacement}\label{replace_merge_additional}
    Если $J_1 = J_2 \cup \{\,k\,\}$, то композиция элементов $E(t, I, J_1) * E(t, I, J_2)$ может быть заменена одним
    элементом $E(t, I \cup \{\,k\,\}, J_2)$.
\end{replacement}

Доказательство корректности замен~\ref{replace_recursive_mirrored} и~\ref{replace_merge_additional}
вытекает из доказательства корректности замен~\ref{replace_duplicates} и~\ref{replace_recursive}.
Примеры замен~\ref{replace_recursive_mirrored} и~\ref{replace_merge_additional} показаны
на рис.~\ref{pic_merge_additional_example}.

\Figure[ht]
    \centering
    \includegraphics[scale=1.2]{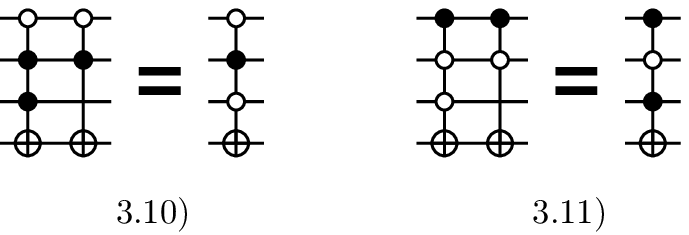}
    \caption{\small Частные случаи замены композиций \gate{} <<слиянием>>.}\label{pic_merge_additional_example}
\end{figure}

\subsection{Алгоритм снижения сложности обратимой схемы}
\label{subsection_common_scheme_complexity_reduction}

\forceindent Предложенные эквивалентные замены композиций \gate{} позволяют в некоторых случаях
снизить сложность обратимой схемы.
В основном для этого используется замена, исключающая дублирующие \gate~(\ref{replace_duplicates}),
и замены <<слиянием>> (\ref{replace_merge}, \ref{replace_recursive_mirrored} и~\ref{replace_merge_additional}).
Замены \ref{replace_reduce}--\ref{replace_dependend_by_swap_mirrored} позволяют получить новую обратимую схему
с новыми \gate, для которой можно снова попробовать использовать замены~\ref{replace_duplicates},
\ref{replace_merge}, \ref{replace_recursive_mirrored} и~\ref{replace_merge_additional}.
В случае, когда сложность уже невозможно снизить, можно использовать замены~\ref{replace_recursive_simple}
и~\ref{replace_recursive}, чтобы заменить все элементы $E(t,I,J)$ на классические элементы NOT и $k$-CNOT и получить
обратимую схему, не содержащую элементов $E(t,I,J)$.

Пусть обратимая схема представляет собой композицию элементов $\compose_{i = 1}^l {E_i}$, где $l$~--- сложность
схемы. Если композиция  элементов $E_i * E_j$ удовлетворяют условию какой-либо замены, $i < j$, и при этом существует такой
индекс $s$, $i \leqslant s < j$, что элементы $E_i$ и $E_k$ являются коммутирующими для всех $i < k \leqslant s$,
и элементы $E_j$ и $E_k$ являются коммутирующими для всех $s < k < j$, то элементы $E_i$ и $E_j$ можно исключить из схемы,
а результат замены композиции $E_i * E_j$ вставить в схему между элементами $E_s$ и $E_{s+1}$.

В качестве примера, на рис.~\ref{pic_minimization_process} показан процесс применения эквивалентных замен композиций \gate{}
для некоторой абстрактной схемы.

\Figure[ht]
    \centering
    \includegraphics[scale=1.2]{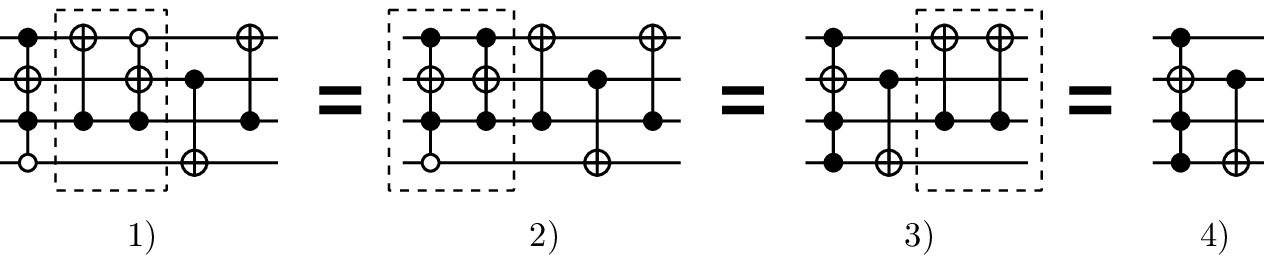}
    \caption
    {
        \small Процесс снижения сложности схемы:
        1) исходная схема; 2) схема после применения замены~\ref{replace_dependend_by_swap};
        3) схема после применения замены~\ref{replace_merge_additional};
        4) схема после применения замены~\ref{replace_duplicates}.
    }\label{pic_minimization_process}
\end{figure}

Оценим временн\'{у}ю сложность $T(A_{\text{reduction}})$ данного алгоритма снижения сложности обратимой схемы
$\frS^{(0)}$. Пусть в начале работы алгоритма $A_{\text{reduction}}$ сложность $L(\frS^{(0)}) = L_0$,
в конце $L(\frS^{(k)}) = L_k$, при этом обратимые схемы $\frS^{(0)}$ и $\frS^{(k)}$ задают одно и тоже преобразование,
$L_k \leqslant L_0$, $k$~--- количество шагов алгоритма $A_{\text{reduction}}$.

На $i$-м шаге алгоритма текущая обратимая схема $\frS^{(i)}$ имеет сложность $L_i$.
Оценим сложность работы алгоритма $T_i$ на данном шаге.\label{complexity_reduction_algorithm_description}
В начале необходимо найти два \gate{} в схеме, удовлетворяющих условию какой-либо замены.
В худшем случае для этого потребуется рассмотреть все возможные пары \gate{} в схеме $\frS^{(i)}$
($O(L_i^2)$ операций).
Затем для найденной пары \gate{} (если она существует) необходимо выяснить, могут ли рассматриваемые \gate{}
быть перемещены друг к другу в схеме без изменения её результирующего преобразования ($O(L_i)$ операций).
Следовательно, $T_i \gtrsim L_i^2$.
Заметим, что данная оценка является грубой, т.\,к. в случае, если рассматриваемая эквивалентная замена
не снижает сложность, а позволяет получить схему $\frS^{(i)'}$ с новыми \gate,
то алгоритм $A_{\text{reduction}}$ может быть применён ещё раз уже к этой схеме. И так до тех пор,
пока не будет получена схема $\frS^{(i+1)}$ сложности $L_{i+1} < L_i$;
если же такая обратимая схема не может быть получена, алгоритм заканчивает работу.
Данные требования позволяют говорить о сходимости описываемого алгоритма.
В зависимости от того, сколько раз будет применён этот рекурсивный поиск схемы $\frS^{(i+1)}$
для эквивалентных замен, не снижающих сложность, временн\'{а}я сложность $T_i$ может вырасти на несколько порядков
по сравнению с величиной $O(L_i^2)$.

Суммарная временн\'{а}я сложность $T(A_{\text{reduction}})$ может быть оценена следующим образом:
$$
    T(A_{\text{reduction}}) = \sum_{i=0}^k {T_i} \gtrsim \sum_{i=0}^k {L_i^2} \; .
$$
В худшем случае на каждом шаге алгоритма сложность снижается на 1, т.\,е. $L_0 - L_k = k$.
Таким образом, можно повысить нижнюю оценку для $T(A_{\text{reduction}})$:
\begin{equation}
    T(A_{\text{reduction}}) \gtrsim \sum_{i=L_k}^{L_0} {i^2} = \sum_{i=0}^{L_0} {i^2} - \sum_{i=0}^{L_k - 1} {i^2}
        = \frac{L_0(L_0 + 1)(2L_0 + 1) - (L_k - 1)L_k(2L_k - 1)}{6} \; .
    \label{formula_minimization_algorithm_time_complexity}
\end{equation}

Если $L_k = O(L_0)$, то формулу~\eqref{formula_minimization_algorithm_time_complexity} можно упростить:
$T(A_{\text{reduction}}) \gtrsim L_0^2$. Если же $L_k = o(L_0)$, то $T(A_{\text{reduction}}) \gtrsim L_0^3 \mathop / 3$.

Отсюда можно сделать следующий вывод: время работы описанного алгоритма снижения сложности обратимой схемы
нелинейно зависит от начального значения сложности этой схемы.
Поэтому в некоторых случаях может потребоваться синтезировать обратимую схему с минимально возможной сложностью,
чтобы уменьшить время, затрачиваемое на дальнейшее снижение сложности.
Этому вопросу будут посвящены следующие разделы данной главы.

\subsection{Снижение сложности схемы на этапе синтеза за счёт увеличения времени синтеза}
\label{subsection_complexity_reduction_during_synthesis}{}

\forceindent
Если в предыдущем разделе~\ref{subsection_common_scheme_complexity_reduction} предложенный алгоритм снижения
сложности обратимой схемы не зависел от алгоритма, при помощи которого данная схема была синтезирована,
то в данном разделе речь пойдёт о различных способах снижения сложности синтезируемой схемы на этапе синтеза
конкретно для алгоритма синтеза~\algref{alg_my_common}, описанного в разделе~\ref{subsection_my_synthesis_algorithms}.
Всего будет рассмотрено три таких способа: поиск грани булева куба, выбор между умножением подстановок справа либо слева,
поиск наилучшего представления цикла в виде произведения циклов меньшей длины.

\myparagraph[label_boolean_edge_search_paragraph]{Поиск грани булева куба}

\forceindent
На с.~\pageref{formula_bad_case_for_my_synthesis_algorithm} было описано обратимое преобразование $f$, которое может быть
задано одним элементом CNOT, но для которого алгорим~\algref{alg_my_common} синтезирует схему сложности $O(n2^n)$:
\begin{equation}
    f\left ( \langle x_1, x_2, \ldots, x_n \rangle \right ) = \langle x_1, x_2 \oplus x_1, x_3, \ldots, x_n \rangle \; .
    \label{formula_bad_case_for_my_synthesis_algorithm_second}
\end{equation}
Такой неоптимальный результат синтеза связан с тем, что алгоритм~\algref{alg_my_common} работает с представлением
подстановки $h_f$, задаваемой преобразованием $f$, в виде произведения пар независимых транспозиций. Каждая такая пара
транспозиций рассматривается независимо от других, что и приводит к избыточной сложности синтезированной схемы.
В качестве очевидного решения данной проблемы напрашивается переупорядочивание независимых транспозиций в представлении
подстановки $h_f$ по какому-либо признаку таким образом, чтобы для некоторой последовательности транспозиций в этом представлении
нашлась бы такая задающая её обратимая подсхема, использование которой в синтезированной схеме позволило бы снизить
суммарную сложность.

Для того, чтобы выделить такой <<объединяющий>> признак, введём некоторые обозначения.
Произвольную подстановку $h$ можно представить в виде произведения независимых циклов.
Такое представление единственно с точностью до порядка следования циклов в произведении и порядка следования элементов
в этих циклах. К примеру, $h = (ab) \circ (cde) = (dec) \circ (ba)$.
Для произвольной подстановки $h$ сумму длин циклов в её представлении в виде произведения
независимых циклов будем называть \textit{длиной} этой подстановки.

Будем говорить, что вектор $\vv x$ принадлежит циклу $c$, если $\vv x$ принадлежит множеству элементов этого цикла.
К примеру, $\vv y \in (\vv x, \vv y, \vv z)$, $\vv w \notin (\vv x, \vv y, \vv z)$. Аналогичным образом будем говорить,
что транспозиция $\tau = (\vv x, \vv y)$ принадлежит циклу $c$, если оба вектора этой транспозиции принадлежат циклу $c$.
В этом случае найдётся такая подстановка $h'$, что $c = \tau \circ h'$ и длина $h'$ на 1 меньше, чем длина цикла $c$:
\begin{enumerate} \label{formula_left_multiplications_simple}
    \item
        $c = (\vv x, \vv y, \vv a_1, \ldots, \vv a_k) = (\vv x, \vv y) \circ (\vv x, \vv a_1, \ldots, \vv a_k)$.
    \item
        $c = (\vv x, \vv a_1, \ldots, \vv a_k, \vv y, \vv b_1, \ldots, \vv b_p) = 
        (\vv x, \vv y) \circ (\vv x, \vv b_1, \ldots, \vv b_p) \circ (\vv y, \vv a_1, \ldots, \vv a_k)$.
\end{enumerate}
Если подстановка $h$ представляется в виде произведения независимых циклов $h = \mcirc_i {c_i}$, то запись $\tau \in h$
будет обозначать, что существует такой цикл $c_j$ в этом произведении, что $\tau \in c_j$.

Для двух векторов $\vv x, \vv y \in \ZZ_2^n$ обозначим через $\Delta(\vv x, \vv y)$ вектор,
равный $\Delta(\vv x, \vv y) = \langle x_1 \oplus y_1, \ldots, x_n \oplus y_n \rangle \in \ZZ_2^n$.
Будем называть $\Delta(\vv x, \vv y)$ \textit{вектором разницы} (или просто \textit{разницей}) векторов $\vv x$ и $\vv y$.
Для транспозиции $\tau = (\vv x, \vv y)$ разницу векторов $\vv x$ и $\vv y$ будем обозначать также через $\Delta(\tau)$.

Теперь для произвольного цикла $c$ можно ввести множество транспозиций $T(\vv d, c)$ следующим образом:
$$
    T(\vv d, c) \subseteq \{\,\tau=(\vv x_i, \vv y_i)\mid \tau \in c, \Delta(\tau)=\vv d\,\} \; .
$$
Другими словами, множество $T(\vv d, c)$~--- это множество транспозиций (не обязательно всех возможных),
принадлежащих циклу $c$, разница векторов для которых равняется заданному вектору $\vv d$.
Среди всех возможных множеств $T(\vv d, c)$ существует такое множество $T^*(\vv d, c)$, которое удовлетворяет двум условиям:
\begin{equation}
    \left \{\,
        \begin{aligned}
            c = \left( \mcirc_{\tau \in T^*(\vv d, c)}{\tau} \right ) \circ h' \; , \\
            \forall \tau' \in h' \colon \Delta(\tau') \neq \vv d\; .
        \end{aligned}
    \right .
\end{equation}
Таким образом, любое из возможных множеств $T^*(\vv d, c)$ позволяет задать представление цикла $c$ в виде произведения
транспозиций с $\Delta(\tau) = \vv d$ и, возможно, некоторой новой подстановки $h'$,
не имеющей транспозиций с разницей векторов, равной $\vv d$. При этом длина подстановки $h'$ меньше длины цикла $c$
на величину $|T^*(\vv d, c)|$.

Обозначим через $T_{\max}^*(\vv d, c)$ множество максимальной мощности среди всех возможных множеств $T^*(\vv d, c)$
(при этом возможна ситуация, когда таких множеств будет несколько).
Отметим\label{best_cycle_disjoint_question_presented},
что задача нахождения множества $T_{\max}^*(\vv d, c)$ не является тривиальной, способы его нахождения будут
исследоваться далее. На данном этапе будем считать, что нам подходит любое множество $T^*(\vv d, c)$.
Оно может быть получено следующим относительно простым способом:
\begin{enumerate}
    \item Пусть $T(\vv d, c) = \varnothing$.
    \item Найти какую-либо транспозицию $\tau \in c$, для которой $\Delta(\tau) = \vv d$, добавить её в множество $T(\vv d, c)$.
    \item Вычислить подстановку $h' = \tau \circ c$, найти представление подстановки $h'$ в виде произведения независимых циклов.
    \item Для каждого цикла из этого произведения повторить шаги (2)--(4).
    \item $T^*(\vv d, c) = T(\vv d, c)$.
\end{enumerate}

Для подстановки $h$, представимой в виде произведения независимых циклов $c_i$,
введём множество $T^*(\vv d, h)$ следующим образом:
$$
    T^*(\vv d, h) = \bigsqcup_i T^*(\vv d, c_i) \; .
$$
Тогда верно следующее равенство:
$$
    h = \left( \mcirc_{\tau \in T^*(\vv d, h)}{\tau} \right) \circ h' \; ,
$$
где $T^*(\vv d, h') = \varnothing$.

Теперь перейдём к рассмотрению множества векторов $\mathbb B_{\vv d, h} \subseteq \ZZ_2^n$, задаваемого следующим образом:
$$
    \mathbb B_{\vv d, h} = \left \{\,\vv x \in \ZZ_2^n\mid \exists \tau \in T^*(\vv d, h)\colon \vv x \in \tau \,\right\} \; .
$$
Другими словами, $\mathbb B_{\vv d, h}$~--- множество векторов, принадлежащих транспозициям из множества $T^*(\vv d, h)$.
Множество $\ZZ_2^n$ можно рассматривать как булев куб $\mathbb B^n$.
Будем искать грань этого куба $\mathbb B^{n,i_1, \ldots, i_k}_{\sigma_1, \ldots, \sigma_k} \subseteq \mathbb B_{\vv d, h}$,
удовлетворяющую условию:
\label{condition_for_boolean_edge_view}
$\vv d_{i_j} = 0$ для всех $i_1, \ldots, i_k$.
Другими словами, вектора транспозиций из множества $T^*(\vv d, h)$ не различаются в координатах $i_1, \ldots, i_k$.

Если такая грань $\mathbb B^{n,i_1, \ldots, i_k}_{\sigma_1, \ldots, \sigma_k} \subseteq \mathbb B_{\vv d, h}$
размерности $(n-k)$ существует, то для каждого вектора $\vv x \in \mathbb B^{n,i_1, \ldots, i_k}_{\sigma_1, \ldots, \sigma_k}$
найдётся парный ему вектор $\vv y \in \mathbb B^{n,i_1, \ldots, i_k}_{\sigma_1, \ldots, \sigma_k}$, такой что
транспозиция $(\vv x, \vv y) \in T^*(\vv d, h)$. Это следует из построения множества $\mathbb B_{\vv d, h}$ и ограничений,
которые мы наложили на грань $\mathbb B^{n,i_1, \ldots, i_k}_{\sigma_1, \ldots, \sigma_k}$.
Таким образом, данная грань включает в себя вектора $2^{n-k-1}$ транспозиций, т.\,к.
$\left | \mathbb B^{n,i_1, \ldots, i_k}_{\sigma_1, \ldots, \sigma_k} \right | = 2^{n-k}$.
Обозначим эти транспозиции $\tau_1, \ldots, \tau_{2^{n-k-1}}$.
Введём два множества, состоящих из $i_1, \ldots, i_k$:
\begin{enumerate}
    \item $I = \{\,i_j\mid \sigma_j = 1\,\}$~--- множество прямых контролирующих входов;
    \item $J = \{\,i_j\mid \sigma_j = 0\,\}$~--- множество инвертированных контролирующих входов;
\end{enumerate}
Тогда можно легко показать, что подстановка $g = \mcirc_{i=1}^{2^{n-k-1}}\tau_i$
задаётся обратимой схемой $\frS_g$:
\begin{equation}
    \frS_g = \compose_{i\colon \vv d_i = 1} {E(t_i,I,J)} \; ,
    \label{formula_scheme_for_boolean_edge}
\end{equation}
сложность которой $L(\frS_g) = w(\vv d)$, где $w(\vv d)$ обозначает вес Хэмминга вектора $\vv d$.
Следовательно, $L(\frS_g) \leqslant n$ (при использовании обобщённых элементов $E(t,I,J)$).
При этом выполняется равенство $h = g \circ h'$, где длина подстановки $h'$ меньше длины подстановки $h$ ровно на $2^{n-k-1}$.

С другой стороны, все транспозиции $\tau_1, \ldots, \tau_{2^{n-k-1}}$ независимы, потому что для них выполняется
условие $\Delta(\tau_i) = \vv d$. Согласно алгоритму~\algref{alg_my_common}, их можно разбить на $2^{n-k-2}$
пар независимых транспозиций, каждая из которых задаётся обратимой схемой сложности $O(n)$
(см. формулу~\eqref{formula_independent_pair_realization_complexity} на
с.~\pageref{formula_independent_pair_realization_complexity}). Следовательно, алгоритм~\algref{alg_my_common}
задал бы подстановку $g = \mcirc_{i=1}^{2^{n-k-1}}\tau_i$ обратимой схемой $\frS_g$
сложности $L(\frS_g) = 2^{n-k-2} \cdot O(n)$.

Отсюда следует, что для рассматриваемой грани $\mathbb B^{n,i_1, \ldots, i_k}_{\sigma_1, \ldots, \sigma_k}$
сложность обратимой подсхемы $\frS_g$, задающей подстановку $g = \mcirc_{i=1}^{2^{n-k-1}}\tau_i$,
можно уменьшить примерно в $2^{n-k-2}$ раз.

\bigskip
\bigskip
Опишем теперь кратко, как можно изменить алгоритм~\algref{alg_my_common}, чтобы применялся поиск грани булева куба.
Для заданной чётной подстановки $h \in A(\ZZ_2^n)$:
\begin{enumerate}
    \item \label{boolean_edge_algorithm_diff_search_step}
        Найти все вектора $\vv d_i \in \ZZ_2^n$, такие, что найдётся транспозиция $\tau_i \in h$, для которой
        $\Delta(\tau_i) = \vv d_i$.
    \item \label{boolean_edge_algorithm_sets_search_step}
        Для каждого найденного вектора $\vv d_i$ построить множества $T^*(\vv d_i, h)$ и $\mathbb B_{\vv d_i, h}$.
    \item \label{boolean_edge_algorithm_edges_search_step}
        Для каждого множества $\mathbb B_{\vv d_i, h}$ найти в булеве кубе $\mathbb B^n$
        грань $\mathbb B_i \subseteq \mathbb B_{\vv d_i, h}$ максимальной размерности, удовлетворяющую условиям,
        описанным на с.~\pageref{condition_for_boolean_edge_view}.
    \item \label{boolean_edge_algorithm_choose_edge_step}
        Среди всех найденных граней $\mathbb B_i$ выбрать грань $\mathbb B_j$ максимальной размерности $k$.
    \item \label{boolean_edge_algorithm_pair_implementation_step}
        Если $k < 2$ ($|\mathbb B_j| \leqslant 2$), то
        \begin{itemize}
            \item воспользоваться алгоритмом~\algref{alg_my_common} для синтеза подсхемы, задающей пару транспозиций $p$;
            \item получить новую подстановку $h' = p \circ h$.
        \end{itemize}
    \item \label{boolean_edge_algorithm_edge_implementation_step}
        Если $k \geqslant 2$ ($|\mathbb B_j| \geqslant 4$), то
        \begin{itemize}
            \item найти все $2^{k-1}$ транспозиций $\tau_1, \ldots, \tau_{2^{k-1}} \in T^*(\vv d_j, h)$,
                векторы которых принадлежат множеству $\mathbb B_j$;
            \item для подстановки $g = \mcirc_{i=1}^{2^{k-1}}\tau_i$ синтезировать схему $\frS_g$ согласно
                формуле~\eqref{formula_scheme_for_boolean_edge} и включить её в качестве подсхемы в итоговую схему;
            \item получить новую подстановку $h' = g \circ h$.
        \end{itemize}
    \item \label{boolean_edge_algorithm_reduction_step}
        Сделать замену $h = h'$ и повторить все шаги, начиная с 1-го.
        Полученная подстановка $h'$ будет иметь меньшую длину по сравнению с исходной подстановкой $h$,
        что позволяет говорить о сходимости описываемого алгоритма синтеза.
\end{enumerate}

Оценим примерную сложность каждого шага этого модифицированного алгоритма.
Обозначим длину подстановки $h$ через $l_h$. Для большинства подстановок $h \in \ZZ_2^n$ длина $l_h = O(2^n)$.
При выполнении шагов~\ref{boolean_edge_algorithm_diff_search_step} и~\ref{boolean_edge_algorithm_sets_search_step}
необходимо рассмотреть все возможные транспозиции, входящие в подстановку $h$. Для этого потребуется рассмотреть пары
векторов каждого с каждым, а значит, временн\'{а}я сложность выполнения этих шагов равна
$T_\text{1-2} = O(l_h^2) = O(2^{2n})$.

Всего существует $3^n$ различных граней булева куба $\mathbb B^n$.
На шаге~\ref{boolean_edge_algorithm_sets_search_step} может быть найдено не более $2^n$
различных множеств $\mathbb B_{\vv d_i, h}$, каждое из которых мощностью не более $2^n$ векторов.
Для того, чтобы понять, является ли множество векторов грани булева куба подмножеством какого-либо множества
$\mathbb B_{\vv d_i, h}$, необходимо просмотреть все векторы множества $\mathbb B_{\vv d_i, h}$.
Поэтому временн\'{у}ю сложность шага~\ref{boolean_edge_algorithm_edges_search_step}
можно грубо оценить сверху, как $T_3 \leqslant 12^n$.

Временн\'{а}я сложность шага~\ref{boolean_edge_algorithm_choose_edge_step} не превышает количества различных множеств
$\mathbb B_{\vv d_i, h}$, т.\,е. $T_4 \leqslant 2^n$.

Временн\'{а}я сложность шага~\ref{boolean_edge_algorithm_pair_implementation_step} равна $T_5 = O(n)$,
как было показано на с.~\pageref{formula_independent_pair_realization_complexity} для алгоритма~\algref{alg_my_common}.

На шаге~\ref{boolean_edge_algorithm_edge_implementation_step} необходимо получить все транспозиции подстановки $h$
из найденной грани булева куба $\mathbb B_j$, поэтому временн\'{а}я сложность этого шага $T_6 \leqslant 2^n$.

В самом худшем случае на шаге~\ref{boolean_edge_algorithm_reduction_step} длина подстановки будет сокращаться на 2,
если всё время будет выполняться шаг~\ref{boolean_edge_algorithm_pair_implementation_step}.
Таким образом, последний шаг будет выполняться не более $l_h \mathop / 2$ раз, а суммарная временн\'{а}я сложность
описанного алгоритма ограниченна сверху следующей величиной:
\begin{equation}
    T_\text{edge} \leqslant \frac{l_h}{2} \left( T_\text{1-2} + T_3 + T_4 + T_5 \right ) = O(24^n) \; .
\end{equation}
Шаг~\ref{boolean_edge_algorithm_diff_search_step} выполняется всегда, поэтому $T_\text{edge}$ ограничено снизу
временн\'{о}й сложностью этого шага
\begin{equation}
    T_\text{edge} \geqslant O(2^{2n}) \; .
\end{equation}

На первый взгляд, такое увеличение временн\'{о}й сложности алгоритма синтеза не оправдано снижением сложности
синтезированной схемы. Однако в случае, когда сложность синтезированной схемы равна $L = O(n2^n)$,
то, согласно формуле~\eqref{formula_minimization_algorithm_time_complexity}, время, затраченное на снижение
сложности такой схемы при помощи алгоритма, описанного на с.~\pageref{complexity_reduction_algorithm_description},
может превысить $O(L^3) = O(n^3 2^{3n})$. При этом такой алгоритм может дать в результате обратимую схему с б\'{о}льшей
сложностью, чем если бы её сложность была снижена на этапе синтеза приведённым выше способом.

\myparagraph{Умножение справа подстановки на произведение транспозиций}

\forceindent
На с.~\pageref{formula_left_multiplications_simple} было показано, что любой цикл длины больше двух можно представить
в виде произведения некоторой транспозиции и подстановки. Это было сделано при помощи умножения слева транспозиции
на подстановку. Легко показать, что любой цикл длины больше двух может быть также получен при помощи умножения справа
транспозиции на подстановку:
\begin{gather}
    (\vv x, \vv a_1, \ldots, \vv a_k, \vv y, \vv b_1, \ldots, \vv b_p) = 
        (\vv x, \vv y) \circ (\vv x, \vv b_1, \ldots, \vv b_p) \circ (\vv y, \vv a_1, \ldots, \vv a_k)
    \label{formula_left_multiplication} \; , \\
    (\vv x, \vv a_1, \ldots, \vv a_k, \vv y, \vv b_1, \ldots, \vv b_p) = 
        (\vv y, \vv b_1, \ldots, \vv b_p) \circ (\vv x, \vv a_1, \ldots, \vv a_k) \circ (\vv x, \vv y)
    \label{formula_right_multiplication} \; .
\end{gather}
Отсюда следует, что и произвольная подстановка может быть получена при помощи как произведения слева, так и произведения
справа транспозиции и некоторой новой подстановки.
По формулам~\eqref{formula_left_multiplication} и~\eqref{formula_right_multiplication} видно, что при этом новая
подстановка будет иметь различный вид.

На шагах~\ref{boolean_edge_algorithm_pair_implementation_step} и~\ref{boolean_edge_algorithm_edge_implementation_step}
на с.~\pageref{boolean_edge_algorithm_pair_implementation_step} использовалось умножение слева
транспозиции и произведения нескольких транспозиций на новую подстановку. Однако если использовать умножение справа,
то далее на шаге~\ref{boolean_edge_algorithm_diff_search_step} могут быть получены другие вектора $\vv d_i$,
а на шаге~\ref{boolean_edge_algorithm_sets_search_step}~--- другие множества $T^*(\vv d_i, h)$ и $\mathbb B_{\vv d_i, h}$,
имеющие, возможно, б\'{о}льшую мощность.
Всё это в итоге может повлиять на выбор грани $\mathbb B_j$ на шаге~\ref{boolean_edge_algorithm_choose_edge_step}
и на сложность синтезированной схемы.

Не представляется возможным однозначно выбирать левое или правое произведение на каждом этапе работы алгоритма синтеза,
чтобы в итоге получить обратимую схему с меньшей сложностью.
Однако можно использовать следующий подход:
\begin{itemize}
    \item
        вначале на шагах~\ref{boolean_edge_algorithm_pair_implementation_step}
        и~\ref{boolean_edge_algorithm_edge_implementation_step}
        применить левое произведение, получить новую подстановку $h_\text{left}$ и найти для неё грань булева куба
        $\mathbb B_\text{left}$ максимальной размерности на следующем этапе алгоритма;

    \item
        затем на шагах~\ref{boolean_edge_algorithm_pair_implementation_step}
        и~\ref{boolean_edge_algorithm_edge_implementation_step}
        применить правое произведение, получить новую подстановку $h_\text{right}$ и найти для неё грань булева куба
        $\mathbb B_\text{right}$ максимальной размерности на следующем этапе алгоритма;
        
    \item
        выбрать между $\mathbb B_\text{left}$ и $\mathbb B_\text{right}$ грань наибольшей размерности и соответствующее
        ей левое или правое произведение для данного этапа работы алгоритма синтеза;
        
    \item
        перейти к следующему этапу работы алгоритма синтеза.
\end{itemize}
Если грани $\mathbb B_\text{left}$ и $\mathbb B_\text{right}$ имеют одинаковую размерность, то для однозначности можно
всегда выбирать в таком случае левое произведение.

Данный подход увеличивает временн\'{у}ю сложность алгоритма синтеза в 2 раза, однако в некоторых случаях он позволяет
получить в результате работы такого алгоритма обратимую схему с меньшей сложностью, по сравнению с
алгоритмом синтеза, использующим только левое или только правое произведение.

\myparagraph{Увеличение мощности множества транспозиций с одинаковой разницей векторов}

\forceindent
На с.~\pageref{best_cycle_disjoint_question_presented} было сказано, что для фиксированного вектора $\vv d$ и цикла $c$
задача нахождения множества $T_{\max}^*(\vv d, c)$ не является тривиальной. Покажем это на простом примере.

Пусть $c = (\vv x_1, \vv x_2, \vv x_3,\vv y_1, \vv y_3, \vv y_2 )$, $\Delta(\vv x_i, \vv y_i) = \vv d$,
$\Delta(\vv x_i, \vv y_j) \neq \vv d$ при $i \neq j$,
$\Delta(\vv x_i, \vv x_j) \neq \vv d$, $\Delta(\vv y_i, \vv y_j) \neq \vv d$.
Тогда существует множество $T(\vv d, c) = \{\,(\vv x_1, \vv y_1), (\vv x_2, \vv y_2), (\vv x_3, \vv y_3)\,\}$.
Воспользовавшись формулой~\eqref{formula_left_multiplication}, мы можем представить цикл $c$ двумя способами:
\begin{enumerate}
    \item
        $c = (\vv x_1, \vv y_1) \circ (\vv x_1, \vv y_3, \vv y_2) \circ (\vv y_1, \vv x_2, \vv x_3)$;

    \item
        $c = (\vv x_2, \vv y_2) \circ (\vv x_3, \vv y_3) \circ (\vv x_1, \vv x_2) \circ
             (\vv x_3, \vv y_2) \circ (\vv y_1, \vv y_3)$.
\end{enumerate}
Таким образом, мы нашли два множества:
\begin{align*}
    T_1^*(\vv d, c) &= \{\,(\vv x_1, \vv y_1)\,\} \; , \\
    T_2^*(\vv d, c) &= \{\,(\vv x_2, \vv y_2), (\vv x_3, \vv y_3)\,\} \; .
\end{align*}
Очевидно, что в алгоритме синтеза желательно всегда получать такие множества, как $T_2^*(\vv d, c)$, а не $T_1^*(\vv d, c)$,
т.\,к. это, возможно, позволит найти грань булева куба большей размерности.

Для решения данной задачи можно использовать следующий алгоритм.
Пусть вначале $T(\vv d, c) = \varnothing$, тогда
\begin{enumerate}
    \item
        Зафиксировать представление цикла $c$, чтобы можно было пронумеровать элементы этого цикла слева направо от 1 до $l_c$,
        где $l_c$~--- длина цикла $c$. Сами элементы цикла будем обозначать через $\vv a_i$, где $1 \leqslant i \leqslant l_c$.
        В итоге цикл $c$ можно однозначно записать как $c = (\vv a_1, \ldots, \vv a_{l_c})$.
    \item
        Для заданного значения вектора разницы $\vv d$ посторить множество упорядоченных пар индексов:
        $$
            M = \{\,(i, j)\mid i < j, \Delta(\vv a_i, \vv a_j) = \vv d\,\} \; .
        $$
    \item
        Используя фильтрующую функцию $\phi(x, (i, j))$ вида
        $$
            \phi(x, (i, j)) =
                \begin{cases}
                1, &\text{если $i \leqslant x \leqslant j$} \; , \\
                0  &\text{иначе} \; .
                \end{cases}
        $$
        и весовую функцию $w(x)$ вида
        $$
            w(x) = \sum_{(i, j) \in M} \phi(x, (i, j)) \; ,
        $$
        найти такую пару индексов $(i, j) \in M$, для которой значение функции
        $f(i, j) = w(i) + w(j)$
        является минимальным (если таких пар несколько, выбрать любую).
        
    \item
        Добавить транспозицию $\tau = (\vv a_i, \vv a_j)$ в множество $T(\vv d, c)$.
        
    \item
        Получить новую подстановку $h' = \tau \circ c$ (либо $h' = c \circ \tau$,
        в зависимости от требуемого вида произведения).
        Найти представление подстановки $h'$ в виде произведения независимых циклов,
        для каждого из которых повторить данный алгоритм, начиная с самого первого шага.
        Отметим, что длина полученной подстановки $h'$ будет меньше, чем исходного цикла $c$,
        поэтому можно говорить о сходимости данного алгоритма.
\end{enumerate}
В итоге получим множество $T(\vv d, c) = T^*(\vv d, c)$, мощность которого близка к $T_\text{max}^*(\vv d, c)$.

Цикл $c$ можно рассматривать как отрезок, разделённый на части элементами множества $M$.
Эти элементы также можно рассматривать как отрезки, которые, возможно, частично перекрывают друг друга.
Тогда весовая функция $w(x)$ показывает, скольким отрезкам принадлежит точка $x$.
Таким образом, предложенный алгоритм представляет собой поиск отрезка, у которого концы принадлежат наименьшему количеству
отрезков. 

Поясним данный алгоритм на примере цикла, описанного в начале данного параграфа:
\begin{itemize}
    \item
        $c = (\vv x_1, \vv x_2, \vv x_3,\vv y_1, \vv y_3, \vv y_2 )$, $\Delta(\vv x_i, \vv y_i) = \vv d$,
        $\Delta(\vv x_i, \vv y_j) \neq \vv d$ при $i \neq j$,
        $\Delta(\vv x_i, \vv x_j) \neq \vv d$, $\Delta(\vv y_i, \vv y_j) \neq \vv d$;
        
    \item
        фиксируем представление цикла $c$: $\vv a_1 = \vv x_1, \ldots, \vv a_6 = \vv y_2$;
        
    \item
        строим множество $M = \{\,(1, 4), (2, 6), (3, 5)\,\}$;
        
    \item
        ищем значение функции $w(x)$: $w(1) = 1$, $w(2) = 2$, $w(3) = 3$, $w(4) = 3$, $w(5) = 2$, $w(6) = 1$;
        
    \item
        ищем значение функции $f(i, j)$: $f(1, 4) = 4$, $f(2, 6) = 3$, $f(3, 5) = 5$;
        
    \item
        добавляем в множество $T(\vv d, c)$ транспозицию $\tau_1 = (\vv a_2, \vv a_6) = (\vv x_2, \vv y_2)$;
        
    \item
        $h' = \tau_1 \circ c = (\vv x_1, \vv x_2) \circ (\vv x_3, \vv y_1, \vv y_3, \vv y_2)$;
        
    \item
        в подстановке $h'$ только для транспозиции $\tau_2 = (\vv x_3, \vv y_3)$ значение разницы
        $\Delta(\tau_2) = \vv d$, добавляем $\tau_2$ в множество $T(\vv d, c)$;
        
    \item
        $h'' = \tau_2 \circ h' = (\vv x_1, \vv x_2) \circ (\vv x_3, \vv y_2) \circ (\vv y_1, \vv y_3)$;
        
    \item
        для всех транспозиций $\tau \in h''$ значение разницы $\Delta(\tau) \neq \vv d$, закончить работу.
\end{itemize}
В этом примере получено множество $T(\vv d, c) = \{\,(\vv x_2, \vv y_2), (\vv x_3, \vv y_3)\,\} = T_2^*(\vv d, c)$,
что и требовалось от алгоритма.

Предложенный алгоритм по увеличению мощности множества транспозиций с одинаковой разницей векторов
позволяет в некоторых случаях находить грани булева куба большей размерности
(см. параграф~\ref{label_boolean_edge_search_paragraph}, шаг~\ref{boolean_edge_algorithm_sets_search_step}
на с.~\pageref{boolean_edge_algorithm_sets_search_step}) и синтезировать обратимую схему с меньшей
сложностью. Если длина входной подстановки $h$ равна $l_h$, то временн\'{а}я сложность данного алгоритма
равна $O(l_h \log_2 l_h)$: основное время занимает составление множества $M$, т.\,к. приходится сравнивать друг с другом
элементы циклов подстановки $h$; поиск минимального значения функции $f(i, j)$ делается за линейное время
при обходе всех элементов циклов подстановки $h$. При $l_h = O(2^n)$ временн\'{а}я сложность равна $O(n2^n)$.

\subsection{Экспериментальные результаты снижения сложности схемы}

\forceindent
Для того, чтобы проверить целесообразность применимости на практике различных способов снижения
сложности обратимой схемы, было разработано программное обеспечение,
реализующее алгоритм синтеза~\algref{alg_my_common}. На стадии синтеза схемы данным программным обеспечением применяются
все способы снижения сложности, описанные в разделе~\ref{subsection_complexity_reduction_during_synthesis}.
На заключительной стадии работы к схеме применяются правила эквивалентных замен по алгоритму, описанному
в разделе~\ref{subsection_common_scheme_complexity_reduction}.

В качестве тестовой функции было решено использовать функцию $\mathbf{rd53}\colon \ZZ_2^5\to\ZZ_2^3$, принимающую на вход
двоичный вектор размерности 5 и выдающую на выходе двоичное представление веса Хэмминга этого вектора.
По определению, такая функция не может быть обратимой и, как следствие, не может быть задана обратимой схемой.
Однако можно найти обратимую схему, реализующую \textbf{rd53}. Для этой цели авторами работы~\cite{miller_spectral}
была задана обратимая функция $f\colon \ZZ_2^7 \to \ZZ_2^7$, часть входов/выходов которой соответствует входам/выходам
функции \textbf{rd53}. В работе~\cite{miller_transform_based} авторами была найдена обратимая схема сложности 12,
задающая функцию $f$ и, как следствие, реализующая функцию \textbf{rd53} (рис.~\ref{pic_rd53_scheme}).

\Figure[ht]
    \centering
    \includegraphics[scale=1.2]{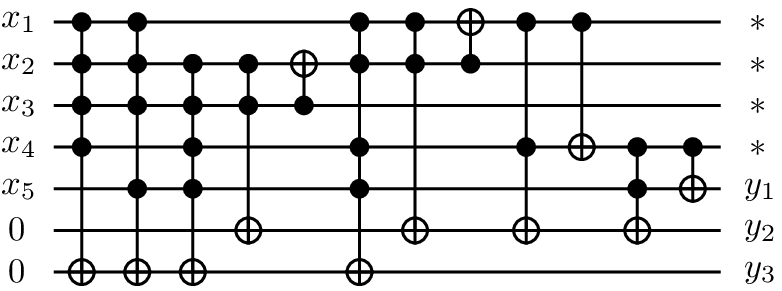}
    \caption
    {
        \small Обратимая схема, реализующая функцию \textbf{rd53}.
    }\label{pic_rd53_scheme}
\end{figure}

На сегодняшний день данная обратимая схема имеет минимальную сложность среди всех обратимых схем,
реализующих функцию \textbf{rd53}, имеющих 7 входов и состоящих из элементов $k$-CNOT.
Данная схема была принята в качестве образца для сравнения сложности при проведении эксперимента.
Обозначим эту схему $\frS_\mathbf{rd53}$, тогда $L(\frS_\mathbf{rd53}) = 12$.
Сам эксперимент заключался в следующем:
\begin{enumerate}
    \item
        Схема $\frS_\mathbf{rd53}$ представляется в виде композиции элементов $k$-CNOT:
        $$
            \frS_\mathbf{rd53} = \compose_{j=1}^{L(\frS_\mathbf{rd53})}{E(t_j, I_j)} \; .
        $$
    \item
        Для всех значений $l$ от 1 до $L(\frS_\mathbf{rd53})$ формируется подсхема
        $\frS_l = \compose_{j=1}^{l}{E(t_j, I_j)}$ (первые $l$ элементов схемы $\frS_\mathbf{rd53}$).
        
    \item
        Для каждой схемы $\frS_l$ находится задаваемая ей подстановка $h_l \in A(\ZZ_2^7)$.
        
    \item
        Каждая подстановка $h_l$ подаётся в качестве входа разработанному программному обеспечению по синтезу обратимых схем.
        
    \item
        Замеряется время синтеза схемы $\frS'_l$ и её сложность $L(\frS'_l)$, которая затем сравнивается
        со сложностью эталонной схемы $\frS_l$.
\end{enumerate}
Во время проведения эксперимента измерялись следующие величины:
\begin{itemize}
    \item
        $T_s$ и $L_s$~--- время синтеза схемы $\frS'_l$ без применения правил эквивалентных замен композиций \gate{}
        и её сложность соответственно (схема содержит обобщённые элементы $E(t, I, J)$);
    \item
        $T_{r,1}$ и $L_{r,1}$~--- время, затраченное на снижение сложности синтезированной схемы при помощи
        правил эквивалентных замен композиций \gate{} (за исключением замен~\ref{replace_recursive_simple}
        и~\ref{replace_recursive}), и сложность полученной схемы соответственно
        (схема содержит обобщённые элементы $E(t, I, J)$);
    \item
        $T_{r,2}$ и $L_{r,2}$~--- время, затраченное на снижение сложности синтезированной схемы при помощи
        всех правил эквивалентных замен композиций \gate, и сложность полученной схемы соответственно
        (схема содержит только элементы $k$-CNOT).
\end{itemize}

Результаты измерений представлены в таблице~\ref{table_rd53_synthesis_results}.
Отметим, что в данной таблице представлены наилучшие результаты, при которых достигалась наименьшая сложность
синтезированной схемы (за счёт увеличенного времени работы).
В некоторых случаях для этого требовалось менять набор правил эквивалентных замен композиций \gate{}
и последовательность их применения. Также стоит отметить, что алгоритм~\algref{alg_my_common} без применения
различных способов снижения сложности обратимой схемы синтезирует все схемы $\frS'_l$ со сложностью
порядка $7 \cdot 2^7 = 896$ элементов.

{
    \renewcommand{\baselinestretch}{1.2}

    \Table[ht]
        \small
        \centering
        \begin{tabular}{|*{7}{c|}}
            \hline
            $l$ & $T_s$, мс & $L_s$ & $T_{r,1}$, мс & $L_{r,1}$ & $T_{r,2}$, мс & $L_{r,2}$ \tabularnewline
            \hline
            
             1 & 0,09 &  1 &  0,00 &  1 & 0,01 &  1 \tabularnewline
            \hline
             2 & 0,10 &  2 &  0,00 &  2 & 0,02 &  2 \tabularnewline
            \hline
             3 & 0,27 &  3 &  0,00 &  3 & 0,01 &  3 \tabularnewline
            \hline
             4 & 0,31 &  5 &  0,00 &  4 & 0,01 &  4 \tabularnewline
            \hline
             5 & 1,11 &  6 &  0,01 &  5 & 0,02 &  5 \tabularnewline
            \hline
             6 & 1,03 &  7 &  0,01 &  6 & 0,02 &  6 \tabularnewline
            \hline
             7 & 1,54 &  9 &  0,02 &  7 & 0,03 &  7 \tabularnewline
            \hline
             8 & 2,98 & 10 &  0,02 &  8 & 0,04 &  8 \tabularnewline
            \hline
             9 & 3,85 & 16 &  3,25 &  9 & 0,05 &  9 \tabularnewline
            \hline
            10 & 6,39 & 16 &  5,13 & 10 & 0,07 & 13 \tabularnewline
            \hline
            11 & 6,83 & 24 & 13,46 & 11 & 0,07 & 12 \tabularnewline
            \hline
            12 & 14,2 & 30 & 29,75 & 16 & 0,08 & 13 \tabularnewline
            \hline

        \end{tabular}
        \caption{
            \small Экспериментальные результаты работы разработанного программного обеспечения
                по синтезу подсхем схемы $\frS_\mathbf{rd53}$.
        }\label{table_rd53_synthesis_results}
    \end{table}

} 
Проанализировав данные таблицы~\ref{table_rd53_synthesis_results}, можно сделать вывод, что применение
правил эквивалентных замен композиций \gate~\ref{replace_recursive_simple} и~\ref{replace_recursive} позволяет
в некоторых случаях ($l = 12$) снизить сложность синтезированной схемы. Этот вывод не является очевидным,
т.\,к. по данным правилам замен все обобщённые элементы $E(t, I, J)$ в схеме заменяются на композицию элементов $E(t, I)$,
что должно было бы увеличить, а не снизить сложность.

Из таблицы~\ref{table_rd53_synthesis_results} видно, что разработанное программное обеспечение, реализующее
алгоритм синтеза~\algref{alg_my_common} с применением различных способов снижения сложности обратимой схемы,
позволяет синтезировать схемы, имеющие сложность, близкую к минимальной. При этом время синтеза очень мало.
Для сравнения, в работе~\cite{miller_transform_based} схема $\frS_\mathbf{rd53}$ была получена авторами за 1,84 секунды
на процессоре Pentium III 750 MHz, в то время как схема $\frS'_{12}$, имеющая всего на 1 элемент $k$-CNOT больше,
получена автором за 15 миллисекунд на процессоре Core i7 2,40 GHz.

В работе~\cite{maslov_rm_synthesis} был описан достаточно простой и быстрый алгоритм синтеза обратимых схем
без дополнительной памяти, использующий спектр Рида-Маллера (полином Жегалкина). Данный алгоритм эффективен для синтеза
булевых преобразований, описываемых полиномами Жегалкина малой степени, но становится неэффективным для полиномов
высокой степени. В разработанном программном обеспечении~\cite{my_program} были объединены алгоритм~\algref{alg_my_common}
и алгоритм из работы~\cite{maslov_rm_synthesis}, а также различные способы снижения сложности обратимой схемы,
описанные в данной главе и в Главе~\ref{section_asymptotic_bounds}. Алгоритм~\algref{alg_my_common} в данном
программном обеспечении предназначен для работы с остаточными полиномами Жегалкина, имеющими высокую степень.
При помощи данного программного обеспечения удалось получить новую обратимую схему, реализующую функцию \textbf{rd53}.
Полученная схема имеет 7 входов, её сложность равна 11, квантовый вес равен 96 (см. рис.~\ref{pic_smallest_rd53}).
На сегодняшний день эта обратимая схема имеет минимальную сложность среди всех известных обратимых схем,
состоящих из обобщённых элементов Тоффоли с инвертированными входами и реализующих функцию \textbf{rd53}
с 7-ю входами~\cite{maslov_benchmarks,revlib}. Более того, все остальные схемы имеют сложность 12 или выше.

\Figure[ht]
    \centering
    \includegraphics[scale=1.2]{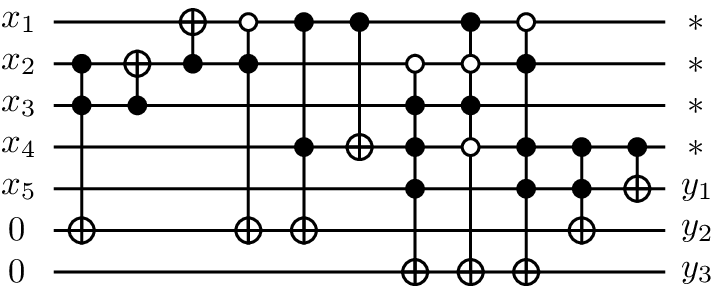}
    \caption
    {
        \small Полученная обратимая схема с наименьшей известной сложностью, реализующая функцию \textbf{rd53}
        и имеющая 7 входов.
    }\label{pic_smallest_rd53}
\end{figure}

Также был проведён ряд экспериментов по синтезу обратимых схем, реализующих различные булевы отображения, при помощи
разработанного программного обеспечения~\cite{my_program}. Спецификация функций, их название и уже построенные обратимые
схемы для них брались с сайтов~\cite{maslov_benchmarks,revlib}.
Удалось получить более 40 новых обратимых схем, имеющих либо меньшее количество входов, либо меньшую сложность,
либо меньший квантовый вес, по сравнению с известными результатами.
Измерение характеристик всех обратимых схем (полученных и известных) производилось
при помощи программного обеспечения RCViewer+~\cite{rcviewer}. Результаты экспериментов сведены в следующих таблицах.

В таблице~\ref{table_exp_less_input_count} приведены характеристики полученных обратимых схем, имеющих меньшее количество
входов по сравнению с известными результатами;
в таблице~\ref{table_exp_less_gate_complexity}~--- имеющих меньшую сложность по сравнению с известными результатами;
в таблице~\ref{table_exp_less_quantum_cost}~--- имеющих меньший квантовый вес по сравнению с известными результатами.
В первой колонке всех таблиц приведено имя функции, взятое с сайта~\cite{maslov_benchmarks}. Далее таблицы разбиваются на две
части: в левой части указаны характеристики полученных схем, в правой части~--- наилучших известных схем.
Обозначения: $n$~--- количество входов схемы, $L(\frS)$~--- сложность схемы, $W(\frS)$~--- квантовый вес схемы.
Нижний индекс <<$\min$>> приписывается к той характеристике, по которой идёт сравнение.
К примеру, в таблице~\ref{table_exp_less_gate_complexity} идёт сравнение по сложности схемы, поэтому колонка сложности
имеет обозначение $L_{\min}(\frS)$. Все остальные характеристики указаны для этой же схемы с данной минимальной характеристикой.

\vspace{-0.3cm}
{
    \renewcommand{\baselinestretch}{1.2}
    \begin{center}
        \small
        \Longtable{|c||c|c|c||c|c|c|}
            \hline
            
            \multirow{2}{*}{
                \textbf{Название функции}
            } &
            \multicolumn{3}{c||}{\textbf{Полученные схемы}} &
            \multicolumn{3}{c|}{\textbf{Существующие схемы}} \\
            \cline{2-7}
            
            & $n_{\min}$ & $L(\frS)$ & $W(\frS)$ & $n_{\min}$ & $L(\frS)$ & $W(\frS)$\\
            \hline
            \endhead
            
            gf2\verb"^"3mult & \textbf{7} & 73 & 740 &
                \multirow{3}{*}{\textbf{9}} & \multirow{3}{*}{11} & \multirow{3}{*}{47} \\
            \cline{1-4}
            gf2\verb"^"3mult & \textbf{7} & 79 & 712 & & & \\
            \cline{1-4}
            gf2\verb"^"3mult & \textbf{7} & 145 & 704 & & & \\
            \hline
            gf2\verb"^"4mult & \textbf{9} & 415 & 47649 &
                \multirow{2}{*}{\textbf{12}} & \multirow{2}{*}{19} & \multirow{2}{*}{83} \\
            \cline{1-4}
            gf2\verb"^"4mult & \textbf{9} & 1834 & 5914 & & & \\
            \hline
            nth\_prime9\_inc & \textbf{9} & 3942 & 19313 & \textbf{10} & 7522 & 17975 \\
            \hline
            rd73 & \textbf{9} & 296 & 43421 & \multirow{2}{*}{\textbf{10}} & \multirow{2}{*}{20} & \multirow{2}{*}{64} \\
            \cline{1-4}
            rd73 & \textbf{9} & 835 & 4069 & & & \\
            \hline            
            rd84 & \textbf{11} & 679 & 359384 & \multirow{2}{*}{\textbf{15}} & \multirow{2}{*}{28} & \multirow{2}{*}{98} \\
            \cline{1-4}
            rd84 & \textbf{11} & 2560 & 12397 & & & \\
            \hline

            \caption{
                \small Характеристики обратимых схем, полученных при помощи разработанного программного обеспечения%
                ~\cite{my_program}, имеющих меньшее количество входов по сравнению с известными результатами.
            }\label{table_exp_less_input_count}
        \end{longtable}
    \end{center}
} 

\clearpage
{
    \renewcommand{\baselinestretch}{1.2}
    \begin{center}
        \small
        \centering
        \Longtable{|c||c|c|c||c|c|c|}
            \hline
            
            \multirow{2}{*}{
                \textbf{Название функции}
            } &
            \multicolumn{3}{c||}{\textbf{Полученные схемы}} &
            \multicolumn{3}{c|}{\textbf{Существующие схемы}} \\
            \cline{2-7}
            
            & $n$ & $L_{\min}(\frS)$ & $W(\frS)$ & $n$ & $L_{\min}(\frS)$ & $W(\frS)$\\
            \hline
            \endhead

            2of5 & 6 & \textbf{9} & 268 & \multirow{2}{*}{6} & \multirow{2}{*}{\textbf{15}} & \multirow{2}{*}{107} \\
            \cline{1-4}
            2of5 & 6 & \textbf{10} & 118 & & & \\
            \hline
            2of5 & 7 & \textbf{11} & 32 & 7 & \textbf{12} & 32 \\
            \hline
            3\_17 & 3 & \textbf{4} & 14 & \multirow{2}{*}{3} & \multirow{2}{*}{\textbf{6}} & \multirow{2}{*}{12} \\
            \cline{1-4}
            3\_17 & 3 & \textbf{5} & 13 & & & \\
            \hline
            4b15g\_2 & 4 & \textbf{12} & 57 & 4 & \textbf{15} & 31 \\
            \hline
            4b15g\_4 & 4 & \textbf{12} & 49 & \multirow{2}{*}{4} & \multirow{2}{*}{\textbf{15}} & \multirow{2}{*}{35} \\
            \cline{1-4}
            4b15g\_4 & 4 & \textbf{14} & 47 & & & \\
            \hline
            4b15g\_5 & 4 & \textbf{14} & 72 & 4 & \textbf{15} & 29 \\
            \hline
            4mod5 & 5 & \textbf{4} & 13 & 5 & \textbf{5} & 7 \\
            \hline
            5mod5 & 6 & \textbf{7} & 429 & 6 & \textbf{8} & 84 \\
            \hline
            6sym & 7 & \textbf{14} & 1308 & \multirow{2}{*}{7} & \multirow{2}{*}{\textbf{36}} & \multirow{2}{*}{777} \\
            \cline{1-4}
            6sym & 7 & \textbf{15} & 825 & & & \\
            \hline
            9sym & 10 & \textbf{73} & 61928 & \multirow{2}{*}{10} & \multirow{2}{*}{\textbf{129}} & \multirow{2}{*}{6941} \\
            \cline{1-4}
            9sym & 10 & \textbf{74} & 31819 & & & \\
            \hline
            ham7 & 7 & \textbf{19} & 77 & 7 & \textbf{25} & 49 \\
            \hline
            hwb12 & 12 & \textbf{42095} & 134316 & 12 & \textbf{55998} & 198928 \\
            \hline
            nth\_prime7\_inc & 7 & \textbf{427} & 10970 &
                \multirow{3}{*}{7} & \multirow{3}{*}{\textbf{1427}} & \multirow{3}{*}{3172} \\
            \cline{1-4}
            nth\_prime7\_inc & 7 & \textbf{474} & 10879 & & & \\
            \cline{1-4}
            nth\_prime7\_inc & 7 & \textbf{824} & 2269 & & & \\
            \hline
            nth\_prime8\_inc & 8 & \textbf{977} & 10218 &
                \multirow{2}{*}{8} & \multirow{2}{*}{\textbf{3346}} & \multirow{2}{*}{7618} \\
            \cline{1-4}
            nth\_prime8\_inc & 8 & \textbf{1683} & 6330 & & & \\
            \hline
            nth\_prime9\_inc & 10 & \textbf{2234} & 22181 & 10 & \textbf{7522} & 17975 \\
            \hline
            nth\_prime10\_inc & 11 & \textbf{5207} & 50152 & 11 & \textbf{16626} & 40299 \\
            \hline
            nth\_prime11\_inc & 12 & \textbf{11765} & 124408 & 12 & \textbf{35335} & 95431 \\
            \hline
            rd53 & 7 & \textbf{11} & 96 & 7 & \textbf{12} & 120 \\
            \hline
            
            \caption{
                \small Характеристики обратимых схем, полученных при помощи разработанного программного обеспечения%
                ~\cite{my_program}, имеющих меньшую сложность по сравнению с известными результатами.
            }\label{table_exp_less_gate_complexity}
        \end{longtable}
    \end{center}
} 

\vspace{-1cm}
Полученные экспериментальные результаты позволяют утверждать, что предложенные в данной и последующей главах
способы снижения сложности обратимых схем являются применимыми на практике и в некоторых случаях позволяют синтезировать
обратимые схемы, имеющие лучшие характеристики, по сравнению с известными результатами.

\clearpage
{
    \renewcommand{\baselinestretch}{1.2}
    \begin{center}
        \small
        \centering
        \Longtable{|c||c|c|c||c|c|c|}
            \hline
            
            \multirow{2}{*}{
                \textbf{Название функции}
            } &
            \multicolumn{3}{c||}{\textbf{Полученные схемы}} &
            \multicolumn{3}{c|}{\textbf{Существующие схемы}} \\
            \cline{2-7}
            
            & $n$ & $L(\frS)$ & $W_{\min}(\frS)$ & $n$ & $L(\frS)$ & $W_{\min}(\frS)$\\
            \hline
            \endhead

            2of5 & 7 & 12 & \textbf{31} & 7 & 12 & \textbf{32} \\
            \hline
            6sym & 7 & 41 & \textbf{206} & 7 & 36 & \textbf{777} \\
            \hline
            9sym & 10 & 347 & \textbf{1975} & 10 & 210 & \textbf{4368} \\
            \hline
            hwb7 & 7 & 603 & \textbf{1728} & 7 & 331 & \textbf{2611} \\
            \hline
            hwb8 & 8 & 1594 & \textbf{4852} & 8 & 2710 & \textbf{6940} \\
            \hline
            hwb9 & 9 & 3999 & \textbf{12278} & 9 & 6563 & \textbf{16173} \\
            \hline
            hwb10 & 10 & 8247 & \textbf{26084} & 10 & 12288 & \textbf{35618} \\
            \hline
            hwb11 & 11 & 21432 & \textbf{69138} & 11 & 32261 & \textbf{90745} \\
            \hline
            hwb12 & 12 & 42095 & \textbf{134316} & 12 & 55998 & \textbf{198928} \\
            \hline
            nth\_prime7\_inc & 7 & 824 & \textbf{2269} & 7 & 1427 & \textbf{3172} \\
            \hline
            nth\_prime8\_inc & 8 & 1683 & \textbf{6330} & 8 & 3346 & \textbf{7618} \\
            \hline
            rd53 & 7 & 12 & \textbf{82} & \multirow{3}{*}{7} & \multirow{3}{*}{12} & \multirow{3}{*}{\textbf{120}} \\
            \cline{1-4}
            rd53 & 7 & 12 & \textbf{95} & & & \\
            \cline{1-4}
            rd53 & 7 & 11 & \textbf{96} & & & \\
            \hline

            \caption{
                \small Характеристики обратимых схем, полученных при помощи разработанного программного обеспечения%
                ~\cite{my_program}, имеющих меньший квантовый вес по сравнению с известными результатами.
            }\label{table_exp_less_quantum_cost}
        \end{longtable}
    \end{center}
} 

\sectionenumerated[section_asymptotic_bounds]{Оценка сложности, глубины и квантового веса обратимых схем}

\forceindent
Теория схемной сложности берёт своё начало с работы К.~Шеннона~\cite{shannon}. В ней он предложил в качестве меры сложности
булевой функции рассматривать сложность минимальной контактной схемы, реализующей эту функцию.
О.\,Б.~Лупановым была установлена~\cite{lupanov_one_method} асимптотика сложности $L(n) \sim \rho 2^n \mathop / n$
булевой функции от $n$ переменных в произвольном конечном полном базисе элементов с произвольными положительными весами,
где $\rho$ обозначает минимальный приведённый вес элементов базиса.

В работе~\cite{karpova} рассматривается вопрос о вычислениях с ограниченной памятью. Н.\,А.~Карповой было доказано,
что в базисе классических \gate, реализующих все \mbox{$p$-местные} булевы функции,
асимптотическая оценка функции Шеннона сложности схемы с тремя и более регистрами памяти зависит от значения $p$,
но не изменяется при увеличении количества используемых регистров памяти. Также было показано, что существует булева функция,
которая не может быть реализована в маломестных базисах с использованием менее, чем двух регистров памяти.

О.\,Б.~Лупановым были рассмотрены схемы из \gate{} с задержками~\cite{lupanov_delay}. Было доказано, 
что в регулярном базисе \gate{} любая булева функция может быть реализована схемой,
имеющей задержку $T(n) \sim \tau n$, где $\tau$~--- минимум приведённых задержек всех элементов базиса, при сохранении
асимптотически оптимальной сложности. Однако не рассматривался вопрос зависимости $T(n)$
от количества используемых регистров памяти. Хотя задержка и глубина схемы в некоторых работах определяется по-разному%
~\cite{hrapchenko_new}, в рассматриваемой модели обратимой схемы их, по мнению автора, можно отождествить.

Вопрос асимптотической сложности обратимых схем рассматривался в работе~\cite{maslov_thesis}. В ней было показано,
что сложность обратимой схемы $\frS$, состоящей из элементов $E(t,I,J)$, удовлетворяет неравенству $L(\frS) \leqslant n2^n$.
При этом было показано, что существует обратимая схема $\frS$ сложности $L(\frS) \geqslant 2^n \mathop / \ln 3 + o(2^n)$.
Для обратимых схем, состоящих из элементов mEXOR (терминология авторов), в той же работе~\cite{maslov_thesis} были
доказаны следующие оценки сложности схемы по Шеннону: $2^n \mathop / \ln 3 + o(2^n) \leqslant L(n)
\leqslant 2^{n+1} - 4$.
В работе~\cite{vinokurov} была доказана нижняя оценка сложности $L(\frS) \geqslant n2^n \mathop / (n + \log_2 n -1)$
обратимых схем, состоящих из обобщённых элеметов Тоффоли $k$-CNOT.
В работе~\cite{shende_synthesis} была доказана нижняя асимптотическая оценка $\Omega(n2^n \mathop / \log n)$
сложности обратимой схемы, состоящей из \gate{} множества $\Omega_n^2$ и не имеющей дополнительных входов.
В работе~\cite{maslov_rm_synthesis} была доказана наилучшая известная на сегодняшний день верхняя асимптотическая оценка
сложности $L(\frS) \lesssim 5n2^n$ обратимой схемы $\frS$, состоящей из \gate{} множества $\Omega_n^2$
и не имеющей дополнительных входов.
Однако стоит отметить, что всеми авторами рассматривались только обратимые схемы без дополнительной памяти.

В данной главе рассматривается вопрос асимптотической сложности, глубины и квантового веса обратимых схем,
реализующих некоторое преобразование $\ZZ_2^n \to \ZZ_2^n$ и состоящих из элементов NOT, CNOT и 2-CNOT.
Вводится множество $F(n,q)$ всех отображений $\ZZ_2^n \to \ZZ_2^n$, которые могут быть реализованы такими
обратимыми схемами с $(n+q)$ входами. Оцениваются сложность, глубина и квантовый вес обратимой схемы,
реализующей отображение $f \in F(n,q)$ с использованием $q$ дополнительных входов (дополнительной памяти).
Определяются функции Шеннона сложности $L(n,q)$, глубины $D(n,q)$ и квантового веса $W(n,q)$ обратимой схемы
как функции от $n$ и количества дополнительных входов схемы $q$.

Использование дополнительных входов должно, по-видимому, снижать сложность обратимых схем, но никаких конкретных
оценок до настоящего времени известно не было. Удалось получить такие оценки, согласно которым
сложность и глубина обратимой схемы, в отличие от схем из классических \gate, существенно
зависят от количества дополнительных входов (аналог регистров памяти, см.~\cite{karpova}).

\subsection{Функции Шеннона сложности, глубины и квантового веса}

\forceindent
Обозначим через $P_2(n,n)$ множество всех булевых отображений $\ZZ_2^n \to \ZZ_2^n$.
Обозначим через $F(n,q) \subseteq P_2(n,n)$ множество всех отображений $\ZZ_2^n \to \ZZ_2^n$, которые могут быть реализованы
обратимыми схемами с $(n+q)$ входами, состоящими из \gate{} множества $\Omega_{n+q}^2$.
Множество подстановок из $S(\ZZ_2^n)$, задаваемых всеми \gate{} множества $\Omega_n^2$,
генерирует знакопеременную $A(\ZZ_2^n)$ и симметрическую $S(\ZZ_2^n)$ группы подстановок при $n > 3$ и $n \leqslant 3$,
соответственно~\cite{shende_synthesis, my_lemma_prove}.
Отсюда следует, что $F(n,0)$ совпадает с множеством отображений, задаваемых всеми подстановками из $A(\ZZ_2^n)$
и $S(\ZZ_2^n)$ при $n > 3$ и $n \leqslant 3$, соответственно.
С другой стороны, несложно показать, что при $q \geqslant n$ верно равенство $F(n,q) = P_2(n,n)$.
Другими словами, имея $n$ или более дополнительных входов, можно реализовать с помощью обратимой схемы любое отображение
из $P_2(n,n)$. Это следует из утверждения~\ref{predicate_max_additional_outputs}
на с.~\pageref{predicate_max_additional_outputs}.

Рассмотрим произвольное отображение $f \in F(n,q)$. Среди всех обратимых схем, состоящих из \gate{} множества $\Omega_{n+q}^2$
и реализующих отображение $f$ с использованием $q$ дополнительных входов, мы можем найти схему $\frS_l$ с минимальной
сложностью, схему $\frS_d$ с минимальной глубиной и схему $\frS_w$ с минимальным квантовым весом.
Обозначим через $L(f,q)$, $D(f,q)$ и $W(f,q)$ минимальную сложность, минимальную глубину и минимальный квантовый вес,
соответственно, обратимой схемы, состоящей из \gate{} множества $\Omega_{n+q}^2$ и реализующей отображение $f \in F(n,q)$
с использованием $q$ дополнительных входов. Тогда $L(f, q) = L(\frS_l)$, $D(f,q) = D(\frS_d)$ и $W(f,q) = W(\frS_w)$.
Определим функции Шеннона сложности $L(n,q)$, глубины $D(n,q)$ и квантового веса $W(n,q)$ обратимой схемы:
\begin{gather*}
    L(n,q) = \max_{f \in F(n,q)} {L(f,q)} \; , \\
    D(n,q) = \max_{f \in F(n,q)} {D(f,q)} \; , \\
    W(n,q) = \max_{f \in F(n,q)} {W(f,q)} \; .
\end{gather*}
Из неравенства~\eqref{formula_quantum_weight_connection_with_complexity_common_case} следует, что
\begin{equation}\label{formula_quantum_weight_connection_with_complexity}
    W(n,q) \geqslant \left( \min_{E \in \Omega_{n+q}^2} {W(E)} \right) \cdot L(n,q) \; .
\end{equation}

Как было сказано на с.~\pageref{page_quantum_weight}, чем меньше в обратимой схеме элементов $k$-CNOT с $k > 1$,
тем проще её реализовать при помощи квантовых технологий. Связано это с тем, что
$W(\text{NOT}) = W(\text{CNOT}) = 1$ и $W(\text{2-CNOT}) = 5$~\cite{barenco_elementary_gates}.
Следовательно, представляет собой интерес подсчёт в обратимой схеме отдельно количества различных \gate{}
Для этого введём соответствующие сложности:
$\LC(\frS)$~--- количество элементов NOT и CNOT в схеме $\frS$;
$\LT(\frS)$~--- количество элементов 2-CNOT (Тоффоли) в схеме $\frS$.
Аналогичным образом введём величины $\LC(f,q)$, $\LT(f,q)$, $\LC(n,q)$ и $\LT(n,q)$.
Также обозначим через $\WC$ величину $W(\text{CNOT})$, а через $\WT$~--- величину $W(\text{2-CNOT})$.
Тогда можно считать, что верно следующее равентсво:
\begin{equation}
    W(n,q) = \WC \cdot \LC(n,q) + \WT \cdot \LT(n,q) \; .
    \label{formula_quantum_weigth_as_sum}
\end{equation}

Для двух вещественных положительных функций $f(n)$ и $g(n)$ от натуральной переменной $n$
мы будем использовать следующие обозначения~\cite[с.~355]{yablonsky}:
\begin{enumerate}
    \item
        Функция $f(n)$ \textit{асимптотически больше или равна} функции $g(n)$,
        т.\,е. $f(n) \gtrsim g(n)$, если для любого $\varepsilon > 0$ найдётся $N = N(\varepsilon)$ такое,
        что при любом $n \geqslant N$ верно неравенство $(1- \varepsilon) \cdot g(n) \leqslant f(n)$.
    
    \item
        В случае, если $f(n) \gtrsim g(n)$ и $g(n) \gtrsim f(n)$, говорят, что $f(n)$ и $g(n)$
        \textit{асимптотически равны} (\textit{эквивалентны}) и пишут $f(n) \sim g(n)$.
        В этом случае предел $f(n) \mathop / g(n)$ существует и равен 1.
    
    \item
        Если $0 < c_1 < f(n) \mathop / g(n) < c_2$, то говорят, что функции $f(n)$ и~$g(n)$
        \textit{эквивалентны с точностью до порядка}. В этом случае пишут $f(n) \asymp g(n)$.    
\end{enumerate}

При помощи мощностного метода Риордана--Шеннона в следующих разделах будут доказаны общие нижние оценки
для функций $L(n,q)$ и $D(n,q)$. Будет дано описание двух алгоритмов синтеза, позволяющих получить обратимую схему
с наилучшей асимптотической сложностью и глубиной в случае отсутствия дополнительных входов в схеме и в случае их использования.
Будет доказано, что $L(n,0) \asymp n2^n \mathop / \log_2 n$ и $D(n,0) \gtrsim 2^n \mathop / (3\log_2 n)$.
Будет также описан аналогичный методу Лупанова~\cite{lupanov_delay} подход к синтезу обратимой схемы, для которого 
сложность синтезированной схемы $L(n, q_0) \asymp 2^n$ при использовании $q_0 \sim n 2^{n- o(n)}$
дополнительных входов, а глубина $D(n, q_1) \lesssim 3n$ при использовании $q_1 \sim 2^n$ дополнительных входов.
На основании полученных оценок будут даны некоторые оценки для функций $\LC(n,q)$ и $\LT(n,q)$, а также для $W(n,q)$.

\subsection{Нижние оценки}

\forceindent
Как было сказано в предыдущем разделе, множество подстановок из $S(\ZZ_2^n)$, задаваемых всеми \gate{} множества $\Omega_n^2$,
генерирует знакопеременную $A(\ZZ_2^n)$ при $n > 3$.
В работе~\cite{gluhov} было показано, что длина $L(G,M)$ группы подстановок $G$ относительно системы образующих $M$
удовлетворяет неравенству
\begin{equation}\label{formula_gluhov_lower_bound}
    L(G,M) \geqslant \left \lceil \log_{|M|} |G| \right \rceil \; .
\end{equation}
В нашем случае $G = A(\ZZ_2^n)$, $|G| = (2^n)! \mathop / 2$, $|M| = |\Omega_n^2|$. Поскольку мощность множества $\Omega_n^2$ равна
\begin{equation}\label{formula_size_of_set_omega_n_2}
    |\Omega_n^2| = \sum_{k = 0}^2 {(n-k)\binom{n}{k}} = \frac{n^3}{2} \left( 1 + o(1) \right) \; ,
\end{equation}
то мы можем вывести простую нижнюю оценку для $L(n,0)$:
\begin{gather}
    L(n,0) \gtrsim \frac{\log_2 ((2^n)! \mathop / 2)}{\log_2 (n^3 \mathop / 2)}
           \gtrsim \frac{\log_2 2^{n2^n} - \log_2 e^{2^n}}{3 \log_2 n} \notag  \; , \\
    L(n,0) \gtrsim \frac{n2^n}{3 \log_2 n} \; .
    \label{formula_simple_lower_bound}
\end{gather}

Нижняя оценка~\eqref{formula_gluhov_lower_bound} в работе~\cite{gluhov} строго доказана не была и, по мнению автора,
основывается на не совсем верном предположении, что достаточно рассмотреть только все возможные произведения подстановок из $M$
длины ровно $L(G,M)$, чтобы получить все элементы группы подстановок $G$. Данное предположение верно только для системы
образующих $M$, содержащей тождественную подстановку. В противном случае, необходимо рассматривать все возможные произведения
подстановок из $M$ длины менее $L(G,M)$ в том числе. Очевидно, что множество подстановок,
задаваемых всеми \gate{} множества $\Omega_n^2$, не содержит тождественной подстановки.

Для того, чтобы получить общую нижнюю оценку $L(n,q)$, также необходимо учитывать те булевы отображения,
которые могут быть реализованы обратимой схемой с $(n+q)$ входами. Таких отображений не более $A_{n+q}^n$
(количество размещений из $(n+q)$ по $n$ без повторений).

Перейдём к первой теореме данной главы.
\begin{theorem}[нижняя оценка сложности обратимой схемы]\label{theorem_complexity_lower_bound}
    $$
        L(n,q) \geqslant \frac{2^n(n-2)}{3\log_2(n+q)} - \frac{n}{3} \; .
    $$
\end{theorem}
\begin{proof}
    Доказательство основано на использовании мощностного метода Риор\-дана-Шеннона.    
    Пусть $r = |\Omega_n^2|$. Из формулы~\eqref{formula_size_of_set_omega_n_2} следует, что
    \begin{gather}
        r = \sum_{k=0}^2{(n-k)\binom{n}{k}} = \frac{n^3 - n^2 + 2n}{2} \notag  \; , \\
        \frac{n^2(n-1)}{2} + 1 < r \leqslant \frac{n^3}{2} \text{ \,при\, } n \geqslant 2 \; .
            \label{formula_cardinality_of_omega_n_2_bounds}
    \end{gather}

    Обозначим через $\EuScript C^*(n,s) = r^s$ и $\EuScript C(n,s)$ количество всех обратимых схем,
    которые состоят из \gate{} множества $\Omega_n^2$ и сложность которых равна $s$ и не превышает $s$, соответственно.
    Тогда
    \begin{align*}
        \EuScript C(n,s) &= \sum_{i=0}^s {\EuScript C^*(n,i)} = \frac{r^{s+1} - 1}{r-1}
            \leqslant \left( \frac{n^3}{2} \right)^{s+1} \cdot \frac{2}{n^2(n-1)}  \; , \\
        \EuScript C(n,s) &\leqslant \left( \frac{n^3}{2} \right)^s \cdot \left(1 + \frac{1}{n-1}\right)
            \text{ \,при\, } n \geqslant 2 \; .
    \end{align*}

    Как было сказано выше, каждой обратимой схеме с $(n+q)$ входами
    соответствует не более $A_{n+q}^n$ различных булевых отображений $\ZZ_2^n \to \ZZ_2^n$.
    Следовательно, верно следующее неравенство:
    $$
       \EuScript C(n+q,L(n,q)) \cdot A_{n+q}^n \geqslant |F(n,q)| \; .
    $$
    Поскольку $|F(n,q)| \geqslant |A(\ZZ_2^n)| = (2^n)! \mathop / 2$ и $A_{n+q}^n \leqslant (n+q)^n$, то
    $$
        \left( \frac{(n+q)^3}{2} \right)^{L(n,q)} \cdot \left(1 + \frac{1}{n+q-1}\right)
            \cdot (n+q)^n \geqslant (2^n)! \mathop / 2 \; .
    $$
    Несложно убедиться, что при $n > 0$ верно неравенство $(2^n)! > (2^n \mathop / e)^{2^n}$.
    Следовательно,
    $$
        L(n,q) \cdot (3\log_2(n+q) - 1) + \log_2 \left(1 + \frac{1}{n+q-1}\right) +
            n \log_2(n+q) \geqslant 2^n(n - \log_2 e) \; .
    $$    
    Отсюда следует неравенство из условия теоремы
    $$
        L(n,q) \geqslant \frac{2^n(n-2)}{3\log_2(n+q)} - \frac{n}{3} \; .
    $$
\end{proof}

\begin{theorem}[нижняя оценка глубины обратимой схемы]\label{theorem_depth_lower_bound}
    $$
        D(n,q) \geqslant \frac{2^n(n-2)}{3(n+q)\log_2(n+q)} - \frac{n}{3(n+q)} \; .
    $$
\end{theorem}
\begin{proof}
    Следует из Теоремы~\ref{theorem_complexity_lower_bound} и неравенства~\eqref{formula_complexity_and_depth_connection}.
\end{proof}

\begin{theorem}[нижняя оценка квантового веса обратимой схемы]\label{theorem_quantum_weight_lower_bound}
    $$
        W(n,q) \geqslant \min \left( \WC, \WT \right) \cdot \left( \frac{2^n(n-2)}{3\log_2(n+q)} - \frac{n}{3} \right) \; .
    $$
\end{theorem}
\begin{proof}
    Следует из Теоремы~\ref{theorem_complexity_lower_bound}
    и неравенства~\eqref{formula_quantum_weight_connection_with_complexity}.
\end{proof}

В работе~\cite{my_complexity_bounds} была сделана попытка снизить константу 3 в знаменателе нижней оценки сложности $L(n,q)$
из Теоремы~\ref{theorem_complexity_lower_bound} при помощи свойства эквивалентности некоторых обратимых схем с точки зрения
задаваемых ими преобразований. Для этой цели была выдвинута следующая гипотеза о структуре обратимых схем из \gate{}
множества $\Omega_n^2$:
\begin{hypothesis}\label{hypothesis_k_independent_gates_in_circuit}
    Почти каждая обратимая схема, состоящая из элементов NOT, CNOT и 2-CNOT и имеющая $n \to \infty$ входов,
    может быть представлена в виде композиции подсхем сложности $k = o(n)$
    (кроме последней, у которой сложность $L \leqslant k$),
    таких что в каждой подсхеме все элементы являются попарно коммутирующими.
    Количество обратимых схем, для которых это неверно, пренебрежимо мало.
\end{hypothesis}
Покажем, что данная гипотеза неверна для обратимых схем со сложностью $O(n)$ и выше.

Пусть $r = |\Omega_n^2|$ (см. формулу~\eqref{formula_cardinality_of_omega_n_2_bounds}).
Обозначим через $M(n,k)$ множество обратимых схем с $n$ входами из Гипотезы~\ref{hypothesis_k_independent_gates_in_circuit}.
Очевидно, что $|M(n,2)| \geqslant |M(n,k)|$. Будем строить схему $\frS \in M(n,2)$ следующим способом:
сначала выбираем $r$ способами первый элемент в паре, затем выбираем второй элемент, коммутирующий с предыдущим;
так делаем $\lfloor L(\frS) \mathop / 2 \rfloor$ раз. Оставшийся элемент в случае нечётной сложности схемы $\frS$ выбираем
$r$ способами.

В работе~\cite{my_complexity_bounds} использовалось следующее свойство коммутируемости: два элемента
$E(t_1, I_1)$ и $E(t_2, I_2)$ из множества $\Omega_n^2$ являются коммутирующими, если $t_1 \notin I_2$ и $t_2 \notin I_1$.
Если рассматривать только элементы 2-CNOT из множества $\Omega_n^2$,
то второй элемент в паре можно выбрать не менее $\binom{n-2}{1} \cdot \binom{n-2}{2}$ способами.
Следовательно, верно следующее неравенство:
$$
    |M(n,2)| \geqslant \left( r \cdot \binom{n-2}{1} \cdot \binom{n-2}{2}
        \right)^{\lfloor s \mathop / 2 \rfloor} \; ,
$$
где $s = L(\frS)$. Оценим величину $Q(n,s) = |M(n,2)| \mathop / r^s$:
\begin{gather}\label{formula_q_n_s_lower_bound}
    Q(n,s) = \frac{|M(n,2)|}{r^s} \geqslant \left(
        \frac{n-2}{n} \cdot \frac{n-2}{n} \cdot \frac{n-3}{n} \right)^{(s-1)/2}
        \geqslant \left(1- \frac{3}{n}\right)^{2(s-1)} \; .
\end{gather}

Таким же образом можно оценить величину $Q(n,s)$ сверху: при выборе второго \gate{} пары
не будут рассматриваться как минимум те элементы 2-CNOT, у которых один из контролирующих входов равен контролируемому
выходу предыдущего элемента. Таких элементов Тоффоли $\binom{n-1}{2}$ штук. Следовательно, верно неравенство
$$
    Q(n,s) = \frac{|M(n,2)|}{r^s} \leqslant \left(
        \frac{r}{r} \cdot \frac{r - \binom{n-1}{2}}{r} \right)^{s/2} =
        \left( 1 - \frac{n-2}{n^2} \right)^{s/2} \; .
$$

Отсюда следует, что
\begin{equation}\label{formula_q_n_s_general}
    \lim_{n \to \infty} e^{-6s \mathop / n} \leqslant \lim_{n \to \infty} Q(n,s)
        \leqslant \lim_{n \to \infty} e^{-s \mathop / (2n)} \; .
\end{equation}
Таким образом, $r^s e^{-6s \mathop / n} \lesssim |M(n,2)| \lesssim r^s e^{-s \mathop / (2n)}$.

Из неравенства, аналогичного неравенству~\eqref{formula_q_n_s_lower_bound},
в работе~\cite{my_complexity_bounds} был сделан ошибочный вывод, что $|M(n,2)| \sim |M(n,k)| \sim r^s$
для всех значений $s$, если $k = o(n)$.
Из~\eqref{formula_q_n_s_general} видно, что $|M(n,2)| \sim r^s$ только при $s = o(n)$.
И уже при $s \geqslant n$ количество обратимых схем, не соответствующих утверждению гипотезы, перестаёт быть пренебрежимо малым.
Следовательно, при оценке сложности обратимых схем нельзя опираться на Гипотезу~\ref{hypothesis_k_independent_gates_in_circuit}.

На практике б\'{о}льший интерес представляют верхние оценки сложности, глубины и квантового веса обратимых схем.
Этот вопрос будет изучаться в следующих двух разделах.

\subsection{Верхние оценки для схем без дополнительной памяти}\label{subsection_upper_bound_no_memory}
\refstepcounter{synthalgchapter}

\forceindent
Как уже неоднократно было сказано ранее, обратимая схема без дополнительных входов может реализовывать только чётные подстановки.
В работе~\cite{maslov_rm_synthesis} был описан алгоритм синтеза таких схем, позволяющий получить обратимую схему
$\frS$, реализующую заданную чётную подстановку $h \in A(\ZZ_2^n)$ со сложностью $L(\frS) \lesssim 5n 2^n$.
Эту оценку можно считать простейшей верхней оценкой для $L(n,0)$:
\begin{equation}\label{formula_upper_bound_not_my_simple}
    L(n,0) \lesssim 5n 2^n \; .
\end{equation}

В разделе~\ref{subsection_my_synthesis_algorithms} был описан алгоритм~\algref{alg_my_common} синтеза обратимых схем,
состоящих из \gate{} множества $\Omega_n^2$.
Было доказано, что для синтезированной данным алгоритмом схемы $\frS$ её сложность $L(\frS) \lesssim 7n 2^n$.
Алгоритм~\algref{alg_my_common} основывается на синтезе пары независимых транспозиций. Однако если обобщить данный подход
для синтеза б\'{о}льшего количества независимых транспозиций, то верхнюю оценку~\eqref{formula_upper_bound_not_my_simple}
можно существенно улучшить.

\begin{theorem}[о сложности обратимой схемы без дополнительных входов]\label{theorem_complexity_upper_no_memory}
    $$
        L(n, 0) \leqslant \frac{3n2^{n+4}}{\log_2 n - \log_2 \log_2 n - \log_2 \phi(n)}
            \left( 1 + \varepsilon(n) \right)  \; ,
    $$
    где $\phi(n)$~--- любая сколь угодно медленно растущая функция такая, что $\phi(n) < n \mathop / \log_2 n$,
    $$
        \varepsilon(n) = \frac{1}{6\phi(n)} + \left(\frac{8}{3} - o(1)\right)
            \frac{\log_2 n \cdot \log_2 \log_2 n}{n} \; .
    $$
\end{theorem}
\begin{proof}

    Опишем новый алгоритм синтеза~\nextalg{alg_asymp_no_mem}, являющийся обобщением алгоритма~\algref{alg_my_common}.
    Рассмотрим произвольную чётную подстановку $h \in A(\ZZ_2^n)$ и задаваемое ей булево преобразование
    $f_h\colon \ZZ_2^n \to \ZZ_2^n$.
    Основная идея алгоритма синтеза заключается в представлении подстановки $h$ в виде произведения
    транспозиций таким образом, чтобы их можно было разбить на группы по $K$ независимых транспозиций в каждой группе:
    \begin{equation}
        h = G_1 \circ G_2 \circ \ldots \circ G_t \circ h' \; ,
        \label{formula_permutation_decomposition_main}
    \end{equation}
    где $G_i = (\vv x_{i,1}, \vv y_{i,1}) \circ \ldots \circ (\vv x_{i,K}, \vv y_{i,K})$~--- $i$-я группа
    $K$ независимых транспозиций, $\vv x_{i,j}, \vv y_{i,j} \in \ZZ_2^n$; $h'$~--- некоторая остаточная подстановка.
    Далее мы покажем, что любая группа $G_i$ может быть задана композицией одного элемента $k$-CNOT с большим количеством
    контролирующих входов $k$ и множества элементов CNOT и 2-CNOT.

    Подстановка $h \in A(\ZZ_2^n)$ может быть представлена в виде произведения независимых циклов,
    сумма длин которых не превышает $2^n$ (см. с.~\pageref{formula_cycle_decomposition}).
    Имея это представление, мы можем получать независимые транспозиции из циклов по формулам
    \eqref{formula_decompostion_of_two_cycles}--\eqref{formula_decompostion_of_k_cycle}
    следующим образом:
    \begin{gather*}
        (i_1, i_2, \ldots, i_{l_1}) \circ (j_1, j_2, \ldots, j_{l_2}) = (i_1, i_2) \circ (j_1, j_2) \circ
            (i_1, i_3, \ldots, i_{l_1}) \circ (j_1, j_3, \ldots, j_{l_2})  \; , \\
        (i_1, i_2, \ldots, i_l) = (i_1, i_2) \circ (i_3, i_4) \circ (i_1, i_3, i_5, i_6, \ldots, i_l) \; .
    \end{gather*}
    
    Если посмотреть на представление подстановки~\eqref{formula_permutation_decomposition_main}
    и на формулы \eqref{formula_decompostion_of_two_cycles}--\eqref{formula_decompostion_of_k_cycle},
    то можно сделать вывод, что $K$ независимых транспозиций нельзя получить из остаточной подстановки $h'$,
    только если она в своём представлении имеет строго меньше $K$ независимых циклов, и каждый из этих циклов имеет длину
    строго меньше 5. Таким образом, сумма длин циклов в представлении подстановки $h'$ не превосходит $4(K-1)$.

    Обозначим через $M_g$ множество подвижных точек подстановки $g \in S(\ZZ_2^n)$:
    $$
        M_g = \{\,\vv x \in \ZZ_2^n\mid g(\vv x) \ne \vv x\,\} \; .
    $$
    Тогда можно утверждать, что $|M_h| \leqslant 2^n$ и $|M_{h'}| \leqslant 4(K-1)$.
    
    Из формул~\eqref{formula_decompostion_of_two_cycles} и~\eqref{formula_permutation_decomposition_main}
    следует, что в представлении подстановки $h$ в виде произведения траспозиций можно получить не более
    $|M_h| \mathop / K$ групп, в каждой из которых $K$ независимых транспозиций,
    а в представлении подстановки $h'$ в виде произведения траспозиций
    можно получить не более $|M_{h'}| \mathop / 2$ пар независимых транспозиций и не более одной пары зависимых транспозиций.
    Пара зависимых транспозиций $(i,j) \circ (i, k)$ выражается через произведение двух пар независимых транспозиций:
    $$
        (i,j) \circ (i, k) = ((i,j) \circ (r, s)) \circ ((r,s) \circ (i,k)) \; .
    $$

    Обозначим через $g^{(i)}$ подстановку, представляющую собой произведение $i$ независимых транспозиций,
    а через $f_{g^{(i)}}$~--- задаваемое ей булево преобразование.
    Тогда можно оценить сверху величину $L(f_h,0)$ следующим образом:
    \begin{gather}
        L(f_h,0) \leqslant \frac{|M_h|}{K} \cdot L(f_{g^{(K)}},0) + \left(\frac{|M_{h'}|}{2} + 2 \right )
            \cdot L(f_{g^{(2)}},0) \notag  \; , \\
        L(f_h,0) \leqslant \frac{2^n}{K} \cdot L(f_{g^{(K)}},0) + 2K \cdot L(f_{g^{(2)}},0)
            \label{formula_upper_bound_of_L_h_common} \; .
    \end{gather}
    Всё, что нам остаётся сделать,~--- оценить сверху величину $L(f_{g^{(K)}},0)$.

    Рассмотрим произвольную подстановку $g^{(K)}$.
    Обозначим через $k$ величину $|M_{g^{(K)}}|$, тогда $k = 2K$.
    Суть описываемого алгоритма заключается в действии сопряжением на подстановку $g^{(K)}$
    таким образом, чтобы получить некоторую новую подстановку, соответствующую одному обобщённому элементу Тоффоли.
    Напомним, что действие сопряжением не меняет цикловой структуры подстановки, поэтому подстановка $g^{(K)}$ в результате
    действия сопряжением всегда будет оставаться произведением $K$ независимых транспозиций.
    Любой элемент $E$ из множества $\Omega_n^2$ задаёт подстановку $h_{E}$ на множестве двоичных векторов $\ZZ_2^n$.
    Для этой подстановки верно равенство $h^{-1}_{E} = h_{E}$.
    Следовательно, применение к $g^{(K)}$ действия сопряжением подстановкой $h_{E}$,
    записываемое как $h^{-1}_{E} \circ g^{(K)} \circ h_{E}$, соответствует
    присоединению элемента $E$ к началу и к концу текущей обратимой подсхемы.

    Пусть $g^{(K)} = (\vv x_1, \vv y_1) \circ \ldots \circ (\vv x_K, \vv y_K)$.
    Составим матрицу $A$ следующим образом:
    \begin{equation}
        A =
            \left(
                \begin{matrix}
                    \vv x_1 \\
                    \vv y_1 \\
                    \ldots \\
                    \vv x_K \\
                    \vv y_K
                \end{matrix}
            \right )
          =
            \left(
                \begin{matrix}
                    a_{1,1}   & \ldots & a_{1,n}   \\
                    a_{2,1}   & \ldots & a_{2,n}   \\
                    \hdotsfor{3}                   \\
                    a_{k-1,1} & \ldots & a_{k-1,n} \\
                    a_{k,1}   & \ldots & a_{k,n}   \\
                \end{matrix}
            \right ) \; .
        \label{formula_matrix_for_permutation}
    \end{equation}

    Наложим на значение $k$ следующее ограничение: $k$ должно быть степенью двойки, $2^{\lfloor \log_2 k \rfloor} = k$.
    Если $k \leqslant \log_2 n$, то в матрице $A$ существует не более $2^k$ и не менее $\log_2 k$ попарно различных столбцов.
    Без ограничения общности будем считать, что такими столбцами являются первые $d \leqslant 2^k$ столбцов матрицы.
    Тогда для любого $j$-го столбца, $j > d$, найдётся равный ему $i$-й столбец, $i \leqslant d$.
    Следовательно, применив к подстановке $g^{(K)}$ действие сопряжением подстановкой,
    задаваемой элементом $C_{i;j}$,
    можно обнулить $j$-й столбец в матрице $A$ (для этого потребуется 2 элемента CNOT).
    Обнуляя таким образом все столбцы с индексами больше $d$, использовав $L_1 = 2(n-d)$ элементов CNOT%
    \phantomsection\label{page_l1_complexity},
    мы получим новую подстановку $g_1^{(K)}$ и соответствующую ей матрицу $A_1$ следующего вида:
    $$
        A_1 =
            \left(\phantom{\begin{matrix} 0\\0\\0\\0\\0\\ \end{matrix}} \right.
            \begin{matrix}
                a_{1,1}   & \ldots & a_{1,d}   \\
                a_{2,1}   & \ldots & a_{2,d}   \\
                \hdotsfor{3}                     \\
                a_{k-1,1} & \ldots & a_{k-1,d} \\
                a_{k,1}   & \ldots & a_{k,d}   \\
            \end{matrix}
            \phantom{\begin{matrix} 0\\0\\0\\0\\0\\ \end{matrix}}
            \overbrace{
                \begin{matrix}
                    0 &\ldots & 0   \\
                    0 &\ldots & 0   \\
                    \hdotsfor{3}    \\
                    0 &\ldots & 0   \\
                    0 &\ldots & 0   \\
                \end{matrix}
            }^{n - d}
            \left.\phantom{\begin{matrix} 0\\0\\0\\0\\0\\ \end{matrix}} \right) \; .
    $$

    Теперь для всех $a_{1,i} = 1$ применяем к $g_1^{(K)}$ действие сопряжением подстановкой,
    задаваемой элементом $N_i$. Для этого потребуется $L_2 \leqslant 2d$ элементов NOT%
    \phantomsection\label{page_l2_complexity}.
    В итоге получим подстановку $g_2^{(K)}$ и соответствующую ей матрицу
    $A_2$ (элементы матрицы обозначены через $b_{i,j}$, чтобы показать их возможное отличие от элементов матрицы $A_1$):
    $$
        A_2 =
            \left(\phantom{\begin{matrix} 0\\0\\0\\0\\0\\ \end{matrix}} \right.
            \begin{matrix}
                0         & \ldots & 0         \\
                b_{2,1}   & \ldots & b_{2,d}   \\
                \hdotsfor{3}                   \\
                b_{k-1,1} & \ldots & b_{k-1,d} \\
                b_{k,1}   & \ldots & b_{k,d}   \\
            \end{matrix}
            \phantom{\begin{matrix} 0\\0\\0\\0\\0\\ \end{matrix}}
            \overbrace{
                \begin{matrix}
                    0 &\ldots & 0   \\
                    0 &\ldots & 0   \\
                    \hdotsfor{3}    \\
                    0 &\ldots & 0   \\
                    0 &\ldots & 0   \\
                \end{matrix}
            }^{n - d}
            \left.\phantom{\begin{matrix} 0\\0\\0\\0\\0\\ \end{matrix}} \right) \; .
    $$

    Следующим шагом является приведение матрицы $A_2$ к \textit{каноническому виду}, где каждая строка, если её записать
    в обратном порядке, представляет собой запись в двоичной системе счисления числа <<номер строки минус 1>>.

    Все строки матрицы $A_2$ различны. Первая строка уже имеет канонический вид, поэтому мы последовательно будем приводить
    оставшиеся строки к каноническому виду, начиная со второй.
    Предположим, что текущая строка имеет номер $i$, и все строки с номерами от $1$ до $(i-1)$ имеют канонический вид.
    Возможны два случая:
    \begin{enumerate}
        \item
            Существует ненулевой элемент в $i$-й строке с индексом $j > \log_2 k$: $b_{i,j} = 1$.
            В этом случае для всех элементов матрицы $b_{i, j'}$, $j' \ne j$, $j' \leqslant d$,
            не равных $j'$-ой цифре в двоичной записи числа $(i-1)$, мы применяем к $g_2^{(K)}$ действие сопряжением
            подстановкой, задаваемой элементом $C_{j;j'}$. Для этого потребуется не более $2d$ элементов CNOT.
            После этого нам остаётся только обнулить $j$-й элемент текущей строки.
            Для этого мы применяем к $g_2^{(K)}$ действие сопряжением подстановкой, задаваемой
            элементом $C_{I;j}$, где $I$~--- множество индексов ненулевых цифр в двоичной записи числа $(i-1)$.
            К примеру, если $i = 6$, то $I = \{\,1, 3\,\}$.
            Поскольку $|I| \leqslant \log_2 k$, мы можем заменить данный элемент $C_{I;j}$
            композицией не более $8 \log_2 k$ элементов 2-CNOT~\cite{barenco_elementary_gates}.
            Следовательно, для данного действия сопряжением нам потребуется не более $16 \log_2 k$ элементов 2-CNOT.
            
            Итак, суммируя количество используемых \gate{}, мы получаем, что для приведения $i$-й строки
            к каноническому виду в данном случае требуется $L_3^{(i)} \leqslant 2d + 16 \log_2 k$
            элементов из множества $\Omega_n^2$.

        \item
            Не существует ненулевого элемента в $i$-й строке с индексом $j > \log_2 k$: $b_{i,j} = 0$ для всех $j > \log_2 k$.
            В этом случае мы применяем к $g_2^{(K)}$ действие сопряжением подстановкой,
            задаваемой элементом $C_{I;\log_2 k + 1}$, где $I$~--- множество индексов ненулевых элементов
            текущей строки.
            Т.\,к. все строки матрицы различны и при этом все предыдущие строки находятся в каноническом виде,
            мы можем утверждать, что значение элемента матрицы $b_{j,\log_2 k + 1}$ после данного действия сопряжением
            будет изменено только в случае, если $j \geqslant i$.
            Поскольку $|I| \leqslant \log_2 k$, мы можем заменить данный элемент $C_{I;j}$
            композицией не более $8 \log_2 k$ элементов 2-CNOT~\cite{barenco_elementary_gates}.
            Следовательно, для данного действия сопряжением нам потребуется не более $16 \log_2 k$ элементов 2-CNOT.
            После этого мы можем перейти к предыдущему случаю.

            Итак, суммируя количество используемых \gate{}, мы получаем, что для приведения $i$-й строки
            к каноническому виду в данном случае требуется $L_3^{(i)} \leqslant 2d + 32 \log_2 k$
            элементов из множества $\Omega_n^2$.
    \end{enumerate}

    После приведения матрицы $A_2$ к каноническому виду, мы получим новую подстановку $g_3^{(K)}$
    и соответствующую ей матрицу $A_3$ следующего вида:
    $$
        A_3 =
            \left(\phantom{\begin{matrix} 0\\0\\0\\0\\0\\ \end{matrix}} \right.
            \overbrace{
                \begin{matrix}
                    0 & 0 & 0 & \ldots & 0 \\
                    1 & 0 & 0 & \ldots & 0 \\
                    \hdotsfor{5}           \\
                    0 & 1 & 1 & \ldots & 1 \\
                    1 & 1 & 1 & \ldots & 1 \\
                \end{matrix}
            }^{\log_2 k}        
            \phantom{\begin{matrix} 0\\0\\0\\0\\0\\ \end{matrix}}
            \overbrace{
                \begin{matrix}
                    0 &\ldots & 0   \\
                    0 &\ldots & 0   \\
                    \hdotsfor{3}    \\
                    0 &\ldots & 0   \\
                    0 &\ldots & 0   \\
                \end{matrix}
            }^{n - \log_2 k}
            \left.\phantom{\begin{matrix} 0\\0\\0\\0\\0\\ \end{matrix}} \right) \; .
    $$
    Для этого в сумме потребуется $L_3$ \gate{} множества $\Omega_n^2$%
    \phantomsection\label{page_l3_complexity}:
    $$
        L_3 = \sum_{i=2}^k {L_3^{(i)}} \leqslant k(2d + 32 \log_2 k) \; .
    $$
    При этом мы получили ещё одно ограничение на значение $k$: значение $\log_2 k$ должно быть строго меньше $n$,
    иначе не всегда будет возможно привести матрицу $A_2$ к каноническому виду.

    На последнем шаге для каждого $i > \log_2 k$ мы применяем к $g_3^{(K)}$ действие сопряжением подстановкой,
    задаваемой элементом $N_i$. Для этого нам потребуется $L_4 = 2(n - \log_2 k)$ элементов NOT%
    \phantomsection\label{page_l4_complexity}.
    В итоге получим подстановку $g_4^{(K)}$ и соответствующую ей матрицу $A_4$ следующего вида:
    $$
        A_4 =
            \left(\phantom{\begin{matrix} 0\\0\\0\\0\\0\\ \end{matrix}} \right.
            \overbrace{
                \begin{matrix}
                    0 & 0 & 0 & \ldots & 0 \\
                    1 & 0 & 0 & \ldots & 0 \\
                    \hdotsfor{5}           \\
                    0 & 1 & 1 & \ldots & 1 \\
                    1 & 1 & 1 & \ldots & 1 \\
                \end{matrix}
            }^{\log_2 k}        
            \phantom{\begin{matrix} 0\\0\\0\\0\\0\\ \end{matrix}}
            \overbrace{
                \begin{matrix}
                    1 &\ldots & 1   \\
                    1 &\ldots & 1   \\
                    \hdotsfor{3}    \\
                    1 &\ldots & 1   \\
                    1 &\ldots & 1   \\
                \end{matrix}
            }^{n - \log_2 k}        
            \left.\phantom{\begin{matrix} 0\\0\\0\\0\\0\\ \end{matrix}} \right) \; .
    $$

    Подстановка $g_4^{(K)}$ задаётся одним элементом $C_{n,n-1, \ldots, \log_2 k + 1; 1}$.
    Этот элемент имеет $(n - \log_2 k)$ контролирующих входов, поэтому он может быть заменён композицией не более
    $L_5 \leqslant 8(n - \log_2 k)$ элементов 2-CNOT~\cite{barenco_elementary_gates}%
    \phantomsection\label{page_l5_complexity}.
    
    Мы получили подстановку $g_4^{(K)}$, применяя к $g^{(K)}$ действие сопряжением подстановками определённого вида.
    Если мы применим к $g_4^{(K)}$ действие сопряжением в точности теми же подстановками, но в обратном порядке,
    мы получим $g^{(K)}$. В терминах синтеза обратимой логики это означает, что мы должны присоединить
    ко входу и выходу элемента $C_{n,n-1, \ldots, \log_2 k + 1; 1}$ все те \gate{},
    что мы использовали в наших преобразованиях исходной матрицы $A$, но в обратном порядке,
    и как результат, мы получим обратимую схему $\frS_K$, задающую подстановку $g^{(K)}$.
    
    Таким образом, можно утверждать, что $L(g^{(K)},0) \leqslant L(\frS_K)$ и
    \begin{gather}
        L(g^{(K)},0) \leqslant \sum_{i=1}^5 {L_i} \leqslant 2(n-d) + 2d + k(2d + 32 \log_2 k) + 2(n-\log_2 k) + 8(n - \log_2 k)
            \notag  \; , \\
        L(g^{(K)},0) \leqslant 12n + k2^{k+1} + 32k\log_2 k - 10 \log_2 k
            \label{formula_upper_bound_in_synthesis_algorithm} \; .
    \end{gather}
    Отсюда также следует, что $L(g^{(2)},0) \leqslant 12n + 364$.

    Подставляя полученные верхние оценки в формулу~\eqref{formula_upper_bound_of_L_h_common},
    мы получаем следующую верхнюю оценку для $L(f_h,0)$:
    \begin{equation}\label{formula_L_h_upper_bound_general_case_with_k}
        L(f_h,0) \leqslant \frac{2^{n+1}}{k}(12n + k2^{k+1} + 32k\log_2 k - 10 \log_2 k) + k(12n + 364) \; .
    \end{equation}

    Описанным алгоритмом требуется, чтобы $k$ было степенью двойки и чтобы $\log_2 k$ было строго меньше $n$.
    Пусть $m = \log_2 n - \log_2 \log_2 n - \log_2 \phi(n)$ и $k = 2^{\lfloor \log_2 m \rfloor}$,
    где $\phi(n)$~--- любая сколь угодно медленно растущая функция такая, что $\phi(n) < n \mathop / \log_2 n$.
    Тогда $m / 2 \leqslant k \leqslant m$ и
    \begin{gather*}
        L(f_h,0) \leqslant \frac{2^{n+2}}{m}(12n + 2m2^m + (32 - o(1))m\log_2 m)  \; , \\
        L(f_h,0) \leqslant \frac{3n2^{n+4}}{m} 
            \left( 1 + \frac{2^m\log_2 n}{6n} + \left(\frac{8}{3} - o(1)\right)
            \frac{\log_2 n \cdot \log_2 \log_2 n}{n} \right) \; .
    \end{gather*}

    Отсюда следует итоговая верхняя оценка для $L(f_h,0)$:
    $$
        L(f_h,0) \leqslant \frac{3n2^{n+4}}{\log_2 n - \log_2 \log_2 n - \log_2 \phi(n)}
            \left( 1 + \varepsilon(n) \right)  \; ,
    $$
    где функция $\varepsilon(n)$ равна
    $$
        \varepsilon(n) = \frac{1}{6\phi(n)} +\left(\frac{8}{3} - o(1)\right)
            \frac{\log_2 n \cdot \log_2 \log_2 n}{n} \; .
    $$

    Поскольку мы описали алгоритм синтеза обратимой схемы $\frS$ для произвольной подстановки $h \in A(\ZZ_2^n)$,
    то функция $L(n,0)$ ограничена сверху также, как и функция $L(f_h,0)$ в формуле выше.
\end{proof}

\begin{theorem}\label{theorem_complexity_no_memory_common}
    $$
        L(n,0) \asymp \frac{n2^n}{\log_2n} \; .
    $$
\end{theorem}
\begin{proof}
    Следует из Теорем~\ref{theorem_complexity_lower_bound} и~\ref{theorem_complexity_upper_no_memory}.
\end{proof}

Стоит отметить, что если представлять подстановку $h \in A(\ZZ_2^n)$ в виде произведения пар независимых транспозиций,
то в этом случае задающая её обратимая схема $\frS_h$ согласно формуле~\eqref{formula_L_h_upper_bound_general_case_with_k}
будет иметь сложность $L(\frS_h) \lesssim 6n 2^n$.
Данная сложность схемы асимптотически ниже, чем сложность $L(\frS) \lesssim 7n 2^n$ обратимой схемы,
синтезируемой алгоритмом~\algref{alg_my_common}, но выше, чем сложность $L(\frS) \lesssim 5n 2^n$ обратимой схемы,
синтезируемой алгоритмом из работы~\cite{maslov_rm_synthesis}.

Для того, чтобы пояснить основную часть алгоритма синтеза~\algref{alg_asymp_no_mem}, рассмотрим подстановку
$g^{(2)} = (\langle 1,0,0,1 \rangle, \langle 0, 0, 0, 0 \rangle) \circ (\langle 1,1,1,1 \rangle, \langle 0, 1, 1, 0 \rangle)$.
Данная подстановка задаётся обратимой схемой
$\frS = C_{1;4} * C_{2;3} * N_1 * N_3 * N_4 * C_{3,4;1} * N_4 * N_3 * N_1 * C_{2;3} * C_{1;4}$.
Процесс получения схемы $\frS$ показан на рис.~\ref{pic_conjugation_process}.

\Figure[ht]
    \centering
    \begin{tabular}{ccccccc}
        \smallskip &
        $\frS$ & &
        $C_{1;4} * \frS * C_{1;4}$ & &
        $C_{2;3} * \frS_1 * C_{2;3}$ &
        \\

        $A =$ &
        $
        \left(
            \begin{matrix}
                1 & 0 & 0 & 1 \\
                0 & 0 & 0 & 0 \\
                1 & 1 & 1 & 1 \\
                0 & 1 & 1 & 0 \\
            \end{matrix}
        \right )$ &
        $\Rightarrow$ &
        $\left(
            \begin{matrix}
                1 & 0 & 0 & 0 \\
                0 & 0 & 0 & 0 \\
                1 & 1 & 1 & 0 \\
                0 & 1 & 1 & 0 \\
            \end{matrix}
        \right )$ &
        $\Rightarrow$ &
        $\left(
            \begin{matrix}
                1 & 0 & 0 & 0 \\
                0 & 0 & 0 & 0 \\
                1 & 1 & 0 & 0 \\
                0 & 1 & 0 & 0 \\
            \end{matrix}
        \right )$ &
        $\Rightarrow$
    \end{tabular}

    \bigskip    
    \begin{tabular}{ccccccc}
        \smallskip &
        $N_1 * \frS_2 * N_1$ & &
        $N_3 * \frS_3 * N_3$ & &
        $N_4 * \frS_4 * N_4$ &
        \\
        
        $\Rightarrow$ &
        $
        \left(
            \begin{matrix}
                0 & 0 & 0 & 0 \\
                1 & 0 & 0 & 0 \\
                0 & 1 & 0 & 0 \\
                1 & 1 & 0 & 0 \\
            \end{matrix}
        \right )$ &
        $\Rightarrow$ &
        $
        \left(
            \begin{matrix}
                0 & 0 & 1 & 0 \\
                1 & 0 & 1 & 0 \\
                0 & 1 & 1 & 0 \\
                1 & 1 & 1 & 0 \\
            \end{matrix}
        \right )$ &
        $\Rightarrow$ &
        $
        \left(
            \begin{matrix}
                0 & 0 & 1 & 1 \\
                1 & 0 & 1 & 1 \\
                0 & 1 & 1 & 1 \\
                1 & 1 & 1 & 1 \\
            \end{matrix}
        \right )$ &
    \end{tabular}

    \smallskip
    \bigskip

    $
        N_4 * N_3 * N_1 * C_{2;3} * C_{1;4} * \frS * C_{1;4} * C_{2;3} * N_1 * N_3 * N_4 = C_{3,4;1}
    $
    \caption
    {
        Процесс получения обратимой схемы $\frS$, задающей подстановку
        $g^{(2)} = (\langle 1,0,0,1 \rangle, \langle 0, 0, 0, 0 \rangle)
        \circ (\langle 1,1,1,1 \rangle, \langle 0, 1, 1, 0 \rangle)$.
    }\label{pic_conjugation_process}
\end{figure}

Для глубины обратимой схемы мы можем получить похожую, однако не асимптотически оптимальную верхнюю оценку.
\begin{theorem}[о глубине обратимой схемы без дополнительных входов]\label{theorem_depth_upper_no_memory}
    $$
        D(n, 0) \leqslant \frac{n2^{n+5}}{\log_2 n - \log_2 \log_2 n - \log_2 \phi(n)}
            \left( 1 + \varepsilon(n) \right) \;  ,
    $$
    где $\phi(n)$~--- любая сколь угодно медленно растущая функция такая, что $\phi(n) < n \mathop / \log_2 n$,
    $$
        \varepsilon(n) = \frac{1}{4\phi(n)} +(4 + o(1))\frac{\log_2 n \cdot \log_2 \log_2 n}{n} \; .
    $$
\end{theorem}
\begin{proof}
    Из описания алгоритма синтеза~\algref{alg_asymp_no_mem} следует, что некоторые операции можно делать с логарифмической
    или линейной глубиной. К примеру, обнуление столбцов матрицы может быть произведено с логарифмической глубиной
    (см. рис.~\ref{pic_zeroing_out_with_logarithmic_depth}). Также действие сопряжением подстановками,
    задаваемыми элементами NOT, может быть реализовано обратимой подсхемой с константной глубиной.
    
    \Figure[ht]
        \centering
        \includegraphics[scale=1.2]{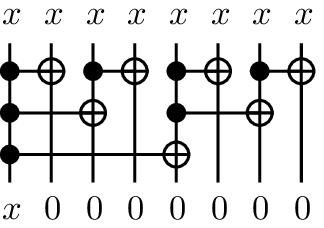}
        \caption
        {
            \small Обнуление дублирующих входов с логарифмической глубиной (входы схемы сверху).
        }\label{pic_zeroing_out_with_logarithmic_depth}
    \end{figure}

    Отсюда следует, что $D_1 = 2 \lceil \log_2(n-d) \rceil$ (против $L_1 = 2(n-d)$, см. с.~\pageref{page_l1_complexity}),
    $D_2 \leqslant 2$ (против $L_2 \leqslant 2d$, см. с.~\pageref{page_l2_complexity}) и $D_4 = 2$
    (против $L_4 = 2(n-\log_2 k)$, см. с.~\pageref{page_l4_complexity}).
    Для других этапов алгоритма синтеза~\algref{alg_asymp_no_mem} получаются подсхемы, у которых глубина равна сложности:
    $D_3 = L_3 \leqslant k(2d + 32 \log_2 k)$, $D_5 = L_5 \leqslant 8(n - \log_2 k)$
    (см. с.~\pageref{page_l3_complexity}).

    Используя данные значения глубин подсхем, можно вывести следующую оценку сверху:
    $$
        D(g^{(K)},0) \leqslant \sum_{i=1}^5 {D_i} \leqslant 2 \log_2 n + k(2^{k+1} + 32 \log_2 k)
            + 8(n - \log_2 k) + 6 \; .
    $$
    Отсюда также следует, что $D(g^{(2)},0) \leqslant 8n + 2\log_2 n + 374$.

    Подставляя данные значения в формулу для глубины обратимой схемы,
    аналогичную формуле~\eqref{formula_upper_bound_of_L_h_common}, мы получаем следующую верхнюю оценку для $D(f_h,0)$:
    $$
        D(f_h,0) \leqslant \frac{2^{n+1}}{k}(8n + (2+o(1))\log_2 n + k2^{k+1} + 32k\log_2 k - 8\log_2 k) \; .
    $$

    При $m = \log_2 n - \log_2 \log_2 n - \log_2 \phi(n)$ и $k = 2^{\lfloor \log_2 m \rfloor}$,
    где $\phi(n)$~--- любая сколь угодно медленно растущая функция такая, что $\phi(n) < n \mathop / \log_2 n$,
    верно следующее неравенство для $D(f_h,0)$:
    \begin{gather*}
        D(f_h,0) \leqslant \frac{n2^{n+5}}{m} 
            \left( 1 + \frac{m 2^m}{4n} + (4 + o(1))\frac{m \log_2 m}{n} \right)  \; , \\
        D(f_h,0) \leqslant \frac{n2^{n+5}}{\log_2 n - \log_2 \log_2 n - \log_2 \phi(n)}
            \left( 1 + \varepsilon(n) \right) \; ,
    \end{gather*}
    где функция $\varepsilon(n)$ равна
    $$
        \varepsilon(n) = \frac{1}{4\phi(n)} + (4 + o(1)) \frac{\log_2 n \cdot \log_2 \log_2 n}{n} \; .
    $$

    Поскольку описанный алгоритм синтеза~\algref{alg_asymp_no_mem} позволяет получить обратимую схему $\frS$
    для произвольной подстановки $h \in A(\ZZ_2^n)$, то функция $D(n,0)$ ограничена сверху также,
    как и функция $D(f_h,0)$ в формуле выше.
\end{proof}

Стоит отметить, что полученные верхние и нижние оценки глубины обратимой схемы без дополнительных входов
достаточно неточны: они не являются эквивалентными с точностью до порядка, в отличие от
оценок для сложности такой схемы (см. Теорему~\ref{theorem_complexity_no_memory_common}).

Оценим теперь квантовый вес обратимой схемы без дополнительной памяти.
\begin{theorem}[о квантовом весе обратимой схемы без дополнительных входов]\label{theorem_quantum_weight_upper_no_memory}
    $$
        W(n, 0) \leqslant \frac{n2^{n+4} \left( \WC(1 + \varepsilon_C(n)) + 2\WT(1 + \varepsilon_T(n)) \right)}
            {\log_2 n - \log_2 \log_2 n - \log_2 \phi(n)} \;  ,
    $$
    где $\phi(n)$~--- любая сколь угодно медленно растущая функция такая, что $\phi(n) < n \mathop / \log_2 n$,
    \begin{align*}
        \varepsilon_C(n) &= \frac{1}{2\phi(n)} - \left( \frac{1}{2} - o(1) \right) \cdot \frac{\log_2 \log_2 n }{n}  \; , \\
        \varepsilon_T(n) &= (4 - o(1))\frac{\log_2 n \cdot \log_2 \log_2 n}{n} \; .
    \end{align*}
\end{theorem}
\begin{proof}
    Согласно формуле~\eqref{formula_quantum_weigth_as_sum}, нам необходимо оценить величины $\LC(n,0)$ и $\LT(n,0)$.
    Из описания алгоритма синтеза~\algref{alg_asymp_no_mem} следует, что
    \begin{align*}
        \LC_1 &= 2(n-d) \; ,            & \LT_1 &= 0 \; ,\\
        \LC_2 &\leqslant 2d \; ,        & \LT_2 &= 0 \; ,\\
        \LC_3 &\leqslant 2kd \; ,       & \LT_3 &\leqslant 32k \log_2 k \; , \\
        \LC_4 &= 2(n - \log_2 k) \; ,   & \LT_4 &= 0 \; ,\\
        \LC_5 &= 0 \; ,                 & \LT_5 &\leqslant 8(n - \log_2 k) \; .
    \end{align*}
    
    Суммируя эти величины по отдельности, мы получаем следующие оценки сверху:
    \begin{align*}
        \LC(g^{(K)},0) &\leqslant \sum_{i=1}^5 {\LC_i} \leqslant 2(n-d) + 2d + 2kd + 2(n - \log_2 k) \; , \\
        \LT(g^{(K)},0) &\leqslant \sum_{i=1}^5 {\LT_i} \leqslant 32k \log_2 k + 8(n - \log_2 k) \; , \\
        \LC(g^{(K)},0) &\leqslant 4n + k 2^{k+1} - 2 \log_2 k \; , \\
        \LT(g^{(K)},0) &\leqslant 8n + 32k \log_2 k - 8 \log_2 k \; .
    \end{align*}
    Отсюда также следует, что $\LC(g^{(2)},0) \leqslant 4n + 124$ и $\LT(g^{(2)},0) \leqslant 8n + 240$.

    Подставляя данные значения в формулу~\eqref{formula_upper_bound_of_L_h_common},
    мы получаем следующие верхние оценки для $\LC(f_h,0)$ и $\LT(f_h,0)$:
    \begin{align*}
        \LC(f_h,0) &\leqslant \frac{2^{n+1}}{k} (4n + k 2^{k+1} - 2 \log_2 k) + k(4n + 124) \; , \\
        \LT(f_h,0) &\leqslant \frac{2^{n+1}}{k} (8n + 32k \log_2 k - 8 \log_2 k) + k(8n + 240) \; .
    \end{align*}

    При $m = \log_2 n - \log_2 \log_2 n - \log_2 \phi(n)$ и $k = 2^{\lfloor \log_2 m \rfloor}$,
    где $\phi(n)$~--- любая сколь угодно медленно растущая функция такая, что $\phi(n) < n \mathop / \log_2 n$,
    верны следующие неравенства для $\LC(f_h,0)$ и $\LT(f_h,0)$:
    \begin{align*}
        \LC(f_h,0) &\leqslant \frac{n2^{n+4}}{\log_2 n - \log_2 \log_2 n - \log_2 \phi(n)}
            \left( 1 + \varepsilon_C(n) \right)  \; , \\
        \LT(f_h,0) &\leqslant \frac{n2^{n+5}}{\log_2 n - \log_2 \log_2 n - \log_2 \phi(n)}
            \left( 1 + \varepsilon_T(n) \right)  \; ,
    \end{align*}
    где функции $\varepsilon_C(n)$ и $\varepsilon_T(n)$ равны
    \begin{align*}
        \varepsilon_C(n) &= \frac{1}{2\phi(n)} - \left( \frac{1}{2} - o(1) \right) \cdot \frac{\log_2 \log_2 n }{n} \; , \\
        \varepsilon_T(n) &= (4 - o(1))\frac{\log_2 n \cdot \log_2 \log_2 n}{n} \; .
    \end{align*}

    Поскольку описанный алгоритм синтеза~\algref{alg_asymp_no_mem} позволяет получить обратимую схему $\frS$
    для произвольной подстановки $h \in A(\ZZ_2^n)$, то функция $W(n,0)$ ограничена сверху также,
    как и функция $W(f_h,0)$. Из верхних оценок $\LC(f_h,0)$ и $\LT(f_h,0)$,
    а также из формулы~\eqref{formula_quantum_weigth_as_sum} следует верхняя оценка для
    $W(n,0)$ из условия теоремы.
\end{proof}

\subsection{Верхние оценки для схем с дополнительной памятью}\label{subsection_upper_bound_with_memory}

\forceindent
Элемент $k$-CNOT при $k < (n-1)$, где $n$~--- количество значимых входов схемы,
можно заменить композицией не более $8k$ элементов 2-CNOT~\cite{barenco_elementary_gates},
если не использовать дополнительные входы. Однако если использовать $(k-2)$ дополнительных входов, то элемент $k$-CNOT
при любом значении $k < n$ можно заменить композицией $(2k-3)$ элементов 2-CNOT с уборкой вычислительного мусора
(см. рис.~\ref{pic_reducing_complexity}).
При этом после такой замены на всех незначимых выходах будет значение 0, поэтому их можно будет использовать в дальнейшем.
Если же элемент $k$-CNOT заменить композицией $(k-1)$ элементов 2-CNOT с использованием $(k-2)$
дополнительных входов, то на незначимых выходах после замены могут быть значения, отличные от 0. Как следствие,
эти незначимые выходы нельзя будет использовать в дальнейшем.

\Figure[ht]
    \centering
    \includegraphics[scale=1.2]{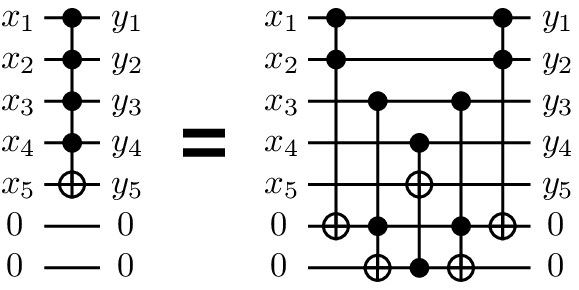}
    \caption
    {
        \small Замена одного элемента 4-CNOT композицией элементов 2-CNOT при помощи двух дополнительных входов.
    }\label{pic_reducing_complexity}
\end{figure}

Таким образом, если в алгоритме синтеза, описанном в предыдущем разделе, использовать ровно $(n-3)$ дополнительных входов,
то в формуле~\eqref{formula_upper_bound_in_synthesis_algorithm} слагаемое $12n = 4n + 8n$ можно заменить на $6n = 4n + 2n$.
В этом случае из формулы~\eqref{formula_L_h_upper_bound_general_case_with_k} следует, что
$L(n,n-3) \leqslant 3n2^{n+3}(1+o(1)) \mathop / \log_2n$.
Если же в описанном алгоритме синтеза использовать $q_0 \geqslant (n-3) 2^{n+2} / (\log_2 n - \log_2 \log_2 n - \log_2 \phi(n))$
дополнительных входов, где $\phi(n)$~--- любая сколь угодно медленно растущая функция такая,
что $\phi(n) < n \mathop / \log_2 n$,
то в формуле~\eqref{formula_upper_bound_in_synthesis_algorithm} слагаемое $12n = 4n + 8n$ можно заменить на $5n = 4n + n$.
В этом случае из формулы~\eqref{formula_L_h_upper_bound_general_case_with_k} следует, что
$L(n,q_0) \leqslant 5n 2^{n+2} / \log_2 n$. Однако можно получить существенно меньшую верхнюю оценку для $L(n,q)$ при
использовании гораздо меньшего количества дополнительных входов, что и будет показано далее.

О.\,Б. Лупановым был предложен~\cite{lupanov_delay} асимптотически оптимальный алгоритм синтеза для произвольной
булевой функции в базисе \gate{} $\{\,\neg, \wedge, \vee\,\}$. Им было доказано, что любая булева функция от $n$ переменных
может быть реализована в схеме из \gate{} данного базиса со сложностью, эквивалентной $2^n \mathop / n$,
и с задержкой, эквивалентной с точностью до порядка $n$.
Воспользуемся данным результатом и применим аналогичный подход для синтеза обратимых схем, состоящих из \gate{}
множества $\Omega_n^2$ и реализующих булево отображение $f \in F(n,q)$ с использованием $q$ дополнительных входов.

Базис $\{\,\neg, \oplus, \wedge\,\}$ является функционально полным, следовательно, в нём можно реализовать любое
булево отображение $f \in F(n,q)$. Выразим каждый элемент этого базиса через композицию элементов
NOT, CNOT и 2-CNOT. Из рис.~\ref{pic_basis} видно,
что это может быть сделано при помощи не более двух \gate{} и с глубиной не выше 2 при использовании
одного дополнительного входа.
\Figure[ht]
    \centering
    \includegraphics[scale=1.2]{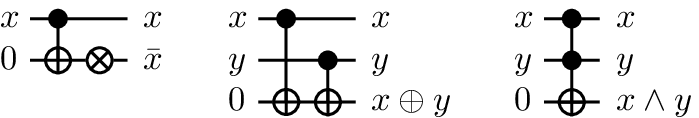}
    \caption
    {
        \small Выражение \gate{} базиса $\{\,\neg, \oplus, \wedge\,\}$ через\\
        композицию элементов NOT, CNOT и 2-CNOT.
    }\label{pic_basis}
\end{figure}

На рис.~\ref{pic_scheme_realization} на с.~\pageref{pic_scheme_realization} схематично изображена обратимая схема,
реализующая некоторое булево отображение. 
На данном рисунке все незначимые выходы, которые могут содержать или не содержать вычислительный мусор, были помечены символом~*.
В большинстве случаев, значения на этих выходах не обнуляются. Очистить вычислительный мусор, очевидно,
можно только в том случае, когда булево отображение $f$, реализуемое обратимой схемой $\frS$,
биективно.
В этом случае, продублировав часть оригинальной схемы (обозначим её через $\frS_*$), можно обнулить значения на всех
незначимых выходах, за исключением выходов, соответствующих входам отображения $f$.
Затем можно присоединить обратимую схему $\frS^{-1}$, реализующую обратное отображение $f^{-1}$,
которая позволит обнулить значения на незначимых выходах, соответствующих входам отображения $f$,
но которая также, возможно, породит свой собственный вычислительный мусор на незначимых выходах.
Этот вычислительный мусор может быть убран дублированием части схемы $\frS^{-1}$ (обозначим её через $\frS^{-1}_*$).
Таким образом, итоговая обратимая схема $\frS_{\text{res}}$ без вычислительного мусора на незначимых выходах,
имеет сложность $L(\frS_{\text{res}}) \leqslant 4 \cdot \max(L(\frS), L(\frS^{-1}))$
и глубину $D(\frS_{\text{res}}) \leqslant 4 \cdot \max(D(\frS), D(\frS^{-1}))$ (см. рис.~\ref{pic_clearing_out_garbage}).

\Figure[ht]
    \centering
    \includegraphics[scale=1.2]{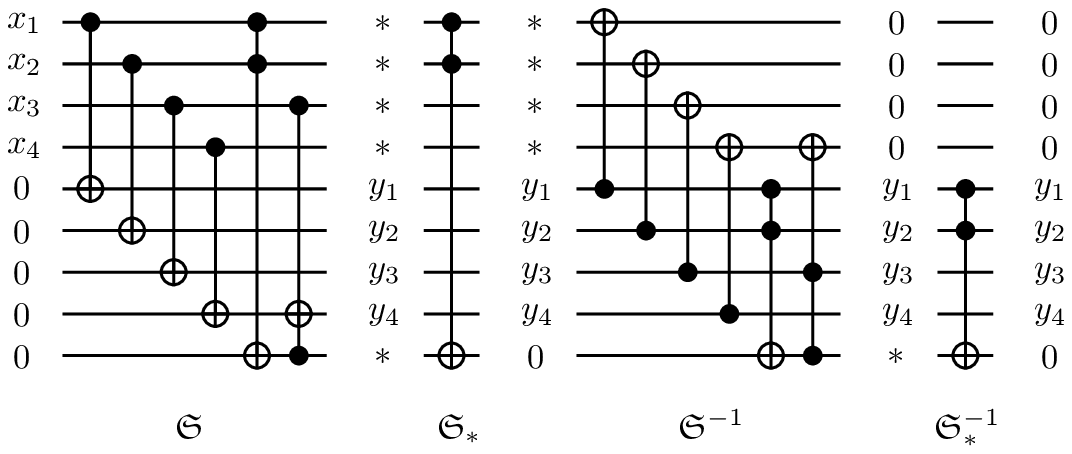}
    \caption
    {
        \small Пример обратимой схемы $\frS_{\text{res}} = \frS * (\frS_*) * \frS^{-1} * (\frS^{-1}_*)$
        с уборкой вычислительного мусора на незначимых выходах.
    }\label{pic_clearing_out_garbage}
\end{figure}

Исходя из этих соображений, все асимптотические оценки для сложности и глубины будут даваться далее по тексту
для обратимых схем с вычислительным мусором на незначимых выходах. Умножив эти оценки на 4, можно получить аналогичные оценки уже для
обратимых схем без вычислительного мусора на незначимых выходах.

\myparagraph{Снижение сложности схемы}

\forceindent
В дальнейшем нам потребуется следующая лемма о сложности обратимой схемы, реализующей все конъюнкции $n$ переменных
вида $x_1^{a_1} \wedge \ldots \wedge x_n^{a_n}$, $a_i \in \ZZ_2$.
\begin{lemma}\label{lemma_complexity_of_all_conjunctions_of_n_variables}
    Все конъюнкции $n$ переменных вида $x_1^{a_1} \wedge \ldots \wedge x_n^{a_n}$, $a_i \in \ZZ_2$,
    можно реализовать обратимой схемой $\frS_n$, состоящей из \gate{} множества $\Omega_{n+q}^2$
    и имеющей сложность $L(\frS_n) \sim 2^n$ при использовании $Q(\frS_n) \sim 2^n$ дополнительных входов.
\end{lemma}
\begin{proof}
    Сперва мы реализуем все инверсии $\bar x_i$, $1 \leqslant i \leqslant n$.
    Это может быть сделано при помощи $L_1 = 2n$ элементов NOT и CNOT при использовании $q_1 = n$ дополнительных входов.

    Искомую обратимую схему $\frS_n$ мы строим следующим образом (см. рис.~\ref{pic_lemma_complexity_of_all_conjunctions}):
    с помощью обратимых схем $\frS_{\lceil n \mathop / 2 \rceil}$ и $\frS_{\lfloor n \mathop / 2 \rfloor}$
    мы реализуем все конъюнкции $\lceil n \mathop / 2 \rceil$ первых и $\lfloor n \mathop / 2 \rfloor$ последних
    переменных. Затем мы реализуем конъюнкции значимых выходов этих двух схем каждого с каждым.
    Для этого потребуется $L_2 = 2^n$ элементов 2-CNOT и $q_2 = 2^n$ дополнительных входов.

    Отсюда следует, что
    $$
        L(\frS_n) = Q(\frS_n) = 2^n(1+o(1)) + L(\frS_{\lceil n \mathop / 2 \rceil}) + L(\frS_{\lfloor n \mathop / 2 \rfloor})
            = 2^n(1 + o(1)) \; .
    $$
\end{proof}

\Figure[ht]
    \centering
    \includegraphics[scale=1.2]{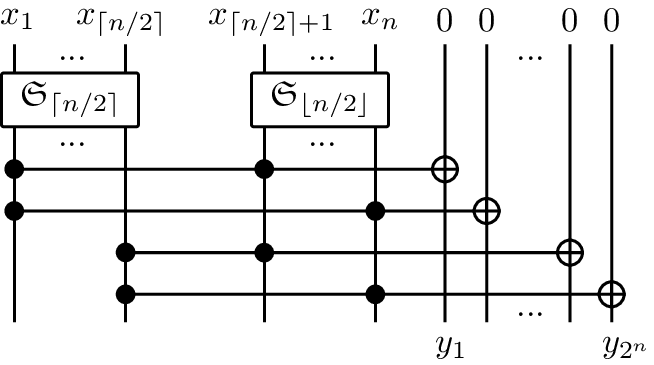}
    \caption
    {
        \small Структура обратимой схемы, реализующей все конъюнкции от $n$ переменных\\
        с минимальной сложностью (входы схемы сверху).
    }\label{pic_lemma_complexity_of_all_conjunctions}
\end{figure}

Перейдём теперь непосредственно к теореме данного параграфа.
\begin{theorem}[о сложности обратимой схемы с дополнительными входами]\label{theorem_complexity_upper_with_memory}
    $$
        L(n,q_0) \lesssim 2^n \text{ \,при\, } q_0 \sim n 2^{n-\lceil n \mathop / \phi(n)\rceil} \; ,
    $$
    где $\phi(n)$ и $\psi(n)$~--- любые сколь угодно медленно растущие функции такие,
    что $\phi(n) \leqslant n \mathop / (\log_2 n + \log_2 \psi(n))$.
\end{theorem}
\begin{proof}
    Опишем алгоритм синтеза~\nextalg{alg_asymp_with_mem_complexity_min}, который основан на методе Лупанова,
    и основная цель которого заключается в снижении сложности обратимой схемы при помощи использования
    дополнительных входов в схеме.
    
    Произвольное булево отображение $f\colon \ZZ_2^n \to \ZZ_2^n$ можно представить в виде некоторых $n$ булевых функций
    $f_i\colon \ZZ_2^n \to \ZZ_2$ от $n$ переменных
    \begin{equation}\label{formula_function_decomposition_by_n_functions}
        f(\vv x) = \langle f_1(\vv x), f_2(\vv x), \ldots, f_n(\vv x) \rangle \; .
    \end{equation}
    Каждую функцию $f_i(\vv x)$ можно разложить по последним $(n-k)$ переменным:
    \begin{equation}\label{formula_function_decomposition_by_last_variables}
        f_i(\vv x) = \bigoplus_{a_{k+1}, \ldots, a_n \in \ZZ_2} {x_{k+1}^{a_{k+1}} \wedge \ldots \wedge x_n^{a_n}}
            \wedge f_i(\langle x_1, \ldots, x_k, a_{k+1}, \ldots, a_n \rangle) \; .
    \end{equation}

    Каждая из $n2^{n-k}$ булевых функций $f_i(\langle x_1, \ldots, x_k, a_{k+1}, \ldots, a_n \rangle)$, $1 \leqslant i \leqslant n$,
    является функцией от $k$ переменных $x_1, \ldots, x_k$, её можно получить при помощи аналога
    СДНФ, в котором дизъюнкции заменяются на сложение по модулю два:
    \begin{equation}\label{formula_analog_sdnf}
        f_i(\langle x_1, \ldots, x_k, a_{k+1}, \ldots, a_n \rangle) = f_{i,j} = \bigoplus_{
            \substack{\boldsymbol \sigma \in \ZZ_2^k \\f_{i,j}(\boldsymbol \sigma) = 1}}
            x_1^{\sigma_1} \wedge \ldots \wedge x_k^{\sigma_k} \; .        
    \end{equation}
    
    Все $2^k$ конъюнкций вида $x_1^{\sigma_1} \wedge \ldots \wedge x_k^{\sigma_k}$ можно разделить на группы,
    в каждой из которых будет не более $s$ конъюнкций. Обозначим через $p = \lceil 2^k \mathop / s \rceil$ количество таких групп.
    Используя конъюнкции одной группы, мы можем реализовать не более $2^s$ булевых функций по формуле~\eqref{formula_analog_sdnf}.
    Обозначим через $G_i$ множество булевых функций, которые могут быть реализованы при помощи конъюнкций $i$-й группы,
    $1 \leqslant i \leqslant p$. Тогда $|G_i| \leqslant 2^s$.
    Следовательно, мы можем переписать формулу~\eqref{formula_analog_sdnf} следующим образом:
    \begin{equation}\label{formula_analog_sdnf_improved}
        f_i(\langle x_1, \ldots, x_k, a_{k+1}, \ldots, a_n \rangle) = \bigoplus_{
            \substack{t=1 \ldots p\\ g_{j_t} \in G_t\\ 1 \leqslant j_t \leqslant |G_t|}} g_{j_t}(\langle x_1, \ldots, x_k\rangle) \; .
    \end{equation}
    
    Отсюда следует, что
    \begin{equation}
        f_i(\vv x) = \bigoplus_{a_{k+1}, \ldots, a_n \in \ZZ_2} x_{k+1}^{a_{k+1}} \wedge \ldots \wedge x_n^{a_n} \wedge
            \left( \bigoplus_{ \substack{t=1 \ldots p\\ g_{j_t} \in G_t\\ 1 \leqslant j_t \leqslant |G_t|}}
            g_{j_t}(\langle x_1, \ldots, x_k\rangle) \right) \; .
        \label{formula_f_i_with_braces_first}
    \end{equation}

    Отметим, что все булевы функции множества $G_i$ можно реализовать, используя такой же подход, что и в Лемме%
    ~\ref{lemma_complexity_of_all_conjunctions_of_n_variables}.
    Из рис.~\ref{pic_lemma_complexity_of_all_conjunctions} видно, что в этом случае каждый элемент 2-CNOT просто заменяется
    композицией двух элементов CNOT.
    Суммарно нам потребуется $L \sim 2^{s+1}$ элементов CNOT и $q \sim 2^s$ дополнительных входов.
    
    Алгоритм синтеза~\algref{alg_asymp_with_mem_complexity_min} конструирует обратимую схему $\frS$, реализующую
    булево отображение $f$~\eqref{formula_function_decomposition_by_n_functions}, при помощи следующих подсхем
    (см. рис.~\ref{pic_five_schemes}):
    \begin{enumerate}
        \item\label{item_first_subcircuit_min_complexity}
            Подсхема $\frS_1$, реализующая все конъюнкции $k$ первых переменных $x_i$, согласно Лемме%
            ~\ref{lemma_complexity_of_all_conjunctions_of_n_variables},
            со сложностью $L_1 \sim 2^k$ и $q_1 \sim 2^k$ дополнительными входами.
            Подсхема $\frS_1$ почти вся состоит из элементов 2-CNOT (количество остальных элементов пренебрежимо мало).

        \item
            Подсхема $\frS_2$, реализующая все булевы функции $g \in G_i$ для всех $i \in \ZZ_p$
            по формуле~\eqref{formula_analog_sdnf} со сложностью $L_2 \sim p2^{s+1}$ и
            $q_2 \sim p2^s$ дополнительными входами (см. замечание выше про реализацию всех булевых функций множества $G_i$).
            Подсхема $\frS_2$ состоит только из элементов CNOT.
        
        \item
            Подсхема $\frS_3$, реализующая все $n2^{n-k}$ координатных функций $f_{i,j}(\vv x)$,
            $i \in \ZZ_{2^{n-k}}$, $j \in \ZZ_n$, по формуле~\eqref{formula_analog_sdnf_improved}
            со сложностью $L_3 \leqslant pn 2^{n-k}$ и $q_3 = n 2^{n-k}$ дополнительными входами.
            Подсхема $\frS_3$ состоит только из элементов CNOT.
            
        \item
            Подсхема $\frS_4$, реализующая все конъюнкции $(n-k)$ последних переменных $x_i$, согласно Лемме%
            ~\ref{lemma_complexity_of_all_conjunctions_of_n_variables},
            со сложностью $L_4 \sim 2^{n-k}$ и $q_4 \sim 2^{n-k}$ дополнительными входами.
            Подсхема $\frS_4$ почти вся состоит из элементов 2-CNOT (количество остальных элементов пренебрежительно мало).

        \item
            Подсхема $\frS_5$, реализующая булево отображение $f$
            по формулам~\eqref{formula_function_decomposition_by_n_functions} и~\eqref{formula_function_decomposition_by_last_variables}
            со сложностью $L_5 \leqslant n 2^{n-k}$ и $q_5 = n$ дополнительными входами.
            Подсхема $\frS_5$ состоит только из элементов 2-CNOT.
    \end{enumerate}
    
    \Figure[ht]
        \centering
        \includegraphics[scale=1.2]{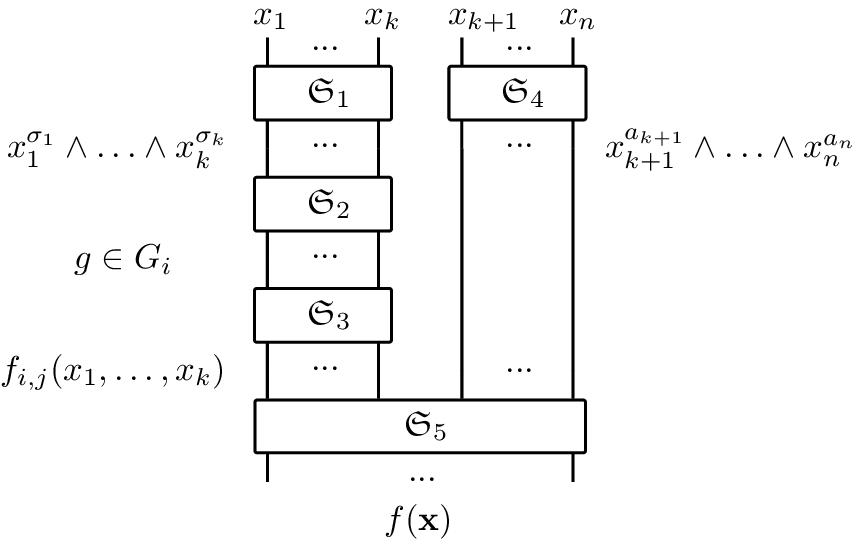}
        \caption
        {
            \small Структура обратимой схемы $\frS$, синтезируемой алгоритмом~\algref{alg_asymp_with_mem_complexity_min}
            (входы схемы сверху).
        }\label{pic_five_schemes}
    \end{figure}

    Будем искать значения параметров $k$ и $s$, удовлетворяющие следующим условиям:
    $$
        \left\{
            \begin{array}{lr}
                s = n - 2k  \; , & \\
                k = \lceil n \mathop / \phi(n) \rceil \;, & \text{где $\phi(n)$~--- некоторая растущая функция}\; , \\
                1 \leqslant s < n  \; , &\\
                1 \leqslant k < n \mathop / 2  \; , &\\
                \frac{2^k}{s} \geqslant \psi(n) \;, & \text{где $\psi(n)$~--- некоторая растущая функция} \; .
            \end{array}
        \right.
    $$
    В этом случае $p = \lceil 2^k \mathop / s \rceil \sim 2^k \mathop / s$
    и $2^{\lceil n \mathop / \phi(n) \rceil} \geqslant s\psi(n)$,
    откуда следует, что при $\phi(n) \leqslant n \mathop / (\log_2n + \log_2 \psi(n))$ параметры $k$ и $s$ будут удовлетворять
    условиям выше.
    
    Суммируя сложности обратимых подсхем $\frS_1$--$\frS_5$ и количество используемых ими дополнительных входов,
    мы получаем следующие оценки для искомой обратимой схемы $\frS$:
    \begin{gather*}
        L(\frS) \sim 2^k + p2^{s+1} + pn 2^{n-k} + 2^{n-k} + n 2^{n-k}
            \sim 2^k + \frac{2^{n-k+1}}{s} + \frac{n2^n}{s} \; , \\
        Q(\frS) \sim 2^k + p2^s + n2^{n-k} + 2^{n-k} + n \sim 2^k + \frac{2^{n-k}}{s} + n 2^{n-k} \; .
    \end{gather*}
    
    Следовательно, при $k = \lceil n \mathop / \phi(n) \rceil$ и $s = n - 2k$,
    где $\phi(n) \leqslant n \mathop / (\log_2n + \log_2 \psi(n))$ и $\psi(n)$~--- некоторые растущие функции,
    верны следующие соотношения:
    \begin{gather*}
        L(\frS) \sim 2^{\lceil n \mathop / \phi(n) \rceil} + \frac{2^{n+1}}{n(1-o(1))2^{\lceil n \mathop / \phi(n) \rceil}}
            + \frac{n2^n}{n(1-o(1))} \sim 2^n \; , \\
        Q(\frS) \sim 2^{\lceil n \mathop / \phi(n) \rceil} + \frac{2^n}{n(1-o(1))2^{\lceil n \mathop / \phi(n) \rceil}}
            + \frac{n 2^n}{2^{\lceil n \mathop / \phi(n) \rceil}} \sim \frac{n 2^n}{2^{\lceil n \mathop / \phi(n) \rceil}} \; .
    \end{gather*}
    
    Поскольку мы описали алгоритм синтеза обратимой схемы $\frS$ для произвольного булева отображения $f \in F(n,q)$,
    то $L(n,q_0) \leqslant L(\frS) \sim 2^n$, где $q_0 \sim n2^{n - \lceil n \mathop / \phi(n) \rceil}$.
\end{proof}

\begin{theorem}\label{theorem_complexity_with_memory_common}
    $$
        L(n,q_0) \asymp 2^n \text{ \,при\, } q_0 \sim n 2^{n-\lceil n \mathop / \phi(n)\rceil} \; ,
    $$
    где $\phi(n)$ и $\psi(n)$~--- любые сколь угодно медленно растущие функции такие,
    что $\phi(n) \leqslant n \mathop / (\log_2 n + \log_2 \psi(n))$.
\end{theorem}
\begin{proof}
    Следует из Теорем~\ref{theorem_complexity_lower_bound} и~\ref{theorem_complexity_upper_with_memory}.
\end{proof}

Оценим теперь квантовый вес обратимой схемы с дополнительной памятью.
\begin{theorem}[о квантовом весе обратимой схемы с дополнительными входами]\label{theorem_quantum_weight_upper_with_memory}
    $$
        W(n,q_0) \lesssim \WC \cdot 2^n + \WT \cdot n2^{n - \lceil n \mathop / \phi(n)\rceil},
            \text{ \,при\, } q_0 \sim n 2^{n-\lceil n \mathop / \phi(n)\rceil}  \;  ,
    $$
    где $\phi(n)$ и $\psi(n)$~--- любые сколь угодно медленно растущие функции такие,
    что $\phi(n) \leqslant n \mathop / (\log_2 n + \log_2 \psi(n))$.
\end{theorem}
\begin{proof}
    Согласно формуле~\eqref{formula_quantum_weigth_as_sum}, нам необходимо оценить величины $\LC(n,q_0)$ и $\LT(n,q_0)$.
    Из описания алгоритма синтеза~\algref{alg_asymp_with_mem_complexity_min} видно, что
    \begin{align*}
        \LC_1 & = O(k) \; ,                 & \LT_1 &\sim 2^k \; , \\
        \LC_2 & \sim p2^{s+1} \; ,          & \LT_2 &= 0 \; , \\
        \LC_3 & \leqslant pn 2^{n-k} \; ,   & \LT_3 &= 0 \; , \\
        \LC_4 & = O(n-k) \; ,               & \LT_4 &\sim 2^{n-k} \; , \\
        \LC_5 & = 0 \; ,                    & \LT_5 &\leqslant n 2^{n-k} \; .
    \end{align*}
    При $k = \lceil n \mathop / \phi(n) \rceil$ и $s = n - 2k$,
    где $\phi(n)$ и $\psi(n)$~--- любые сколь угодно медленно растущие функции такие,
    что $\phi(n) \leqslant n \mathop / (\log_2n + \log_2 \psi(n))$,
    верны следующие соотношения:
    \begin{align*}
        \LC(n,q_0) &\lesssim O(k) + p2^{s+1} + pn 2^{n-k} + O(n-k) \sim \frac{2^{k+s+1}}{s} + \frac{n2^n}{s} \; , \\
        \LT(n,q_0) &\lesssim 2^k + 2^{n-k} + n 2^{n-k} \sim 2^k + n2^{n-k} \; ,
    \end{align*}
    \begin{align*}
        \LC(n,q_0) &\lesssim \frac{2^{n+1}}{(n-o(n))2^{\lceil n \mathop / \phi(n) \rceil}} + \frac{n2^n}{n - o(n)}
            \sim 2^n \; , \\
        \LT(n,q_0) &\lesssim 2^{\lceil n \mathop / \phi(n) \rceil} + \frac{n2^n}{2^{\lceil n \mathop / \phi(n) \rceil}}
            \sim \frac{n2^n}{2^{\lceil n \mathop / \phi(n) \rceil}} \; .
    \end{align*}

    Из этих верхних оценок и из формулы~\eqref{formula_quantum_weigth_as_sum} следует верхняя
    оценка для функции $W(n,q_0)$ из условия теоремы.
\end{proof}

Стоит отметить, что в случае $\WT = O(\WC) = const$ верно соотношение
$$
    W(n, q_0) \asymp \LC(n, q_0) \sim L(n, q_0) \; ,
$$
где $q_0 \sim n 2^{n-\lceil n \mathop / \phi(n)\rceil}$, $\phi(n)$ и $\psi(n)$~--- любые сколь угодно медленно растущие функции
такие, что $\phi(n) \leqslant n \mathop / (\log_2 n + \log_2 \psi(n))$.
Другими словами, количество элементов 2-CNOT в обратимой схеме,
синтезируемой алгоритмом~\algref{alg_asymp_with_mem_complexity_min}, пренебрежимо мало по сравнению со сложностью этой схемы.
В случае же когда использование дополнительных входов запрещено, количество элементов 2-CNOT в обратимой схеме,
согласно Теоремам \ref{theorem_quantum_weight_lower_bound}, \ref{theorem_complexity_no_memory_common}
и~\ref{theorem_quantum_weight_upper_no_memory}, эквивалентно с точностью до порядка сложности этой схемы.

Для булева отображения $f\colon \ZZ_2^n \to \ZZ_2^m$ можно получить аналогичные верхние оценки сложности реализующей его
обратимой схемы.
\begin{theorem}\label{theorem_arbitrary_boolean_transformation_complexity}
    Любое булево отображение $f\colon \ZZ_2^n \to \ZZ_2^m$ можно реализовать с помощью обратимой схемы $\frS$,
    имеющей сложность $L(\frS) \lesssim m2^n \mathop / n$, при использовании
    $q \sim (m+1) 2^{n-\lceil n \mathop / \phi(n)\rceil}$
    дополнительных входов, где $\phi(n)$ и $\psi(n)$~--- любые сколь угодно медленно растущие функции такие,
    что $\phi(n) \leqslant n \mathop / (\log_2 n + \log_2 \psi(n))$.
\end{theorem}
\begin{proof}
    Алгоритм синтеза~\algref{alg_asymp_with_mem_complexity_min} можно модифицировать таким образом, чтобы он мог
    синтезировать обратимую схему, реализующую заданное булево отображение $f\colon \ZZ_2^n \to \ZZ_2^m$.
    
    Булево отображение $f(\vv x)$ в формуле~\eqref{formula_function_decomposition_by_n_functions}
    теперь будет представляться системой не $n$,
    а $m$ булевых функций. Следовательно, подсхема $\frS_3$ будет уже иметь сложность $L_3 \leqslant pm 2^{n-k}$
    и использовать $q_3 = m 2^{n-k}$ дополнительных входов, а подсхема $\frS_5$ иметь сложность
    $L_5 \leqslant m 2^{n-k}$ и использовать $q_5 = m$ дополнительных входов.
    
    Используя эти обновлённые значения при подсчёте суммарной сложности схемы $\frS$ и количества дополнительных входов,
    можно убедиться, что даже при $m = 1$ верна верхняя оценка $L(\frS) \lesssim m2^n \mathop / n$ из условия теоремы
    при использовании $q \sim (m+1) 2^{n-\lceil n \mathop / \phi(n)\rceil}$ дополнительных входов,
    где $\phi(n)$ и $\psi(n)$~--- любые сколь угодно медленно растущие функции такие,
    что $\phi(n) \leqslant n \mathop / (\log_2 n + \log_2 \psi(n))$.
\end{proof}

\myparagraph{Снижение глубины схемы}

\forceindent
Мы описали алгоритм синтеза~\algref{alg_asymp_with_mem_complexity_min}, основная цель которого была снижение
сложности синтезированной обратимой схемы при помощи использования дополнительных входов в схеме.
Однако мы можем использовать аналогичный подход для снижения глубины обратимой схемы.
Обозначим такой алгоритм синтеза через~\nextalg{alg_asymp_with_mem_depth_min}.
Важнейшей особенностью данного алгоритма является копирование значений с некоторых выходов на дополнительных входы
с логарифмической глубиной (см. рис.~\ref{pic_parallel_copy}). В результате все эти дополнительные входы могут
быть использованы независимо друг от друга, что позволяет достичь глубины 1 для следующей операции.
Всё, что нам нужно сделать для этого,~--- скопировать значение достаточное количество раз.
\Figure[ht]
    \centering
    \includegraphics[scale=1.2]{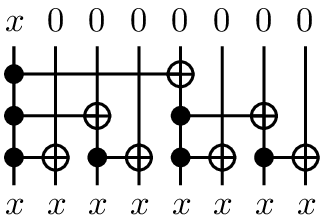}
    \caption
    {
        \small Копирование значения $x$ на дополнительные входы\\
        с логарифмической глубиной (входы схемы сверху).
    }\label{pic_parallel_copy}
\end{figure}

Докажем теперь лемму, схожую с Леммой~\ref{lemma_complexity_of_all_conjunctions_of_n_variables},
о глубине обратимой схемы, реализующей все конъюнкции $n$ переменных вида
$x_1^{a_1} \wedge \ldots \wedge x_n^{a_n}$, $a_i \in \ZZ_2$.
\begin{lemma}\label{lemma_depth_of_all_conjunctions_of_n_variables}
    Все конъюнкции $n$ переменных вида $x_1^{a_1} \wedge \ldots \wedge x_n^{a_n}$, $a_i \in \ZZ_2$,
    могут быть реализованы обратимой схемой $\frS_n$, состоящей из \gate{} множества $\Omega_{n+q}^2$ и
    имеющей глубину $D(\frS_n) \sim n$ при использовании $Q(\frS_n) \sim 3 \cdot 2^n$ дополнительных входов.
    Сложность такой схемы $L(\frS_n) \sim 3 \cdot 2^n$.
\end{lemma}
\begin{proof}
    Сперва реализуем все инверсии $\bar x_i$, $1 \leqslant i \leqslant n$.
    Это можно сделать с глубиной $D_1 = 2$ при использовании $L_1 = 2n$ элементов NOT и CNOT
    и $q_1 = n$ дополнительных входов.

    Искомую обратимую схему $\frS_n$ мы строим таким же образом, что и в Лемме%
    ~\ref{lemma_complexity_of_all_conjunctions_of_n_variables}, используя подсхемы $\frS_{\lceil n \mathop / 2 \rceil}$
    и $\frS_{\lfloor n \mathop / 2 \rfloor}$ (см. рис.~\ref{pic_lemma_depth_of_all_conjunctions}).
    Каждый значимый выход этих двух подсхем будет использован не более чем в $2 \cdot 2^{n \mathop / 2}$ конъюнкциях со значимыми выходами
    другой подсхемы, поэтому все конъюнкции могут быть реализованы с глубиной
    $D_2 \leqslant 2 + n \mathop / 2$ при использовании $2^{n+1}$ элементов CNOT, $2^n$ элементов 2-CNOT и
    $q_2 = 3 \cdot 2^n$ дополнительных входов.

    Отсюда следует, что
    \begin{gather*}
        D(\frS_n) \sim 2 + \frac{n}{2} + \max(D(\frS_{\lceil n \mathop / 2 \rceil}), D(\frS_{\lfloor n \mathop / 2 \rfloor}))
            \sim n \; , \\
        L(\frS_n) \sim 3 \cdot 2^n + L(\frS_{\lceil n \mathop / 2 \rceil}) + L(\frS_{\lfloor n \mathop / 2 \rfloor})
            \sim 3 \cdot 2^n \; , \\
        Q(\frS_n) \sim 3 \cdot 2^n + Q(\frS_{\lceil n \mathop / 2 \rceil}) + Q(\frS_{\lfloor n \mathop / 2 \rfloor})
            \sim 3 \cdot 2^n \; .
    \end{gather*}
\end{proof}

\Figure[ht]
    \centering
    \includegraphics[scale=1.2]{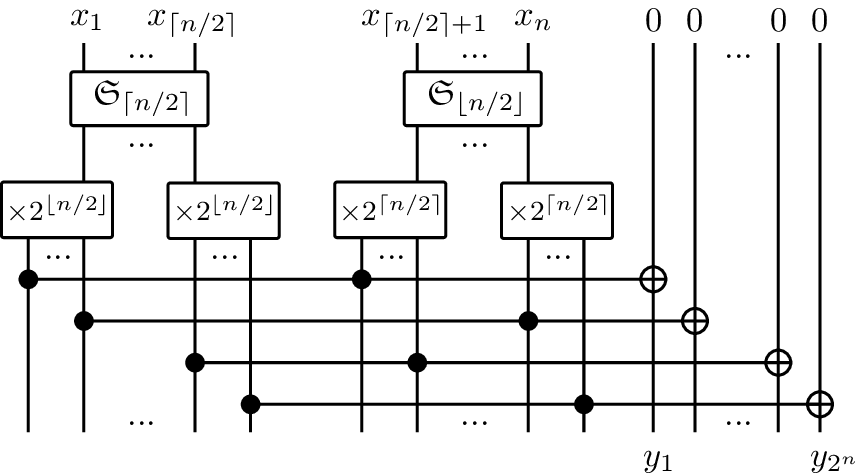}
    \caption
    {
        \small Структура обратимой схемы, реализующей все конъюнкции от $n$ переменных\\
        с минимальной глубиной (входы схемы сверху).
    }\label{pic_lemma_depth_of_all_conjunctions}
\end{figure}

Перейдём теперь к первой теореме данного параграфа.
\begin{theorem}\label{theorem_depth_upper_with_memory_3n}
    $$
        D(n,q_1) \lesssim 3n \text{ \,при\, } q_1 \sim 2^n \; .
    $$
    Обратимая схема $\frS$, реализующая отображение $f \in F(n, q_1)$ с глубиной $D(\frS) \sim 3n$,
    имеет сложность $L(\frS) \sim 2^{n+1}$
    и квантовый вес $W(\frS) \sim \WC \cdot 2^{n+1} + \WT \cdot n2^{n - \lceil n \mathop / \phi(n)\rceil}$,
    где $\phi(n)$ и $\psi(n)$~--- любые сколь угодно медленно растущие функции такие,
    что $\phi(n) \leqslant n \mathop / (\log_2 n + \log_2 \psi(n))$.
\end{theorem}
\begin{proof}    
    Опишем алгоритм синтеза~\algref{alg_asymp_with_mem_depth_min}.
    Он схож с алгоритмом~\algref{alg_asymp_with_mem_complexity_min}, поэтому мы опустим описание
    представления булева отображения $f \in P_2(n,n)$
    по формулам~\eqref{formula_function_decomposition_by_n_functions}--\eqref{formula_function_decomposition_by_last_variables}.

    Отметим, что все булевы функции множества $G_i$ можно реализовать с глубиной $s$, используя такой же подход, что и в Лемме%
    ~\ref{lemma_depth_of_all_conjunctions_of_n_variables}. Для этого нам потребуется $L \sim 3 \cdot 2^s$ элементов CNOT,
    из которых $2^{s+1}$ \gate{} нужны для копирования значений на дополнительные входы,
    и $q \sim 2^{s+1}$ дополнительных входов (см. рис.~\ref{pic_structure_of_s_3_depth}).
    \Figure[ht]
        \centering
        \includegraphics[scale=1.2]{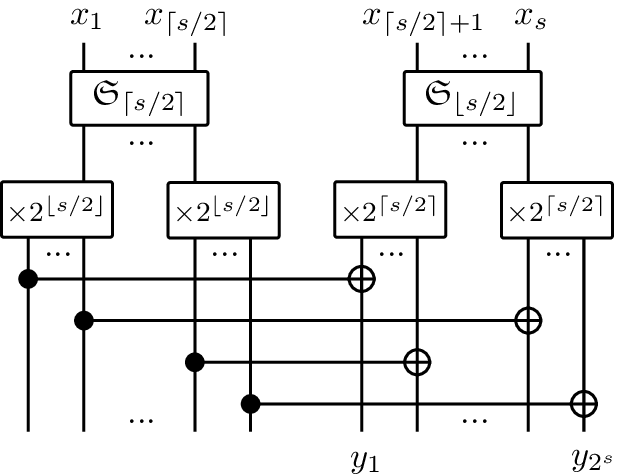}
        \caption
        {
            \small Структура обратимой подсхемы, реализующей все булевы функции $g \in G_i$\\
            с минимальной глубиной (входы схемы сверху).
        }\label{pic_structure_of_s_3_depth}
    \end{figure}

    Алгоритм~\algref{alg_asymp_with_mem_depth_min} конструирует обратимую схему $\frS$, реализующую
    булево отображение $f$~\eqref{formula_function_decomposition_by_n_functions},
    при помощи следующих подсхем (см. рис.~\ref{pic_six_schemes}):

    \Figure[ht]
        \centering
        \includegraphics[scale=1.2]{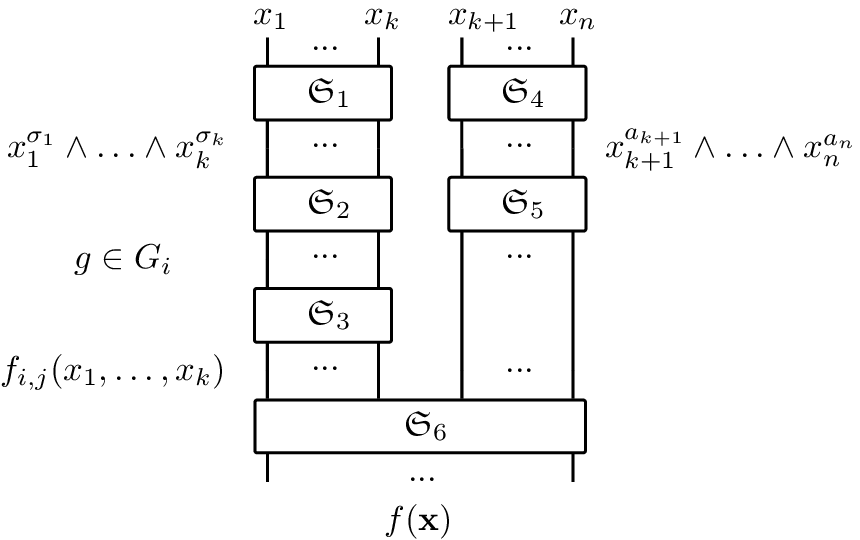}
        \caption
        {
            \small Структура обратимой схемы $\frS$, синтезируемой алгоритмом~\algref{alg_asymp_with_mem_depth_min}
            (входы схемы сверху).
        }\label{pic_six_schemes}
    \end{figure}

    \begin{enumerate}
        \item\label{item_first_subcircuit_min_depth}
            Подсхема $\frS_1$, реализующая все конъюнкции $k$ первых переменных $x_i$, согласно Лемме%
            ~\ref{lemma_depth_of_all_conjunctions_of_n_variables},
            с глубиной $D_1 \sim k$, сложностью $L_1 \sim 3 \cdot 2^k$ и $q_1 \sim 3 \cdot 2^k$
            дополнительными входами. Подсхема $\frS_1$ содержит $2^k$ элементов 2-CNOT.

        \item
            Подсхема $\frS_2$, реализующая все булевы функции $g \in G_i$ для всех $i \in \ZZ_p$
            по формуле~\eqref{formula_analog_sdnf} с глубиной $D_2 \sim s$, сложностью $L_2 \sim 3p2^s$
            и $q_2 \sim p2^{s+1}$ дополнительными входами
            (см. замечание выше про реализацию всех булевых функций множества $G_i$).
            Подсхема $\frS_2$ состоит только из элементов CNOT.
        
        \item
            Подсхема $\frS_3$, реализующая все $n2^{n-k}$ координатных функций $f_{i,j}(\vv x)$,
            $i \in \ZZ_{2^{n-k}}$, $j \in \ZZ_n$, по формуле~\eqref{formula_analog_sdnf_improved}.
            Особенностью данной подсхемы является то, что некоторая булева функция $g \in G_t$ может использоваться больше
            одного раза. Максимальное количество использования функции $g$ не превосходит $n2^{n-k}$.
            Следовательно, сперва нам необходимо скопировать значения со значимых выходов подсхемы $\frS_2$
            для всех таких булевых функций.
            Это можно сделать с глубиной $(n-k +\log_2 n)$ при использовании не более $pn2^{n-k}$ \gate{}
            и $pn2^{n-k}$ дополнительных входов (см. рис.~\ref{pic_parallel_copy}).
            Затем производится сложение по модулю 2 полученных выходов с глубиной $\log_2 p$,
            сложностью $(p-1)n2^{n-k}$
            и без использования дополнительных входов (см. рис.~\ref{pic_logarithmic_xor}).
            Таким образом, подсхема $\frS_3$ имеет глубину $D_3 \sim n-k + \log_2 p$,
            сложность $L_3 \sim (2p-1)n 2^{n-k}$ и $q_3 \sim pn 2^{n-k}$ дополнительных входов.
            Подсхема $\frS_3$ состоит только из элементов CNOT.
            
        \item
            Подсхема $\frS_4$, реализующая все конъюнкции $(n-k)$ последних переменных $x_i$, согласно Лемме%
            ~\ref{lemma_depth_of_all_conjunctions_of_n_variables},
            с глубиной $D_4 \sim (n-k)$, сложностью $L_4 \sim 3 \cdot 2^{n-k}$ и $q_4 \sim 3 \cdot 2^{n-k}$
            дополнительными входами. Подсхема $\frS_4$ содержит $2^{n-k}$ элементов 2-CNOT.
            
        \item
            Подсхема $\frS_5$, необходимая для копирования $(n-1)$ раз значения каждого значимого выхода подсхемы $\frS_4$.
            Это можно сделать с глубиной $D_5 \sim \log_2 n$, сложностью $L_5 = (n-1) \cdot 2^{n-k}$
            и $q_5 = (n-1)2^{n-k}$ дополнительными входами.
            Подсхема $\frS_5$ состоит только из элементов CNOT.

        \item
            Подсхема $\frS_6$, реализующая булево отображение $f$
            по формулам~\eqref{formula_function_decomposition_by_n_functions} и~\eqref{formula_function_decomposition_by_last_variables}.
            Структура данной подсхемы следующая: все $n2^{n-k}$ координатных функций $f_{i,j}(\vv x)$ группируются
            по $2^{n-k}$ функций (всего $n$ групп, соответствующих $n$ выходам отображения $f$).
            Функции одной группы объединяются по две.
            В каждой паре функций производится конъюнкция соответствующих значимых выходов подсхем $\frS_3$ и $\frS_5$
            при помощи двух элементов 2-CNOT. При этом для каждой пары функций используется один дополнительный вход
            для хранения промежуточного результата (см. рис.~\ref{pic_s_6_logarithmic_depth}).
            Таким образом, данный этап требует глубины 2, $n2^{n-k}$ элементов 2-CNOT и $n2^{n-k-1}$ дополнительных входов.
            Затем в каждой из $n$ групп полученных значений происходит суммирование по подулю 2 при помощи элементов CNOT
            с логарифмической глубиной (см. рис.~\ref{pic_logarithmic_xor} и~\ref{pic_s_6_logarithmic_depth}).
            Следовательно, этот этап требует глубины $(n-k-1)$ при использовании $n(2^{n-k-1} - 1)$ элементов CNOT
            и без использования дополнительных входов,
            т.\,к. можно обойтись уже существующими выходами для суммирования по модулю 2.
            
            В итоге получаем подсхему $\frS_6$ с глубиной $D_6 \sim (n-k)$, сложностью $L_6 \sim 3n 2^{n-k-1}$
            и $q_6 \sim n 2^{n-k-1}$ дополнительными входами.
    \end{enumerate}
    \noindent
    Отметим, что подсхемы $\frS_1$--$\frS_3$ и $\frS_4$--$\frS_5$ могут работать параллельно,
    т.\,к. они работают с непересекающимися подмножествами множества входов $x_1, \ldots, x_n$
    обратимой схемы $\frS$ (см. рис.~\ref{pic_six_schemes}).

    \Figure[ht]
        \centering
        \includegraphics[scale=1.2]{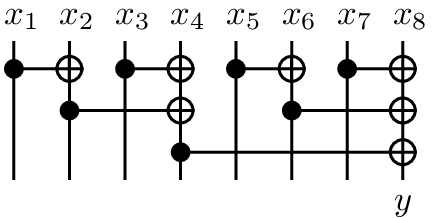}
        \caption
        {
            \small Реализация с помощью обратимой схемы функции $y = x_1 \oplus \ldots \oplus x_8$\\
            с логарифмической глубиной (часть подсхемы $\frS_3$ алгоритма синтеза~\algref{alg_asymp_with_mem_depth_min};\\
            входы схемы сверху).
        }\label{pic_logarithmic_xor}
    \end{figure}

    \Figure[ht]
        \centering
        \includegraphics[scale=1.2]{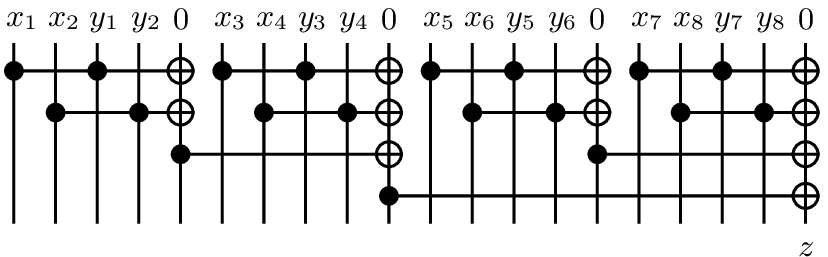}
        \caption
        {
            \small Реализация с помощью обратимой схемы функции $z = \bigoplus_{i=1}^8{x_i \wedge y_i}$\\
            с логарифмической глубиной (часть подсхемы $\frS_6$ алгоритма синтеза~\algref{alg_asymp_with_mem_depth_min};\\
            входы схемы сверху).
        }\label{pic_s_6_logarithmic_depth}
    \end{figure}

    Будем искать значения параметров $k$ и $s$, удовлетворяющие следующим условиям:
    $$
        \left\{
            \begin{array}{lr}
                k + s = n \; , & \\
                1 \leqslant k < n \; , & \\
                1 \leqslant s < n \; , & \\
                2^k \mathop / s \geqslant \psi(n) \;, & \text{где $\psi(n)$~--- некоторая растущая функция} \; .
            \end{array}
        \right.
    $$
    В этом случае $p = \lceil 2^k \mathop / s \rceil \sim 2^k \mathop / s$.
                
    Суммируя глубины, сложности и количество дополнительных входов всех подсхем $\frS_1$--$\frS_6$,
    мы получаем следующие оценки для характеристик обратимой схемы $\frS$.
    
    Глубина
    \begin{gather}
        D(\frS) \sim \max(k + s + n - k + \log_2 p \;;\; n-k + \log_n) + n - k  \; ,  \notag \\
        D(\frS) \sim 2n + s \label{formula_depth_general_upper_bound} \; .
    \end{gather}
    
    Cложность
    \begin{gather}
        L(\frS) \sim 3 \cdot 2^k + 3p2^s + (2p-1)n2^{n-k} + 3 \cdot 2^{n-k} + n2^{n-k} + 3n2^{n-k-1}  \; , \notag \\
        L(\frS) \sim 3 \cdot \frac{2^n}{2^s} + \frac{3 \cdot 2^n}{s} + \frac{n2^{n+1}}{s} \sim \frac{n2^{n+1}}{s}  \; .
            \label{formula_complexity_for_depth_general_upper_bound}
    \end{gather}
    
    Количество дополнительных входов
    \begin{gather}
        Q(\frS) \sim 3 \cdot 2^k + p2^{s+1} + pn2^{n-k} + 3 \cdot 2^{n-k} + n2^{n-k} + n2^{n-k-1}  \; , \notag \\
        Q(\frS) \sim 3 \cdot \frac{2^n}{2^s} + \frac{2^{n+1}}{s} + \frac{n2^n}{s} \sim \frac{n2^n}{s}  \; .
            \label{formula_memory_for_depth_general_upper_bound}
    \end{gather}
    
    Из описания алгоритма~\algref{alg_asymp_with_mem_depth_min} следует, что
    \begin{align*}
        \LC_1 & \sim 2^{k+1} \; ,         & \LT_1 &\sim 2^k \; , \\
        \LC_2 & \sim 3p2^s \; ,           & \LT_2 &= 0 \; , \\
        \LC_3 & \sim pn 2^{n-k+1} \; ,    & \LT_3 &= 0 \; , \\
        \LC_4 & \sim 2^{n-k+1} \; ,       & \LT_4 &\sim 2^{n-k} \; , \\
        \LC_5 & \sim n2^{n-k} \; ,        & \LT_5 &= 0 \; , \\
        \LC_6 & \sim n2^{n-k-1} \; ,      & \LT_6 &= n 2^{n-k} \; .
    \end{align*}
    Отсюда следует, что
    \begin{align}
        \LC(\frS) &\sim 2^{k+1} + \frac{n2^{n+1}}{s} \sim \frac{n2^{n+1}}{s}
            \label{formula_lc_complexity_for_depth_general_upper_bound} \; , \\
        \LT(\frS) &\sim 2^k + n2^{n-k}
            \label{formula_lt_complexity_for_depth_general_upper_bound} \; .
    \end{align}

    Пусть $k = \lceil n \mathop / \phi(n) \rceil$, где $\phi(n)$~--- любая сколь угодно медленно растущая функция такая,
    что $\phi(n) = o(n)$.
    Тогда $s = n - \lceil n \mathop / \phi(n)\rceil$ и
    $$
        2^k \geqslant s\psi(n) \Rightarrow k \geqslant \log_2 s + \log_2 \psi(n) \Rightarrow
            \phi(n) \leqslant \frac{n}{\log_2 s + \log_2 \psi(n) - 1} \; .
    $$
    Мы всегда можем выбрать любые сколь угодно медленно растущие функции $\phi(n)$ и $\psi(n)$ такие,
    что $\phi(n) \leqslant n \mathop / (\log_2 n + \log_2 \psi(n))$.
    Таким образом, мы получаем следующие оценки для характеристик обратимой схемы $\frS$:
    \begin{gather*}
        D(\frS) \sim 2n + n - \lceil n \mathop / \phi(n)\rceil \sim 3n \; , \\
        L(\frS) \sim \LC(\frS) \sim \frac{n2^{n+1}}{n - \lceil n \mathop / \phi(n)\rceil} \sim 2^{n+1} \; , \\
        \LT(\frS) \sim 2^{\lceil n \mathop / \phi(n)\rceil} + n2^{n - \lceil n \mathop / \phi(n)\rceil}
            \sim n2^{n - \lceil n \mathop / \phi(n)\rceil} \; , \\
        Q(\frS) \sim \frac{n2^n}{n - \lceil n \mathop / \phi(n)\rceil} \sim 2^n \; .
    \end{gather*}
    \noindent
    Поскольку описанный алгоритм~\algref{alg_asymp_with_mem_depth_min} позволяет синтезировать обратимую схему $\frS$
    для любого отображения $f \in F(n,q)$, можно утверждать, что $D(n,q_1) \leqslant D(\frS) \sim 3n$, где $q_1 \sim 2^n$.
    
    Также, синтезируемая схема $\frS$ имеет сложность $L(\frS) \sim 2^{n+1}$
    и квантовый вес $W(\frS) \sim \WC \cdot 2^{n+1} + \WT \cdot n2^{n - \lceil n \mathop / \phi(n)\rceil}$,
    где $\phi(n)$ и $\psi(n)$~--- любые сколь угодно медленно растущие функции такие,
    что $\phi(n) \leqslant n \mathop / (\log_2 n + \log_2 \psi(n))$.
\end{proof}

В следующей теореме показывается, что можно снизить асимптотическую глубину обратимой схемы с $3n$ до $2n$.
\begin{theorem}[о глубине обратимой схемы с дополнительными входами]\label{theorem_depth_upper_with_memory_2n}
    $$
        D(n,q_2) \lesssim 2n \text{ \,при\, } q_2 \sim \phi(n)2^n  \; ,
    $$
    где $\phi(n)$~--- любая сколь угодно медленно растущая функция такая, что $\phi(n) = o(n)$.
    Обратимая схема $\frS$, реализующая отображение $f \in F(n, q_2)$ с глубиной $D(\frS) \sim 2n$,
    имеет сложность $L(\frS) \sim \phi(n)2^{n+1}$
    и квантовый вес $W(\frS) \sim \WC \cdot \phi(n)2^{n+1} + \WT \cdot 2^{n - \lceil n \mathop / \phi(n)\rceil}$.
\end{theorem}
\begin{proof}
    Основано на результатах из доказательства Теоремы~\ref{theorem_depth_upper_with_memory_3n}.

    Пусть $s = \lceil n \mathop / \phi(n) \rceil$, где $\phi(n)$~--- любая сколь угодно медленно растущая функция
    такая, что $\phi(n) = o(n)$.
    В этом случае $k = n - \lceil n \mathop / \phi(n)\rceil$ и
    $$
        \psi(n) \leqslant \frac{2^k}{s} \leqslant \frac{\phi(n)2^{n-o(n)}}{n} \; .
    $$
    Из этого неравенства видно, что для любой сколь угодно медленно растущей функции $\phi(n)$ такой, что $\phi(n) = o(n)$,
    мы всегда можем подобрать растущую функцию $\psi(n)$.
    Из формул \eqref{formula_depth_general_upper_bound}--\eqref{formula_lt_complexity_for_depth_general_upper_bound}
    следует, что для таких значений параметров $k$ и $s$ верны следующие оценки:
    \begin{gather*}
        D(\frS) \sim 2n + \lceil n \mathop / \phi(n) \rceil \sim 2n \; , \\
        L(\frS) \sim \LC(\frS) \sim \frac{n2^{n+1}}{\lceil n \mathop / \phi(n) \rceil} \sim \phi(n)2^{n+1} \; , \\
        \LT(\frS) \sim 2^{n - \lceil n \mathop / \phi(n)\rceil} + n2^{\lceil n \mathop / \phi(n) \rceil}
            \sim 2^{n - \lceil n \mathop / \phi(n)\rceil} \; , \\
        Q(\frS) \sim \frac{n2^n}{\lceil n \mathop / \phi(n) \rceil} \sim \phi(n)2^n \; .
    \end{gather*}

    Поскольку описанный алгоритм~\algref{alg_asymp_with_mem_depth_min} позволяет синтезировать обратимую схему $\frS$
    для любого отображения $f \in F(n,q)$, можно утверждать, что $D(n,q_2) \leqslant D(\frS) \sim 2n$,
    где $q_2 \sim \phi(n)2^n$.

    Также, синтезируемая схема $\frS$ имеет сложность $L(\frS) \sim \phi(n)2^{n+1}$
    и квантовый вес $W(\frS) \sim \WC \cdot \phi(n)2^{n+1} + \WT \cdot 2^{n - \lceil n \mathop / \phi(n)\rceil}$.
\end{proof}

\myparagraph{Общая верхняя оценка сложности и глубины}

В разделe~\ref{subsection_upper_bound_no_memory} и предыдущих параграфах данного раздела были получены
верхние оценки для функции $L(n,q)$ в двух частных случаях: $q = 0$ (Теорема~\ref{theorem_complexity_upper_no_memory})
и $q \sim n2^{n-o(n)}$ (Теорема~\ref{theorem_complexity_upper_with_memory}).
Из Теоремы~\ref{theorem_complexity_lower_bound} следует, что нижняя оценка для функции $L(n,q)$ уменьшается с увеличением значения $q$.
Следовательно, можно предположить, что и верхняя оценка для функции $L(n,q)$ также уменьшается с увеличением значения $q$.
Доказательству этого утверждения и посвящён данный параграф.

Рассмотрим обратимую схему $\frS_n$ на рис.~\ref{pic_lemma_complexity_of_all_conjunctions}
на с.~\pageref{pic_lemma_complexity_of_all_conjunctions},
реализующую все конъюнкции $n$ переменных вида
$x_1^{a_1} \wedge \ldots \wedge x_n^{a_n}$, $a_i \in \ZZ_2$. Схема имеет $n$ значимых входов и $2^n$ значимых выходов;
в неё входят подсхемы $\frS_{\lceil n \mathop / 2 \rceil}$ и $\frS_{\lfloor n \mathop / 2 \rfloor}$,
реализующие конъюнкции от меньшего числа переменных.
Если все $2^n$ конъюнкций на значимых выходах основной схемы реализовать одновременно, а не по мере необходимости,
то $L(\frS_n) \sim 2^n$ и $Q(\frS_n) \sim 2^n$,
как было доказано в Лемме~\ref{lemma_complexity_of_all_conjunctions_of_n_variables}.
С другой стороны, мы можем конструировать конъюнкции по мере необходимости по одной, а не все сразу, используя только лишь значимые выходы
подсхем $\frS_{\lceil n \mathop / 2 \rceil}$ и $\frS_{\lfloor n \mathop / 2 \rfloor}$ и всего один дополнительный вход,
который и будет хранить значение нужной нам конъюнкции. После того, как все необходимые операции с этим значимым выходом будут
осуществлены, мы можем его обнулить, применив те же функциональные элементы, что и для его получения, но в обратном порядке.
Таким образом, в рассматриваемом нами случае для получения каждой конъюнкции потребуется не более двух элементов 2-CNOT,
а для получения $t$ конъюнкций (последовательно, по мере необходимости)~--- не более $2t$ элементов 2-CNOT.
Следовательно, $L(\frS_n) \lesssim O(2^{n \mathop / 2}) + 2t$, $Q(\frS_n) \lesssim O(2^{n \mathop / 2}) + 1$.
Такой же подход можно применить к подсхемам
$\frS_{\lceil n \mathop / 2 \rceil}$ и $\frS_{\lfloor n \mathop / 2 \rfloor}$, а также к подсхемам этих подсхем.
Если вообще не хранить промежуточных значений, а конструировать конъюнкции по мере необходимости, имея лишь входы
$x_1, \ldots, x_n, \bar x_1, \ldots, \bar x_n$, то для получения каждой конъюнкции $x_1^{a_1} \wedge \ldots \wedge x_n^{a_n}$,
очевидно, потребуется не более $2(n-1)$ элементов 2-CNOT, а дополнительных входов потребуется всего $(n-1)$ на все конъюнкции.
К примеру, на рис.~\ref{pic_construct_conjunctions_on_demand} показан пример конструирования 
конъюнкции $\bar x_1 \bar x_2 \bar x_3 x_4 x_5 \bar x_6 x_7 x_8$ с
использованием промежуточных значений $\bar x_1 \bar x_2$, $\bar x_3 x_4$, $x_5 \bar x_6$
и $x_7 x_8$ и последующим обнулением значений на незначимых выходах.

\Figure[ht]
    \centering
    \includegraphics[scale=1.2]{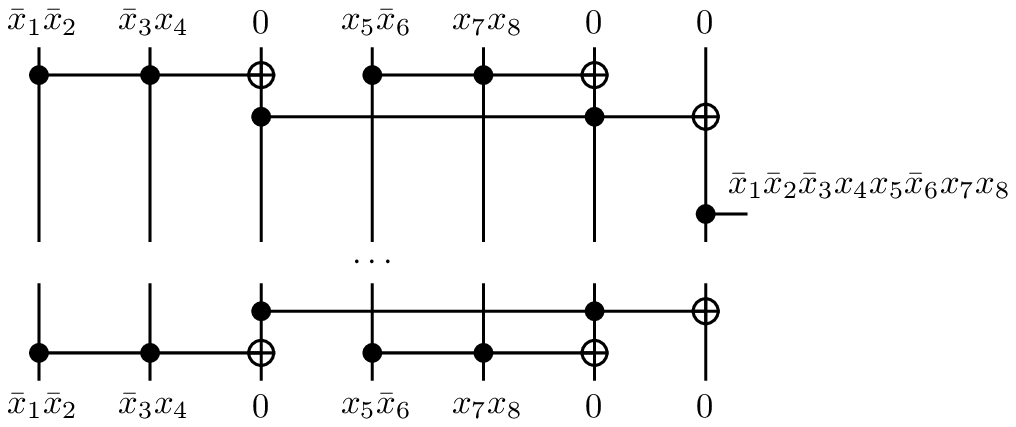}
    \caption
    {
        \small Пример конструирования конъюнкции с использованием промежуточных значений
        и последующим обнулением значений на незначимых выходах (входы схемы сверху).
    }\label{pic_construct_conjunctions_on_demand}
\end{figure}

Рассмотрим в общем случае обратимую схему $\frS_{CONJ(n,q)}$, которая реализует конъюнкции
$n$ переменных вида $x_1^{a_1} \wedge \ldots \wedge x_n^{a_n}$, $a_i \in \ZZ_2$,
при условии, что для хранения промежуточных значений отведено $q$ дополнительных входов,
а значения $\bar x_1, \ldots, \bar x_n$ уже получены ранее.
Обозначим через $L_{CONJ}(n, q, t)$ сложность схемы $\frS_{CONJ(n,q)}$,
реализующей по мере необходимости $t$ конъюнкций, не обязательно различных, причём значение $t$ может быть любым,
в том числе больше $2^n$.
Также обозначим через $Q_{CONJ}(n, q, t)$ общее количество необходимых дополнительных входов для такой обратимой схемы.
Из рассуждений выше можно вывести следующие простые оценки:
\begin{align}
    L_{CONJ}(n, 0, t) & \leqslant 2(n-1)t
        \label{formula_L_CONJ_0}   \; , \\
    Q_{CONJ}(n, 0, t) & = n-1
        \label{formula_Q_CONJ_0}   \; .
\end{align}
Для $q_0 \sim 2^{\lceil n \mathop / 2 \rceil} + 2^{\lfloor n \mathop / 2 \rfloor}$ верны соотношения
\begin{gather*}
    L_{CONJ}(n, q_0, t) \leqslant q_0 + 2t  \; , \\
    Q_{CONJ}(n, q_0, t) \leqslant q_0 + 1  \; .
\end{gather*}

Выведем зависимость значения функции $L_{CONJ}(n, q, t)$ от значения $q$.
\begin{lemma}\label{lemma_L_CONJ_bound}
    Для любого значения $q$ такого, что $2n < q < 2n2^{4n}$, верны соотношения
    \begin{align}
        L_{CONJ}(n, q, t) & \leqslant q + \frac{8nt}{\log_2 q - \log_2 n - 1}
            \label{formula_L_CONJ_q}   \; , \\
        Q_{CONJ}(n, q, t) & \leqslant q + n - 1
            \label{formula_Q_CONJ_q}   \; .
    \end{align}
\end{lemma}
\begin{proof}
    Соотношение $Q_{CONJ}(n, q, t) \leqslant q + n - 1$ следует из соотношения~\eqref{formula_Q_CONJ_0}.

    \Figure[ht]
        \centering
        \includegraphics[scale=1.2]{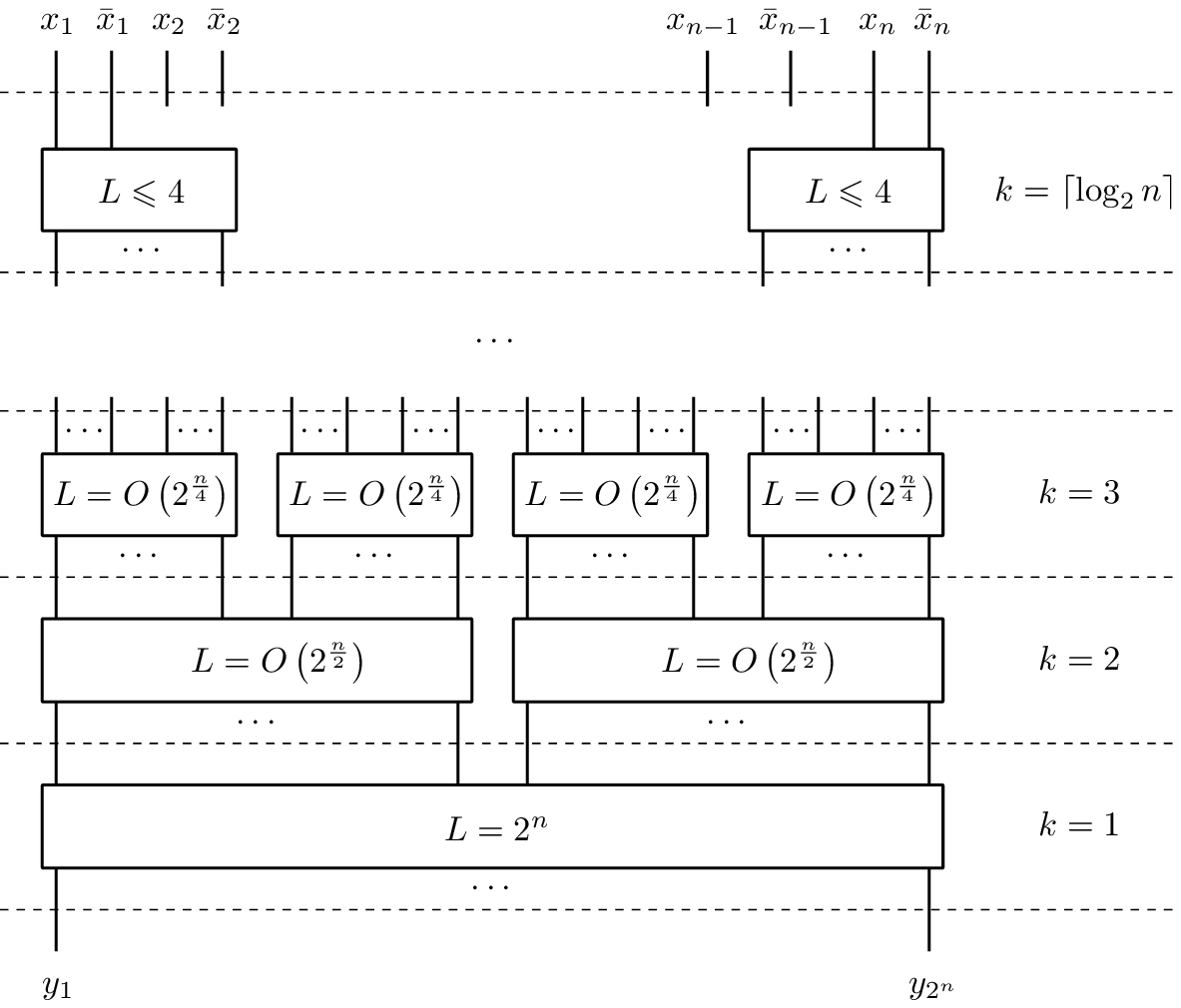}
        \caption
        {
            \small Общая структура обратимой схемы $\frS_{CONJ(n,q)}$ (входы схемы сверху).
        }\label{pic_S_CONJ_structure}
    \end{figure}

    Рассмотрим структуру искомой обратимой схемы $\frS_{CONJ(n,q)}$ на рис.~\ref{pic_S_CONJ_structure}:
    она разбита на $K = \lceil \log_2 n \rceil$ уровней, нумерация ведётся снизу вверх.
    На уровне номер $k$ расположены $2^{k-1}$ обратимых подсхем, все они имеют примерно одинаковое количество значимых входов и выходов
    и реализуют все конъюнкции от некоторого подмножества переменных $x_1, \ldots, x_n$, причём подмножества для разных схем
    одного уровня не пересекаются, их объединение равно всему множеству $\{\,x_1, \ldots, x_n\,\}$, а мощности
    данных подмножеств примерно равны.

    \Table[ht]
        \centering
        \begin{tabular}{|m{1.7 cm}|m{2 cm}|m{3 cm}|m{2.4 cm}|m{5.5 cm}|}
            \hline
            \centering \textbf{Уровень} &
            \centering \textbf{Подсхема} &
            \centering \textbf{Количество входов} &
            \centering \textbf{Количество выходов} &
            \centering \textbf{Подмножество переменных, для которых реализованы все конъюнкции} \tabularnewline
            \hline

            \centering 1 &
            \centering $\frS_{1;1}$ &
            \centering $2^3 + 2^4$ &
            \centering $2^7$ &
            \centering $\{\,x_1, \ldots, x_7\,\}$ \tabularnewline
            \hline

            \multirow{2}{1.7 cm}{\centering 2} &
            \centering $\frS_{2;1}$ &
            \centering $2^1 + 2^2$ &
            \centering $2^3$ &
            \centering $\{\,x_1, x_2, x_3\,\}$ \tabularnewline
            \cline{2-5}

            &
            \centering $\frS_{2;2}$ &
            \centering $2^2 + 2^2$ &
            \centering $2^4$ &
            \centering $\{\,x_4, x_5, x_6, x_7\,\}$ \tabularnewline
            \hline

            \multirow{4}{1.7 cm}{\centering 3} &
            \centering $\frS_{3;1}$ &
            \centering 2 ($x_1$, $\bar x_1$) &
            \centering $2^1$ &
            \centering $\{\,x_1\,\}$ \tabularnewline
            \cline{2-5}

            &
            \centering $\frS_{3;2}$ &
            \centering 4 ($x_2$, $\bar x_2$, $x_3$, $\bar x_3$) &
            \centering $2^2$ &
            \centering $\{\,x_2, x_3\,\}$ \tabularnewline
            \cline{2-5}

            &
            \centering $\frS_{3;3}$ &
            \centering 4 ($x_4$, $\bar x_4$, $x_5$, $\bar x_5$) &
            \centering $2^2$ &
            \centering $\{\,x_4, x_5\,\}$ \tabularnewline
            \cline{2-5}

            &
            \centering $\frS_{3;4}$ &
            \centering 4 ($x_6$, $\bar x_6$, $x_7$, $\bar x_7$) &
            \centering $2^2$ &
            \centering $\{\,x_6, x_7\,\}$ \tabularnewline
            \hline
        \end{tabular}
        \caption
        {
            Описание структуры обратимой схемы $\frS_{CONJ}$ при $n=7$.
        }\label{table_S_CONJ_example}
    \end{table}
        
    Для пояснения структуры схемы $\frS_{CONJ}$ рассмотрим частный её случай для $n = 7$.
    Схема имеет $K = 3$ уровня, каждый её уровень расписан в Таблице~\ref{table_S_CONJ_example}.
    Из данной таблицы видно, что если некоторая подсхема $\frS_{k;i}$ на уровне $k$ имеет $2^m$ значимых выходов,
    то на уровне $(k+1)$ есть ровно две подсхемы $\frS_{k+1;j}$ и $\frS_{k+1;j+1}$, подключённые к ней,
    первая из которых имеет $2^{\lfloor m \mathop / 2 \rfloor}$ значимых выходов, а вторая~---
    $2^{\lceil m \mathop / 2 \rceil}$ значимых выходов.
    Структура подсхемы $\frS_{k;i}$ проста: она реализует конъюнкции каждого значимого выхода подсхемы $\frS_{k+1;j}$ с каждым значимым
    выходом подсхемы $\frS_{k+1;j+1}$ (см. рис.~\ref{pic_sub_S_CONJ_structure}). Следовательно, сложность такой подсхемы будет равна
    $2^{\lfloor m \mathop / 2 \rfloor} \cdot 2^{\lceil m \mathop / 2 \rceil} = 2^m$ (используются только элементы 2-CNOT).
    
    \Figure[ht]
        \centering
        \includegraphics[scale=1.2]{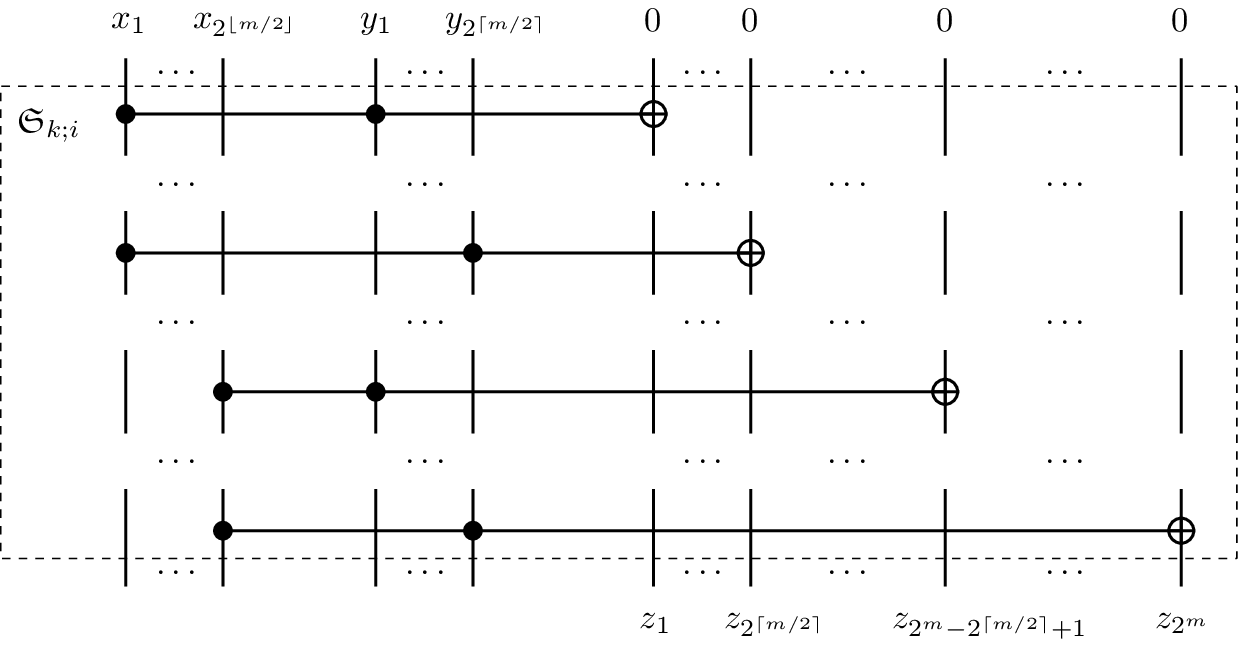}
        \caption
        {
            \small Структура подсхемы $\frS_{k;i}$ обратимой схемы $\frS_{CONJ}$ (входы схемы сверху).
        }\label{pic_sub_S_CONJ_structure}
    \end{figure}

    Вернёмся к общей схеме $\frS_{CONJ}$. Нам дано $q$ дополнительных входов для хранения промежуточных значений.
    Разумнее всего потратить их для хранения значений на выходах подсхем самых высоких уровней, поскольку видно, что чем меньше
    уровень схемы $\frS_{CONJ}$, тем больше требуется дополнительных входов для хранения промежуточных значений.
    
    Рассмотрим случай, когда мы имеем возможность хранить все промежуточные значения.
    Обозначим через $L_k$ количество элементов на $k$-м уровне схемы. К примеру, $L_1 = 2^n$,
    $L_2 = 2^{\lfloor n \mathop / 2 \rfloor} + 2^{\lceil n \mathop / 2 \rceil}$.

    Оценим значение $L_k$. Поскольку
    $$
        \left\lceil \frac{n}{2} \right\rceil = \left\lfloor \frac{n + 1}{2} \right\rfloor \leqslant \frac{n + 1}{2}  \; ,
    $$
    то
    $$
        L_2 \leqslant 2 \cdot \max\left(2^{\left\lfloor \frac{n}{2} \right\rfloor}, 2^{\left\lceil \frac{n}{2} \right\rceil} \right) =
            2 \cdot 2^{\left\lceil \frac{n}{2} \right\rceil} \leqslant 2 \cdot 2 ^ {\frac{n}{2} + \frac{1}{2}}  \; .
    $$
    Отсюда следует, что
    \begin{gather*}
        L_3 \leqslant 4 \cdot 2 ^ {\frac{n}{4} + \frac{1}{4} + \frac{1}{2}}  \; ,  \\
        L_4 \leqslant 8 \cdot 2 ^ {\frac{n}{8} + \frac{1}{8} + \frac{1}{4} + \frac{1}{2}}  \; ,  \\
        L_k \leqslant 2^k \cdot 2 ^ {n \mathop / 2^{k-1}}  \; .
    \end{gather*}
    Обозначим $\delta_k = 2^k \cdot 2 ^ {n \mathop / 2^{k-1}}$. Значение переменной $k$ лежит в диапазоне
    $[1, \ldots, K]$, $K = \lceil \log_2 n \rceil$, $k \in \mathbb N$.
    Сделаем переобозначение переменной: $k = K - s$, тогда $s = s(k) = K - k$. Если $k$ обозначает номер уровня схемы
    при нумерации от выходов ко входам (снизу вверх), то $(s+1)$ будет означать номер уровня схемы при нумерации от входов к выходам
    (сверху вниз).
    Значение переменной $s$ лежит в диапазоне $[0, \ldots, K - 1]$, $s \in \mathbb Z_+$.
    В этом случае
    $$
        \delta_k = \frac{2^K}{2^s} \cdot 2^{(2n \cdot 2^s) \mathop / 2^K}
            \leqslant \frac{2n}{2^s} \cdot 2^{2^{s+1}} = \Delta_s  \; .
    $$
    Следовательно, мы получили цепочку неравенств
    $$
        L_k \leqslant \delta_k \leqslant \Delta_{s(k)} = \frac{2n}{2^s} \cdot 2^{2^{s+1}}  \; .
    $$
    
    Выпишем первые члены ряда $\{\,\Delta_{s(k)}\,\}$: $\{\,8n, 16n, 128n, \ldots\,\}$. Видно, что с ростом $s$ значение
    $\Delta_s$ растёт всё быстрее. Более того, можно утверждать, что для любого $s \geqslant 1$ верно соотношение
    $$
        \sum_{i=0}^{s-1} {\Delta_i} \leqslant \frac{\Delta_s}{2}  \; .
    $$
    Отсюда следует, что
    \begin{gather}
        \sum_{i=0}^s {\Delta_i} \leqslant \frac{3\Delta_s}{2}  \; , \notag \\
        \sum_{i=K}^{K-s} L_i \leqslant \frac{3n}{2^s} \cdot 2^{2^{s+1}}  \; .
            \label{formula_complexity_of_last_layers}
    \end{gather}
    Другими словами, сумма сложностей всех подсхем на последних $(s+1)$ уровнях (при нумерации снизу вверх) не превышает
    $3n \cdot 2^{-s + 2^{s+1}}$.
    
    Вернёмся снова к общей схеме $\frS_{CONJ}$.
    Из рис.~\ref{pic_construct_conjunctions_on_demand} видно,
    что для конструирования по мере необходимости одного значимого выхода схемы $\frS_{CONJ}$ на первых $r$ уровнях
    будет использовано не более $(1 + 2 + 4 + \ldots + 2^{r-1})=(2^r - 1)$ элементов 2-CNOT. Столько же \gate{} потребуется для обнуления
    значений на незначимых выходах. Следовательно, при условии, что количество уровней, для которых подсхемы
    надо конструировать по мере необходимости, не превышает $r$, верно соотношение
    \begin{equation}
        L_{CONJ}(n, q, t) \leqslant q + 2(2^r - 1)\cdot t \leqslant q + t \cdot 2^{r+1}  \; .
            \label{formula_bound_for_L_CONJ_with_r}
    \end{equation}
    Наличие слагаемого $q$ в данном соотношении, очевидно, следует из того факта, что для получения значения на одном выходе
    любой подсхемы $\frS_{k;i}$ требуется ровно один элемент 2-CNOT.
    Если мы можем хранить не более $q$ промежуточных значений на выходах подсхем,
    то для их хранения потребуется не более $q$ элементов 2-CNOT.

    Нам требуется оценить значение $r$. Пусть для данного по условию задачи значения $q$ выполняется неравенство
    \begin{equation}
        \frac{3n}{2^s} \cdot 2^{2^{s+1}} \leqslant q < \frac{3n}{2^{s+1}} \cdot 2^{2^{s+2}}  \;
            \label{formula_bounds_for_q_in_S_CONJ}
    \end{equation}
    для некоторого значения $s \in [0, \ldots, K - 1]$, $s \in \mathbb Z_+$. Тогда мы можем утверждать,
    что данного значения $q$ достаточно для хранения значений на всех выходах подсхем на последних $(s+1)$ уровнях
    при нумерации уровней снизу вверх (см. соотношение~\eqref{formula_complexity_of_last_layers}),
    а количество первых уровней, для которых подсхемы нужно конструировать по мере необходимости,
    не превышает $(k-1)$, поскольку $(s+1) = K - (k-1)$.
    Следовательно, $r \leqslant k-1$.
    Из правого неравенства соотношения~\eqref{formula_bounds_for_q_in_S_CONJ} следует, что
    \begin{gather*}
        \log_2 q < \log_2 3 + \log_2 n - s - 1 + 2^{s+2}  \; , \\
        \frac{2^K}{2^k} = 2^s > \frac{\log_2 q - (\log_2 3 + \log_2 n - s - 1)}{4}  \; , \\
        (\log_2 q - (\log_2 3 + \log_2 n - s - 1))2^k < 4 \cdot 2^K  \; .
    \end{gather*}
    Поскольку $K = \lceil \log_2 n \rceil$, $s \geqslant 0$, $r + 1 \leqslant k$, то при $q > 2n$
    \begin{equation}
        2^{r+1} < \frac{8n}{\log_2 q - \log_2 n - 1}  \; .
        \label{formula_for_2_pow_r_in_lemma_L_CONJ_bound}
    \end{equation}
    Из этого неравенства и неравенства~\eqref{formula_bound_for_L_CONJ_with_r}
    следует оценка утверждения Леммы
    $$
        L_{CONJ}(n, q, t) \leqslant q + \frac{8nt}{\log_2 q - \log_2 n - 1}  \; .
    $$
    Величина $r$ является неотрицательной, целочисленной, её минимальное значение равно нулю,
    поэтому из неравенства~\eqref{formula_for_2_pow_r_in_lemma_L_CONJ_bound} следует ограничение $q < 2n2^{4n}$.
\end{proof}
Хотя оценка Леммы~\ref{lemma_L_CONJ_bound} и верна при $q < 2n2^{4n}$,
однако рассматривать схему $\frS_{CONJ}$ с количеством дополнительных входов $q$ таким, что $2^n = o(q)$,
нецелесообразно, поскольку если все $2^n$ конъюнкций на значимых выходах схемы реализовать одновременно,
а не по мере необходимости, то $L(\frS_{CONJ}) \sim 2^n$ и $Q(\frS_{CONJ}) \sim 2^n$,
как было доказано в работе~\cite[Лемма~1]{my_dm_complexity}.
Другими словами, имея $q \lesssim 2^n$ дополнительных входов, мы всегда сможем построить схему $\frS_{CONJ}$,
а увеличение значения $q$ не приведёт к снижению её сложности (следует из построения схемы).

Отметим, что по аналогии со схемой $\frS_{CONJ}$ можно построить схему $\frS_{XOR}$, которая для заданных входов
$x_1, \ldots, x_n$ на своих значимых выходах реализует по мере необходимости $t$ некоторых, не обязательно различных
значений $x_1 \wedge a_1 \oplus \ldots \oplus x_n \wedge a_n$,
$a_i \in \ZZ_2$. Для этого просто надо каждый элемент 2-CNOT в схеме $\frS_{CONJ}$ заменить на два элемента CNOT
(см. рис.~\ref{pic_basis}). Следовательно,
\begin{gather*}
    L_{XOR}(n, q, t) \leqslant 2q + \frac{16nt}{\log_2 q - \log_2 n - 1}  \; , \\
    Q_{XOR}(n, q, t) \leqslant q + n - 1  \; .
\end{gather*}

Теперь мы можем доказать основную теорему данного раздела.
\begin{theorem}[общая верхняя оценка сложности обратимой схемы с дополнительными входами]\label{theorem_L_n_q_bound_for_arbitrary_q}
    Для любого значения $q$ такого, что $8n < q \lesssim 2^{n-\lceil n \mathop / \phi(n)\rceil + 1}$,
    где $\phi(n)$ и $\psi(n)$~--- любые сколь угодно медленно растущие функции такие, что
    $\phi(n) \leqslant n \mathop / (\log_2 n + \log_2 \psi(n))$,
    верно соотношение
    $$
        L(n,q) \lesssim 2^n + \frac{8n2^n}{\log_2 (q-4n) - \log_2 n - 2} \; .
    $$
\end{theorem}
\begin{proof}
    Опишем алгоритм синтеза~\nextalg{alg_asymp_with_mem_complexity_min_general_case},
    который является модификацией алгоритма синтеза~\algref{alg_asymp_with_mem_complexity_min} из 
    Теоремы~\ref{theorem_complexity_upper_with_memory} и который предназначен для синтезирования обратимых схем
    в условиях ограничения на количество используемых дополнительных входов.
    
    Напомним, что произвольное булево отображение $f\colon \ZZ_2^n \to \ZZ_2^n$ можно представить в виде некоторых $n$ булевых функций
    $f_i\colon \ZZ_2^n \to \ZZ_2$ от $n$ переменных
    $$
        f(\vv x) = \langle f_1(\vv x), f_2(\vv x), \ldots, f_n(\vv x) \rangle \; ,
    $$
    где
    \begin{equation}\label{formula_f_i_with_braces}
        f_i(\vv x) = \bigoplus_{a_{k+1}, \ldots, a_n \in \ZZ_2}
            \left( \bigoplus_{ \substack{t=1 \ldots p\\ g_{j_t} \in G_t\\ 1 \leqslant j_t \leqslant |G_t|}}
            x_{k+1}^{a_{k+1}} \wedge \ldots \wedge x_n^{a_n} \wedge g_{j_t}(\langle x_1, \ldots, x_k\rangle) \right) \; .        
    \end{equation}
    (Описание внутренних переменных в формуле~\eqref{formula_f_i_with_braces} можно найти в доказательстве
    Теоремы~\ref{theorem_complexity_upper_with_memory} на с.~\pageref{formula_function_decomposition_by_last_variables}.)

    Общая структура обратимой схемы $\frS_f$, которая реализует отображение $f$
    и которая синтезируется алгоритмом~\algref{alg_asymp_with_mem_complexity_min_general_case},
    показана на рис.~\ref{pic_scheme_structure_for_case_of_limited_memory}.
    
    \Figure[!ht]
        \includegraphics[scale=0.87]{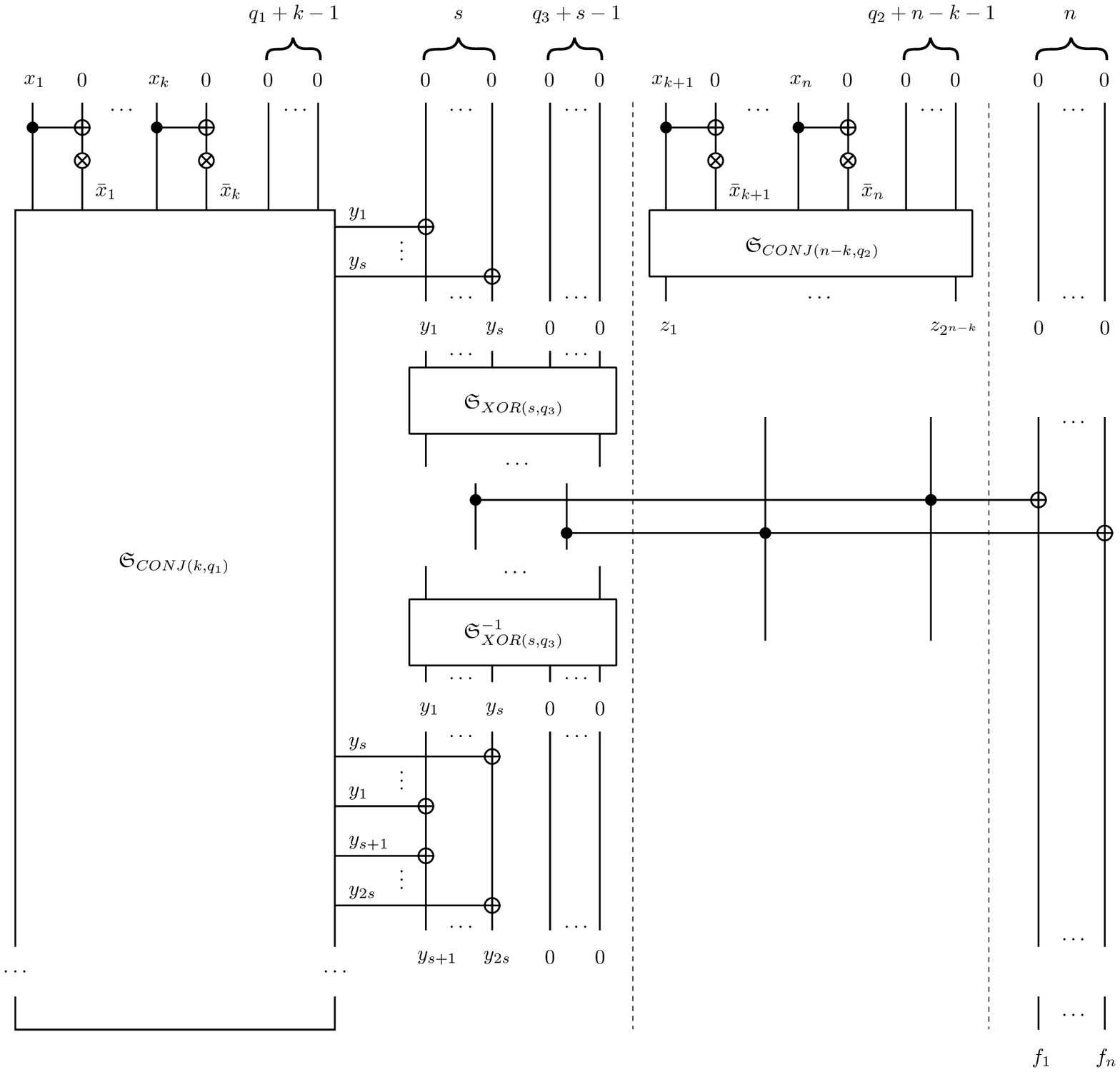}
        \caption
        {
            \small Структура обратимой схемы $\frS_f$, реализующей отображение $f\colon \ZZ_2^n \to \ZZ_2^n$
            в условиях ограничения на количество используемых дополнительных входов (входы схемы сверху).
        }\label{pic_scheme_structure_for_case_of_limited_memory}
    \end{figure}    
    
    Сперва реализуем отрицания для всех входных значений $x_1, \ldots, x_n$ со сложностью $2n$ (по элементу NOT и CNOT на каждый вход),
    задействовав $n$ дополнительных входов.
    
    Разобьём множество значимых входов схемы $x_1, \ldots, x_n$ на две группы: $\{\,x_1, \ldots, x_k\,\}$
    и $\{\,x_{k+1}, \ldots, x_n\,\}$.
    Первую группу входов вместе с их отрицаниями подадим на подсхему $\frS_1 = \frS_{CONJ(k,q_1)}$ для реализации некоторых $t_1$
    конъюнкций $x_1^{a_1} \wedge \ldots \wedge x_k^{a_k}$, $a_i \in \ZZ_2$, отведя данной подсхеме $q_1$ дополнительных входов
    для хранения промежуточных значений.
    Вторую группу входов вместе с их отрицаниями подадим на подсхему $\frS_2 = \frS_{CONJ(n-k,q_2)}$ для реализации некоторых $t_2$
    конъюнкций $x_{k+1}^{a_{k+1}} \wedge \ldots \wedge x_n^{a_n}$, $a_i \in \ZZ_2$, отведя данной подсхеме $q_2$ дополнительных входов
    для хранения промежуточных значений.
    
    Будем реализовывать все $2^k$ различных конъюнкций на значимых выходах подсхемы $\frS_1$ последовательно.
    Полученные значения будем хранить, используя дополнительные входы.
    Как только будут получены очередные $s$ конъюнкций, соответствующие им $s$ значимых выходов
    подаём на значимые входы подсхемы $\frS_{3;i} = \frS_{XOR(s,q_3)}$ для получения значений некоторых $t_3$ функций от переменных
    $x_1, \ldots, x_k$, отведя данной подсхеме $q_3$ дополнительных входов для хранения промежуточных значений.
    Всего будет не более $p = \lceil 2^k \mathop / s \rceil$ различных подсхем $\frS_{3;i}$.
    Как только работа с очередной подсхемой $\frS_{3;i}$ будет закончена, значение на $q_3$ незначимых выходах обнуляем,
    применяя те же самые \gate{}, что и для получения подсхемы, но в обратном порядке.
    Затем обнуляем значения на $s$ значимых выходах,
    служивших значимыми входами подсхеме $\frS_{3;i}$, реализуя ещё раз полученные ранее $s$ конъюнкций при помощи подсхемы $\frS_1$
    (см. рис.~\ref{pic_scheme_structure_for_case_of_limited_memory}).
    Тем самым мы сможем не увеличивать количество используемых дополнительных входов, а использовать одни и те же дополнительные входы,
    увеличивая однако при этом сложность соответствующих подсхем в два раза.
    
    Из формулы~\eqref{formula_f_i_with_braces} следует,
    что имея значения некоторого значимого выхода подсхемы $\frS_2$ и некоторого значимого выхода подсхемы $\frS_{3;i}$,
    мы можем реализовать одно слагаемое во внутренней скобке, используя ровно один элемент 2-CNOT, контролируемый выход которого
    будет одним из $n$ значимых выходов нашей конструируемой схемы $\frS_f$
    (см. рис.~\ref{pic_scheme_structure_for_case_of_limited_memory}).
    Рассматриваемое нами отображение $f$ имеет $n$ выходов, количество групп конъюнкций от первых $k$ переменных $x_1, \ldots, x_k$
    равно $p$, количество различных конъюнкций от последних $(n-k)$ переменных $x_{k+1}, \ldots, x_n$ равно $2^{n-k}$.
    Следовательно, схемная сложность реализации функции $f_i$ по формуле~\eqref{formula_f_i_with_braces}
    равна $p2^{n-k}$, а отображения $f$ в целом равна $L_4 = pn2^{n-k}$, при этом потребуется ровно $n$
    дополнительных входов для хранения выходных значений отображения $f$.
    
    Таким образом, мы можем вывести соотношение для $L(f,q)$ следующего вида:
    \begin{equation}
        L(f,q) = 2n + L_{CONJ}(k, q_1, t_1) + L_{CONJ}(n-k, q_2, t_2) + 2p \cdot L_{XOR}(s, q_3, t_3) + pn2^{n-k}  \; ,
        \label{formula_L_f_q_common}
    \end{equation}
    и для $Q(\frS_f)$ следующего вида:
    \begin{equation}
        Q(\frS_f) = q = n + Q_{CONJ}(k, q_1, t_1) + Q_{CONJ}(n-k, q_2, t_2) + Q_{XOR}(s, q_3, t_3) + n  \; .
        \label{formula_Q_f_q_common}
    \end{equation}
    
    Отметим, что каждая из $2^k$ различных конъюнкций на значимых выходах подсхемы $\frS_1$ будет получена ровно два раза,
    следовательно, $t_1 = 2^{k+1}$.
    
    Поскольку каждый значимый выход подсхемы $\frS_2$ используется в качестве входа для $pn$ элементов 2-CNOT,
    а значимый выход подсхемы $\frS_{3;i}$ может использоваться в качестве входа для $2^{n-k}$ элементов 2-CNOT,
    возникает два различных способа конструирования нашей искомой схемы $\frS_f$.
    \begin{enumerate}
        \item \label{item_first_case_for_theorem_L_n_q_bound_for_arbitrary_q}
            В первом случае мы минимизируем значение $t_2$: для каждой группы конъюнкций от первых $k$ переменных $x_1, \ldots, x_k$
            мы один раз конструируем очередной значимый выход подсхемы $\frS_2$, а затем конструируем для него
            $n$ значимых выходов подсхемы $\frS_{3;i}$. Тогда можно утверждать, что $t_2 = p2^{n-k}$, $t_3 \leqslant n2^{n-k}$.
        \item
            Во втором случае мы минимизируем значение $t_3$: для каждой группы конъюнкций от первых $k$ переменных $x_1, \ldots, x_k$
            мы один раз конструируем очередной значимый выход подсхемы $\frS_{3;i}$, а затем конструируем для него нужные значимые выходы
            подсхемы $\frS_2$. Таких выходов может быть один, а может быть и $2^{n-k}$. Однако мы точно можем утверждать,
            что $t_2 \leqslant pn2^{n-k}$, $t_3 \leqslant 2^s$.
    \end{enumerate}
    
    Оценим в общем случае значение $L(f,q)$:
    \begin{multline*}
        L(f,q) \leqslant 2n + pn2^{n-k} + q_1 + q_2 + 4p q_3 +
             \frac{8k2^{k+1}}{\log_2 q_1 - \log_2 k - 1} + \\
                + \frac{8(n-k)t_2}{\log_2 q_2 - \log_2 (n-k) - 1}
                + \frac{32pst_3}{\log_2 q_3 - \log_2 s - 1} \; .
    \end{multline*}
    
    Будем искать такие значения $k$ и $s$, что $p = \lceil 2^k \mathop / s \rceil \sim 2^k \mathop / s$. Тогда
    \begin{multline*}
        L(f,q) \lesssim 2n + \frac{n2^n}{s} + q_1 + q_2 + \frac{4q_3 2^k}{s} +
             \frac{8k2^{k+1}}{\log_2 q_1 - \log_2 k - 1} + \\
                + \frac{8(n-k)t_2}{\log_2 q_2 - \log_2 (n-k) - 1}
                + \frac{32t_3 2^k}{\log_2 q_3 - \log_2 s - 1} \; .
    \end{multline*}
   
    \begin{enumerate}
        \item
            Пусть $t_2 = p2^{n-k} \sim 2^n \mathop / s$, $t_3 \leqslant n2^{n-k}$. В этом случае 
            \begin{multline}
                L(f,q) \lesssim 2n + \frac{n2^n}{s} + q_1 + q_2 + \frac{4q_3 2^k}{s} +
                     \frac{8k2^{k+1}}{\log_2 q_1 - \log_2 k - 1} + \\
                        + \frac{8 \cdot 2^n}{\log_2 q_2 - \log_2 (n-k) - 1}
                        + \frac{32n2^n}{\log_2 q_3 - \log_2 s - 1} \; .
                    \label{formula_L_f_q_first_way}
            \end{multline}
            Положим $s = n - k$, $k = \lceil n \mathop / \phi(n) \rceil$,
            где $\phi(n)$ и $\psi(n)$~--- любые сколь угодно медленно растущие функции такие, что
            $\phi(n) \leqslant n \mathop / (\log_2 n + \log_2 \psi(n))$.
            В этом случае будет верно неравенство $2^k \mathop / s \geqslant \psi(n)$.

            Поскольку $q_3 \lesssim 2^s$ и $q_2 \lesssim 2^{n-k} = o(2^n)$, то верно соотношение
            \begin{equation}
                2n + \frac{n2^n}{s} + q_2 + \frac{4q_3 2^k}{s} \lesssim 2n + \frac{n2^n}{n - o(n)} + q_2 + \frac{2^{n+2}}{n-o(n)}
                    \lesssim 2^n  \; .
                \label{formula_least_member_bound_first_case_in_theorem_L_n_q_bound_for_arbitrary_q}
            \end{equation}

            Положим $q_1 = 0$, $q_2 = q_3$.    
            Согласно формуле~\eqref{formula_L_CONJ_0}, $L_{CONJ}(n, 0, t) \leqslant 2(n-1)t$,
            следовательно, мы можем заменить в соотношении~\eqref{formula_L_f_q_first_way}
            сложность подсхемы $\frS_1$ на $k2^{k+2}$:
            $$
                L(f,q) \lesssim 2^n + k2^{k+2} + \frac{8 \cdot 2^n}{\log_2 q_2 - \log_2 (n-k) - 1}
                    + \frac{32n2^n}{\log_2 q_3 - \log_2 s - 1} \; .
            $$
            Очевидно, что $k2^{k+2} = 4\lceil n \mathop / \phi(n) \rceil \cdot 2^{\lceil n \mathop / \phi(n) \rceil} = o(2^n)$
            и $8 \cdot 2^n \mathop / (\log_2 q_3 - \log_2 s - 1) = o \left(32n2^n \mathop / (\log_2 q_3 - \log_2 s - 1) \right)$,
            поэтому верно соотношение 
            $$
                L(f,q) \lesssim 2^n + \frac{32n2^n}{\log_2 q_3 - \log_2 s - 1} \; .
            $$

            Согласно Лемме~\ref{lemma_L_CONJ_bound}, $Q_{CONJ}(n, q, t) \leqslant q + n - 1$,
            поэтому соотношение~\eqref{formula_Q_f_q_common} можо переписать в виде
            \begin{equation}
                q \leqslant n + k - 1 + q_2 + n - k -1 + q_3 + s - 1 + n = 4n - k + 2q_3 - 3 < 4n + 2q_3  \; .
                \label{formula_q_3_bound_in_theorem_L_n_q_bound_for_arbitrary_q}
            \end{equation}
            Следовательно, $\log_2 q_3 > \log_2 (q - 4n) - 1$. Отсюда получаем соотношение
            $$
                L(f,q) \lesssim 2^n + \frac{32n2^n}{\log_2 (q - 4n) - \log_2 n - 2} \; ,
            $$
            которое верно при $\log_2 (q - 4n) > \log_2 n + 2$. Таким образом, $q > 8n$.
            С другой стороны $q < 4n + 2q_3 \lesssim 2^{n-\lceil n \mathop / \phi(n)\rceil + 1}$.
            \smallskip

        \item
            Пусть $t_2 \leqslant pn2^{n-k} \sim n2^n \mathop / s$, $t_3 \leqslant 2^s$. В этом случае 
            \begin{multline}
                L(f,q) \lesssim 2n + \frac{n2^n}{s} + q_1 + q_2 + \frac{4q_3 2^k}{s} +
                     \frac{8k2^{k+1}}{\log_2 q_1 - \log_2 k - 1} + \\
                        + \frac{8n2^n}{\log_2 q_2 - \log_2 (n-k) - 1}
                        + \frac{32 \cdot 2^n}{\log_2 q_3 - \log_2 s - 1} \; .
                    \label{formula_L_f_q_second_way}
            \end{multline}
            
            Как и в первом способе, положим $s = n - k$, $k = \lceil n \mathop / \phi(n) \rceil$,
            где $\phi(n)$ и $\psi(n)$~--- любые сколь угодно медленно растущие функции такие, что
            $\phi(n) \leqslant n \mathop / (\log_2 n + \log_2 \psi(n))$;
            $q_1 = 0$, $q_2 = q_3$.
            Тогда из рассуждений, приведённых при описании первого способа, следует соотношение
            $$
                L(f,q) \lesssim 2^n + \frac{8n2^n}{\log_2 q_2 - \log_2 (n-k) - 1}
                        + \frac{32 \cdot 2^n}{\log_2 q_3 - \log_2 s - 1} \; .
            $$

            Очевидно, что
            $32 \cdot 2^n \mathop / (\log_2 q_2 - \log_2 (n-k) - 1) = o \left(8n2^n \mathop / (\log_2 q_2 - \log_2 (n-k) - 1) \right)$,
            поэтому верно соотношение 
            $$
                L(f,q) \lesssim 2^n + \frac{8n2^n}{\log_2 q_2 - \log_2 (n-k) - 1}  \; .
            $$
            
            Поскольку $\log_2 q_3 > \log_2 (q - 4n) - 1$, то и $\log_2 q_2 > \log_2(q - 4n) - 1$.
            Отсюда получаем соотношение
            $$
                L(f,q) \lesssim 2^n + \frac{8n2^n}{\log_2 (q - 4n) - \log_2 n - 2} \; ,
            $$
            которое верно при $\log_2 (q - 4n) > \log_2 n + 2 \Rightarrow q > 8n$.

            Видно, что второй способ синтеза асимптотически лучше первого.
            \smallskip
    \end{enumerate}

    Поскольку мы описали алгоритм синтеза обратимой схемы для произвольного отображения $f$, то
    $$
        L(n,q) \leqslant L(f,q) \lesssim 2^n + \frac{8n2^n}{\log_2 (q-4n) - \log_2 n - 2} \;
    $$
    для любого значения $q$ такого, что $8n < q \lesssim 2^{n-\lceil n \mathop / \phi(n)\rceil + 1}$.
\end{proof}

\begin{corollary}
    Для любого значения $q$ такого, что $8n < q \lesssim 2^{n-\lceil n \mathop / \phi(n)\rceil + 1}$,
    где $\phi(n)$ и $\psi(n)$~--- любые сколь угодно медленно растущие функции такие, что
    $\phi(n) \leqslant n \mathop / (\log_2 n + \log_2 \psi(n))$,
    верны соотношения
    \begin{align}
        W(n,q) &\lesssim \WT \cdot \left(2^n + \frac{8 \cdot 2^n}{\log_2 (q - 4n) - \log_2 n - 2} \right)
            + \frac{32 \WC n2^n}{\log_2 (q - 4n) - \log_2 n - 2} \; ,
            \label{theorem_W_n_q_bound_for_arbitrary_q_first_case} \\
        W(n,q) &\lesssim \WT \cdot \left(2^n + \frac{8n2^n}{\log_2 (q-4n) - \log_2 n - 2}\right)
            + \frac{32 \WC 2^n}{\log_2 (q-4n) - \log_2 n - 2} \; .
            \label{theorem_W_n_q_bound_for_arbitrary_q_second_case}
    \end{align}
\end{corollary}
\begin{proof}
    В обратимой схеме $\frS_f$ из доказательства Теоремы~\ref{theorem_L_n_q_bound_for_arbitrary_q} элементы NOT и CNOT
    используются только для получения отрицаний для всех входных значений $x_1, \ldots, x_n$ и в подсхемах $\frS_{3;i}$.
    Следовательно, оценка~\eqref{theorem_W_n_q_bound_for_arbitrary_q_first_case} верна для случая
    $t_2 = p2^{n-k}$, $t_3 \leqslant n2^{n-k}$, описанного на с.~\pageref{item_first_case_for_theorem_L_n_q_bound_for_arbitrary_q},
    а оценка~\eqref{theorem_W_n_q_bound_for_arbitrary_q_second_case} верна для случая
    $t_2 \leqslant pn2^{n-k}$, $t_3 \leqslant 2^s$, описанного там же.
\end{proof}

Теперь можно оценить порядок роста функции $L(n,q)$.
\begin{theorem}\label{theorem_L_asymp_general}
    Для любого значения $q$ такого, что $n^2 \lesssim q \lesssim 2^{n-\lceil n \mathop / \phi(n)\rceil + 1}$,
    где $\phi(n)$ и $\psi(n)$~--- любые сколь угодно медленно растущие функции такие, что
    $\phi(n) \leqslant n \mathop / (\log_2 n + \log_2 \psi(n))$,
    верно соотношение
    $$
        L(n,q) \asymp \frac{n2^n}{\log_2 q} \; .
    $$
\end{theorem}
\begin{proof}
Следует из Теоремы~\ref{theorem_complexity_lower_bound} и Теоремы~\ref{theorem_L_n_q_bound_for_arbitrary_q}.
\end{proof}

Если рассматривать только знакопеременную группу $A(\ZZ_2^n)$ и обратимые схемы, реализующие отображения $\ZZ_2^n \to \ZZ_2^n$
из этой группы, то можно установить следующий порядок роста функции Шеннона $L_A(n,q)$ сложности таких обратимой схемы.
\begin{theorem}\label{theorem_L_A_asymp}
    Для любого значения $q$ такого, что $0 \leqslant q \lesssim 2^{n-\lceil n \mathop / \phi(n)\rceil + 1}$,
    где $\phi(n)$ и $\psi(n)$~--- любые сколь угодно медленно растущие функции такие, что
    $\phi(n) \leqslant n \mathop / (\log_2 n + \log_2 \psi(n))$,
    верно соотношение
    $$
        L_A(n,q) \asymp \frac{n2^n}{\log_2 (n+q)} \; .
    $$
\end{theorem}
\begin{proof}
Следует из Теоремы~\ref{theorem_complexity_no_memory_common}
и Теоремы~\ref{theorem_L_asymp_general}, поскольку без дополнительных входов рассматриваемые нами обратимые схемы реализуют
только чётные подстановки.
\end{proof}

Из доказательства Теоремы~\ref{theorem_L_n_q_bound_for_arbitrary_q} также можно получить верхнюю оценку для функции $D(n,q)$
в случае $q > 8n$, $q \lesssim 2^{n-o(n)}$, но для этого необходимо сперва доказать вспомогательную лемму.

\begin{lemma}\label{lemma_D_CONJ_bound}
    Для любого значения $q$ такого, что $2n < q < 2n2^{4n}$, верны соотношения
    \begin{align*}
        D_{CONJ}(n, q, t) &\leqslant q + 2t(2+\log_2 n - \log_2 (\log_2 q - \log_2 n - 1)) \; . \\
        D_{CONJ}(n, 0, t) &\leqslant 2t \cdot \lceil \log_2 n \rceil  \; .
    \end{align*}
\end{lemma}
\begin{proof}
    Рассмотрим схему $\frS_{CONJ(n,q)}$ из Леммы~\ref{lemma_L_CONJ_bound}.
    Согласно формуле~\eqref{formula_bound_for_L_CONJ_with_r}, верно неравенство
    $L_{CONJ}(n, q, t) \leqslant q + 2t \cdot 2^r$. Промежуточные значения, хранимые на $q$ дополнительных входах,
    можно получить с глубиной не более $q$.
    Также очевидно, что сконструировать по мере необходимости один значимый выход схемы $\frS_{CONJ}$ на первых $r$ уровнях
    можно с глубиной $r$, см. рис.~\ref{pic_construct_conjunctions_on_demand}.
    Отсюда следует, что
    $$
        D_{CONJ}(n, q, t) \leqslant q + 2tr  \; .
    $$

    Согласно формуле~\eqref{formula_for_2_pow_r_in_lemma_L_CONJ_bound}, при $q > 2n$ верно неравенство
    $$
        2^r < \frac{4n}{\log_2 q - \log_2 n - 1}  \; ,
    $$
    откуда следует, что
    \begin{gather*}
        r < 2 + \log_2 n - \log_2 (\log_2 q - \log_n - 1)  \; , \\
        D_{CONJ}(n, q, t) \leqslant q + 2t(2+\log_2 n - \log_2 (\log_2 q - \log_2 n - 1)) \; .
    \end{gather*}
    
    Соотношение $D_{CONJ}(n, 0, t) \leqslant 2t \cdot \lceil \log_2 n \rceil$ следует из соотношения~\eqref{formula_L_CONJ_0}
    и того факта, что сконструировать одну конъюнкцию $x_1^{a_1} \wedge \ldots \wedge x_n^{a_n}$ можно с логарифмической глубиной
    $\lceil \log_2 n \rceil$.
\end{proof}

Аналогично, для обратимой схемы $\frS_{XOR(n,q)}$ верно неравенство
$$
    D_{XOR}(n, q, t) = 2D_{CONJ}(n, q, t)
        \leqslant 2q + 4t(2+\log_2 n - \log_2 (\log_2 q - \log_2 n - 1))
$$
для любого значения $q$ такого, что $2n < q < 2n2^{4n}$.

Итак, докажем последнюю теорему данного раздела.
\begin{theorem}[общая верхняя оценка глубины обратимой схемы с дополнительными входами]\label{theorem_D_n_q_bound_for_arbitrary_q}
    Для любого значения $q$ такого, что $8n < q \lesssim 2^{n-\lceil n \mathop / \phi(n)\rceil + 1}$,
    где $\phi(n)$ и $\psi(n)$~--- любые сколь угодно медленно растущие функции такие, что
    $\phi(n) \leqslant n \mathop / (\log_2 n + \log_2 \psi(n))$,
    верно соотношение
    $$
        D(n,q) \lesssim 2^{n+1}(2,5 + \log_2 n - \log_2 (\log_2 (q - 4n) - \log_2 n - 2))  \; .
    $$
\end{theorem}
\begin{proof}
    Рассмотрим обратимую схему $\frS_f$ из доказательства Теоремы~\ref{theorem_L_n_q_bound_for_arbitrary_q},
    синтезированную алгоритмом~\algref{alg_asymp_with_mem_complexity_min_general_case}.
    Из соотношения~\eqref{formula_L_f_q_common} можно вывести аналогичное соотношение для глубины $D(f,q)$ вида
    $$
        D(f,q) = 2 + D_{CONJ}(k, q_1, t_1) + D_{CONJ}(n-k, q_2, t_2) + 2p \cdot D_{XOR}(s, q_3, t_3) + pn2^{n-k}  \; .
    $$

    Положим $s = n - k$, $k = \lceil n \mathop / \phi(n) \rceil$,
    где $\phi(n)$ и $\psi(n)$~--- любые сколь угодно медленно растущие функции такие, что
    $\phi(n) \leqslant n \mathop / (\log_2 n + \log_2 \psi(n))$
    В этом случае будут верны соотношения $p \sim 2^k \mathop / s \geqslant \psi(n)$
    и $pn2^{n-k} \sim n2^n \mathop / s \sim 2^n$.
    
    Положим $q_1 = 0$. Поскольку $t_1 = 2^{k+1}$, то
    $$
        D_{CONJ}(k, q_1, t_1) \leqslant 2t_1 \cdot \lceil \log_2 k \rceil
            \leqslant \lceil \log_2 \lceil n \mathop / \phi(n)\rceil \rceil \cdot 2^{\lceil n \mathop / \phi(n) \rceil + 2} = o(2^n)  \; .
    $$
    Таким образом, верно соотношение
    $$
        D(f,q) \lesssim 2^n + D_{CONJ}(n-k, q_2, t_2) + \frac{2^{k+1}}{s} \cdot D_{XOR}(s, q_3, t_3)  \; .
    $$
   
    Рассмотрим те же два случая для $t_2$ и $t_3$, что и на с.~\pageref{item_first_case_for_theorem_L_n_q_bound_for_arbitrary_q}.
    \begin{enumerate}
        \item
            Пусть $t_2 = p2^{n-k} \sim 2^n \mathop / s$, $t_3 \leqslant n2^{n-k}$.
            В этом случае
            \begin{gather*}
                D_{CONJ}(n-k, q_2, t_2) \leqslant q_2 + \frac{2^{n+1}}{s}(2+\log_2 s - \log_2 (\log_2 q_2 - \log_2 s - 1))  \; , \\
                D_{XOR}(s, q_3, t_3) \leqslant 2 q_3 + n2^{n-k+2}(2+\log_2 s - \log_2 (\log_2 q_3 - \log_2 s - 1))  \; .
            \end{gather*}
            
            Положим $q_2 = q_3$. Обозначим $d = 2+\log_2 s - \log_2 (\log_2 q_2 - \log_2 s - 1)$, тогда
            $$
                D(f,q) \lesssim 2^n + q_2 + \frac{d2^{n+1}}{s} +
                    \frac{q_3 2^{k+2}}{s} + \frac{dn2^{n+3}}{s}  \; .
            $$
            
            Согласно формуле~\eqref{formula_least_member_bound_first_case_in_theorem_L_n_q_bound_for_arbitrary_q},
            верно соотношение
            $$
                2^n + q_2 + \frac{4q_3 2^k}{s} \lesssim 2^n  \; .
            $$
            Отсюда получаем, что
            $$
                D(f,q) \lesssim 2^n + \frac{dn2^{n+3}}{s} \lesssim
                    2^n + 2^{n+3}(2+\log_2 n - \log_2 (\log_2 q_3 - \log_2 n - 1))  \; .
            $$

            Из соотношения~\eqref{formula_q_3_bound_in_theorem_L_n_q_bound_for_arbitrary_q} следует, что
            $\log_2 q_3 > \log_2 (q - 4n) - 1$. Таким образом, получаем итоговую оценку сверху вида
            $$
                D(f,q) \lesssim 2^n(17 + 8(\log_2 n - \log_2 (\log_2 (q - 4n) - \log_2 n - 2)))  \; ,
            $$
            которая верна при $\log_2 (q - 4n) > \log_2 n + 2 \Rightarrow q > 8n$.
            \smallskip

        \item
            Пусть $t_2 \leqslant pn2^{n-k} \sim n2^n \mathop / s$, $t_3 \leqslant 2^s$. В этом случае
            \begin{gather*}
                D_{CONJ}(n-k, q_2, t_2) \leqslant q_2 + \frac{n2^{n+1}}{s}(2+\log_2 s - \log_2 (\log_2 q_2 - \log_2 s - 1))  \; , \\
                D_{XOR}(s, q_3, t_3) \leqslant 2 q_3 + 2^{s+2}(2+\log_2 s - \log_2 (\log_2 q_3 - \log_2 s - 1))  \; .
            \end{gather*}
            
            Положим $q_2 = q_3$. Обозначим $d = 2+\log_2 s - \log_2 (\log_2 q_2 - \log_2 s - 1)$, тогда
            $$
                D(f,q) \lesssim 2^n + q_2 + \frac{dn2^{n+1}}{s} + \frac{q_3 2^{k+2}}{s} + \frac{d2^{n+3}}{s}  \; .
            $$
            
            Согласно формуле~\eqref{formula_least_member_bound_first_case_in_theorem_L_n_q_bound_for_arbitrary_q},
            верно соотношение
            $$
                2^n + q_2 + \frac{4q_3 2^k}{s} \lesssim 2^n  \; .
            $$
            Отсюда получаем, что
            $$
                D(f,q) \lesssim 2^n + \frac{dn2^{n+1}}{s} \lesssim
                    2^n + 2^{n+1}(2+\log_2 n - \log_2 (\log_2 q_2 - \log_2 n - 1))  \; .
            $$

            Из соотношения~\eqref{formula_q_3_bound_in_theorem_L_n_q_bound_for_arbitrary_q} следует, что
            $\log_2 q_2 > \log_2 (q - 4n) - 1$. Таким образом, получаем итоговую оценку сверху вида
            $$
                D(f,q) \lesssim 2^{n+1}(2,5 + \log_2 n - \log_2 (\log_2 (q - 4n) - \log_2 n - 2))  \; ,
            $$
            которая верна при $\log_2 (q - 4n) > \log_2 n + 2 \Rightarrow q > 8n$.

            Видно, что второй способ синтеза асимптотически лучше первого.
            \smallskip
    \end{enumerate}
    
    Поскольку мы описали алгоритм синтеза обратимой схемы для произвольного отображения $f$, то
    $$
        D(n,q) \leqslant D(f,q) \lesssim 2^{n+1}(2,5 + \log_2 n - \log_2 (\log_2 (q - 4n) - \log_2 n - 2)) \;
    $$
    при $q > 8n$. Ограничение $q \lesssim 2^{n-\lceil n \mathop / \phi(n)\rceil + 1}$
    следует из Теоремы~\ref{theorem_L_n_q_bound_for_arbitrary_q}.
\end{proof}

При увеличении количества дополнительных входов с $q \sim 2^{n-o(n)}$ до $q \sim 2^n$
верхняя асимптотическая оценка функции $D(n,q)$ снижается с экспоненциальной до линейной,
согласно Теоремам~\ref{theorem_complexity_upper_with_memory} ($D(n,q) \leqslant L(n,q)$) и~\ref{theorem_depth_upper_with_memory_3n}.
Однако выведение зависимости верхней оценки функции $D(n,q)$ от $q$ для всех значений $q$ таких,
что $2^{n-o(n)} \lesssim q \lesssim 2^n$, выходит за рамки данной работы.

Теперь мы можем сформулировать основное утверждение данной главы.
\begin{predicate}\label{predicate_reduction_of_complexity_with_help_of_additional_memory}
    Использование дополнительной памяти в обратимых схемах, состоящих из элементов NOT, CNOT и 2-CNOT,
    почти всегда позволяет существенно снизить сложность, глубину и квантовый вес таких схем.
\end{predicate}
\begin{proof}
    Следует из Теорем~\ref{theorem_complexity_lower_bound}--\ref{theorem_D_n_q_bound_for_arbitrary_q}.
\end{proof}

Воспользуемся данным результатом, чтобы показать на примере реализации обратимой схемой
алгоритма дискретного логарифмирования по основанию примитивного элемента в конечном поле характеристики 2,
как использование дополнительных входов схемы позволяет снизить её сложность.

\sectionenumerated{Примеры применения обратимых схем}

\forceindent
В данной главе будет показано применение обратимых схем при решении задачи схемной реализации некоторых
вычислительно асимметричных преобразований. Будет подробно рассмотрен алгоритм дискретного логарифмирования
по основанию примитивного элемента в конечном поле характеристики 2 на примере фактор-кольца $\fp$,
где $f(x)$~--- неприводимый многочлен, и его реализация обратимой схемой.
Будет показано, как использование дополнительных входов схемы позволяет снизить её сложность.
В заключении данной главы будет рассмотрен вопрос схемной сложности реализации алгоритма, обратного к заданному,
и будет сделана попытка объяснить разницу в схемной сложности для прямого и обратного алгоритмов через
необратимость и потерю части информации во время работы прямого алгоритма.

\subsection{Дискретное логарифмирование в конечном поле характеристики 2}
\label{subsection_discrete_logarithm_in_polynomial_field}

\forceindent
Задачу дискретного логарифмирования в поле $\FF$ можно переформулировать следующим образом:
для элементов $\vv x, \vv y \in \FF^*$ найти такое $k$, что $\vv y = \vv x^k$, либо показать, что такого $k$ не существует.
Значение $k$ является значением дискретного логарифма элемента $\vv y$ по основанию $\vv x$: $k = \log_{\vv x}{\vv y}$.

Группа $\FF^*$ является циклической $\Rightarrow$ существует примитивный элемент $\balpha$,
такой что $\FF^* = \langle \balpha \rangle$. Отсюда следует, что для любого $\vv x \in \FF^*$ существует такая
степень $k < |\FF^*|$, что $\vv x = \balpha^k$.
Тогда задачу нахождения $\log_{\vv x}{\vv y}$ можно свести к решению уравнения
$$
    mk_{\vv x} = k_{\vv y}\pmod{|\FF^*|} \; ,
$$
где $m = \log_{\vv x}{\vv y}$, $k_{\vv x} = \log_{\balpha}{\vv x}$, $k_{\vv y} = \log_{\balpha}{\vv y}$.
Тогда $m = k_{\vv y}k_{\vv x}^{-1}\pmod{|\FF^*|}$.

На сегодняшний день неизвестно, существует ли полиномиальный алгоритм для задачи дискретного логарифмирования
в общем случае~\cite{discrete_log_brief_survey}.
Тем не менее было предложено несколько алгоритмов дискретного логарифмирования с субэкспоненциальной временн\'{о}й сложностью
для группы $G$. Обозначим через $L(N,c,a) = \exp(c (\log N)^a (\log \log N)^{1-a})$,
$L(a) = L(|G|,c,a)$, где $c$~--- некоторая константа.
Алгоритм Шенкса (<<baby-step giant-step>>), описанный в работе~\cite{shanks},
имеет временн\'{у}ю сложностью $L(1 / 2)$~\cite{discrete_log_brief_survey}.
Алгоритм Адлемана, представленный в работе~\cite{adleman_subexp_discrete_log},
также имеет временн\'{у}ю сложность $L(1 / 2)$~\cite{coppersmith_fast_discrete_log}.
В нескольких работах~\cite{gordon_discrete_log_in_gfp, adleman_function_field_sieve,
joux_function_field_sieve, joux_number_field_sieve} были представлены различные алгоритмы дискретного логарифмирования в 
конечных полях с временн\'{о}й сложностью $L(1 / 3)$.
В работе~\cite{coppersmith_fast_discrete_log} коллективом авторов был предложен алгоритм дискретного логарифмирования
в поле характеристики 2 с временн\'{о}й сложностью $L(1 / 3)$.
В работе~\cite{joux_index_calculus_discrete_log} был описан алгоритм дискретного логарифмирования в полях очень
малой характеристики с временн\'{о}й сложностью $L(1 / 4 + o(1))$.

Наконец, на конференции EUROCRYPT в 2014 году был предложен эвристический квази-полиномиальный алгоритм
дискретного логарифмирования в полях малой характеристики, который при определённых допущениях имеет временн\'{у}ю сложность
$n^{O(\log n)}$~\cite{quasi_polynomial_discrete_log}.
Также стоит отметить, что существуют квантовые алгоритмы дискретного логарифмирования (к примеру, алгоритм Шора)
с полиномиальной временн\'{о}й сложностью~\cite{shor}.

Однако автору не удалось найти какие-либо опубликованные результаты по реализации алгоритма дискретного логарифмирования
в схеме, состоящей из классических, неквантовых \gate, и оценок сложности такой схемы.

\bigskip
\bigskip
Рассмотрим некоторые особенности решения задачи нахождения $\log_{\balpha}{\vv x}$ в фактор-кольце $\fp$,
где $f(x)$~--- неприводимый многочлен степени $n$, являющегося полем характеристики 2.

Мощность $M$ мультипликативной группы этого поля равна $M = |\fpm| = 2^n - 1$.
Поскольку $M$ нечётно, то для любого чётного числа $m = 2k$ существует $m^{-1}\pmod{M}$.
Отсюда следует, что уравнение $\vv x^m = \vv a$ в поле $\fp$ разрешимо для всех $\vv a \in \fpm$ при чётном $m$,
$\vv x = \vv a^{m^{-1}\pmod{M}}$. Значение $m^{-1}\pmod{M}$ можно найти, к примеру, при помощи алгоритма Евклида.

\begin{predicate}
    Все элементы $\vv a \in \fpm$ являются квадратичными вычетами.
\end{predicate}
\begin{proof}
    Поскольку $2^{-1} = \frac{M+1}{2}\pmod{M}$, то $\vv x = \vv a^{(M+1) \mathop / 2}$ является решением уравнения
    $\vv x^2 = \vv a$ в поле $\fp$.
\end{proof}

\begin{corollary}
Для извлечения квадратного корня из $\vv a \in \fpm$ необходимо возвести $\vv a$ в степень $\frac{M+1}{2}$.
\end{corollary}

Для произвольного $\vv x = \balpha^k$ степень $k$ можно представить в двоичном виде с $n$ разрядами:
$$
    \balpha^k = \balpha^{[k_n \ldots k_2 k_1]} \; .
$$
При возведении в квадрат элемента $\vv x$ все разряды $k_i$ в двоичной записи смещаются на одну позицию влево:
$$
    (\balpha^k)^2 = \balpha^{[k_n \ldots k_2 k_1 0]}
        = \balpha^{k_n 2^n} \balpha^{[k_{n-1} \ldots k_2 k_1 0]} \; .
$$
Заметим, что $2^n\pmod{2^n - 1} \equiv 1\pmod{2^n - 1} \Rightarrow \balpha^{k_n 2^n} = \balpha^{k_n}$.
Отсюда следует равенство
\begin{equation}
    (\balpha^{[k_n \ldots k_2 k_1]})^2 = \balpha^{[k_{n-1} \ldots k_2 k_1 k_n]} \; .
    \label{formula_degree_rotation}
\end{equation}
Другими словами, при возведении в квадрат элемента $\balpha^k$ происходит \textit{циклический} сдвиг разрядов влево
в двоичной записи степени $k$. Несложно показать, что при извлечении квадратного корня из $\balpha^k$
происходит циклический сдвиг вправо разрядов в двоичной записи степени $k$:
\begin{equation}
    \sqrt{\balpha^{[k_n \ldots k_2 k_1]}} = \balpha^{[k_1 k_n \ldots k_3 k_2]} \; .
    \label{formula_sqrt_rotation}
\end{equation}

Рассмотрим для произвольного элемента $\vv x \in \fpm$ следующую последовательность:
$$
    \lbrace \vv x, \vv x^2, \vv x^4, \ldots, \vv x^{2^{n-1}} \rbrace \; .
$$
Очевидно, что $\vv x^{2^n} = \vv x$, поэтому выше рассмотрено только $n$ последовательных квадратов,
которые в общем случае все различны. Рассмотрим, при каком условии $\vv x^{2^i} = \vv x$, где $i \ne n-1$.
Пусть $\vv x = \balpha^k$, $k = [k_n \ldots k_2 k_1]$.
Согласно формуле~\eqref{formula_degree_rotation}, при возведении в квадрат происходит циклический сдвиг разрядов в двоичной
записи степени $k$, поэтому верно следующее равенство:
$$
    \vv x^{2^i} = \balpha^{k'}, \text{ }k'={[k_{n-i}\ldots k_2 k_1 k_n \ldots k_{n-i + 1}]} \; .
$$
Отсюда можно сделать вывод, что $k' = k \Leftrightarrow i \mid n$ и число $k$ в двоичной записи имеет
периодическую структуру с периодом, равным $i$ разрядов.
К примеру, при $n = 6$ число $k = [011011]$ имеет период из трёх разрядов, равный $[011]$.
Поэтому для такого значения $k$ элемент $\vv x = \balpha^k$ в кубе будет равен самому себе: $\vv x^3 = \vv x$.

Введём множество $R({\vv x})$ для элемента $\vv x \in \fpm$ следующим образом:
\begin{equation}
    R({\vv x}) = \{\,\vv x^{2^i} \mid 0 \leqslant i \leqslant n-1, \vv x^{2^i} \neq \vv x \text{ при } i \neq 0\,\} \; .
    \label{formula_cyclic_set_for_x}
\end{equation}
Все элементы множества $R({\vv x})$ различны.
Зная любой элемент $\vv y \in R({\vv x})$, можно получить всё множество $R({\vv x})$
путём последовательного возведения в квадрат элемента $\vv y$.
Из построения~\eqref{formula_cyclic_set_for_x} множества $R({\vv x})$ следует, что
\begin{equation}
    m \cdot |R({\vv x})| = n, \text{ } m \in \mathbb N^+  \; .
\end{equation}
Другими словами, мощность множества $R({\vv x})$ делит $n$.

Таким образом, все элементы мультипликативной группы $\fpm$ разбиваются на непересекающиеся множества $R({\vv x_i})$
с, возможно, неравными мощностями.
Умножив элемент множества $R({\vv x_i})$ на примитивный элемент $\balpha$, получим новый элемент,
который либо принадлежит этому же множеству $R({\vv x_i})$, либо новому $R({\vv x_j})$.

Для элемента $\vv x = \balpha^0 = 1$ множество $R({\vv x})$ состоит только из одного этого элемента.
Отсюда следует одно утверждение для простого $n$, являющееся тривиальным следствием равенства
$q^{p-1} \equiv 1 \pmod{p}$ для $q=2$, где $p$~--- простое число:
\begin{predicate}
    $2^n - 2$ делится нацело на $n$, если $n$~--- простое.
\end{predicate}
\begin{proof}
Докажем при помощи свойств множеств $R({\vv x})$.
Для всех $\balpha^k$ мощность множества $R({\balpha^k})$ должна делить $n$.
По условию, $n$ простое $\Rightarrow |R({\balpha^k})|$ равно либо $n$, либо $1$.
Только для $k = 0$ и $k = 2^n - 1$ мощность $|R({\balpha^k})|=1$. Однако $\balpha^0 = \balpha^{2^n - 1} = 1$.
Отсюда следует, что $|R({\balpha^k})| = n$ при $k \neq 0$, $k < 2^n - 1$.

Выше было показано, что все элементы мультипликативной группы $\fpm$ разбиваются на непересекающиеся множества $R({\vv x_i})$,
из которых только $R({\balpha^0})$ имеет мощность 1, остальные множества имеют мощность $n$.
Мощность мультипликативной группы $|\fpm| = 2^n - 1 \Rightarrow |\left( \fpm \right) \setminus \{ \balpha^0 \}| = 2^n - 2$,
откуда следует исходное утверждение.
\end{proof}

Если вычисление степени $k$ для некоторого $\vv x = \balpha^k$ происходит путём опробования всех возможных значений $k$,
то для поля $\fp$ в общем случае можно опробовать не $\sim 2^n$ значений, а $\sim 2^n \mathop / n$:
\label{brute_force_complexity_in_polynomial_field}
при опробовании нового значения $k'$ получить $\vv x' = \balpha^{k'}$ и $n-1$ последовательных квадратов
$(\vv x')^2, \ldots, (\vv x')^{2^{n-1}}$. Если среди полученных $n$ значений будет $\vv x$, то значение $k$
вычисляется из $k'$ не более чем за $n$ шагов при помощи циклического сдвига влево двоичной записи $k'$.

С другой стороны, если окажется, что среди $\vv x, \vv x^2, \ldots, \vv x^{2^{n-1}}$ есть повторяющиеся элементы,
то можно утверждать, что степень $k$, $\vv x = \balpha^k$, обладает периодом $m < n$.
Пусть $n$ не является простым числом, тогда $n = ml$, $l \in \mathbb N^+$.
В этом случае пространство перебора для степени $k$ сокращается до $\sim 2^m \mathop / m= o(2^n \mathop / n)$ значений.
Тем не менее, количество значений степени $k$, обладающих периодом $m < n$, ничтожно мало по сравнению с количеством
всех возможных значений степени $k$. Отсюда следует, что для произвольного $\vv x = \balpha^k$ вероятность того,
что $k$ обладает периодом $m < n$, стремится к нулю.


\subsection{Схемная реализация дискретного логарифмирования в конечном поле характеристики 2}

\forceindent
В данном разделе будут рассмотрены различные способы реализации с помощью обратимых схем алгоритма возведения в степень и
дискретного логарифмирования в конечном поле характеристики 2 на примере фактор-кольца $\fp$,
где $f(x)$~--- неприводимый многочлен степени $n$.

Без ограничения общности дальнейших рассуждений будем считать, что для неприводимого многочлена $f(x)$ элемент поля
$\balpha = x$, $\balpha \in \fp$, является примитивным, поскольку конечные поля с одинаковым количеством элементов
изоморфны друг другу.

Сформулируем \textit{прямую задачу} (задачу возведения в степень).\\
\textbf{Дано}: $k \in \ZZ_M$, где $M = 2^n - 1$, $n$~--- степень неприводимого многочлена $f(x)$.\\
\textbf{Получить}: $\balpha^k$.

Сформулируем \textit{обратную задачу} (задачу дискретного логарифмирования).\\
\textbf{Дано}: $\boldsymbol\beta \in \fpm$.\\
\textbf{Получить}: $k \in \ZZ_M \colon \balpha^k = \boldsymbol\beta$,
где $M = 2^n - 1$, $n$~--- степень неприводимого многочлена $f(x)$.

Как было показано в предыдущих главах, обратимые схемы, состоящие из элементов $E(t, I, J)$, задают некоторую подстановку
на множестве двоичных векторов. Построим взаимно однозначное отображение из множества двоичных векторов $\ZZ_2^n$
в множество вычетов $\ZZ_M$ и в множество элементов поля $\fp$.

Зададим отображение $f_\NN\colon \ZZ_2^n \to \ZZ_{2^n}$ следующим образом:
\begin{equation}
    f_\NN(\langle v_1, \ldots, v_n \rangle) = \sum_{i=1}^{n}{v_i 2^{i-1}} \; .
\end{equation}
В этом случае вектор $\vv v = \langle v_1, \ldots, v_n \rangle$ соответствует записи $[v_n, \ldots, v_1]$
в двоичной системе счисления числа $\sum_{i=1}^{n}{v_i 2^{i-1}}$.

Зададим отображение $f_\FF\colon \ZZ_2^n \to \fp$ следующим образом:
\begin{equation}
    f_\FF(\langle v_1, \ldots, v_n \rangle) = \bigoplus_{i=1}^{n}{v_i x^{i-1}} \; .
\end{equation}
В этом случае вектор $\vv v = \langle v_1, \ldots, v_n \rangle$ представляет собой вектор коэффициентов
многочлена $\bigoplus_{i=1}^{n}{v_i x^{i-1}}$.

Используя эти два отображения, зададим отображение $f_\mathrm{pow}\colon \ZZ_2^n \to \ZZ_2^n$ следующим образом:
\begin{equation}
    f_\mathrm{pow}(\vv v) = \vv u \colon f_\FF(\vv u) = \balpha^{f_\NN(\vv v)},
                \text{ если }\vv v \neq \langle 1, \ldots, 1 \rangle \; .
    \label{formula_discrete_power_base_function}
\end{equation}
На входном значении $\vv v = \langle 1, \ldots, 1 \rangle$ значение $f_\mathrm{pow}(\vv v)$ явно определять не будем,
т.\,к. в этом случае $f_\NN(\vv v) \notin \ZZ_M$, где $M = 2^n - 1$, $n$~--- степень неприводимого многочлена $f(x)$.

Также зададим отображение $f_\mathrm{log}\colon \ZZ_2^n \to \ZZ_2^n$ следующим образом:
\begin{equation}
    f_\mathrm{log}(\vv v) = \vv u \colon f_\FF(\vv v) = \balpha^{f_\NN(\vv u)},
                \text{ если }\vv v \neq \langle 0, \ldots, 0 \rangle \; .
    \label{formula_discrete_logarithm_base_function}
\end{equation}
Как и в предыдущем случае, на входном значении $\vv v = \langle 0, \ldots, 0 \rangle$
значение $f_\mathrm{log}(\vv v)$ явно определять не будем, т.\,к. в этом случае $f_\FF(\vv v) \notin \fpm$.

Тогда можно утверждать, что обратимая схема является решением \textit{прямой задачи},
если она реализует отображение $f_\mathrm{pow}$~\eqref{formula_discrete_power_base_function},
и является решением \textit{обратной задачи}, если она реализует отображение
$f_\mathrm{log}$~\eqref{formula_discrete_logarithm_base_function}.

Обозначим для краткости через $\Omega^2_*$ множество всех возможных элементов NOT, CNOT и 2-CNOT в обратимой схеме,
число входов которой заранее неизвестно. Этим мы ограничимся обратимыми \gate, имеющими не более двух прямых контролирующих
входов.


\myparagraph{Схемы без дополнительной памяти}

\forceindent
Рассмотрим сперва решение \textit{прямой задачи} при помощи обратимых схем, не использующих дополнительную
память. Доопределим отображение $f_\mathrm{pow}$ до биекции:
\begin{equation}
    f'_\mathrm{pow}(\vv v) =
        \begin{cases}
        f_\mathrm{pow}(\vv v),&\text{если }\vv v \neq \langle 1, \ldots, 1 \rangle  \; , \\
        \langle 0, \ldots, 0 \rangle & \text{иначе} \; .
        \end{cases}
    \label{formula_discrete_power_bijection_function}
\end{equation}
Получившееся отображение $f'_\mathrm{pow}$ является биекцией, однако оно уже не является отображением,
задаваемым непосредственно функцией возведения в степень в поле $\fp$.
Это следует из того, что $f_\NN(\langle 1, \ldots, 1 \rangle) = 2^n - 1$ и
$\balpha^{2^n - 1} = \balpha^0 = 1$, в то время как $f_\FF(f'_\mathrm{pow}(\langle 1, \ldots, 1 \rangle)) = 0$.

Отметим, что значение отображения $(f'_\mathrm{pow})^{-1}(\vv v)$ совпадает со значением отображения
$f_\mathrm{log}(\vv v)$ при $\vv v \neq \langle 0, \ldots, 0 \rangle$.
Отсюда следует, что обратимая схема, задающая преобразование $f'_\mathrm{pow}$, является решением \textit{прямой задачи} и
позволяет получить обратимую схему, являющуюся решением \textit{обратной задачи}, с той же сложностью
путём зеркального отображения слева направо данной схемы.

Данный подход позволяет оценить сверху сложность обратимой схемы, реализующей алгоритм
дискретного логарифмирования в поле $\fp$.
Если построенная схема будет иметь минимальную сложность среди всех обратимых схем,
реализующих отображение $f_\mathrm{log}$, то можно будет утверждать, что при данном значении $n$
\textit{прямая} и \textit{обратная задачи} решаются с одинаковой схемной сложностью.

В таблице~\ref{table_discrete_log_complexity_simple_case} приведены экспериментальные результаты синтеза
обратимых схем, состоящих из элементов $E(t,I,J)$ и реализующих отображение $f_\mathrm{log}$
без использования дополнительной памяти при различных значениях неприводимого многочлена $f(x)$.
Обозначения: $n$~--- степень многочлена $f(x)$, $L(\frS)$~--- сложность обратимой схемы.
Синтез производился при помощи разработанного программного обеспечения,
основанного на алгоритме синтеза~\algref{alg_my_common}.

{
    \renewcommand{\baselinestretch}{1.2}

    \Table[ht!]
        \small
        \centering
        \begin{tabular}{|*{3}{c|}|*{3}{c|}}
            \hline
            $n$ & $f(x)$ & $\balpha$ & $L(\frS)$ & $L(\frS^*)$ & $n^3$ \tabularnewline \hline
            
            2 & $x^2 + x + 1$                   & $x$           & 3     & 3    & 8    \tabularnewline \hline
            3 & $x^3 + x + 1$                   & $x$           & 6     & 7    & 27   \tabularnewline \hline
            3 & $x^3 + x^2 + 1$                 & $x$           & 8     & 7    & 27   \tabularnewline \hline
            4 & $x^4 + x + 1$                   & $x$           & 23    & 18   & 64   \tabularnewline \hline
            4 & $x^4 + x^3 + x^2 + x + 1$       & $x + 1$       & 18    & 15   & 64   \tabularnewline \hline
            4 & $x^4 + x^3 + 1$                 & $x$           & 22    & 17   & 64   \tabularnewline \hline
            5 & $x^5 + x^2 + 1$                 & $x$           & 53    & 41   & 125  \tabularnewline \hline
            5 & $x^5 + x^4 + x^3 + x^2 + 1$     & $x$           & 53    & 42   & 125  \tabularnewline \hline
            5 & $x^5 + x^4 + x^2 + x + 1$       & $x$           & 55    & 37   & 125  \tabularnewline \hline
            5 & $x^5 + x^3 + x^2 + x + 1$       & $x$           & 60    & 41   & 125  \tabularnewline \hline
            6 & $x^6 + x + 1$                   & $x$           & 178   & 85   & 216  \tabularnewline \hline
            6 & $x^6 + x^4 + x^2 + x + 1$       & $x + 1$       & 168   & 91   & 216  \tabularnewline \hline
            6 & $x^6 + x^5 + x^2 + x + 1$       & $x$           & 156   & 85   & 216  \tabularnewline \hline
            6 & $x^6 + x^3 + 1$                 & $x + 1$       & 145   & 90   & 216  \tabularnewline \hline
            7 & $x^7 + x + 1$                   & $x$           & 415   & 184  & 343  \tabularnewline \hline
            7 & $x^7 + x^3 + 1$                 & $x$           & 407   & 190  & 343  \tabularnewline \hline
            7 & $x^7 + x^5 + x^2 + x + 1$       & $x$           & 400   & 191  & 343  \tabularnewline \hline
            7 & $x^7 + x^6 + x^4 + x + 1$       & $x$           & 358   & 191  & 343  \tabularnewline \hline
            8 & $x^8 + x^4 + x^3 + x^2 + 1$     & $x$           & 951   & 422  & 512  \tabularnewline \hline
            8 & $x^8 + x^6 + x^5 + x^2 + 1$     & $x$           & 987   & 417  & 512  \tabularnewline \hline
            8 & $x^8 + x^7 + x^6 + x + 1$       & $x$           & 1019  & 414  & 512  \tabularnewline \hline
            8 & $x^8 + x^6 + x^3 + x^2 + 1$     & $x$           & 943   & 401  & 512  \tabularnewline \hline
            9 & $x^9 + x^4 + 1$                 & $x$           & 2698  & 858  & 729  \tabularnewline \hline
            9 & $x^9 + x^8 + x^4 + x + 1$       & $x$           & 2691  & 873  & 729  \tabularnewline \hline
            9 & $x^9 + x^8 + 1$                 & $x^2 + x + 1$ & 2780  & 892  & 729  \tabularnewline \hline
            9 & $x^9 + x^7 + x^6 + x^4 + 1$     & $x$           & 2679  & 849  & 729  \tabularnewline \hline
            10 & $x^{10} + x^3 + 1$             & $x$           & 6312  & 1840 & 1000 \tabularnewline \hline
            10 & $x^{10} + x^9 + x^5 + x + 1$   & $x + 1$       & 6419  & 1873 & 1000 \tabularnewline \hline
            10 & $x^{10} + x^6 + x^2 + x + 1$   & $x + 1$       & 6437  & 1858 & 1000 \tabularnewline \hline
            10 & $x^{10} + x^8 + x^7 + x^6 + 1$ & $x^2 + x + 1$ & 6289  & 1847 & 1000 \tabularnewline \hline
            11 & $x^{11} + x^2 + 1$             & $x$           & 14659 & 3947 & 1331 \tabularnewline \hline
            11 & $x^{11} + x^5 + x^3 + x + 1$   & $x$           & 14429 & 3952 & 1331 \tabularnewline \hline
            11 & $x^{11} + x^7 + x^6 + x^5 + 1$ & $x$           & 14636 & 3941 & 1331 \tabularnewline \hline
            11 & $x^{11} + x^7 + x^5 + x^3 + 1$ & $x$           & 14559 & 3921 & 1331 \tabularnewline \hline

        \end{tabular}
        \caption{
            \small Сложность $L(\frS)$ и $L(\frS^*)$ обратимых схем, состоящих из элементов $E(t,I,J)$
                и реализующих отображение $f_\mathrm{log}$ без использования и с использованием
                дополнительной памяти соответственно при различных значениях
                неприводимого многочлена $f(x)$.
        }\label{table_discrete_log_complexity_simple_case}
    \end{table}

} 

Из таблицы~\ref{table_discrete_log_complexity_simple_case} видно, что с ростом $n$ сложность
обратимой схемы увеличивается примерно в 2 раза.
Эти результаты согласуются с теоремой~\ref{theorem_complexity_no_memory_common}, при этом можно сделать вывод,
что при реализации отображения $f_\mathrm{log}$ обратимой схемой, состоящей из элементов $E(t,I,J)$
и не использующей дополнительную память, описанный подход не позволяет выявить существенных отличий данного отображения
от произвольного.


\myparagraph[paragraph_schemes_with_additional_memory_for_discrete_log]{Схемы c дополнительной памятью}

\forceindent
Согласно утверждению~\ref{predicate_reduction_of_complexity_with_help_of_additional_memory},
для большинства булевых отображений использование в реализующих их обратимых схемах дополнительной памяти
позволяет снизить их сложность.
Для решения задачи возведения в степень $k \in \ZZ_{2^n}$ известен быстрый алгоритм, использующий не более $2n$ умножений.
В работе~\cite{zakablukov_zasorina_chikin} было показано, что существует обратимая схема,
состоящая из \gate{} множества $\Omega^2_*$,
реализующая умножение многочленов в поле $\fp$ со сложностью $2n^2$ при использовании $(2n-1)$ дополнительных входов.
\label{complexity_of_power_and_multiply_in_polinomial_field}
Таким образом, можно построить обратимую схему, состоящую из \gate{} множества $\Omega^2_*$,
решающую задачу возведения в степень (прямую задачу) в поле $\fp$ со сложностью $4n^3$ при использовании
$(4n^2 - 2n)$ дополнительных входов.

Как уже было сказано в предыдущем параграфе, на данный момент неизвестно, существует ли полиномиальный алгоритм решения
задачи дискретного логарифмирования в общем случае.
Тем не менее представляется интересным исследование снижения сложности обратимой схемы,
являющейся решением обратной задачи, за счёт использования дополнительной памяти.

Отображение $f_\mathrm{log}$ можно представить в виде совокупности $n$ независимых координатных функций
$f_\mathrm{log}^{(i)}\colon \ZZ_2^n \to \ZZ_2$.
Один из способов реализации таких функций при помощи обратимых схем с дополнительной памятью использовался в алгоритме%
~\algref{alg_asymp_with_mem_complexity_min}.
Однако если мы допускаем, что обратимые схемы могут состоять не только из \gate{} множества $\Omega^2_*$,
но и обобщённых элементов $E(t,I,J)$, то можно предложить второй способ реализации указанных координатных функций,
который в некоторых случаях даёт меньшую сложность.

В параграфе~\ref{label_boolean_edge_search_paragraph} был описан способ снижения сложности обратимой схемы,
основанный на поиске грани булева куба.
Множество двоичных векторов $\vv v_j$, для которых $f_\mathrm{log}^{(i)}(\vv v_j) = 1$, можно разбить
на непересекающиеся подмножества, каждое из которых будет представлять собой грань булева куба.
Реализуя каждое из этих подмножеств с помощью обратимых элементов $E(t,I,J)$
(см. параграф~\ref{label_boolean_edge_search_paragraph}), можно построить обратимую подсхему,
реализующую координатную функцию $f_\mathrm{log}^{(i)}$. Объединив эти обратимые подсхемы для всех координатных функций,
мы получим обратимую схему, реализующую отображение $f_\mathrm{log}$.\label{simple_boolean_edge_search_algorithm}
В результате будет использовано всего $n$ дополнительных входов (по количеству координатных функций).

Экспериментальные результаты синтеза обратимых схем при помощи разработанного программного обеспечения,
основанного на данном подходе, приведены в таблице~\ref{table_discrete_log_complexity_simple_case}.
Обозначения: $n$~--- степень многочлена $f(x)$, $L(\frS^*)$~--- сложность обратимой схемы,
использующей дополнительную память.

Из таблицы~\ref{table_discrete_log_complexity_simple_case} видно, что с ростом $n$ сложность
обратимой схемы, использующей дополнительную память, увеличивается примерно в 2 раза.
При этом величина сложности растёт медленней по сравнению с обратимыми схемами,
не использующими дополнительную память.
Эти результаты согласуются с теоремой~\ref{theorem_complexity_with_memory_common}, при этом можно сделать вывод,
что при реализации отображения $f_\mathrm{log}$ обратимой схемой, состоящей из элементов $E(t,I,J)$
и использующей дополнительную память, описанный подход не позволяет выявить существенных отличий данного отображения
от произвольного.
Также стоит отметить, что предложенный способ синтеза обратимых схем с дополнительной памятью при помощи поиска
граней булева куба хоть и не позволяет во всех случаях снизить сложность, но даёт возможность получить
обратимую схему, использующую всего $n$ дополнительных входов. Такое количество дополнительных входов намного меньше,
чем количество $q \sim n2^{n - o(n)}$ дополнительных входов обратимой схемы, полученной при помощи алгоритма синтеза%
~\algref{alg_asymp_with_mem_complexity_min} обратимых схем с дополнительной памятью,
описанного на с.~\pageref{alg_asymp_with_mem_complexity_min}.

\bigskip
\noindent\textbf{Использование свойства циклического сдвига при возведении в квадрат}
\smallskip

В разделе~\ref{subsection_discrete_logarithm_in_polynomial_field} на с.~\pageref{brute_force_complexity_in_polynomial_field}
было показано, что для решения обратной задачи (задачи дискретного логарифмирования) в поле $\fp$ можно
перебирать не все $2^n$ возможных значений для степени, а только некоторые $2^n \mathop / n$ значений.
Можно воспользоваться этим свойством для снижения сложности обратимой схемы, являющейся решением данной задачи.

Рассмотрим произвольный элемент $\vv x \in \fpm$, $\vv x = \balpha ^ k$.
На с.~\pageref{formula_cyclic_set_for_x} мы ввели множество $R(\vv x)$ различных элементов из
$\fpm$ следующего вида:
$$
    R(\vv x) = \{\,\vv x^{2^i}\mid 0 \leqslant i \leqslant n - 1, \vv x^{2^i} \neq \vv x \text{ при } i \neq 0\,\} \; .
$$
В большинстве случаев мощность этого множества равна $n$.
Если степень $n$ неприводимого многочлена $f(x)$~--- простое число, то только для $\vv x' = 1$ верно равенство
$|R(\vv x')| = 1$, для всех оставшися $2^n -2$ элементов множества $\fpm$ мощность $|R(\vv x)| = n$.

Для различных элементов множества $R(\vv x)$ можно построить такие же множества. Очевидно, что все такие множества будут равны
между собой:
$$
    R(\vv x) = R(\vv x^{2^i}) \; .
$$
Для каждого из различных множеств $R(\vv x)$ зафиксируем ровно один элемент из него $\vv d \in R(\vv x)$.
Назовём этот элемент $\vv d$ \textit{представителем} множества $R(\vv x)$.

Введём отображение $g\colon \ZZ_2^n \to \ZZ_2^n$ следующим образом:
\begin{equation}
    g(\vv v) = \vv u \colon \balpha^{f_\NN(\vv u)} = \vv d, \text{ где } \vv d \text{ --- представитель множества }
        R(f_\FF(\vv v)) \; .
    \label{formula_discrete_logarithm_rotation_optimization}
\end{equation}
Другими словами, отображение $g$ является решением задачи дискретного логарифмирования не для самого элемента $f_\FF(\vv v)$,
а для представителя $\vv d$ множества $R(f_\FF(\vv v))$.
В зависимости от выбора представителей множеств $R(\vv x)$ будет меняться соответствующим образом и таблица истинности для
отображения $g$. Причём в отличие от таблицы истинности для отображения $f_\mathrm{log}$, в этой таблице будет примерно
в $n$ раз меньше различных значений, что предположительно может позволить снизить сложность обратимой схемы,
реализующей отображение $g$. Обозначим эту схему и её сложность через $\frS_1$ и $L(\frS_1)$ соответственно.

Зная представитель $\vv d$ множества $R(\vv x)$ и степень $k_{\vv d}$, для которой верно равенство
$\balpha^{k_{\vv d}} = \vv d$, можно относительно просто узнать для любого элемента $\vv y \in R(\vv x)$
степень $k_{\vv y}$, для которой верно равенство $\balpha^{k_{\vv y}} = \vv y$.
Рассмотрим, какова сложность обратимой схемы $\frS_2$, решающей данную задачу.
Как было сказано в начале этого параграфа на с.~\pageref{complexity_of_power_and_multiply_in_polinomial_field},
для возведения $\balpha$ в степень $k_{\vv d}$ требуется не более $4n^3$ обратимых \gate{}
Оставшаяся часть схемы представляет из себя $n$ почти одинаковых частей, которые делают следующее:
сравнивают текущий результат $\vv y'$ с входным значением $\vv y$ и, если результат совпадает, применяют к $k_{\vv d}$
циклический сдвиг на определённую величину и копируют это значение на выходы схемы;
в конце возводят текущий результат $\vv y'$ в квадрат и передают его на вход следующей подсхеме.
Циклический свдиг делается на величину, равную порядку подсхемы в схеме минус 1: первая подсхема не делает циклический сдвиг,
вторая делает циклический сдвиг на 1, третья~--- на 2 и т.\,д.

Сравнение двух двоичных векторов длины $n$ можно реализовать при помощи $(2n + 1)$ элементов $E(t,I,J)$:
соответствующие координаты двух векторов копируются при помощи элементов CNOT на отдельный выход в схеме, получается
функция сложения по модулю 2 этих координат; затем при помощи одного элемента $E(t, I, I)$, где множество $I$ равно множеству
тех выходов, которые были получены на предыдущем шаге, получается значение на выходе $t$, равное 1, если два входных вектора
совпадают по всем координатам, и 0 иначе.

Копирование степени $k_{\vv d}$ на выходы схемы с циклическим сдвигом на любую величину можно сделать при помощи $n$
элементов 2-CNOT, у которых в качестве одного из контролирующих входов будет вход $t$, полученный на предыдущем шаге.
И, наконец, возведение в квадрат текущего результата $\vv y'$ можно реализовать при помощи $2n^2$ обратимых элементов,
как было показано в начале этого параграфа на с.~\pageref{complexity_of_power_and_multiply_in_polinomial_field}.

Таким образом, суммарная сложность рассматриваемой обратимой схемы будет равна
\begin{equation}
    L(\frS_2) = 4n^3 + n(2n + 1 + n + 2n^2) = 6n^3 + 3n^2 + n \; .
\end{equation}

Объединяя две рассмотренные схемы $\frS_1$ и $\frS_2$ в одну, можно получить обратимую схему $\frS$, являющуюся решением 
обратной задачи (задачи дискретного логарифмирования). Её сложность $L(\frS)$ будет равна
\begin{equation}
    L(\frS) = L(\frS_1) + 6n^3 + 3n^2 + n \; .
\end{equation}
Таким образом, чтобы описанный подход позволил снизить сложность обратимой схемы, сложность $L(\frS_1)$
должна удовлетворять неравенству
\begin{equation}
    L(\frS_1) < L(\frS^*) - 6n^3 - 3n^2 - n \; ,
    \label{formula_profit_condition_for_rotation_optimization}
\end{equation}
где $\frS^*$~--- обратимая схема, являющаяся решением обратной задачи, полученная при помощи алгоритма синтеза
по таблице истинности, описанного в начале данного параграфа.

{
    \renewcommand{\baselinestretch}{1.16}
    \Table[ht!]
        \small
        \centering
        \begin{tabular}{|*{9}{c|}{|c|}}
            \hline
            
            \multirow{2}{*}{ $n$ } & 
            \multirow{2}{*}{ $f(x)$ } &
            \multirow{2}{*}{ $L(\frS^*)$ } &
            \multicolumn{6}{c||}{ $L(\frS_1)$ } &
            \multirow{2}{*}{ $\sigma$ } \tabularnewline
            
            \cline{4-9}
            & & &
            $k_\mathrm{min}$  &
            $k_\mathrm{max}$  &
            $k_\mathrm{dist}$ &
            $r_1$ & $r_2$ & $r_3$ &             
            \tabularnewline \hline
            
             2 & $x^2+x+1$              &    3 &    3 &    3 &    3 &    3 &    3 &    3 & 1    \tabularnewline \hline
             3 & $x^3+x+1$              &    7 &    5 &    5 &    5 &    5 &    5 &    5 & 1,4  \tabularnewline \hline
             3 & $x^3+x^2+1$            &    7 &    7 &    7 &    7 &    7 &   10 &   10 & 1    \tabularnewline \hline
             4 & $x^4+x+1$              &   18 &    8 &    8 &   12 &    9 &   11 &   13 & 2,25 \tabularnewline \hline
             4 & $x^4+x^3+x^2+x+1$      &   15 &   11 &   11 &   16 &   11 &   13 &   16 & 1,36 \tabularnewline \hline
             4 & $x^4+x^3+1$            &   17 &   11 &   11 &   11 &   11 &   13 &   18 & 1,55 \tabularnewline \hline
             5 & $x^5+x^2+1$            &   41 &   23 &   23 &   33 &   30 &   32 &   36 & 1,78 \tabularnewline \hline
             5 & $x^5+x^4+x^3+x^2+1$    &   42 &   29 &   29 &   45 &   38 &   45 &   50 & 1,45 \tabularnewline \hline
             5 & $x^5+x^4+x^2+x+1$      &   37 &   26 &   26 &   29 &   28 &   33 &   38 & 1,42 \tabularnewline \hline
             5 & $x^5+x^3+x^2+x+1$      &   41 &   22 &   22 &   27 &   25 &   31 &   36 & 1,86 \tabularnewline \hline
             6 & $x^6+x+1$              &   85 &   50 &   51 &   60 &   73 &   75 &   88 & 1,7  \tabularnewline \hline
             6 & $x^6+x^4+x^2+x+1$      &   91 &   48 &   50 &   64 &   62 &   66 &   69 & 1,9  \tabularnewline \hline
             6 & $x^6+x^5+x^2+x+1$      &   85 &   57 &   56 &   80 &   73 &   81 &   88 & 1,52 \tabularnewline \hline
             6 & $x^6+x^3+1$            &   90 &   54 &   40 &   57 &   51 &   62 &   73 & 2,25 \tabularnewline \hline
             7 & $x^7+x+1$              &  184 &  124 &  119 &  138 &  149 &  157 &  169 & 1,55 \tabularnewline \hline
             7 & $x^7+x^3+1$            &  190 &  119 &  119 &  128 &  159 &  160 &  172 & 1,6  \tabularnewline \hline
             7 & $x^7+x^5+x^2+x+1$      &  191 &  128 &  117 &  146 &  155 &  168 &  179 & 1,63 \tabularnewline \hline
             7 & $x^7+x^6+x^4+x+1$      &  191 &  123 &  108 &  169 &  169 &  179 &  189 & 1,77 \tabularnewline \hline
             8 & $x^8+x^4+x^3+x^2+1$    &  422 &  276 &  265 &  341 &  353 &  370 &  394 & 1,59 \tabularnewline \hline
             8 & $x^8+x^6+x^5+x^2+1$    &  417 &  273 &  260 &  378 &  375 &  389 &  403 & 1,6  \tabularnewline \hline
             8 & $x^8+x^7+x^6+x+1$      &  414 &  279 &  272 &  358 &  369 &  385 &  415 & 1,52 \tabularnewline \hline
             8 & $x^8+x^6+x^3+x^2+1$    &  401 &  261 &  257 &  357 &  376 &  384 &  395 & 1,56 \tabularnewline \hline
             9 & $x^9+x^4+1$            &  858 &  600 &  598 &  795 &  839 &  857 &  874 & 1,43 \tabularnewline \hline
             9 & $x^9+x^8+x^4+x+1$      &  873 &  595 &  609 &  814 &  836 &  845 &  850 & 1,47 \tabularnewline \hline
             9 & $x^9+x^8+1$            &  892 &  584 &  596 &  780 &  826 &  858 &  889 & 1,53 \tabularnewline \hline
             9 & $x^9+x^7+x^6+x^4+1$    &  849 &  618 &  605 &  775 &  828 &  849 &  886 & 1,4  \tabularnewline \hline
            10 & $x^{10}+x^3+1$         & 1840 & 1334 & 1311 & 1549 & 1797 & 1812 & 1838 & 1,4  \tabularnewline \hline
            10 & $x^{10}+x^9+x^5+x+1$   & 1873 & 1339 & 1331 & 1763 & 1792 & 1828 & 1834 & 1,41 \tabularnewline \hline
            10 & $x^{10}+x^6+x^2+x+1$   & 1858 & 1312 & 1288 & 1587 & 1789 & 1828 & 1845 & 1,44 \tabularnewline \hline
            10 & $x^{10}+x^8+x^7+x^6+1$ & 1847 & 1332 & 1305 & 1650 & 1820 & 1841 & 1872 & 1,42 \tabularnewline \hline
            11 & $x^{11}+x^2+1$         & 3947 & 2850 & 2891 & 3703 & 3849 & 3881 & 3910 & 1,38 \tabularnewline \hline
            11 & $x^{11}+x^5+x^3+x+1$   & 3952 & 2856 & 2841 & 3444 & 3837 & 3881 & 3940 & 1,39 \tabularnewline \hline
            11 & $x^{11}+x^7+x^6+x^5+1$ & 3941 & 2882 & 2881 & 3591 & 3877 & 3900 & 3945 & 1,37 \tabularnewline \hline
            11 & $x^{11}+x^7+x^5+x^3+1$ & 3921 & 2823 & 2864 & 3396 & 3830 & 3890 & 3943 & 1,39 \tabularnewline \hline

        \end{tabular}
        \caption{
            \small Сравнение сложностей $L(\frS^*)$ и $L(\frS_1)$ обратимых схем, состоящих из элементов $E(t,I,J)$
                и реализующих отображение $f_\mathrm{log}$~\eqref{formula_discrete_logarithm_base_function}
                и $g$~\eqref{formula_discrete_logarithm_rotation_optimization} соответственно,
                при различном выборе представителей $\vv d_i$ множеств $R(\vv x_i)$:
                $k_\mathrm{min}$~--- представитель с минимальной степенью;
                $k_\mathrm{max}$~--- с максимальной степенью;
                $k_\mathrm{dist}$~--- степень представителя имеет минимальное расстояние Хемминга до векторов всех элементов
                $R(\vv x_i)$;
                $r_1, r_2, r_3$~--- случайные представители (по возрастанию сложности);
                $\sigma = L(\frS^*) \mathop / \min {L(\frS_1)}$.
        }\label{table_discrete_log_complexity_rotation_optimization}
    \end{table}

} 

В таблице~\ref{table_discrete_log_complexity_rotation_optimization} приведены экспериментальные результаты синтеза
обратимых схем, состоящих из элементов $E(t,I,J)$ и реализующих отображение
$g$~\eqref{formula_discrete_logarithm_rotation_optimization} при различных значениях неприводимого многочлена $f(x)$
и различном выборе представителей $\vv d_i$ множеств $R(\vv x_i)$.
Обозначения: $n$~--- степень многочлена $f(x)$; $L(\frS^*)$~--- сложность обратимой схемы,
реализующей отображение $f_\mathrm{log}$~\eqref{formula_discrete_logarithm_base_function},
взято из таблицы~\ref{table_discrete_log_complexity_simple_case};
$L(\frS_1)$~--- сложность обратимой схемы,
реализующей отображение $g$~\eqref{formula_discrete_logarithm_rotation_optimization};
$\sigma$~--- отношение $L(\frS^*)$ к минимальному из полученных во время экспериментов значению $L(\frS_1)$.
Синтез производился при помощи разработанного программного обеспечения, основанного на алгоритме
поиска граней булева куба (см. с.~\pageref{simple_boolean_edge_search_algorithm}).

Для каждого элемента $\vv y_j$ каждого из множеств $R(\vv x_i)$ можно найти соответствующую степень $k_j$,
такую что $\balpha^{k_j} = \vv y_j$. Для одного множества $R(\vv x_i)$ все степени $k_j$, соответствующие его элементам,
представляют собой циклический сдвиг друг друга, поэтому вес Хэмминга у них одинаковый.
Во время проведения эксперимента выбор представителя $\vv d_i$ множества $R(\vv x_i)$ производился одним из следующих способов:
\begin{enumerate}
    \item представитель с наименьшим значением $k_j$
        (колонка $k_\mathrm{min}$ в таблице~\ref{table_discrete_log_complexity_rotation_optimization});

    \item представитель с наибольшим значением $k_j$
        (колонка $k_\mathrm{max}$ в таблице~\ref{table_discrete_log_complexity_rotation_optimization});
        
    \item степень $k_j$ и все элементы множества $R(\vv x_i)$ представляются в виде двоичных векторов длины $n$;
        выбирается представитель, для которого минимальна сумма расстояний Хемминга между вектором, представляющим
        $k_j$, и векторами, представляющими все элементы множества $R(\vv x_i)$
        (колонка $k_\mathrm{dist}$ в таблице~\ref{table_discrete_log_complexity_rotation_optimization});
        
    \item случайный представитель.
\end{enumerate}
Было произведено 5 различных экспериментов с выбором случайного представителя $\vv d_i$ множества $R(\vv x_i)$.
В таблице~\ref{table_discrete_log_complexity_rotation_optimization} приведены результаты трёх таких экспериментов:
когда получалась минимальная сложность $L(\frS_1)$ (колонка $r_1$), когда получалась максимальная
сложность $L(\frS_1)$ (колонка $r_3$), и один из экспериментов, при котором получалось промежуточное значение
$L(\frS_1)$ (колонка $r_2$).

Из данных таблицы~\ref{table_discrete_log_complexity_rotation_optimization} видно, что наименьшая сложность
$L(\frS_1)$ синтезированной схемы получалась при выборе представителя с минимальной/максимальной степенью среди возможных.
При этом значение $\sigma = L(\frS^*) \mathop / L(\frS_1)$ колебалось от $1,37$ до $2,25$.
Если предположить, что при растущем значении $n$ величина $L(\frS^*) = 2^n$ и что $\sigma \geqslant 1,35$, то
по неравенству~\eqref{formula_profit_condition_for_rotation_optimization} можно подсчитать, после какого значения $n$
предложенный подход позволяет снизить сложность:
\begin{gather*}
    \frac{2^n}{1,35} < 2^n - 6n^3 - 3n^2 - n  \; , \\
    1,35(6n^3 + 3n^2 + n) < 0,35 \cdot 2^n \Rightarrow n \geqslant 17 \; .
\end{gather*}
Другими словами, при $n \geqslant 17$ предложенный подход по синтезу обратимой схемы, являющейся решением обратной задачи
(задачи дискретного логарифмирования) в поле $\fp$, позволяет снизить сложность этой схемы в $1,35$ раза или больше
по сравнению с обратимой схемой, синтезированной по таблице истинности способом, описанным в начале
параграфа~\ref{paragraph_schemes_with_additional_memory_for_discrete_log}, при условии сохранения указанных выше темпов роста
величин $L(\frS^*)$ и $L(\frS_1)$.


\myparagraph{Верхняя асимптотическая оценка сложности}

\forceindent
Согласно теореме~\ref{theorem_complexity_upper_with_memory}, отображение $f_\mathrm{log}\colon \ZZ_2^n \to \ZZ_2^n$,
которое мы ввели на с.~\pageref{formula_discrete_logarithm_base_function}, может быть реализовано обратимой схемой
$\frS_\mathrm{log}$, состоящей из \gate{} множества $\Omega^2_*$,
со сложностью $L(\frS_\mathrm{log}) \lesssim 2^n$.
Это простейшая верхняя асимптотическая оценка сложности обратимой схемы, реализующей алгоритм дискретного
логарифмирования в поле $\fp$.

Однако в данной оценке сложности схемы не учитывается, что существует полиномиальный алгоритм возведения в степень в поле,
и при синтезе обратимой схемы с указанной выше сложностью
работа бы велась, как с произвольным булевым отображением $\ZZ_2^n \to \ZZ_2^n$.
Покажем, что верхнюю асимптотическую оценку сложности обратимой схемы $\frS_\mathrm{log}$ можно снизить,
если за основу взять схему, реализующую алгоритм возведения в степень в поле $\fp$.

\begin{theorem}
    \label{theorem_upper_bound_discrete_log}
    Существует обратимая схема $\frS_\mathrm{log}$, состоящая из \gate{} множества $\Omega^2_{n+q}$
    и реализующая отображение $f_\mathrm{log}$ со сложностью $L(\frS_\mathrm{log}) \lesssim \frac{2^{n+1} \cdot \log_2 n}{n}$
    при использовании $Q(\frS_\mathrm{log}) \sim 2^{n-\lceil n \mathop / \phi(n)\rceil + 2} \cdot \log_2 n$ дополнительных входов,
    где $\phi(n)$ и $\psi(n)$~--- любые сколь угодно медленно растущие функции такие,
    что $\phi(n) \leqslant n \mathop / (\log_2 n + \log_2 \psi(n))$.
\end{theorem}
\begin{proof}
    Введём пары множеств $A_i$ и $B_i$ при $1 \leqslant i \leqslant n$ следующим образом:
    \begin{align*}
        A_1 &= \{\,1\,\}  \; , \\
        A_i &= A_{i-1} \cup B_{i-1}  \; , \\
        B_i &= \{\,\balpha^{2^{i-1}} \vv x\mid \vv x \in A_i\,\} \; .
    \end{align*}
    При таком построении видно, что $A_i \cap B_i = \varnothing$ при $i < n$.
    Следовательно, $|A_i| = 2^{i-1}$ для всех значений $1 \leqslant i \leqslant n$.

    Отметим, что множество $A_i$ можно трактовать, как множество всех элементов $\balpha ^{k_i}$, где $k_i \in \ZZ_{2^{i-1}}$.
    Множество $B_i$ содержит все элементы множества $A_i$, домноженные на $\balpha ^ {2^{i-1}}$.

    Введём \textit{характеристическую} функцию $\phi_{\mathbb X}\colon \fpm \to \ZZ_2$ на множестве $\mathbb X$ следующим образом:
    $$
        \phi_{\mathbb X}(\vv x) =
        \begin{cases}
            1, & \text{если } \vv x \in \mathbb X  \; , \\
            0  & \text{иначе} \; .
        \end{cases}
    $$

    Теперь мы можем описать принцип работы алгоритма дискретного логарифмирования в поле $\fp$, работающего за $n$ шагов, взяв
    за основу известный нам полиномиальный алгоритм возведения в степень.
    Обозначим этот алгоритм через $\mathbf A_\mathrm{log}$.

    Для входного значения $\vv x \in \fpm$ выполняем следующие действия:
    \begin{enumerate}
        \item Определить начальные значения: $i = n$, $\vv x_i = \vv x$.
        \item\label{discrete_log_algorithm_second_step}
            Вычислить значение функции $\phi_{A_i}(\vv x_i) = \phi_i$.
        \item\label{discrete_log_algorithm_third_step} Вычислить новое значение $\vv x_{i-1}$:
            $$
                \vv x_{i-1} =
                \begin{cases}
                    \balpha ^ {-2^{i-1}} \vv x_i, & \text{если } \phi_i = 1  \; , \\
                    \vv x_i & \text{иначе} \; .
                \end{cases}
            $$
        \item Уменьшить значение $i$ на 1.
        \item\label{discrete_log_algorithm_last_step}
            Если $i = 0$, закончить работу. Иначе перейти к шагу~\ref{discrete_log_algorithm_second_step}.
    \end{enumerate}
    В конце работы этого алгоритма степень $k = [\phi_n \phi_{n-1} \ldots \phi_2 \phi_1]$ (квадратные скобки означают запись в
    двоичной системе счисления) будет решением уравнения $\vv x = \balpha^k$.

    Осталось показать, что алгоритм $\mathbf A_\mathrm{log}$ можно реализовать с помощью обратимой схемы,
    состоящей из \gate{} множества $\Omega^2_*$, с указанными в условии теоремы сложностью и количеством дополнительных входов.

    Шаги \ref{discrete_log_algorithm_second_step}--\ref{discrete_log_algorithm_last_step} данного алгоритма
    будут повторены ровно $n$ раз.
    Сложность всей схемы $L(\frS_\mathrm{log})$ будет складываться из суммарной сложности
    \ref{discrete_log_algorithm_second_step}-го и \ref{discrete_log_algorithm_third_step}-го шагов алгоритма
    $L_2$ и $L_3$ соответственно:
    \begin{equation}
        L(\frS_\mathrm{log}) = L_2 + L_3 \; .
        \label{formula_gate_complexity_of_discrete_log_as_sum}
    \end{equation}
    Остальные шаги (4 и 5) реализуются без использования дополнительных \gate{} при помощи $n$ подсхем,
    реализующих шаги \ref{discrete_log_algorithm_second_step} и \ref{discrete_log_algorithm_third_step}.

    \Figure[ht]
        \centering
        \includegraphics[scale=1.2]{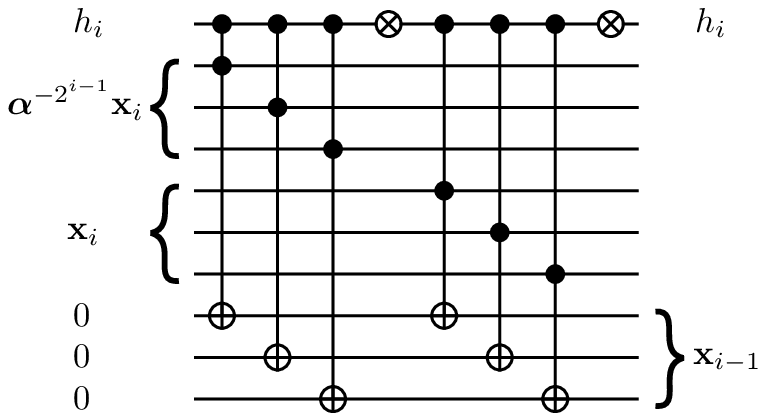}
        \caption
        {
            \small Реализация шага~\ref{discrete_log_algorithm_third_step} алгоритма $\mathbf A_\mathrm{log}$
            с помощью обратимых элементов NOT и 2-CNOT при $n = 3$. 
        }\label{pic_copy_with_control_input}
    \end{figure}

    Значения элементов $\balpha ^ {-2^{n-1}}, \ldots, \balpha ^ {-2}$ для \ref{discrete_log_algorithm_third_step}-го
    шага алгоритма можно вычислить заранее. Умножение $\balpha ^ {-2^{i-1}}$ на $\vv x_i$ можно реализовать при помощи $2n^2$
    \gate{} множества $\Omega^2_*$ при использовании $(2n-1)$ дополнительных входов~\cite{zakablukov_zasorina_chikin}.
    Выбор значения $\vv x_{i-1}$ на этом шаге можно реализовать при помощи $(2n + 1)$ элементов
    NOT и CNOT при использовании $n$ дополнительных входов.
    К примеру, на рис.~\ref{pic_copy_with_control_input} показана обратимая подсхема, реализующая такой выбор при $n = 3$.
    Следовательно, суммарная сложность шага~\ref{discrete_log_algorithm_third_step} равна
    $$
        L_3 = n(2n^2 + 2n + 1) \sim 2n^3 \; ,
    $$
    а суммарное количество использованных дополнительных входов на этом шаге равно
    $$
        q_3 = n(2n-1 + n) \sim 3n^2 \; .
    $$

    Сложность \ref{discrete_log_algorithm_second_step}-го шага алгоритма $L_2$ зависит от сложности
    обратимой схемы, вычисляющей значение $\phi_i$. Обозначим эту сложность через $L_{\phi_i}$.
    Тогда верно равенство
    $$
        L_2 = \sum_{i=1}^n{L_{\phi_i}} \; .
    $$

    Каждой из функций $\phi_{A_i}\colon \fpm \to \ZZ_2$ можно поставить в соответствие булеву функцию
    $f_i\colon \ZZ_2^n \to \ZZ_2$, такую что $\phi_{A_i}(\vv x) = f_i(\vv v)$, $\vv v \in \ZZ_2^n$~--- вектор
    коэффициентов многочлена $\vv x \in \fpm$, $f_\FF(\vv v) = \vv x$.

    Согласно теореме~\ref{theorem_arbitrary_boolean_transformation_complexity}, булеву функцию $f_i$ можно реализовать
    с помощью обратимой подсхемы, состоящей из \gate{} множества $\Omega^2_*$,
    со сложностью $L \lesssim 2^n \mathop / n$ при использовании $q \sim 2^{n-\lceil n \mathop / \phi(n)\rceil + 1}$
    дополнительных входов (см. начало данного параграфа),
    где $\phi(n)$ и $\psi(n)$~--- любые сколь угодно медленно растущие функции такие,
    что $\phi(n) \leqslant n \mathop / (\log_2 n + \log_2 \psi(n))$.
    Таким образом,
    \begin{gather}
        L_{\phi_i} \lesssim 2^n \mathop / n
            \label{formula_function_complexity_upper_bound_table_case}  \; , \\
        q_{\phi_i} \sim 2^{n-\lceil n \mathop / \phi(n)\rceil + 1} \notag \; .
    \end{gather}

    С другой стороны, функцию $f_i$ можно реализовать, используя аналог СДНФ, в котором дизъюнкции заменены на
    сложение по модулю 2:
    $$
        f_i(\langle v_1, \ldots, v_n \rangle) = \bigoplus_{\substack{a_1, \ldots, a_n \in \ZZ_2\\
            f_i(\langle a_1, \ldots, a_n \rangle) = 1}}{ v_1^{a_1} \wedge \ldots \wedge v_n^{a_n}} \; .
    $$
    Каждая из конъюнкций $v_1^{a_1} \wedge \ldots \wedge v_n^{a_n}$ реализуется одним элементом $E(t,I,J)$.
    Этот \gate{} может быть выражен в виде композиции не более $2n$ элементов NOT для инверсии контролирующих входов
    и не более $(2n-3)$ элементов 2-CNOT для замены одного элемента $n$-CNOT (см. рис.~\ref{pic_reducing_complexity}).
    При этом используется всего $(n - 2)$ дополнительных входов.
    Таким образом, при реализации функции $f_i$ верны следующие соотношения:
    \begin{gather*}
        L_{\phi_i} \leqslant (4n-3)|f_i| \; , \\
        q_{\phi_i} = n-2 \; ,
    \end{gather*}
    где $|f_i|$~--- вес вектора значений функции $f_i$.
    На рис.~\ref{pic_dnf_realization_example} приведён пример реализации с помощью обратимой схемы функции
    $g(\langle v_1, v_2, v_3 \rangle) = v_1^0 \wedge v_2 \wedge v_3^0 \oplus v_1 \wedge v_2^0 \wedge v_3$.

    \Figure[ht]
        \centering
        \includegraphics[scale=1.2]{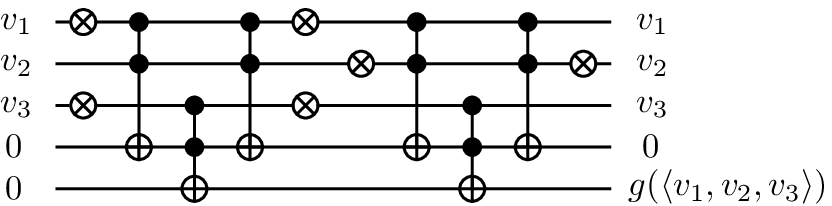}
        \caption
        {
            \small Реализация функции
            $g(\langle v_1, v_2, v_3 \rangle) = v_1^0 \wedge v_2 \wedge v_3^0 \oplus v_1 \wedge v_2^0 \wedge v_3$
            с помощью обратимых элементов NOT и 2-CNOT. 
        }\label{pic_dnf_realization_example}
    \end{figure}

    В начале доказательства было показано, что $|A_i| = 2^{i-1}$ для всех значений $i$ от $1$ до $n$,
    следовательно, вес вектора значений функции $f_i$ равен $|f_i| = 2^{i-1}$.
    Отсюда следует, что величины $L_{\phi_i}$ и $q_{\phi_i}$ ограничены следующим образом:
    \begin{gather}
        L_{\phi_i} \leqslant (4n-3)2^{i-1} \; ,
            \label{formula_function_complexity_upper_bound_function_weight_case} \\
        q_{\phi_i} = n-2 \notag \; .
    \end{gather}

    Оценим, при каких значениях $i$ величина $L_{\phi_i}$
    по формуле~\eqref{formula_function_complexity_upper_bound_function_weight_case}
    будет не больше, чем по формуле~\eqref{formula_function_complexity_upper_bound_table_case}:
    $$
        (4n-3)2^{i-1} \leqslant 2^n \mathop / n \Rightarrow i \leqslant n - 2 \log_2 n \; .
    $$
    Теперь можно оценить величину $L_2$. Пусть $k = n - 2 \log_2 n$, тогда
    \begin{gather*}
        L_2 = \sum_{i=1}^n{L_{\phi_i}} \lesssim
            \sum_{i=1}^k{(4n-3)2^{i-1}} + 
            \sum_{i= k + 1}^n{\frac{2^n}{n}}  \; , \\
        L_2 \lesssim 4n 2^k + (n-k) \frac{2^n}{n} =
            \frac{2^{n+2}}{n} + \frac{2^{n+1} \cdot \log_2 n}{n} \lesssim \frac{2^{n+1} \cdot \log_2 n}{n} \; .
    \end{gather*}
    Аналогичным образом можно оценить количество использованных дополнительных входов:
    \begin{gather*}
        q_2 = \sum_{i=1}^n{q_{\phi_i}} \sim
            \sum_{i=1}^k{(n-2)} + 
            \sum_{i= k + 1}^n{2^{n-\lceil n \mathop / \phi(n)\rceil + 1}}  \; , \\
        q_2 \sim k(n-2) + (n-k) 2^{n-\lceil n \mathop / \phi(n)\rceil + 1}
            \sim 2^{n-\lceil n \mathop / \phi(n)\rceil + 2} \cdot \log_2 n \; .
    \end{gather*}
    
    Таким образом, мы получаем следующие оценки:
    \begin{gather*}
        L(\frS_\mathrm{log}) = L_2 + L_3 \lesssim \frac{2^{n+1} \cdot \log_2 n}{n} + 2n^3
            \sim \frac{2^{n+1} \cdot \log_2 n}{n}  \; , \\
        Q(\frS_\mathrm{log}) = q_2 + q_3 \sim 2^{n-\lceil n \mathop / \phi(n)\rceil + 2} \cdot \log_2 n + 3n^2 \sim
            2^{n-\lceil n \mathop / \phi(n)\rceil + 2} \cdot \log_2 n \; .
    \end{gather*}
\end{proof}

Данная верхняя асимптотическая оценка сложности $(2^{n+1} \cdot \log_2 n) \mathop / n$ намного выше,
чем самая лучшая асимптотическая оценка временн\'{о}й сложности $O(n^{\log_2 n})$ алгоритма дискретного логарифмирования
в поле $\fp$, известная на данный момент~\cite{quasi_polynomial_discrete_log}.
При этом не представляется возможным реализовать с такой же или с любой другой субъэкспоненциальной
сложностью алгоритм из работы~\cite{quasi_polynomial_discrete_log},
т.\,к. в указанном алгоритме используются эвристические и вероятностные подходы, сложно реализуемые в схеме из обратимых \gate{}

\subsection{Вопрос схемной сложности реализации алгоритма, обратного к заданному}

\forceindent
В работе~\cite{zakablukov_zasorina_chikin} был описан способ построения обратимой схемы,
состоящей из \gate{} множества $\Omega^2_*$ и
реализующей двоичный сумматор без порождения вычислительного мусора на своих незначимых выходах. Такая схема позволила авторам
оценить временн\'{ы}е сложности алгоритмов сложения и вычитания, которые оказались линейными от $n$,
через оценку сложности обратимой схемы.

В работе~\cite{zakablukov_zasorina_chikin} также был описан подход к построению обратимой схемы,
состоящей из \gate{} множества $\Omega^2_*$
и реализующей умножение многочленов в поле $\fp$ без порождения вычислительного мусора на своих незначимых выходах.
Как и в предыдущем случае, такая схема позволила оценить временн\'{ы}е сложности алгоритмов умножения и деления многочленов
в поле $\fp$, которые оказались полиномиальными от $n$, через оценку сложности обратимой схемы.

В то же время мы имеем алгоритм возведения в степень примитивного элемента $\balpha$ в поле $\fp$ с полиномиальной
временн\'{о}й сложностью, который также можно реализовать с помощью обратимой схемы, использующей дополнительную память и 
состоящей из \gate{} множества $\Omega_*^2$, с полиномиальной сложностью.
Однако для алгоритма дискретного логарифмирования, находящего для заданного
$\vv x \in \fpm$ степень $k$, такую что $\balpha ^ k = \vv x$ в поле $\fp$, неизвестно о существовании алгоритма
с полиномиальной временн\'{о}й сложностью~\cite{quasi_polynomial_discrete_log}
или о существовании реализующей его обратимой схемы с полиномиальной сложностью.

Возникает вопрос: почему в одних случаях прямой и обратный алгоритм имеют при реализации обратимой схемой
сложности с одинаковой степенью роста (линейной или полиномиальной), и это можно доказать, а в других случаях не удаётся
реализовать обратимой схемой алгоритм, обратный к заданному, со сложностью, имеющей ту же степень роста?

\begin{hypothesis}\label{hypothesis_reverse_algorithm_complexity}
    Обратимая схема, реализующая алгоритм, обратный к заданному, имеет сложность с б\'{о}льшей на порядок степенью роста
    по отношению к сложности обратимой схемы, реализующей прямой алгоритм, если при переходе от прямого алгоритма к обратному
    теряется какая-то часть информации.
\end{hypothesis}

Рассмотрим на примерах, на чём основана данная гипотеза:
\begin{enumerate}
    \item\label{enum_addition_gate_compexity}
        \textbf{Сложение}.
    
        Пусть необходимо сложить два элемента $a$ и $b$ какого-либо кольца $K$.
        Алгоритм сложения можно описать отображением $f_\mathrm{sum}\colon K \times K \to K \times K$:
        $$
            f_\mathrm{sum}(a, b) = (a, b + a) \; .
        $$
        Тогда обратное к нему отображение $f^{-1}_\mathrm{sum}\colon K \times K \to K \times K$ будет ничем иным,
        как отображением $f_\mathrm{sub}$, описывающим алгоритм вычитания:
        $$
            f^{-1}_\mathrm{sum}(a, b + a) = f_\mathrm{sub}(a, b + a) = (a, (b + a) - a) = (a, b) \; .
        $$
        
        Как видно, прямое и обратное преобразования $f_\mathrm{sum}$ и $f_\mathrm{sub}$ являются, во-первых, обратимыми,
        во-вторых, множество входных и выходных значений у них совпадают,
        а в-третьих, обратное преобразование есть по сути то же сложение, только с обратным относительно операции сложения
        элементом: $c - a = c + (-a)$.
        
    \item \textbf{Умножение}.
    
        Пусть необходимо умножить два ненулевых элемента $a$ и $b$ какой-либо группы $H$.
        Алгоритм умножения можно описать отображением $f_\mathrm{mul}\colon H \times H \to H \times H$:
        $$
            f_\mathrm{mul}(a, b) = (a, b * a)\text{ при }a \ne 0, b \ne 0 \; .
        $$
        Тогда обратное к нему отображение $f^{-1}_\mathrm{mul}\colon H \times H \to H \times H$ будет ничем иным,
        как отображением $f_\mathrm{div}$, описывающим алгоритм деления:
        $$
            f^{-1}_\mathrm{mul}(a, b * a) = f_\mathrm{div}(a, b * a) = (a, (b * a) \mathop / a) = (a, b) \; .
        $$
        
        Как и в предыдущем случае, прямое и обратное преобразования $f_\mathrm{mul}$ и $f_\mathrm{div}$ являются,
        во-первых, обратимыми,
        во-вторых, множество входных и выходных значений у них совпадают,
        а в-третьих, обратное преобразование есть по сути то же умножение, только с обратным относительно операции умножения
        элементом: $c \mathop / a = c * a^{-1}$.

    \item \textbf{Возведение в степень}.

        Пусть необходимо возвести в степень $n$ ненулевой элемент $a$ какой-либо группы $H$.
        Алгоритм возведения в степень можно описать отображением $f_\mathrm{pow}\colon \NN \times H \to \NN \times H$:
        $$
            f_\mathrm{pow}(n, a) = (n, 1 * \underbrace{a * \ldots * a}_n)\text{ при }a \ne 0 \; .
        $$
        Однако алгоритм дискретного логарифмирования, который можно описать отображением
        $f_\mathrm{log}\colon H \times H \to \NN \times H$, не является обратным к данному:
        $$
            f_\mathrm{log}(b, a) = (n, a)\text{, где } f_\mathrm{pow}(n, a) = (n, b) \; .
        $$
        Это следует из того, что множество входных и выходных значений у отображений
        $f_\mathrm{pow}$ и $f_\mathrm{log}$ не совпадают.
        
        Также стоит отметить, что для конечной группы $H$ верно равенство
        $$
            f_\mathrm{pow}(n, a) = f_\mathrm{pow}(k \cdot |H| + n, a), \text{ где } k \in \NN \; .
        $$
        В то же время для пары ненулевых элементов $a, b \in H$ отображение $f_\mathrm{log}$ даёт всего лишь одну
        степень $n$, для которой $f_\mathrm{pow}(n, a) = (n, b)$, а не целое множество степеней вида $n + k \cdot |H|$.
        
        Однако если рассмотреть алгоритм извлечения корня $n$-й степени, описываемый отображением
        $f_\mathrm{root}\colon \NN \times H \to \NN \times H$ следующего вида:
        $$
            f_\mathrm{root}(n, a^n) = (n, \sqrt[n]{a^n}) = (n, a)\text{ при }a \ne 0 \; ,
        $$
        то он как раз и будет обратным к алгоритму возведения в степень: $f^{-1}_\mathrm{pow} = f_\mathrm{root}$.
        Эти два преобразования являются,
        во-первых, обратимыми,
        во-вторых, множество входных и выходных значений у них совпадают,
        а в-третьих, обратное преобразование есть по сути то же возведение в степень,
        только с обратным относительно операции возведения в степень элементом: $\sqrt[n]{b} = b ^ {n^{-1}}$.

\end{enumerate}

На с.~\pageref{formula_discrete_logarithm_base_function} мы задали преобразование $f_\mathrm{log}\colon \ZZ_2^n \to \ZZ_2^n$
и ввели два отображения: одно для отображения степени $n$ в двоичный вектор $\NN \to \ZZ_2^n$ и для отображения
элемента поля $\vv x \in \fp$ в двоичный вектор $\fpm \to \ZZ_2^n$.
Именно на данном этапе при переходе от прямого алгоритма (алгоритма возведения в степень) к обратному
(алгоритму дискретного логарифмирования) произошла потеря информации о том, что для любой степени $n \in \NN$
существует обратимая схема, реализующая алгоритм возведения в степень со сложностью не более $4 (\log_2 n)^3$.
Рассматриваемое отображение $f_\mathrm{log}\colon \ZZ_2^n \to \ZZ_2^n$ было ограничено только степенями
$n \in \ZZ_{2^n - 1}$.
Тем самым синтез обратимой схемы, реализующей это отображение, перестал отличаться от синтеза обратимой схемы, реализующей
произвольное отображение $\ZZ_2^n \to \ZZ_2^n$.

Доказательство верхней асимптотической оценки $(2^{n+1} \cdot \log_2 n) \mathop / n$
в теореме~\ref{theorem_upper_bound_discrete_log}
ясно показывает на примере, что не удаётся построить обратимую схему, реализующую алгоритм, обратный к алгоритму
возведения в степень, используя обратимую схему для прямого алгоритма,
без увеличения сложности обратимой схемы на несколько порядков.
Прямой алгоритм принимает значения степени из множества $\NN$, а обратный алгоритм выдаёт в качестве ответа степень из
ограниченного множества $\ZZ_{2^n}$.
Выдвинутая гипотеза~\ref{hypothesis_reverse_algorithm_complexity} заключается в том,
что такое увеличение сложности является следствием потери информации о прямом алгоритме~---
следствием перехода из $\NN$ в $\ZZ_{2^n}$.

С другой стороны, для алгоритма извлечения корня $n$-й степени такого ограничения множества входных/выходных значений
не происходит, информация не теряется, поэтому данный алгоритм, как и алгоритм возведения в степень, тоже должен
реализовываться обратимой схемой с полиномиальной сложностью.
По крайней мере, это выполняется для временн\'{о}й сложности данных алгоритмов:
для степени корня $n$ достаточно найти такую степень $k$,
что $nk \equiv 1\pmod{|H|}$, где $|H|$~--- порядок группы. Другими словами, $k \equiv n^{-1}\pmod{|H|}$.
К примеру, данная степень $k$ может быть найдена при помощи расширенного алгоритма Евклида,
имеющего полиномиальную сложность.
В итоге, при извлечении корня $n$-й степени из элемента $b$ сначала ищется степень $k \equiv n^{-1}\pmod{|H|}$ за полиномиальное
время, а затем $b$ возводится в степень $k$ также за полиномиальное время.

Рассмотрим общую схему для отображения $f\colon H \times G \to H \times G$, где $H$ и $G$~--- некоторые группы.
Пусть данное отображение при помощи опреации $*$ переводит пару $(h, g)$ в некоторую новую пару $(h, g')$ следующим образом:
$$
    f(h, g) = (h, g * h), \text{ } g' = g * h \; .
$$
Тогда обратное к данному преобразование $f^{-1}$ будет выглядеть следующим образом:
$$
    f^{-1}(h, g') = (h, g' * f_\mathrm{inv}(h)) \text{ \;и\; }
    f^{-1}(f(h, g)) = (h, g) \; ,
$$
где $f_\mathrm{inv}(h)$~--- функция обращения элемента $h \in H$ относительно операции $*$.

Видно, что сложность $L(\frS_{f^{-1}})$ обратимой схемы, реализующей отображение $f^{-1}$,
будет больше сложности $L(\frS_f)$ обратимой схемы, реализующей отображение $f$,
на величину $L(\frS_{f_\mathrm{inv}})$, равную сложности обратимой схемы, реализующей отображение $f_\mathrm{inv}$.

Если $L(\frS_{f_\mathrm{inv}})$ и $L(\frS_f)$ имеют одинаковый порядок роста, то и $L(\frS_{f^{-1}})$ и $L(\frS_f)$
будут иметь одинаковый порядок роста (линейный, полиномиальный, экспоненциальный).
Примерами таких обратимых схем могут быть схемы, реализующие алгоритмы сложения/вычитания
и умножения/деления элементов в группе (см. с.~\pageref{enum_addition_gate_compexity}),
а также схемы, реализующие алгоритм возведения в степень $k$
и извлечения корня $k$-й степени из элементов группы, $k \in \ZZ_{2^n}$.
Если же будет доказано, что $L(\frS_f) = o(L(\frS_{f_\mathrm{inv}}))$, то тогда $L(\frS_f) = o(L(\frS_{f^{-1}}))$.

\sectionnoenumeration{Заключение}

\forceindent
В связи с тепловыми потерями во время вычислительного процесса, вызванными необратимостью производимых операций,
обратимость вычислений, по-видимому, станет в ближайшем будущем одним из главных требований, предъявляемых к синтезу
управляющих систем. С другой стороны, обратимость является неотъемлемой частью, к примеру, квантовых вычислений.
Следовательно, уже сейчас необходимо развивать и улучшать существующий математический базис для синтеза
компактных обратимых схем с малым числом входов.

В данной диссертационной работе были рассмотрены существующие переборные и непереборные алгоритмы синтеза обратимых схем,
состоящих из элементов NOT, CNOT и 2-CNOT. Предложен новый быстрый и эффективный алгоритм синтеза обратимой схемы,
задающей подстановку на множестве $\ZZ_2^n$ с малым числом подвижных точек. Предложены различные эквивалентные
замены композиций обратимых \gate{}, описан алгоритм снижения сложности обратимой схемы, использующий данные замены.
Также были рассмотрены различные методы снижения сложности обратимой схемы на этапе её синтеза и показана
эффективность этих методов на практике.
Доказаны асимптотические верхние и нижние оценки сложности, глубины и квантового веса обратимых схем,
состоящих из элементов NOT, CNOT и 2-CNOT.
Показано, что использование дополнительных входов в таких схемах почти всегда позволяет существенно снизить
их сложность, глубину и квантовый вес, в отличие от классических схем, состоящих из необратимых \gate{}
Снижение сложности обратимой схемы за счёт использования дополнительных входов было также показано на примере
реализации алгоритма дискретного логарифмирования в конечном поле характеристики 2.

Однако в виду ограниченного времени не были рассмотрены некоторые вопросы, представляющие научный интерес для дальнейших
исследований.
Одним из таких вопросов является улучшение нижней оценки глубины $D(n,q)$ обратимой схемы,
состоящей из элементов NOT, CNOT и 2-CNOT. В настоящей работе эта оценка была получена из нижней оценки сложности
обратимой схемы $\frS$ при помощи простого соотношения $D(\frS) \geqslant L(\frS) \mathop / n$, где $n$~--- количество входов
схемы. Тем не менее, можно попытаться посчитать количество неэквивалентных обратимых схем заданной глубины.
Если это будет сделано, то можно будет применить мощностной метод Риордана--Шеннона и, предположительно, улучшить
нижнюю оценку для $D(n,q)$.

Вторым направлением дальнейших исследований является поиск асимптотически оптимального метода синтеза обратимых схем без
дополнительной памяти. В доказанной в настоящей работе верхней оценке $L(n,0)$ участвует константа $3 \cdot 2^4$.
Её можно попытаться снизить либо модифицировав предложенный алгоритм синтеза, либо разработав новый алгоритм.
Также можно попытаться посчитать среднее значение сложности для всех обратимых схем с $n$ входами.

В настоящей работе почти не был затронут вопрос синтеза обратимых схем без дополнительной памяти с асимптотически
оптимальной глубиной. Если для сложности обратимых схем удалось получить нижние и верхние оценки, эквивалентные с точностью
до порядка, то для глубины обратимых схем таких оценок получить не удалось. Отчасти это связано с тем,
что предложенный алгоритм синтеза обратимых схем без дополнительной памяти плохо поддаётся модификации для снижения глубины
синтезированной схемы.

Те же соображения касаются и обратимых схем с дополнительной памятью: для них также не было получено эквивалентных
с точностью до порядка нижних и верхних оценок глубины (в отличие от оценок сложности). Дело осложняется ещё и тем,
что нижняя асимптотическая оценка глубины $D(n,q)$ обратимых схем является очень слабой при быстро растущем значении
количества дополнительных входов $q$. Поэтому предлагается сперва улучшить нижнюю оценку глубины обратимых схем, а затем
попытаться улучшить верхнюю оценку глубины.

Ещё одним важным направлением дальнейших исследований является изучение асимметричных преобразований
через построение реализующих их обратимых схем. Если обратимая схема не имеет дополнительных входов,
то очевидно, что обратное преобразование можно реализовать обратимой схемой с той же сложностью: для этого просто надо
зеркально отобразить существующую схему. Если обратимая схема без дополнительной памяти имеет минимальную сложность среди
всех обратимых схем, реализующих прямое преобразование, можно утверждать, что и обратное преобразование реализуется с той
же схемной сложностью.
Однако если схема имеет дополнительные входы и содержит вычислительный мусор на своих незначимых выходах, то в этом случае уже
нельзя так однозначно утверждать о равенстве схемной сложности для обратного преобразования. Для такой схемы можно построить
подсхему по уборке вычислительного мусора. Тогда разница в схемной сложности прямого и обратного преобразований будет
зависеть от сложности подсхем, <<генерирующей>> и <<убирающей>> вычислительный мусор.

В настоящей работе была сделана попытка построить схему по уборке вычислительного мусора для обратимой схемы, реализующей
алгоритм возведения в степень в конечном поле характеристики 2. К сожалению, ситуация осложнилась тем, что вычислительный мусор
оказался равным входному значению показателя степени, т.\,е. по сути не было обратимой подсхемы,
<<генерирующий>> вычислительный мусор. Как следствие, подсхема по уборке вычислительного мусора оказалась
равной обратимой схеме, реализующей алгоритм дискретного логарифмирования в конечном поле характеристики 2.

Тем не менее, несмотря на описанные сложности, объединение в одной обратимой схеме реализаций как прямого, так и обратного
преобразований, представляет, по мнению автора, значительный научный интерес и открывает новые возможности для
изучения сложности асимметричных преобразований.

\renewcommand{\bibname}{Список литературы}

\refstepcounter{fakecounter}\label{page_last}

\end{document}